\documentclass[reqno,centertags, 12pt]{amsart}
\usepackage{amsmath,amsthm,amscd,amssymb,siunitx}
\usepackage{latexsym}
\usepackage{graphicx,wrapfig}
\usepackage{enumitem}
\usepackage{hyperref}
\usepackage[mathscr]{eucal}
\usepackage{tikz}
\usepackage{url}
\newcommand{\bbC}{{\mathbb{C}}}
\newcommand{\bbD}{{\mathbb{D}}}

\newcommand{\bbR}{{\mathbb{R}}}

\newcommand{\bbZ}{{\mathbb{Z}}}

\newcommand{\frA}{{\frak{A}}}
\newcommand{\x}{{\mathbf{x}}}

\newcommand{\calB}{{\mathcal{B}}}
\newcommand{\calC}{{\mathcal{C}}}
\newcommand{\calD}{{\mathcal{D}}}
\newcommand{\calE}{{\mathcal{E}}}

\newcommand{\calH}{{\mathcal H}}
\newcommand{\calI}{{\mathcal I}}

\newcommand{\calK}{{\mathcal K}}
\newcommand{\calL}{{\mathcal L}}
\newcommand{\calM}{{\mathcal M}}
\newcommand{\calP}{{\mathcal P}}
\newcommand{\calR}{{\mathcal R}}
\newcommand{\calS}{{\mathcal S}}
\newcommand{\calT}{{\mathcal T}}

\newcommand{\calV}{{\mathcal V}}
\newcommand{\bdone}{{\boldsymbol{1}}}

\newcommand{\lb}{\label}

\newcommand{\wti}{\widetilde  }

\newcommand{\tr}{\text{\rm{Tr}}}

\newcommand{\dist}{\text{\rm{dist}}}

\newcommand{\loc}{\text{\rm{loc}}}

\newcommand{\ran}{\text{\rm{ran}}}

\newcommand{\supp}{\text{\rm{supp}}}

\newcommand{\bi}{\bibitem}
\newcommand{\hatt}{\widehat}
\newcommand{\beq}{\begin{equation}}
\newcommand{\eeq}{\end{equation}}
\newcommand{\ba}{\begin{align}}
\newcommand{\ea}{\end{align}}

\let\det=\undefined\DeclareMathOperator{\det}{det}

\newcounter{smalllist}

\newcommand{\comm}[1]{}
\newcommand*\circled[1]{\tikz[baseline=(char.base)]{
            \node[shape=circle,draw,inner sep=2pt] (char) {#1};}}
\DeclareMathOperator{\Real}{Re}

\allowdisplaybreaks
\numberwithin{equation}{section}
\newtheorem{theorem}{Theorem}[section]

\newtheorem*{p2.1}{Proposition 2.1}
\newtheorem{proposition}[theorem]{Proposition}
\newtheorem{lemma}[theorem]{Lemma}
\newtheorem{corollary}[theorem]{Corollary}
\theoremstyle{definition}
\newtheorem{example}[theorem]{Example}
\newtheorem{conjecture}[theorem]{Conjecture}

\newtheorem*{remark}{Remark}
\newtheorem*{remarks}{Remarks}
\newtheorem*{definition}{Definition}

\newcommand{\jap}[1]{\langle #1 \rangle}
\newcommand{\norm}[1]{\lVert#1\rVert}

\AtBeginDocument{%
  \LetLtxMacro{\TheRealLabel}{\label}%
  \LetLtxMacro{\TheRealRef}{\ref}%
  \LetLtxMacro{\TheRealPageRef}{\pageref}%
}

\begin{document}

\title[Kato's Work]{Tosio Kato's Work on Non--Relativistic Quantum Mechanics}
\author[B.~Simon]{Barry Simon$^{1,2}$}

\thanks{$^1$ Departments of Mathematics and Physics, Mathematics 253-37, California Institute of Technology, Pasadena, CA 91125.
E-mail: bsimon@caltech.edu}

\thanks{$^2$ Research supported in part by NSF grants DMS-1265592 and DMS-1665526 and in part by Israeli BSF Grant No. 2014337.}

\

\date{\today}

\begin{abstract}  We review the work of Tosio Kato on the mathematics of non--relativistic quantum mechanics and some of the research that was motivated by this.  Topics include analytic and asymptotic eigenvalue perturbation theory, Temple--Kato inequality, self--adjointness results, quadratic forms including monotone convergence theorems, absence of embedded eigenvalues, trace class scattering, Kato smoothness, the quantum adiabatic theorem and Kato's ultimate Trotter Product Formula.
\end{abstract}

\maketitle

\tableofcontents

%%%%%%%%%%%%%%%%%%%%%%%%%%%%%%%%%%%%%%%%%%%%%%%%%%%%%%%%%%%%%%
\section{Introduction} \lb{s1}
%%%%%%%%%%%%%%%%%%%%%%%%%%%%%%%%%%%%%%%%%%%%%%%%%%%%%%%%%%%%%%

Note: There are pictures on pages \pageref{figure1}, \pageref{figure2}, \pageref{figure3}, \pageref{figure4} and \pageref{figure5}.

In 2017, we are celebrating the $100^{th}$ anniversary of the birth of Tosio Kato (August 25, 1917--October 2, 1999).  While there can be arguments as to which of his work is the deepest or most beautiful, there is no question that the most significant is his discovery, published in 1951, of the self--adjointness of the quantum mechanical Hamiltonian for atoms and molecules \cite{KatoHisThm}.  This is the founding document and Kato is the founding father of what has come to be called the theory of Schr\"{o}dinger operators.  So it seems appropriate to commemorate Kato with a comprehensive review of his work on non--relativistic quantum mechanics (NRQM) that includes the context and later impact of this work.

One might wonder why I date this field only from Kato's 1951 paper.  After all, quantum theory was invented in 1925-26 as matrix mechanics in G\"{o}ttingen (by Heisenberg, Born and Jordan) and as wave mechanics in Z\"{u}rich (by Schr\"{o}dinger) and within a few years, books appeared on the mathematical foundations of quantum mechanics by two of the greatest mathematicians of their generation: Hermann Weyl \cite{WeylQM} (not coincidentally, in Z\"{u}rich; indeed the connection between Weyl and Schr\"{o}dinger was more than professional -- Weyl had a passionate love affair with Schr\"{o}dinger's wife) and John von Neumann \cite{vNQM} (von Neumann, whose thesis had been in logic, went to G\"{o}ttingen to work with Hilbert on that subject, but was swept up in the local enthusiasm for quantum theory, in response to which, he developed the spectral theory of unbounded self--adjoint operators and his foundational work).  One should also mention the work of Bargmann and Wigner (prior to Kato, summarized in \cite{SimonBW} with references) on quantum dynamics.  I think of this earlier work as first level foundations and the theory of Schr\"{o}dinger operators as second level.  Another way of explaining the distinction is that the Weyl--von Neumann work is an analog of setting up a formalism for classical mechanics like the Hamiltonian or Lagrangian while the theory initiated by Kato is the analog of celestial mechanics -- the application of the general framework to concrete systems.

\begin{wrapfigure}{l}{0.5\textwidth}
   \centering
   \vspace{-0.3cm}
   \includegraphics[width=0.5\textwidth]{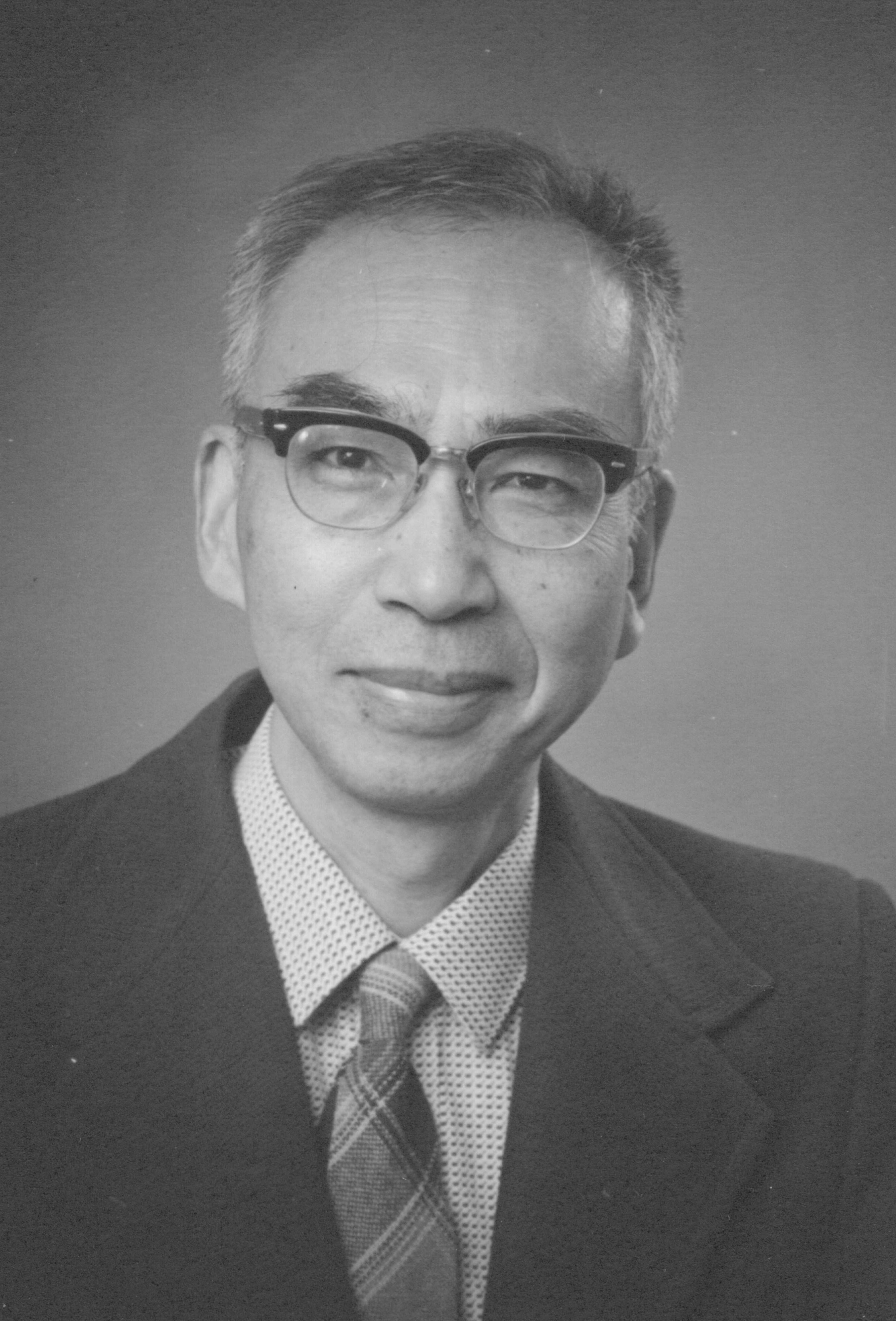}
Kato at Berkeley \lb{figure1}
   \vspace{-0.2cm}
\end{wrapfigure}

When I began this project I decided to write about all of Kato's major contributions to the field in a larger context and this turned into a much larger article than I originally planned.  As such, it is a review of a significant fraction of the work of the last 65 years on the mathematics of NRQM.  Two important areas only touched on or totally missing are N--body systems and the large N limit.  Of course, Kato's self--adjointness work includes N--body systems, and there are papers on bound states in Helium and on properties of many body eigenfunctions.  As we'll see, his theory of smooth perturbations applies to give a complete spectral analysis of certain N--body systems with only one scattering channel and is one tool in the study of general $N$--body systems.  But there is much more to the N--body theory -- for reviews, see \cite{CFKS, DerGer, GrafShen, GusSig, HunzSig}. Except for the 1972 work of Lieb--Simon on Thomas Fermi almost all the large N limit work is after 1980 when Kato mostly left the field; for recent reviews of different aspects of this subfield, see \cite{BPS16, Lewin, Lieb81, LiebBk10, LiebBk05, Roug15, Seir13}.

While this review will cover a huge array of work, it is important to realize it is only a fraction, albeit a substantial fraction, of Kato's opus.  I'd classify his work into four broad areas, NRQM, non--linear PDE's, linear semigroup theory and miscellaneous contributions to functional analysis.  We will not give references to all this work.  The reader can get an  (almost) complete bibliography from MathSciNet or, for papers up to 1987, the dedication of the special issue of JMAA on the occasion of Kato's $70^{th}$ birthday \cite{DHBio} has a bibliography.

Around 1980, one can detect a clear shift in Kato's interest.  Before 1980, the bulk of his papers are on NRQM with a sprinkling in the other three areas while after 1980, the bulk are on nonlinear equations with a sprinkling in the other areas including NRQM.  Kato's nonlinear work includes looking at the Euler, Navier--Stokes, KdV and nonlinear Schr\"{o}dinger equations.  He was a pioneer in existence results -- we note that his famous 1951 paper can be viewed as a result on existence of solutions for the time dependent linear Schr\"{o}dinger equation!  It is almost that when NRQM became too crowded with workers drawn by his work, he moved to a new area which took some time to become popular. Terry Tao said of this work: \emph{the Kato smoothing effect for Schr\"{o}dinger equations is fundamental to the modern theory of nonlinear Schr\"{o}dinger equations, perhaps second only to the Strichartz estimates in importance...Kato developed a beautiful abstract (functional analytic) theory for local well posedness for evolution equations; it is not used directly too much these days because it often requires quite a bit more regularity than we would like, but I think it was influential in inspiring more modern approaches to local existence based on more sophisticated function space estimates.}

And here is what Carlos Kenig told me: \emph{T.Kato played a pioneering role in the study of nonlinear evolution equations. He not only developed an abstract framework for their study, but also introduced the tools to study many fundamental nonlinear evolutions coming from mathematical physics. Some remarkable examples of this are: Kato's introduction of the ``local smoothing effect'' in his pioneering study of the Korteweg-de Vries equation, which has played a key role in the development of the theory of nonlinear dispersive equations.}

{\emph{Kato's unified proof of the global well-posedness of the Euler and Navier-Stokes equations in 2d, which led to the development of the Beale-Kato-Majda blow-up criterion for these equations. Kato's works with Ponce on strong solutions of the Euler and Navier-Stokes equations, which developed the tools for the systematic application of fractional derivatives in the study of evolutions, which now completely permeates the subject. These contributions and many others, have left an indelible and enduring impact for the work of Kato on nonlinear evolutions.}

The basic results on generators of semigroups on Banach spaces date back to the early 1950s  going under the name Feller-Miyadera-Phillips and Hille-Yosida theorems (with a later 1961 paper of Lumer--Phillips).  A basic book with references to this work is Pazy \cite{Pazy}.  This is a subject that Kato returned to often, especially in the 1960s.  Pazy \cite{Pazy} lists 19 papers by Kato on the subject.  There is overlap with the NRQM work and the semigroup work.  Perhaps the most important of these results are the Trotter--Kato theorems (discussed below briefly after Theorem \ref{T3.7}) and the definition of fractional powers for generators of (not necessarily self--adjoint) semigroups.  There are also connections between quantum statistical mechanics and contraction semigroup on operator algebras.  To keep this review within bounds, we will not discuss this work.

The fourth area is a catchall for a variety of results that don't fit into the other bins.  Among these results is an improvement of the celebrated Calder\'{o}n-Vaillancourt bounds on pseudo-differential operators \cite{KatoPsDO}.  In \cite{KatoAbsVal}, Kato proved the absolute value for operators is not Lipschitz continuous even restricted to the self--adjoint operators but for any pair of bounded, even non--self--adjoint, operators one has that
\begin{equation}\label{1.1}
  \norm{|S|-|T|} \le \frac{2}{\pi} \norm{S-T}\left(2+\log\frac{\norm{S}+\norm{T}}{\norm{S-T}}\right)
\end{equation}
(I don't think there is any significance to the fact that the constant is the same as in \eqref{10.22}).

The last of these miscellaneous things that we'll discuss (but far from the last of the miscellaneous results) involves what has come to be called the Heinz--Loewner inequality.  In 1951, Heinz \cite{Heinz} proved for positive operators, $A, B$ on a Hilbert space, one has that $A \le B \Rightarrow \sqrt{A} \le \sqrt{B}$.   Heinz was a student of Rellich and Kato was paying attention to the work of Rellich's group and a year later published a paper \cite{KatoHeinz} with an elegant, simple proof and extended the result to $A \mapsto A^s$ for $0 < s <1$ replacing the square root. Neither of them knew at the time that Loewner \cite{Loewner} had already proven a much more general result in 1934!  Despite the 17 year priority, the monotonicity of the square root is called variably, the Heinz inequality, the Heinz--Loewner inequality or even, sometimes, the Heinz--Kato inequality.  Heinz and Kato found equivalent results to the monotonicity of the square root (one paper with lots of additional equivalent forms is \cite{HeinzEquiv}).  In particular, the following equivalent form is almost universally known as the Heinz--Kato inequality.
\begin{equation}\label{1.2}
  \norm{T\varphi} \le \norm{A\varphi} \quad \norm{T^*\psi} \le \norm{B\psi} \Rightarrow |\jap{\psi,T\varphi}| \le \norm{A^s \varphi} \norm{B^{1-s}\psi}
\end{equation}
Kato returned several times to this subject, most notably \cite{KatoHeinz2} finding a version of the Heinz--Loewner inequality (with an extra constant depending on $s$) for maximal accretive operators on a Hilbert space.

Returning to the timing of Kato's fundamental 1951 paper \cite{KatoHisThm}, I note that he was 34 when it was published (it was submitted a few years earlier as we'll discuss in Section \ref{s7}).  Before it, his most important work was his thesis, awarded in 1951 and published in 1949-51.  One might be surprised at his age when this work was published but not if one understands the impact of the war.  Kato got his BS from the University of Tokyo in 1941, a year in which he published two (not mathematical) papers in theoretical physics.  But during the war, he was evacuated to the countryside.  We were at a conference together one evening and Kato described rather harrowing experiences in the camp he was assigned to, especially an evacuation of the camp down a steep wet hill.  He contracted TB in the camp.  In his acceptance for the Wiener Prize \cite{WP}, Kato says that his work on essential self--adjointness and on perturbation theory were essentially complete ``by the end of the war.'' Recently, several of Kato's notebook were discovered dated 1945 that contain most of results published in Kato \cite{KatoHisThm, KatoThesis} sometimes with different proofs from the later publications (these notes have recently been edited for publication in \cite{KatoNotebooks}).

\begin{wrapfigure}{r}{0.4\textwidth}
   \centering
   \vspace{-0.3cm}
   \includegraphics[width=0.4\textwidth]{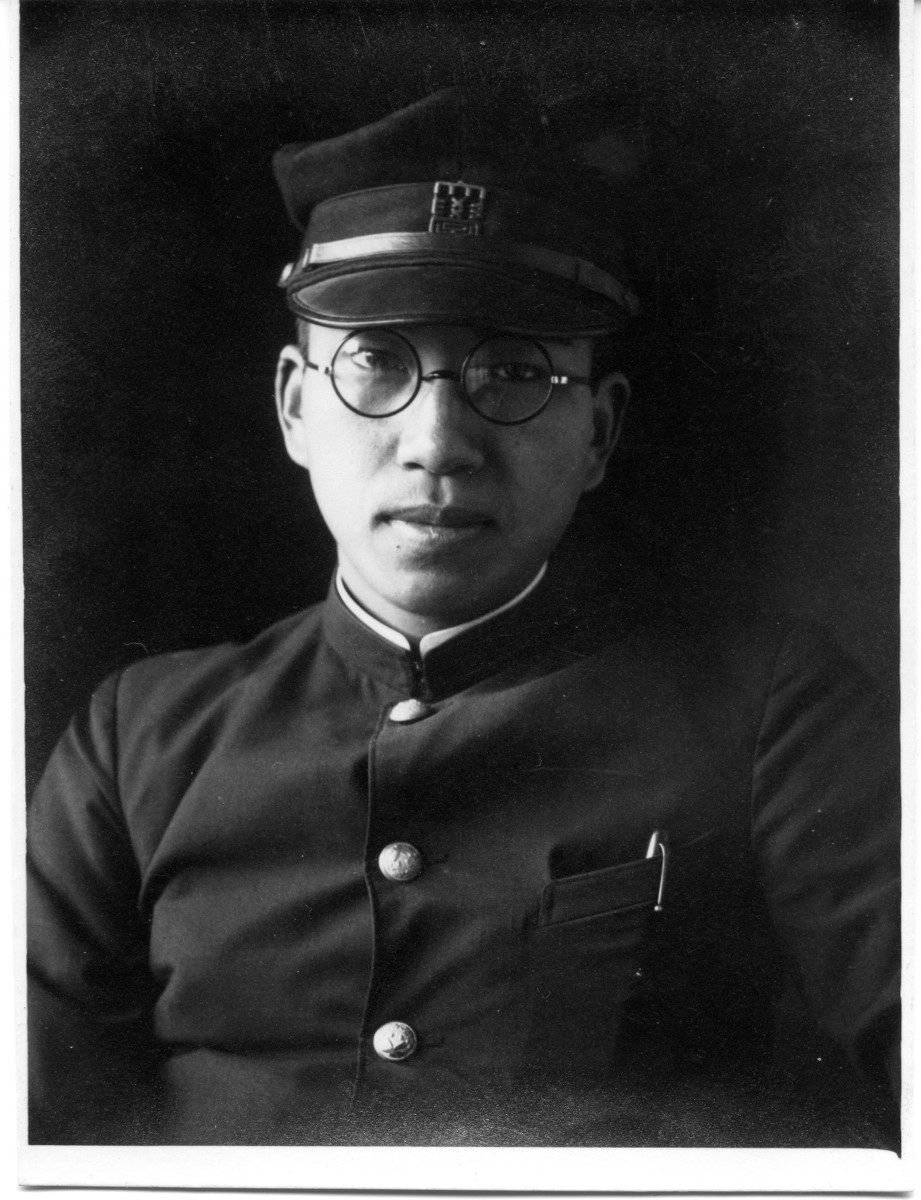}
Kato as a student \lb{figure2}
   \vspace{-0.2cm}
\end{wrapfigure} In 1946, Kato returned to the University of Tokyo as an Assistant (a position common for students progressing towards their degrees) in physics, was appointed Assistant Professor of Physics in 1951 and full professor in 1958. I've sometimes wondered what his colleagues in physics made of him. He was perhaps influenced by the distinguished Japanese algebraic geometer, Kunihiko Kodaira (1915-1997) two years his senior and a 1954 Fields medalist.  Kodaira got a BS in physics after his BA in mathematics and was given a joint appointment in 1944, so there was clearly some sympathy towards pure mathematics in the physics department.  In 1948, Kato and Kodaira wrote a 2 page note \cite{KK} to a physics journal whose point was that every $L^2$ wave function was acceptable for quantum mechanics, something about which there was confusion in the physics literature.

Beginning in 1954, Kato started visiting the United States.  This bland statement masks some drama.  In 1954, Kato was invited to visit Berkeley for a year, I presume arranged by F. Wolf.  Of course, Kato needed a visa and it is likely it would have been denied due to his history of TB.  Fortunately, just at the time (and only for a period of about a year), the scientific attach\'{e} at the US embassy in Tokyo was Otto Laporte (1902-1971) on leave from a professorship in Physics at the University of Michigan.  Charles Dolph (1919-1994), a mathematician at Michigan, learned of the problem and contacted Laporte who intervened to get Kato a visa.  Dolph once told me that he thought his most important contribution to American mathematics was his helping to allow Kato to come to the US.  In 1987, in honor of Kato's $70^{th}$ birthday, there was a special issue of the Journal of Mathematical Analysis and Applications and the issue was jointly dedicated \cite{DHBio} to Laporte (he passed away in 1971) and Kato and edited by Dolph and Kato's student Jim Howland.

During the mid 1950s, Kato spent close to three years visiting US institutions, mainly Berkeley, but also the Courant Institute, American University, National Bureau of Standards and Caltech.  In 1962, he accepted a professorship in Mathematics from Berkeley where he spent the rest of his career and remained after his retirement.  One should not underestimate the courage it takes for a 45 year old to move to a very different culture because of a scientific opportunity.  The reader can consult the Mathematics Genealogy Project (\url{http://www.genealogy.ams.org/id.php?id=32842}) for a list of Kato's students (24 listed there, 3 from Tokyo and 21 from Berkeley; the best known are Ikebe and Kuroda from Tokyo and Balslev and Howland from Berkeley) and \cite{AMSMem} for a memorial article with lots of reminisces of Kato.

One can get a feel for Kato's impact by considering the number of theorems, theories and inequalities with his name on them.  Here are some: Kato's theorem (which usually refers to his result on self--adjointness of atomic Hamiltonians), the Kato--Rellich theorem (which Rellich had first), the Kato-Rosenblum theorem and the Kato--Birman theory (where Kato had the most significant results although, as we'll see, Rosenblum should get more credit than he does), the Kato projection lemma and Kato dynamics (used in the adiabatic theorem), the Putnam--Kato theorem, the Trotter--Kato theorem (which is used for three results; see section \ref{s3}), the Kato cusp condition (see Section \ref{s19}), Kato smoothness theory, the Kato class of potentials and Kato--Kuroda eigenfunction expansions.  To me \emph{Kato's inequality} refers to the self--adjointness technique discussed in Section \ref{s9}, but the term has also been used for the Hardy like inequality with best constant for $r^{-1}$ in three dimensions (which we discuss in Section \ref{s10}), for a result on hyponormal operators that follows from Kato smoothness theory (the book \cite{hypo} has a section called ``Kato's inequality'' on it) and for the above mentioned variant of the Heinz--Loewner inequality for maximal accretive operators.  There are also Heinz--Kato, Ponce--Kato and Kato--Temple inequalities.  In \cite{SS}, Erhard Seiler and I proved that if $f,g \in L^p(\bbR^\nu),\, p\ge 2$, then $f(X)g(-i\nabla)$ is in the trace ideal $\calI_p$.  At the time, Kato and I had correspondence about the issue and about some results for $p<2$.  In \cite{RS3}, Reed and I mentioned that  Kato had this result independently.  Although Kato never published anything on the subject, in recent times, it has come to be called the Kato--Seiler--Simon inequality.

Of course, when discussing the impact of Kato's work, one must emphasize the importance of his book \emph{Perturbation Theory for Linear Operators} \cite{KatoBk} which has been a bible for several generations of mathematicians.  One of its virtues is its comprehensive nature.  Percy Deift told me that Peter Lax told him that Friedrichs remarked on the book: ``Oh, its easy to write a book when you put everything in it!''

We will not discuss every piece of work that Kato did in NRQM -- for example, he wrote several papers on variational bounds on scattering phase shifts whose lasting impact was limited.  And we will discuss Kato's work on the definition of a self--adjoint Dirac Hamiltonian which of course isn't non--relativistic.  It is closely related to the Schr\"{o}dinger work and so belongs here.  Perhaps I should have dropped ``non--relativistic'' from the title but since almost all of Kato's work on quantum theory is non--relativistic and even the Dirac stuff is not quantum field theory, I decided to leave it.

Roughly speaking, this article is in five parts.  Sections \ref{s2}-\ref{s6} discuss eigenvalue perturbation theory in both the analytic (where many of his results were rediscoveries of results of Rellich and Sz-Nagy) and asymptotic (where he was the pioneer).  There is a section on situations where either an eigenvalue is initially embedded in continuous spectrum or where as soon the perturbation is turned on the location of the spectrum is swamped by continuous spectrum (i.e. on the theory of QM resonances).  There are a pair of sections on two issues that Kato studied in connection with eigenvalue perturbation theory: pairs of projections and on the Temple--Kato inequalities.

Next come four sections on self--adjointness.  One focuses on the Kato--Rellich theorem and its applications to atomic physics, one on his work with Ikebe and one on what has come to be called Kato's inequality.  Finally his work on quadratic forms is discussed including his work on monotone  convergence for forms.

After that two pioneering works on aspects of bound states -- his result on non--existence of positive energy bound states in certain two body systems and his paper on the infinity of bound states for Helium, at least for infinite nuclear mass.

Next four sections on scattering and spectral theory which discuss the Kato--Birman theory (trace class scattering), Kato smoothness, Kato--Kuroda eigenfunction expansions and the Jensen--Kato paper on threshold behavior.

Last is a set of three miscellaneous gems: his work on the adiabatic theorem, on the Trotter product formula and his pioneering look at eigenfunction regularity.

I should warn the reader that I use two conventions that are universal among physicists but often the opposite of many mathematicians. First, my (complex) Hilbert space inner product $\jap{\varphi,\psi}$ is linear in $\psi$ and anti--linear in $\varphi$.  Secondly my wave operators are defined by (note $\pm$ vs. $\mp$)
\begin{equation*}
  \Omega^{\pm}(A,B) = \textrm{s}-\lim_{t \to \mp \infty} e^{itA}e^{-itB}P_{ac}(B)
\end{equation*}
In Section \ref{s15}, I'll explain the historical reason for this very strange convention. I should also warn the reader that I use two non--standard abbreviations ``esa'' and ``esa--$\nu$'' (where $\nu$ can be an explicit integer.  They are defined at the start of in Section \ref{s7}).

With apologies to those inadvertently left out, I'd like to thank a number of people for useful information Yosi Avron, Jan Derezi\'{n}ski, Pavel Exner, Rupert Frank, Fritz Gesztesy, Gian Michele Graf, Sandro Graffi, Vincenzo Grecchi, Evans Harrell, Ira Herbst, Bernard Helffer, Arne Jensen, Carlos Kenig, Toshi Kuroda, Peter Lax, Hiroshi Oguri, Sasha Pushnitski, Robert Seiringer, Heinz Siedentop, Israel Michael Sigal, Erik Skibsted, Terry Tao, Dimitri Yafaev and Kenji Yajima.  The pictures here are all from the estate of Mizue Kato, Tosio's wife who passed away in 2011.  Her will gave control of the pictures to H. Fujita, M. Ishiguro and S. T. Kuroda.  I thank them for permission to use the pictures and H. Okamoto for providing digital versions.

%%%%%%%%%%%%%%%%%%%%%%%%%%%%%%%%%%%%%%%%%%%%%%%%%%%%%%%%%%%%%%
\section{Eigenvalue Perturbation Theory, I: Regular Perturbations} \lb{s2}
%%%%%%%%%%%%%%%%%%%%%%%%%%%%%%%%%%%%%%%%%%%%%%%%%%%%%%%%%%%%%%

This is the first of five sections on eigenvalue perturbation theory; this section deals with the analytic case.  Section \ref{s3} begins with examples that delimit some of the possibilities when the analytic theory doesn't apply and that section and the next discuss two sets of those examples after which there are two sections on related mathematical issues which are connected to the subject and where Kato made important contributions.

Eigenvalue perturbation theory in the case where the eigenvalues are analytic (aka regular perturbation theory or analytic perturbation theory) is central to Kato's opus -- it is both a main topic of his famous book on Perturbation Theory and the main subject of his thesis.  We'll begin this section by sketching the modern theory as presented in Kato's book \cite{KatoBk} or as sketched in Simon \cite[Sections 1.4 and 2.3]{OT} (other book presentations include Baumg\"{a}rtel \cite{BaumBk}, Friedrichs \cite{FriedBk},Reed--Simon \cite{RS4} and Rellich \cite{RellichBk}). Then we'll give a Kato--centric discussion of the history.

As a preliminary, we want to recall the theory of spectral projections for general bounded operators, $A$, on a Banach space, $X$. If the spectrum of $A$, $\sigma(A)=\sigma_1\cup\sigma_2$ is a decomposition into disjoint closed sets, one can find a chain (finite sum and/or difference of contours), $\Gamma$, so that if $w(z,\Gamma)$ is the winding number about $z \notin \Gamma$, (i.e. $w(z,\Gamma) = (2\pi i)^{-1} \oint_{\zeta \in \Gamma} (\zeta-z)^{-1} d\zeta$), then $\Gamma \cap \sigma(A) = \emptyset$, $w(z,\Gamma)=0$ or $1$ for all $z \in \bbC\setminus\Gamma$, $w(z,\Gamma)=1$ for $z \in \sigma_1$, and $w(z,\Gamma)=0$ for $z \in \sigma_2$ (see \cite[Section 4.4]{CAA}).

One defines an operator
\begin{equation}\label{2.1}
  P_{\sigma_1} = \frac{1}{2\pi i}\oint_\Gamma \frac{dz}{z-A}
\end{equation}
Then one can prove \cite[Section 2.3]{OT} that $P_{\sigma_1}$ is a projection (i.e. $P_{\sigma_1}^2=P_{\sigma_1}$) commuting with $A$.  Thus $A$ maps each of $\ran\, P_{\sigma_1}$ and $\ran(\bdone-P_{\sigma_1})$ onto themselves and one can prove that
\begin{equation}\label{2.2}
  \sigma(A\restriction\ran\, P_{\sigma_1}) = \sigma_1, \qquad \sigma(A\restriction\ran (\bdone- P_{\sigma_1})) = \sigma_2
\end{equation}

Of particular interest are isolated points, $\lambda$, of $\sigma(A)$ in which case one can consider $\sigma_1 = \{\lambda\}, \, \sigma_2 = \sigma(A)\setminus\{\lambda\}$.  We write $P_{\sigma_1} = P_\lambda$ and $\calH_\lambda=\ran\, P_\lambda$.  If $\dim \calH_\lambda < \infty$, we call $\lambda$ a point of the \emph{discrete spectrum}.  In that case, it is known there is a nilpotent, $N_\lambda$, with $P_\lambda N_\lambda= N_\lambda P_\lambda = N_\lambda$ (and so $N _\lambda\restriction \ran(\bdone-P_\lambda)=0$) so that
\begin{equation}\label{2.3}
  AP_\lambda = \lambda P_\lambda + N_\lambda
\end{equation}

In particular, this implies that $\lambda$ is an eigenvalue.  The $P_\lambda$ are called \emph{eigenprojections} and the $N_\lambda$ are called \emph{eigennilpotents}.  Just as the $P_\lambda$ are first order residues of the poles of $(z-A)^{-1}$ at $z=\lambda$, the $N_\lambda$ are second order residues (and $N_\lambda^k$ is the $(z-\lambda)^{-k-1}$ residue) -- see \cite[Section 2.3]{OT} for more on the subject.

Kato's book \cite{KatoBk} is the standard reference for this beautiful complex analysis approach to Jordan normal forms whose roots go back further.  In 1913, Riesz \cite{Riesz1913}, in one of the first books on operator theory on infinite dimensional spaces, mentioned residues of poles of $(z-A)^{-1}$ could be studied and, in 1930, he noted \cite{Riesz1930} in the Hilbert space case that decompositions of the spectrum into disjoint closed sets induced a decomposition of the space.  Nagumo \cite{Nagumo} used \eqref{2.1} for Banach algebras in 1930.  Gel'fand's great 1941 paper \cite{Gelfand1941} discussed functions, $f$, analytic in a neighborhood of $\sigma(x)$ where $x \in \frA$, a commutative Banach algebra with unit and defined
\begin{equation}\label{2.4}
  f(x) = \frac{1}{2\pi i}\oint_\Gamma \frac{f(z)}{z-x} dz
\end{equation}
where $\Gamma$ surrounds the whole spectrum.

If $\sigma_1 \cup \sigma_2$ is a decomposition, $f$ can be taken to be $1$ in a neighborhood of $\sigma_1$ and $0$ in a neighborhood of $\sigma_2$.   $P_\lambda^2 = P_\lambda$ is then a special case of his functional calculus result $(fg)(x)=f(x)g(x)$.  In 1942-43, this functional calculus was further developed in the United States by Dunford \cite{Dunford1,Dunford2}, Lorch \cite{Lorch} and Taylor \cite{Taylor}.  In his book, Kato calls \eqref{2.4} a Dunford--Taylor integral.

With this formalism out of the way, we can turn to sketch the theory of regular perturbations.  For details see the book presentations of Kato \cite[Chaps. II and VII]{KatoBk}, Reed--Simon \cite[Chap XII]{RS4} and Simon \cite[Sections 1.4 and 2.3]{OT}.

\textbf{Step 1.} \emph{Finite Dimensional Theory.} Let $A(\beta)$ be an analytic family of $n \times n$ matrices for $\beta \in \Omega$, a domain in $\bbC$.  The eigenvalues are solutions of
\begin{equation}\label{2.3A}
  \det(A(\beta)-\lambda) = 0
\end{equation}
so algebroidal functions.  The theory of such functions (see Knopp \cite{Knopp2} or Simon \cite[Section 3.5]{CAA}) implies there is a discrete set of points $S \subset \Omega$ (i.e. with no limit points in $\Omega$) so that all solutions of \eqref{2.3A} are multivalued analytic functions on $\Omega\setminus S$ and so that the number of distinct solutions and their multiplicities are constant on $\Omega\setminus S$.  At points of $S$, the solutions have finite limits and are locally given by all the branches of one or more locally convergent Puiseux series (power series in $(\beta-\beta_0)^{1/p}$ for some $p \in \bbZ_+$).  From the integral formula \eqref{2.1} and its analog for $N_\lambda$, one sees that the eigenprojections and eigennilpotents
are also multivalued analytic functions on $\Omega\setminus S$.  They can have polar singularities at points in $S$, i.e. their Puiseux--Laurent series can have finitely many negative index terms. Indeed, in 1959, Butler \cite{Butler} proved that if some $\lambda(\beta)$ has a fractional power at a point $\beta_0 \in S$, then the Puiseux--Laurent series for $P(\beta)$ must have non--vanishing negative powers.

The set of early significant results include two theorems of Rellich \cite[Part I]{RellichPT}.  If $A(\beta)$ is self--adjoint (i.e. $\Omega$ is invariant under complex conjugations and $A(\bar{\beta})=A(\beta)^*)$, then $\lambda(\beta)$ and $P(\beta)$ are real analytic on $\Omega\cap\bbR$, i.e. no fractional powers in $\lambda(\beta)$ at points of $S \cap\bbR$ and no polar singularities of $P(\beta)$ there.  The first comes from the fact that if a Puiseux series based at $\beta_0 \in \bbR$ has a non--trivial fractional power term, then some branch must have non-real values for some real values of $\beta$ near $\beta_0$ (interestingly enough, in his book, Kato \cite{KatoBk} appeals to Butler's theorem instead of using this simple argument of Rellich).  The second relies on the fact that if $P(\beta)$ has polar terms at $\beta_0$, since there are only finitely many negative index terms, one has that $\lim_{|\beta-\beta_0| \downarrow 0} \norm{P(\beta)} = \infty$ which is inconsistent with the fact that spectral projections for self--adjoint matrices are self-adjoint, so with norm $1$.

For later purposes, we want to note the two leading terms in the perturbations series
\begin{equation}\label{2.4A}
  E(\beta) = E_0 + a_1\beta+a_2\beta^2+\textrm{O}(\beta^3)
\end{equation}
of a simple eigenvalue, $E_0$, of $A+\beta B$ with $A$ and $B$ Hermitian.  Suppose $\{\varphi_j\}_{j=0}^{n-1}$ are an orthonormal basis of eigenvectors of $A$ with $A\varphi_j = E_j\varphi_j$.  Then
\begin{equation}\label{2.4B}
  a_1 = \jap{\varphi_0,B\varphi_0}, \qquad a_2 = \sum_{j\ne 0}\frac{|\jap{\varphi_j,B\varphi_0}|^2}{E_0-E_j}
\end{equation}

One of Kato's contributions is to describe $a_2$ succinctly in the general infinite dimensional case where $E_0$ is discrete but $A$ may have continuous spectrum.  Let $P$ be the projection onto multiples of $\varphi_0$.  Define the reduced resolvent, $S$, of $A$ at $E_0$ by
\begin{equation}\label{2.4D}
  S = (A-E_0)^{-1}(1-P)
\end{equation}
i.e. $S\varphi_0=0$ and $S\psi=\lim_{\epsilon \to 0; \epsilon \ne 0}(A-E_0-\epsilon)^{-1}\psi$ if $\psi\perp\varphi_0$.  Thus for any $\eta$:
\begin{equation}\label{2.4E}
  (A-E_0)S\eta = (1-P)\eta
\end{equation}
In his thesis, Kato \cite{KatoThesis} realized that $a_2$ could be written
\begin{equation}\label{2.4C}
  a_2 = -\jap{\varphi_0,BSB\varphi_0}
\end{equation}

\textbf{Step 2.} \emph{Bounded Analytic Operator Valued Functions.} For $A(\beta)$, a function from a domain $\Omega \subset \bbC$ to the bounded operators on a Banach space, $X$, we say that $A$ is analytic at $\beta_0 \in \Omega$ if it is given by a convergent power series near $\beta_0$.  This is equivalent to $A$ having a complex Fr\'{e}chet derivative or to $A(\beta)x$ being a Banach space valued analytic function for all $x \in X$ or to $\ell(A(\beta)x)$ being a scalar analytic function for all $\ell \in X^*$ and $x \in X$ (see \cite[Theorem 3.1.12]{CAA}).

\textbf{Step 3}. \emph{Analytic Resolvents and Spectral Projections}.  Because the set of invertible maps in $\calL(X)$ is open and on that set, $A \mapsto A^{-1}$ is analytic (by using geometric series), if $A(\beta)$ is an analytic operator valued functions, then $\calR \equiv \{(\beta,z) \,|\, \beta \in \Omega, z \in \bbC, A(\beta)-z\bdone \textrm{ is invertible}\}$ is open in $\Omega \times \bbC$ and the resolvent $(A(\beta)-z)^{-1}$ is analytic there.  It follows that if $\lambda_0$ is an isolated point of the spectrum of $A(\beta_0)$, then there are $\epsilon, \delta$ so that for $|\beta-\beta_0| < \epsilon$ and $|z-\lambda_0| = \delta$, we have that $(\beta,z) \in \calR$ and moreover that $\sigma(A(\beta_0)) \cap \{z\,|\, |z-\lambda_0| \le \delta\} = \{\lambda_0\}$.  We can thus use \eqref{2.1} to define projections $P(\beta)$ for $|\beta-\beta_0| < \epsilon$ so that $A(\beta)P(\beta) = P(\beta)A(\beta)$ and $\sigma(A(\beta) \restriction \ran\, P(\beta)) = \sigma(A(\beta)) \cap \{z\,|\, |z-\lambda_0| \le \delta\}$.  $P(\beta)$ is analytic in $\beta$, so, by shrinking $\epsilon$ if need be, we can suppose that
\begin{equation}\label{2.5}
  |\beta-\beta_0| < \epsilon \Rightarrow \norm{P(\beta)-P(\beta_0)} < 1
\end{equation}

\textbf{Step 4}. \emph{Reduction to a finite dimensional problem}.  A basic fact that we'll prove in Section \ref{s5} (see Theorem \ref{T5.1}) is that when \eqref{2.5} holds, we can define an invertible map $U(\beta)$ for $|\beta-\beta_0| < \epsilon$ analytic in $\beta$ so that
\begin{equation}\label{2.6}
  U(\beta)P(\beta)U(\beta)^{-1} = P(\beta_0)
\end{equation}
Moreover, if $X$ is a Hilbert space and $P(\beta)$ is self--adjoint for $|\beta-\beta_0| < 1$ and $\textrm{Im}\,(\beta-\beta_0) = 0$, then $U(\beta)$ is unitary for such $\beta$.

Because of \eqref{2.6}, $\wti{A}(\beta) \equiv U(\beta)A(\beta)U(\beta)^{-1}$ leaves $\ran\, P(\beta_0)$ invariant and
\begin{equation}\label{2.7A}
  \sigma(\wti{A}(\beta) \restriction \ran\, P(\beta_0)) = \sigma(\wti{A}(\beta)) \cap \{z\,|\, |z-\lambda_0| \le \delta\}
\end{equation}

If now $\lambda_0$ is a point of the discrete spectrum of $A(\beta_0)$, then $P(\beta_0)$ is finite dimensional, so $\wti{A}(\beta) \restriction \ran\, P(\beta_0)$ is a finite dimensional problem and all the results of Step 1 apply.  Moreover, if $X$ is a Hilbert space and $A(\beta)$ is self--adjoint for $\beta$ real, then so is $\wti{A}(\beta)$ and Rellich's Theorems extend. Note that even if $A(\beta)$ is linear in $\beta$, $\wti{A}(\beta)$ will not even be polynomial in $\beta$ so it is important that step 1 be done for general analytic families.

\textbf{Step 5} \emph{Regular Families of Closed Operators}.  For $\beta \in \Omega$, a domain, we consider a family, $A(\beta)$ of closed, densely defined (but not necessarily bounded) operators on a Banach space, $X$.  We say that $A$ is a \emph{regular family} if, for every $\beta_0 \in \Omega$, there is a $z_0 \in \bbC$ and $\epsilon >0$ so that for $|\beta-\beta_0| < \epsilon$, we have that $z_0 \notin \sigma(A(\beta))$ and $\beta \mapsto (A(\beta)-z_0)^{-1}$ is a bounded analytic function near $\beta_0$.  Kato \cite[Section VII.1.2]{KatoBk} has a more general definition that applies even to closed operators between two Banach spaces $X$ and $Y$ but he proves that it is equivalent to the above definition so long as $X=Y$ and every $A(\beta)$ has a non--empty resolvent set (which is no restriction if you want to consider isolated eigenvalues).

With this definition, all the eigenvalue perturbation theory for the bounded case carries over since $\lambda_0$ is a discrete eigenvalue of $A(\beta_0)$ if and only if $(\lambda_0-z_0)^{-1}$ is a discrete eigenvalue of $(A(\beta_0)-z_0)^{-1}$.

\textbf{Step 6} \emph{Criteria for Regular Families}. A \emph{type (A) family} is a function, $A(\beta)$, for $\beta \in \Omega$, a region in $\bbC$, so that $A(\beta)$ is a closed, densely defined operator on a Banach space, $X$, with domain $D(A(\beta)) = \calD$ independent of $\beta$ and so that for all $\varphi \in \calD$ we have that $\beta \mapsto A(\beta)\varphi$ is an analytic vector valued function.  If $A(\beta_0)$ has non--empty resolvent set, it is easy to see that $A(\beta)$ is a regular family for $\beta$ near $\beta_0$.  In particular, if the resolvent set is non--empty for all $\beta \in \Omega$, then $A(\beta)$ is a regular family on $\Omega$.

Of particular interest is the case where $A(\beta) = A_0+\beta B$ where $\calD=D(A_0)$ and $B$ is an operator with $\calD \subset D(B)$.  Then $A(\beta)$ is closed for all $\beta$ small if only if there are $a, b > 0$ so that for all $\varphi \in \calD$, one has that
\begin{equation}\label{2.7}
  \norm{B\varphi} \le a\norm{A_0\varphi} + b\norm{\varphi}
\end{equation}
Thus \eqref{2.7} is a necessary and sufficient condition for a linear $A(\beta)$ to be an analytic family of type (A) near $\beta=0$.

If a bound like \eqref{2.7} holds, we say that $B$ is \emph{$A$-bounded}.  The \emph{relative bound} is the $\inf$ over all $a$ for which \eqref{2.7} holds (typically, if $a_0$ is this $\inf$, the bound only holds for $a > a_0$ and the corresponding $b$'s go to $\infty$ as $a \downarrow a_0$).  There exist unbounded $B$ for which the relative bound is $0$.  There are similar bounds for general analytic families of type (A): $A(\beta) = A +\sum_{n=1}^{\infty} \beta^n B_n$ and $B_n$ obeys $D(B_n) \supset D(A)$ and for some $a,b,c$ and all $\varphi \in D(A)$ one has that
\begin{equation}\label{2.8}
  \norm{B_n\varphi} \le c^{n-1} (a\norm{A\varphi}+b\norm{\varphi})
\end{equation}

There is also a notion of type(B) families on Hilbert space (due to Kato \cite{KatoBk}) where one demands that $A(\beta)$ be m--accretive with $\beta$ independent form domain.

\begin{example} [\emph{$1/Z$ expansion}] \lb{E2.1} A simple example of regular perturbation theory of physical interest concerns two electron ions which in the limit of infinite nuclear mass (ignoring relativistic and spin corrections) is described by
\begin{equation}\label{2.8A}
  H(Z) = -\Delta_1-\Delta_2 - \frac{Z}{r_1}-\frac{Z}{r_2} + \frac{1}{|\boldsymbol{r_1}-\boldsymbol{r_2}|}
\end{equation}
on $L^2(\bbR^6,d^3\boldsymbol{r_1} d^3\boldsymbol{r_2})$.  Under a scale transformation $Z^{-2}H(Z)$ is unitarily equivalent to
\begin{equation}\label{2.8B}
  A(1/Z) = -\Delta_1-\Delta_2 - \frac{1}{r_1}-\frac{1}{r_2} + \frac{1}{Z|\boldsymbol{r_1}-\boldsymbol{r_2}|}
\end{equation}

This is an entire family of type (A) in $1/Z$.  At $1/Z = 0$, the ground state energy is $E_0(0) = -\tfrac{1}{2}$.  For all $Z$, the HVZ theorem (\cite[Section XIII.5]{RS4}) implies that the continuous spectrum of $A(1/Z)$ is $[-\tfrac{1}{4},\infty)$.

Kato was concerned with rigorous estimates on the radius of convergence, $\rho$, of the power series for $E_0(1/Z)$.  He discussed this in his thesis and, in his book \cite[Section VII.4.9]{KatoBk}, was able to show that $\rho > 0.24$ and he noted that this didn't cover the physically important cases $1/Z = 1/2$, i.e, Helium ($Z=2$).  In fact the case $1/Z=1$ is also important because it describes the $H^-$ ion which is known to exist.

There has been considerable physical literature on this example.  Stillinger \cite{Still} found numerically that the perturbation coefficients (not found numerically using perturbation theory but by fitting variationally calculated eigenvalues) are eventually all positive, so there is a singularity on the positive real axis at $\rho$.  As $\beta = 1/Z$ increases, $E(\beta)$ is monotone increasing and known to be real analytic at least until $E$ reaches the bottom of the continuous spectrum, $-\tfrac{1}{4}$, at $\beta=\beta_c$.  Since $H^-$ exists, $\beta_c > 1$.  The best current numerical estimate \cite{Est} suggests that $\rho = \beta_c$ and
\begin{equation*}
  \beta_c = 1.09766083373855980(5)
\end{equation*}
It is known \cite{HOHOS} (see \cite{FBLSCoulomb, Gridnov, FLSCoulomb} for improved results) that at $\beta=\beta_c$, $A(\beta)$ has an eigenvalue at $E(\beta_c) = -\tfrac{1}{4}$.  It would be interesting to understand the nature of the singularity at $\beta=\beta_c$, e.g. is there a convergent Puiseux series?
\end{example}

This completes our discussion of the theory of eigenvalue perturbation theory so we turn to some remarks on its history.  Eigenvalue perturbation theory goes back to fundamental work of Lord Rayleigh on sound waves in 1897 \cite[pp. 115--118]{Ray} and by Schr\"{o}dinger at the dawn of (new) quantum mechanics \cite{Schr} and is often called Rayleigh--Schr\"{o}dinger perturbation theory.

The first substantial rigorous mathematical work on the subject is a five part series of papers by Rellich \cite{RellichPT} published from 1937 to 1942.  It included an exhaustive treatment of the finite dimensional case including what we called Rellich's Theorems on the lack of singularities in the self--adjoint case.  He also noted the simple example:
\begin{equation}\label{2.9}
  A(\beta,\gamma) = \left(
                      \begin{array}{cc}
                        \beta & \gamma \\
                        \gamma & -\beta\\
                      \end{array}
                    \right)
\end{equation}
with eigenvalues $\pm \sqrt{\beta^2+\gamma^2}$ which shows that his analyticity results for the self--adjoint case do not extend to more than one variable.  He also considered the infinite dimensional case where \eqref{2.7} holds ($A$ self--adjoint and $B$ symmetric) and \eqref{2.8} appeared in his papers.  His papers did not use spectral projections but rather some brute force calculations.

B. Sz.-Nagy followed up Rellich's work in two papers published in 1947 and 1951 \cite{Nagy1947, Nagy1951} in which he treated the self--adjoint Hilbert space case and general closed operators on Banach spaces respectively.  The first paper had a 1942 Hungarian language version \cite{Nagy1942}.  He defined type (A) perturbations via \eqref{2.8}.  His main advance is to exploit the definition of spectral projections via \eqref{2.1}.  As a student of F. Riesz, this is not surprising.  This was also the first place that it was proven (in the Hilbert space case) that two orthogonal projections, $P$ and $Q$ with $\norm{P-Q} < 1$ are related via $Q=UPU^{-1}$ for a unitary which is analytic function of $Q$, i.e. he implemented Step 4 above.

Wolf \cite{Wolf} also extended the Nagy approach to the Banach space case is 1952.  Perhaps the most significant aspect of this work is that it served eventually to introduce Kato to Wolf for Wolf was a Professor at Berkeley who was essential to recruiting Kato to come to Berkeley both in 1954 and 1962.

Franti\v{s}ek Wolf (1904--1989) was a Czech mathematician who had a junior position at Charles University in Prague. Wolf had spent time in Cambridge and did some significant work on trigonometric series under the influence of Littlewood.  When the Germans invaded Czechoslovakia in March 1938, he was able to get an invitation to Mittag--Leffler.  He got permission from the Germans for a three week visa but stayed in Sweden!  He was then able to get an instructorship at Macalester College in Minnesota.  He made what turned out to be a fateful decision in terms of later developments.  Because travel across the Atlantic was difficult, he took the trans--Siberian railroad across the Soviet Union and then through Japan and across the Pacific to the US.  This was mid--1941 before the US entered the war and made travel across the Pacific difficult.

Wolf stopped in Berkeley to talk with G. C. Evans (known for his work on potential theory) who was then department chair.  Evans knew of Wolf's work and offered him a position on the spot!!  After the year he promised to Macalester, Wolf returned to Berkeley and worked his way up the ranks.  In 1952, Wolf extended Sz--Nagy's work to the Banach space case.  At about the same time Nagy himself did similar work and in so did Kato. While Wolf and Kato didn't know of each other's work, Wolf learned of Kato's work and that led to his invitation for Kato to visit Berkeley.

Kato's thesis dealt with both analytic and asymptotic perturbation theory (we'll discuss the later in the next section).  It appears that Kato found much of this in about 1944 without knowing about the work of Rellich or Nagy although he did know about Rellich by the time his thesis was written and he learned about the work of Nagy before the publication of the last of his early papers on perturbation theory\cite{KatoPTClosed,KatoAsymPT}.

Interestingly enough, Kato's first published work on the perturbation theory of eigenvalues \cite{Kato2Examples} was a brief 1948 note with examples where the theory didn't apply - these will be discussed in the next section (Examples \ref{E3.5}, \ref{E3.6}).  His thesis was published in a university journal in full \cite{KatoThesis} in 1951 with parts published a year early in broader journals in both English \cite{KatoPT1,KatoPT2} and Japanese \cite{KatoPTJap}.  Two final early papers \cite{KatoPTClosed,KatoAsymPT} dealt with the Banach space case and with further results on asymptotic perturbation theory (discussed further in Section \ref{s6}).

Many of the most significant results in Kato's work on regular eigenvalue perturbation theory had been found (independently but) earlier by Rellich and Nagy.  Kato's work, especially if you include his book \cite{KatoBk}, was more systematic.  His main contribution beyond theirs concerns the use of reduced resolvents.  And, as we'll see, he was the pioneer in the theory of asymptotic perturbation theory.

%%%%%%%%%%%%%%%%%%%%%%%%%%%%%%%%%%%%%%%%%%%%%%%%%%%%%%%%%%%%%%
\section{Eigenvalue Perturbation Theory, II: Asymptotic Perturbation Theory} \lb{s3}
%%%%%%%%%%%%%%%%%%%%%%%%%%%%%%%%%%%%%%%%%%%%%%%%%%%%%%%%%%%%%%

In this section and the next, we discuss situations where the Kato--Nagy--Rellich theory of regular perturbations does not apply.  Lest the reader think this is a strange pathology, we begin with six (!) simple examples, four from the standard physics literature and then two that appeared in Kato's first paper -- a brief note -- on perturbation theory \cite{Kato2Examples}.

\begin{example} [\emph{Anharmonic oscillator and Zeeman effect}] \lb{E3.1} Let
\begin{equation}\label{3.1}
  A_0 = -\frac{d^2}{dx^2}+x^2, \qquad B=x^4, \qquad A(\beta) = A_0 + \beta B
\end{equation}
on $L^2(\bbR,dx)$.  This is an example much beloved by teachers of quantum mechanics since one can compute $a_2$ explicitly since the sum in \eqref{2.4B} is finite (indeed only two terms which can be computed in closed form).  It is also regarded as a paradigm of the simplest quantum field theory, i.e. $\varphi^4_1$ in one space--time dimension (see \cite{GJQFT, SimonQFT}).  A basic fact is that the perturbation series exists to all orders, in fact all the sums in the books \cite{KatoBk,RS4} for individual terms are finite or, alternatively, there exists a simple set of recursion relations \cite{BW1} for the $a_n$ so that formally, the ground state energy is given by
\begin{equation}\label{3.2}
  E_0(\beta) = E_0 + \sum_{n=1}^{\infty} a_n \beta^n
\end{equation}

However, the series in \eqref{3.2} has zero radius of convergence.  One intuition comes from Dyson \cite{Dys} who argued that the perturbation series in quantum electrodynamics shouldn't converge because the theory doesn't make sense if $e^2 < 0$ when electrons attract and there is collapse.  Similarly, $A_0-\beta x^4$ does not define a self--adjoint operator since it is limit circle at $\pm \infty$ (see \cite[Section 7.4]{OT}).  While this is not a proof, one can show (\cite{SimonAHO, LM, LMSW}) that $A(\beta)$ is a type (A) family for $\beta \in \bbC \setminus (-\infty,0]$ (but not at $\beta = 0$), that any eigenvalue, $E_n(\beta)$, of $A(\beta)$ for $\beta >0$ can be analytically continued to all of $\beta \in \bbC \setminus (-\infty,0]$ with limits on $(-\infty,0)$ from either side with $\textrm{Im} E_n(-\beta +i0) >0$ for any $\beta > 0$ (so the continuation is not analytic at $\beta =0$).  \cite{SimonAHO} has much about the analytic structure near $\beta = 0$.

This doesn't quite imply that the series is divergent, only that it can't converge to the right answer.  In fact, one knows that the $a_n$ grow so fast that the series diverges for all $\beta \ne 0$.  Indeed, it is known that
\begin{equation}\label{3.3}
  a_n = 4 \pi^{-3/2} (-1)^{n+1} \left(\tfrac{3}{2}\right)^{n+1/2} \Gamma(n+\tfrac{1}{2}) \left(1+\textrm{O}\left(\tfrac{1}{n}\right)\right)
\end{equation}
This formula with its $n!$ growth is called the \emph{Bender--Wu formula}.  They guessed it \cite{BW1} from a calculation of the first 75 $a_n$ in 1969 and found a non--rigorous argument for it in 1973 \cite{BW2}.  It was proven by Harrell--Simon \cite{HS} in 1980 -- we'll discuss it in  the next section.

There is also literature on the higher order anharmonic oscillator,
\begin{equation}\label{3.3a}
A(\beta) = -\tfrac{d^2}{dx^2} + x^2 + \beta x^{2m}; \quad m=2,3,\dots
\end{equation}
In this case the analogs of Bender--Wu asymptotics have $a_n \sim C (-1)^{n+1} A^n n^\gamma \Gamma((m-1)n)$ for suitable $m$--dependent $A, C, \gamma$.

There is a historically important model that has a similar divergence, namely the Zeeman effect for Hydrogen which describes Hydrogen in a constant magnetic field, $B$, which if $\boldsymbol{B}$ points in the $z$ direction in $\boldsymbol{r}=(x,y,z)$ coordinates is given by the Hamiltonian
\begin{equation}\label{3.4}
  A(B) = -\tfrac{1}{2}\Delta - \tfrac{1}{r} +\tfrac{B^2}{8}(x^2+y^2)+BL_z
\end{equation}
where $L_z$ is the $z$ component of the angular momentum.  For the ground state (where $L_z=0$), one has that
\begin{equation}\label{3.5}
  E_0(B) = \sum_{k=0}^{\infty} E_k B^{2k}
\end{equation}
Avron \cite{AvronZeeman} found a Bender--Wu type formula
\begin{equation}\label{3.6}
  E_k = \left(\frac{4}{\pi}\right)^{5/2} (-1)^{k+1} \pi^{-2k} \Gamma\left(2k+\frac{3}{2}\right) \left(1+\textrm{O}\left(\frac{1}{k}\right)\right)
\end{equation}
with a rigourous proof by Helffer--Sj\"{o}strand \cite{HelfS2}.  In natural units, the magnetic field in early $20^{th}$ century laboratories was very small so lowest order perturbation theory worked very well.
\end{example}

\begin{example} [\emph{Autoionizing States of Two Electron Atoms}] \lb{E3.2} We further consider the Hamiltonian $A(1/Z)$ of Example \ref{E2.1}; see \eqref{2.8B}.  For $1/Z = 0$, $A(0)$ is the Hamiltonian of two uncoupled Hydrogen atoms so its eigenvalues are $E_{n,m} = -\tfrac{1}{4n^2}-\tfrac{1}{4m^2}, \, m,n=1,2,\dots$.  The continuous spectrum starts at $-\tfrac{1}{4}$ (for $n=1, m\to \infty$), so, for example, $E_{2,2}$ at energy $-\tfrac{1}{8}$ is an eigenvalue but not isolated, rather it is embedded in the continuous spectrum on $[-\tfrac{1}{4},\infty)$.  According to the physicist's expectation, this eigenvalue becomes a decaying state,  where in a finite time, one electron drops to the ground state and the other gets kicked out of the atom with the left over energy (i.e. $-\tfrac{1}{8}-(-\tfrac{1}{4})=\tfrac{1}{8}$).  For obvious reasons, these are called autoionizing states.  These states are actually seen as electron scattering resonances (under $e + He^+ \to e+He^+$) or as photo ionization resonances ($\gamma+He \to He^+ + e$) called Auger resonances.

The situation has a complication we'll ignore.  The eigenvalue at energy $-\tfrac{1}{8}$ has multiplicity 16 which one can reduce by using exchange, rotation and parity symmetry.  For our purposes, it is useful to look at states with angular momentum 2 and azimuthal angular momentum 2 which are simple.  In fact, there are states of unnatural parity (with angular momentum 1 but parity +); the continuous spectrum below $-\tfrac{1}{16}$ is only of natural parity states so these unnatural parity eigenvalues are not embedded in continuous spectrum and so they don't disappear.  There are actually 15 subspaces with definite symmetry.  In one, there is a doubly degenerate embedded eigenvalue, in 3 an isolated eigenvalue and in 11 a simple embedded eigenvalue.

According to what is called the Wigner--Weisskopf theory \cite{WW}, these scattering resonances are complex poles of the S--matrix so the perturbed energy, $E(\beta)$ has a non--zero imaginary part
\begin{equation}\label{3.6}
  \textrm{Im}\, E(\beta) = \frac{\Gamma(\beta)}{2}
\end{equation}
where $\Gamma$ is the \emph{width} of the resonance, i.e. $|(E-E_0)+\tfrac{i}{2}\Gamma|^{-2}$ (the impact of a pure pole to a quantum probability) has a distance $\Gamma$ between the two points where it takes half its maximum value.

Physicists argue that $\Gamma = \hbar/\tau$, where $\tau$ is the lifetime of the excited state.  Sometimes Rayleigh--Schr\"{o}dinger perturbation theory is called time--independent perturbation theory because there is a formal textbook argument for computing lifetimes of embedded eigenvalues coupled to the continuum called time--dependent perturbation theory.  In particular, the second order term in this theory is called the Fermi golden rule, discussed, for example, in Landau-Lifshitz \cite[pp. 140-153]{LL}.  Simon \cite{SimonTDPT} has a compact way to write this second order term.  If $A(\beta) = A_0+\beta B$, $A_0\varphi_0 = E_0\varphi_0$ and $\wti{P}_0(\lambda)$ is the spectral projection for $A_0$ with $\{E_0\}$ removed, i.e. $\wti{P}_0(\lambda) = f_\lambda(A)$ where
\begin{equation*}
  f_\lambda(x) = \left\{
                   \begin{array}{ll}
                     1, & x<\lambda, x \ne E_0 \\
                     0, & x \ge \lambda, \textrm{ or } x=E_0
                   \end{array}
                 \right.
\end{equation*}
then
\begin{equation}\label{3.7A}
  \Gamma(\beta) = \Gamma_2 \beta^2 + \textrm{O}(\beta^3)
\end{equation}
\begin{equation}\label{3.7}
  \Gamma_2 = \left.\frac{d}{d\lambda}\jap{B\varphi_0,\wti{P}_0(\lambda)B\varphi_0}\right|_{\lambda=E_0}
\end{equation}
The physics literature arguments for time--dependent perturbation theory are mathematically questionable and there were arguments about what the higher order terms were.

So this example causes lots of problems we'll look at in Section 4: What is a resonance?  What does the perturbation series have to do with the resonance energy? Can one mathematically justify the Fermi golden rule?  What are the higher terms?  Is there a convergent series?

In 1948, Friedrichs \cite{FriedCont} considered a model (related to some earlier work of his \cite{FriedOld}) with operators acting on $L^2([a,b],dx)\oplus \bbC$ with $A_0(f(x),\zeta) = (xf(x),\zeta)$ where $a < 1 < b$ so that $A_0$ has an embedded eigenvalue at $E_0=1$.  $A(\beta)=A_0+\beta B$ where $B$ is the rank two operator $B(f(x),\zeta) = (\zeta h(x), \jap{h,f})$ for some $h \in L^2([a,b],dx)$.  For suitable $h$ and small $\beta > 0$, Friedrichs proved that $A(\beta)$ has no eigenvalues in spite of the fact of a first order perturbation term so the eigenvalue indeed dissolves.  He did not discuss resonances but this was an early attempt to study a model which in his words ``is clearly related to the Auger effect.''
\end{example}

\begin{example} [\emph{Stark Effect}] \lb{E3.3} The Stark Hamiltonian describes the Hydrogen atom in an electric field.  If $F$ is the strength of the field and $\boldsymbol{r}=(x,y,z)$, then the operator on $L^2(\bbR^3)$ has the form
\begin{equation}\label{3.8}
  A(F,Z)=-\Delta-\frac{Z}{r}+Fz
\end{equation}
We will primarily consider $Z=1$.  Schr\"{o}dinger developed eigenvalue perturbation theory \cite{Schr} to apply it to the Stark Hamiltonian.  As with the Zeeman effect, laboratory $F$'s are small so first or second order perturbation theory worked well when compared to experiment and this was regarded as a great success.

Early on, Oppenheimer \cite{Opp} pointed out that when $F \ne 0$, $A(F,Z)$ is not bounded below so that the $A(F=0,Z=1)$ ground state is, as soon as $F \ne 0$, swamped by continuous spectrum.  Put differently, it becomes a finite lifetime state that decays.  He claimed to compute the lifetime but his calculation was wrong.  There are arguments about whether his method was correct but eventually universal agreement that the correct asymptotics for the width, when $Z=1$ and $F$ is small, is that found by Lanczos \cite{Lanc}:
\begin{equation}\label{3.9}
  \Gamma(F) \sim \frac{1}{2F}\exp\left(-\frac{1}{6F}\right)
\end{equation}
which is usually called the \emph{Oppenheimer formula}.

In fact, one can prove that for any $F \ne 0$, and any $Z$ including $Z=0$, $A(F,Z)$ has spectrum $(-\infty,\infty)$ with infinite multiplicity, purely absolutely continuous spectrum.  Titchmarsh \cite{TitBk} proved there are no embedded eigenvalues using the separability in parabolic coordinates we'll use again below, Avron--Herbst \cite{AvronHerbstStark} proved the existence of wave operators from $A(F,Z=0)$ to $A(F,Z)$ (wave operators are discussed in Section \ref{s13}) and Herbst \cite{HerbstStark1} proved that those wave operators were unitaries, $U$, with $UA(F,Z=0)U^{-1} = A(F,Z)$.

In this regard, I should mention what I've called \cite{SimonRazor} \emph{Howland's Razor} after \cite{Howland1, Howland2} and Occam's Razor: ``Resonances cannot be intrinsic to an abstract operator on a Hilbert space but must involve additional structure.''  For $\{A(F,1)\}_{F \ne 0}$ are all unitarily equivalent but we believe they have $F$--dependent resonance energies.  We'll discuss the possible extra structures in the next section.

There is also a Bender--Wu type asymptotics
\begin{equation}\label{3.10}
  E(F) \sim \sum_{n=0}^{\infty} A_{2n} F^{2n}
\end{equation}
\begin{equation}\label{3.11}
  A_{2n} = -6^{2n+1} (2\pi)^{-1} (2n)! \left(1+\textrm{O}\left(\frac{1}{n}\right)\right)
\end{equation}
found formally by Herbst--Simon \cite{HerbS} and proven by Harrell--Simon \cite{HS}.   Interestingly enough, there is a close connection between \eqref{3.11} and the original Bender--Wu formula \eqref{3.3} or rather its analog for
\begin{equation}\label{3.12}
  -\frac{d^2}{dx^2} + x^2 + \beta x^4 -\frac{1}{4x^2}
\end{equation}
whose Bender--Wu formula was found by Banks, Bender and Wu \cite{BBW}.  Jacobi \cite{Jac} discovered that a Coulomb plus linear potential in classical mechanics separates in elliptic coordinates and then Schwarzschild \cite{SchwSch} and Epstein \cite{Epst1} extended this idea to old quantum theory.  In particular, Epstein used parabolic coordinates.  Schr\"{o}dinger \cite{Schr} and Epstein \cite{Epst2} extended this use of parabolic coordinates to the Hamiltonian \eqref{3.8}.  This separation was also used by Titchmarsh \cite{TitPT, TitBk}, Harrell--Simon \cite{HS} and by Graffi--Grecchi and collaborators \cite{GG1, GG2, GG3, GG4, GG5, GG6, GG7}.

Many of the same questions occur as for Example \ref{E3.2} which we'll study in Section 4: What is a resonance? What is the meaning of the divergent perturbation series?  What is the difference between \eqref{3.7A} where $\Gamma(\beta) = \textrm{O}(\beta^2)$ and \eqref{3.9} where $\Gamma(\beta) = \textrm{O}(\beta^k)$ for all $k$.
\end{example}

\begin{example} [\emph{Double Wells}] \lb{E3.4} The standard double well problem is
\begin{equation}\label{3.13}
  A(\beta) = -\frac{d^2}{dx^2}+x^2-2\beta x^3 + \beta^2 x^4
\end{equation}
Writing
\begin{align*}
  V(\beta,x) &\equiv x^2-2\beta x^3 + \beta^2 x^4  \\
             &= x^2(1-x\beta)^2 \\
             &= \beta^2 x^2(\beta^{-1}-x)^2
\end{align*}
we see that if $U_\beta f(x) = f(\beta^{-1}-x)$ which is unitary, then $U_\beta A(\beta) U_\beta^{-1} = A(\beta)$.  If we let $\varphi_0(x) = \pi^{-1/4}\exp(-\tfrac{1}{2} x^2)$, then $\jap{\varphi_0,A(\beta)\varphi_0} = 1 + \textrm{O}(\beta^2)$.  But by symmetry, $\jap{U_\beta\varphi_0,A(\beta)U_\beta\varphi_0} = 1 + \textrm{O}(\beta^2)$ while $\jap{\varphi_0,U_\beta\varphi_0}$ and $\jap{A(\beta)\varphi_0,U_\beta\varphi_0}$ are O$(\exp(-1/(4\beta^2)))$, so very small.  Thus, we see that while $A(\beta=0)$ has simple eigenvalues at $2n+1, \, n=0,1,2,\dots$, for $\beta \ne 0$, $A(\beta)$ has a least two eigenvalues near each $E_n(\beta=0)$.

So far as I know, Kato never discussed anything like double wells in print, but we'll see shortly that it illuminates the meaning of stability, a subject that Kato was the first to emphasize.

This model is closely related to the family on $L^2(\bbR^\nu)$:
\begin{equation}\label{3.14}
  H(\lambda) = -\Delta+\lambda^2 h(x) + \lambda g(x)
\end{equation}
where $h, g$ are $C^\infty$, $g$ is bounded from below, $h \ge \epsilon > 0$ near $\infty$, $h \ge 0$, $h(x)=0$ for only finitely many points and so that at those points the Hessian matrix $\frac{\partial^2 h}{\partial x_i \partial x_j}$ is strictly positive definite.  One is interested in eigenvalues of $H(\lambda)$ as $\lambda \to \infty$.  Notice that when $g=0$,  $\lambda^{-2} H(\lambda) = -\lambda^{-2}\Delta + h$, so this is a quasi--classical ($\hbar \to 0$) limit.  One can rephrase the double well as looking at $-\tfrac{d^2}{dx^2}+\lambda^2 x^2(1-x)^2$ by scaling of space and energy (see Simon \cite{SimonSemiC1}).  There is a considerable literature both on leading asymptotics and on the exponential splitting of the two lowest eigenvalues -- see, for example, Simon \cite{SimonSemiC1, SimonSemiC2} and Helffer--Sj\"{o}strand \cite{HelfS1, HelfS2}.  We note that Witten \cite{WittenMorse} has a proof of the Morse inequalities that relies on this leading quasi--classical limit (see also Cycon et al \cite{CFKS}).
\end{example}

\begin{example} \lb{E3.5} Our last two examples, unlike the first four are neither well-known nor heavily studied.  They are from Kato's first paper on perturbation of eigenvalues, a one page letter to the editor of Progress of Theoretical Physics in 1948.  Both examples, which also appear in his thesis \cite{KatoThesis}, have $A(\beta) = A_0+\beta B$ with
\begin{equation}\label{3.15}
  A_0 = -\jap{\psi,\cdot}\psi
\end{equation}
where $\psi \in L^2(\bbR,dx)$ has $\norm{\psi}_2 = 1$.  He focuses on what happens to the simple eigenvalue $A_0$ has at $E_0 = -1$.

In his first example, he takes $B$ to be multiplication by $x$.  This model is the poor man's Stark effect.  He doesn't mention this connection in the paper but does in the thesis.  He states without proof in the Note (but does have a proof in the thesis) that for $\beta \ne 0$, $A(\beta)$ has no eigenvalues but has a purely continuous spectrum.  He remarked that this example shows that the formal perturbation series may be quite meaningless even if no ``divergence'' occurs.  In his later work, as we'll see in Section \ref{s4}, he did discuss a possible significance of such series.
\end{example}

\begin{example} \lb{E3.6}  $A_0$ is given by \eqref{3.15} but now $B$ is multiplication by $x^2$.  Kato states and proves in his thesis that for $\beta$ small and positive, $A(\beta)$ has a simple eigenvalue near $E=-1$.  Kato proves this by direct calculation rather than the more general strong convergence method in his book which we discuss below.  He then discusses two explicit special $\psi$'s for which the first order term, $\int x^2|\psi(x)|^2 dx$, is infinite.  For $\psi = c(1+x^2)^{-1/2}$, he finds (in the thesis; the paper only has the O($\beta^{1/2}$) term):
\begin{equation}\label{3.16}
  E(\beta) = -1 + \beta^{1/2} - \tfrac{1}{2}\beta + \tfrac{1}{8}\beta^{3/2}+\textrm{O}(\beta^2).
\end{equation}
For $\psi = c|x|^{1/2}(1+x^2)^{-1}$ where the first order integral is only logarithmically divergent, he claims that
\begin{equation}\label{3.17}
  E(\beta) = -1 + \beta \log(\beta) + \textrm{O}(\beta)
\end{equation}
The thesis but not the paper also discusses $\psi = c(1+x^2)^{-1}$ where the first order integral is finite, he claims that
\begin{equation}\label{3.18}
  E(\beta) = -1 + \beta -2\beta^{3/2} + \textrm{O}(\beta^2)
\end{equation}
Kato is primarily a theorem prover and concept developer but occasionally he produces detailed calculational results, often without details; we'll discuss this further in Section \ref{s7}.

This example is quite artificial but in his book \cite{KatoBk}, Kato has an example going back to Rayleigh \cite{Ray}
\begin{equation}\label{3.19}
  A(\beta) = -\frac{d^2}{dx^2}+\beta\frac{d^4}{dx^4}, \quad \beta >0
\end{equation}
with
\begin{equation}\label{3.20}
  \varphi(0) = \varphi'(0) = \varphi(1) = \varphi'(1) = 0
\end{equation}
boundary conditions.  Clearly $A(0)$ should have $A(0) = -\tfrac{d^2}{dx^2}$ but the boundary conditions \eqref{3.20} are too strong to get a self--adjoint operator.  One can show that the right boundary conditions for a strong limit are $\varphi(0) = \varphi(1) = 0$ and that
\begin{equation}\label{3.21}
  E_n(\beta) = n^2 \pi^2 \left[1+4\beta^{1/2}+\textrm{O}(\beta)\right]
\end{equation}
\end{example}

With these examples in mind, we turn to the general theory of asymptotic series.  Recall \cite[Section 15.1]{CAB} that given a function $\beta \mapsto f(\beta)$ on $(0,B)$ and a sequence $\{a_n\}_{n=0}^\infty$, we say that $\sum_{n=0}^{\infty} a_n \beta^n$ is an \emph{asymptotic series to order $N$} if an only if
\begin{equation}\label{3.22}
  f(\beta) -\sum_{n=0}^{N} a_n \beta^n = \textrm{o}(\beta^N)
\end{equation}
Of course, if the series is asymptotic to order $(N+1)$, the right side of \eqref{3.22} can be replaced by O($\beta^{N+1}$).  We'll mainly discuss series asymptotic to infinite order (i.e. to order $N$ for all $N=1,2,\dots$).  It is easy to see that if $f$ has an asymptotic series to infinite order, then $f$ determines all the coefficients $a_n$ uniquely.

The function $g(\beta) = 10^6 \exp(-1/10^6 \beta)$ has a zero asymptotic series.  $f(\beta)$ and $f(\beta)+g(\beta)$ thus have the same asymptotic series so an asymptotic series tells us nothing about the value, $f(\beta_0)$, for a fixed $\beta_0$.   Typically however, for $\beta_0$ small, a few terms approximate $f(\beta_0)$ well but too many terms diverge.  A good example is given \cite[Table after (15.1.18)]{CAB} for the error function $\textrm{Erfc}(x) = \tfrac{2}{\sqrt{\pi}} \int_{x}^{\infty} \exp(-y^2) dy$ for which $h(x) \equiv \pi x \exp(x^2) \textrm{Erfc}(x)$ has an asymptotic series in $1/x$ about $x = \infty$.  At $x=10$, $h(x) = .99507\dots$.  The order $N=2$ asymptotic series is good to 5 decimal places and for $N=108$ to more than 22 decimal places.  But for $N=1000$, the series is about $10^{565}$.  So it is interesting and important to know that a series is asymptotic but if one knows the series and wants to know $f$, it is disappointing not to know more.

One often considers $A(\beta)$ defined in a truncated sector $\{\beta \in \bbC \,|\, 0<|\beta| < B, |\arg \beta| <A\}$ and demands \eqref{3.22} (with $\beta^N$ in the error replaced by $|\beta|^N$) in the whole sector.

In his thesis, Kato \cite{KatoThesis} only considered $A(\beta) = A_0+\beta B$ with $A \ge 0, \,B \ge 0$ where $A(\beta)$ is self--adjoint (with a suitable interpretation of the sum).  He used what are now called Temple--Kato inequalities to obtain asymptotic series to all orders in \cite{KatoThesis, KatoAsymPT}.  We discuss this approach in Section \ref{s6} below.

About the same time, Titchmarsh started a series of papers \cite{TitPT,TitBk} on eigenvalues of second order differential equations including asymptotic perturbation results for $A(\beta) = -\tfrac{d^2}{dx^2}+V(x)+\beta W(x)$ on $L^2(\bbR,dx)$ (or $L^2((0,\infty),dx)$ with a boundary condition at $x=0$).  Typically both $V(x)$ and $W(x)$ go to infinity as $|x| \to \infty$ (so the spectra are discrete) and $W$ goes to $\infty$ faster (so analytic perturbation theory fails; think $V(x) = x^2, W(x) = x^4$).  His work relied heavily on ODE techniques.  They have overlap of applicability with Kato's operator theoretic approach, but Kato's method is more broadly applicable.

In his book, Kato totally changed his approach to be able to say something about the Banach space (and also non--self--adjoint operators in Hilbert space) so he couldn't use the Temple--Kato inequality which relies on the spectral theorem.  There is some overlap of this work from his book and work of Huet \cite{Huet}, Kramer \cite{Kra1, Kra2}, Krieger \cite{Kri} and Simon \cite{SimonAHO}.

Central to Kato's approach is the notion of strong resolvent convergence and of stability.  Kato often discusses this for sequences $A_n$ converging to $A$ in some sense as $n \to \infty$; for our purposes here, it is more natural to consider $A(\beta)$ depending on a positive real parameter as $\beta \downarrow 0$.  To avoid various technicalities, we'll also focus initially on the self--adjoint case were there are a priori bounds on $(B-z)^{-1}$ for $z \in \bbC\setminus\bbR$, although we'll consider some non--self--adjoint operators later.

For (possibly unbounded) self--adjoint $\{A(\beta)\}_{0 < \beta < B}$ and self--adjoint $A_0$, we say that $A(\beta)$ converges in \emph{strong resolvent sense} (srs) if and only if for all $z \in \bbC\setminus\bbR$, we have that $(A(\beta)-z)^{-1} \to (A_0-z)^{-1}$ in the strong (bounded) operator topology.  Here is a theorem, going back to Rellich \cite[Part 2]{RellichPT} describing some results critical for asymptotic perturbation theory:

\begin{theorem} \lb{T3.7} Let $A_0$ be self--adjoint and $\{A(\beta)\}_{0 < \beta < B}$ a family of self--adjoint operators on a Hilbert space, $\calH$.

{\rm (a)} If $\calD \subset \calH$ is a dense subspace with $\calD \subset D(A_0)$ and for all $\beta \in (0,B), \, \calD \subset D(A(\beta))$, and if $\calD$ is a core for $A_0$ and for all $\varphi \in \calD,$ we have that $A(\beta)\varphi \to A_0\varphi$ as $\beta \downarrow 0$, then $A(\beta) \to A_0$ in srs.

{\rm (b)} If $a,b \in \bbR$ are not eigenvalues of $A_0$ and $A(\beta) \to A_0$ in srs, then
\begin{equation}\label{3.23}
  P_{(a,b)}(A(\beta)) \overset{s}{\rightarrow} P_{(a,b)}(A_0)
\end{equation}
where $P_\Omega(B)$ is the spectral projection for $B$ associated to the set $\Omega \subset \bbR$ \cite[Chapter 5 and Section 7.2]{OT}
\end{theorem}

\begin{proof} (a) follows from a simple use of the second resolvent formula; see \cite[Theorem 7.2.11]{OT}.  For (b), one first proves \eqref{3.23} when $P_{(a,b)}$ is replaced by a continuous function \cite[Theorem 7.2.10]{OT} and then approximates $P_{(a,b)}$ with continuous functions \cite[Problem 7.2.5]{OT}.
\end{proof}

\begin{remark} Before leaving the subject of abstract srs results, we should mention two results known as the Trotter--Kato theorem (Kato's ultimate Trotter product formula, the subject of Section \ref{s18}, is also sometimes called the Trotter--Kato theorem).  One version says that if $A_n$ and $A$ are generators of contraction semigroups on a Banach space, $X$, then $e^{-tA_n} \overset{s}{\rightarrow} e^{-tA}$ for all $t>0$ if and only if for one (or for all) $\lambda$ with $\textrm{Re}\,(\lambda) > 0$, one has $(A_n+\lambda)^{-1} \overset{s}{\rightarrow} (A+\lambda)^{-1}$. Related, sometimes part of the statement of the theorem, is that one doesn't require $A$ to exist a priori but only that for some $\lambda$ in the open half plane that $(A_n+\lambda)^{-1}$ have a strong limit whose range is dense.  The basic theorem is then due to Trotter \cite{TrotterSmGp} in his thesis (written under the direction of Feller, whose interest in semigroups was motivated by Markov processes).  Kato's name is often on the theorem because he clarified an obscure point in this second version \cite{KatoPseudo}.  This theorem has also been called the Trotter--Kato--Neveu or Trotter--Kato--Neveu--Kurtz--Sova theorem after related contributions by these authors \cite{Kurtz1, Kurtz2, Neveu, Sova}.  There is another related result of this genre sometimes called the Trotter--Kato theorem.  It says that if $A_n$ is a family of self--adjoint operators, they have a srs limit for some $A$ if and only if $(A_n-z)^{-1}$ has a strong limit with dense range for one $z$ in $\bbC_+$ and one $z$ in $\bbC_-$.
\end{remark}

Returning to perturbation theory, Kato introduced and developed the key notion of stability.  Let $\{A(\beta)\}_{0 < \beta < B}$ (or $\beta$ in a sector) be a family of closed operators in a Banach space, $X$.  Let $A_0$ be a closed operator so that as $\beta \downarrow 0$, $A(\beta)$ converges to $A_0$ in some sense.  Let $E_0$ be an isolated, discrete, eigenvalue of $A_0$.  We say that $E_0$ is \emph{stable} if there exists $\epsilon > 0$ so that $\sigma(A_0) \cap \{z \,|\, |z-E_0| \le \epsilon\} = \{E_0\}$ and so that

(a) $|\beta| < B$ and $|z-E_0| = \epsilon \Rightarrow z \notin \sigma(A(\beta))$ and for each $\varphi \in X$
\begin{equation}\label{3.24}
  \lim_{\beta\downarrow 0} (A(\beta) - z)^{-1}\varphi = (A_0-z)^{-1}\varphi
\end{equation}
uniformly in $\{z \,|\, |z-E_0| = \epsilon\}$

(b) If $P(\beta)$ is given by \eqref{2.1} with $A=A(\beta)$ and with $\Gamma$ the counterclockwise circle indicated at the end of (a), then, for all $\beta$ small, we have that
\begin{equation}\label{3.25}
  \dim \ran\, P(\beta) = \dim \ran\, P(0)
\end{equation}

The uniform strong convergence in (a) implies that
\begin{equation}\label{3.26}
  P(\beta) \overset{s}{\rightarrow} P(0)
\end{equation}
In the self--adjoint case, even without (a), if $A(\beta) \to A_0$ in srs, then
\begin{equation}\label{3.27}
  P_{(E_0-\epsilon,E_0+\epsilon)}(A(\beta)) \overset{s}{\rightarrow} P_{\{E_0\}}(A_0)
\end{equation}
for $\epsilon$ small if $E_0$ is in the discrete spectrum of $A_0$.  $P \mapsto \dim \ran\, P$ is continuous in the topology of norm convergence but it is only lower semicontinuous in the topology of strong operator convergence.  For example, if $P_n$ is the rank one projection onto multiples of the $n$th element of an orthonormal basis, then $P_n \overset{s}{\rightarrow} 0$.  The lower semicontinuity says that
\begin{equation}\label{3.28}
  P_n \overset{s}{\rightarrow} P_\infty \Rightarrow \dim \ran\, P_\infty \le \liminf \dim \ran\, P_n
\end{equation}

Kato was well aware that equality might not hold on the right side of \eqref{3.28} for examples of relevance to physics -- a main example that he mentions is the Stark effect where the right side is infinite.  Double wells show that even if (a) above holds, (b) may fail.  Simon \cite{SimonSemiC1}  describes an extension of stability for multiple well problems.

There are two main ways that one can prove stability in cases where it is true.  One is to note that if $A(\beta) \ge A_0$ as happens if
\begin{equation}\label{3.28a}
  A(\beta) = A_0 + \beta B
\end{equation}
and $B \ge 0$, then $\dim \ran\, P_{(-\infty,a)} (A(\beta)) \le \dim \ran\, P_{(-\infty,a)} (A_0)$.  This and \eqref{3.28} implies stability for $E_0$ below the bottom of the essential spectrum for $A_0$.  This is the typical approach that Kato uses in several places.

The second way one can have stability is illustrated by

\begin{example} [Example \ref{E3.1} (revisited)] \lb{E3.8}  One might have the impression that regular perturbation theory is associated with norm continuity of resolvents and spectral projections and asymptotic perturbation theory always only strong convergence.  While there is some truth to this, Simon \cite{SimonAHO} found the surprising fact that even in situations where perturbation theory diverges, one can have norm convergence of resolvents in a sector. One starts by noting that with $p = \tfrac{1}{i}\tfrac{d}{dx}$, one has that
\begin{align*}
  (p^2+W)^2 &= p^4+W^2+p^2W+Wp^2 \\
            &= p^4+W^2 + 2pWp + [p,[p,W]] \\
            &= p^4 + W^2 + 2pWp - W'' \\
            &\ge \tfrac{1}{2}W^2 - c
\end{align*}
if $W'' \le \tfrac{1}{2}W^2+c$ and $W \ge 0$.  In this way, one sees that for positive constants $c$ and $d$
\begin{equation}\label{3.28Q}
\norm{(p^2+x^2+\beta x^4)\varphi}^2 + c\norm{\varphi}^2 \ge d \left[\norm{x^2\varphi}^2 + \beta^2 \norm{x^4\varphi}^2\right]
\end{equation}
which is called a quadratic estimate.  This, in turn, implies that $\norm{x^2(p^2+x^2+1)^{-1}}$ and $\norm{(p^2+x^2+\beta x^4+1)^{-1}x^2}$ are bounded so that
\begin{align*}
  \norm{(p^2+x^2+\beta x^4&+1)^{-1}-(p^2+x^2+1)^{-1}} \\
                          &= \beta\norm{(p^2+x^2+\beta x^4+1)^{-1}x^4(p^2+x^2+1)^{-1}} \\
                          &\le \beta \norm{(p^2+x^2+\beta x^4+1)^{-1}x^2}\norm{x^2(p^2+x^2+1)^{-1}}
\end{align*}
is O$(\beta) \rightarrow 0$ in norm.  This implies stability by a simple argument.

A similar argument works for $p^2+\gamma x^2+\beta x^4$ for any $\gamma \in \partial\bbD \setminus \{-1\}$ so using scaling and the ideas below, one proves that for each $n$, the $n$th eigenvalue, $E_n(\beta)$,  of $p^2+x^2+\beta x^4$ has an asymptotic series in each sector $\{\beta \,|\, 0<|\beta|<B_A; |\arg \beta| < A\}$ so long as $A \in (0,\tfrac{3\pi}{2})$ \cite{SimonAHO}.

The above argument doesn't work for $\beta x^{2m}; \, m > 2$ but by using that $\norm{\beta x^{2m}(p^2+x^2+\beta x^{2m}+1)^{-1}}$ is bounded, one sees that the norm of the difference of the resolvents is O$(\beta^{1/m})$ which also goes to zero.

\end{example}

To state results on asymptotic series, we focus on getting series for all orders.  Kato \cite{KatoBk} is interested mainly in first and second order, so he needs much weaker hypotheses.  Let $C \ge 1$ be a self--adjoint operator on a Hilbert space, $\calH$.  Then $D^\infty(C) \equiv \cap_{n \ge 0}D(C^n)$ is a countably normed Fr\'{e}chet space with the norms $\norm{\varphi}_n \equiv \norm{C^n\varphi}_\calH$ (see \cite[Section 6.1]{RA}).  A densely defined operator, $X$, on $D^\infty(C)$ is continuous in the Fr\'{e}chet topology if and only if for all $m$, there is $k(m)$ and $c_m$ so that $D^{k(m)}(C) \subset D(X), \, X\left[D^{k(m)}(C)\right] \subset D^m(C)$ and $\norm{X\varphi}_m \le c_m \norm{\varphi}_{k(m)}$.  Typically, for some $\ell$, $k(m)$ can be chosen to be $m+\ell$.

\begin{theorem} \lb{T3.8} Let $C \ge 1$ be a self--adjoint operator on a Hilbert space, $\calH$.  Let $\{A(\beta)\}_{0 \le \beta < B}$ be a family of closed operators with $E_0$ a simple isolated eigenvalue of $A_0 \equiv A(0)$.  Suppose that $D^\infty(C) \cap D(A_0) \subset D(A(\beta))$ for all $\beta$.  Let $V$ be an operator with $D^\infty(C) \cap D(A_0) \subset D(V)$ so that for $\varphi \in D^\infty(C) \cap D(A_0)$, we have that
\begin{equation}\label{3.29}
  A(\beta)\varphi = (A_0+\beta V)\varphi
\end{equation}
Suppose that $E_0$ is stable (in the sense that the spectrum of $A(\beta)$ for $\beta$ small is discrete near $E_0$ and that \eqref{3.25} holds) and that $V$ is a continuous map on $D^\infty(C)$ and that for some $\delta$ with $\sigma(A_0) \cap {\{z \,|\, |z-E_0| = \delta\}} = \{E_0\}$, we have that if $|z-E_0| = \delta$, then $(A_0-z)^{-1}$ is a continuous map of $D^\infty(C)$ and continuous in $z$. Suppose also that if $\varphi_0 \ne 0$ with $A_0\varphi_0 = E_0\varphi_0$, then $\varphi_0 \in D^\infty(C)$. Then, there is a sequence of complex numbers, $\{a_n\}_{n=0}^\infty$, so that the unique eigenvalue, $E(\beta)$, of $A(\beta)$ near $E_0$ is asymptotic to $E_0 + \sum_{n=1}^{\infty} a_n \beta^n$.
\end{theorem}

\begin{remarks} 1.  The proof is easy.  If $P(\beta)$ is the spectral projection for $E(\beta)$, then $P(\beta)\varphi_0 \to \varphi_0$ so for $\beta$ small
\begin{equation}\label{3.30}
  E(\beta) = \frac{\jap{\varphi_0,A(\beta)P(\beta)\varphi_0}}{\jap{\varphi_0,P(\beta)\varphi_0}}
\end{equation}
Thus, it is enough to get asymptotic series for the numerator and denominator. Write $P(\beta)$ as a contour integral and expand $(A(\beta)-z)^{-1}\varphi_0$  in a geometric series with remainder.  Since $\varphi_0 \in D^\infty(C)$, all terms including the remainder are in $\calH$.  The last factor $\norm{(A(\beta)-z)^{-1}}$ is uniformly bounded in $z$ and small $\beta$, so we get an O$(\beta^{N+1})$ error.

2.  The set of algebraic terms obtained by the above proof are the same for asymptotic and analytic perturbation theory so the $a_n$ are given by Rayleigh--Schr\"{o}dinger perturbation theory.

3. Two useful choices for $C$ are $C=A_0+1$ and $C=x^2+1$.  For $A_0=-\tfrac{d^2}{dx^2}+x^2$, there are very good estimates on $\norm{(A_0+1)^m\varphi_0}_2$ (see \cite[Section 6.4]{RA}).  If $A_0 = -\Delta+W+1$, for extremely general $W$'s, it is known that for $z \notin \sigma(A_0)$, $(A_0-z)^{-1}$ has an integral kernel with exponential decay \cite[Theorem B.7.1]{SimonSmgp}, which implies that ${\norm{(1+x^2)^m(A_0-z)^{-1}(1+x^2)^{-m}}}$ is bounded on $L^2(\bbR)$, so $(A_0-z)^{-1}$ is bounded on $D^\infty(1+x^2)$.
\end{remarks}

Asymptotic series have the virtue of uniquely determining the perturbation coefficients from the eigenvalues as functions and they often give good numeric results if $\beta$ is small and one takes only a few terms.  But mathematically, the situation is unsatisfactory -- one would like the coefficients to uniquely determine $E(\beta)$ (as they do in the regular case) or even better, one would like to have an algorithm to compute $E(\beta)$ from $\{a_n\}_{n=0}^\infty$.  This is not an issue that Kato seems to have written about but it is an important part of the picture, so we will say a little about it.

It is a theorem of Carleman \cite{CarlQA} that if $\epsilon > 0$ and $g$ is analytic in $R_{\epsilon,B} = \{z \,|\, |\arg z| < \tfrac{\pi}{2}+ \epsilon,\, 0 < |z| < B\}$, if $|g(z)| \le b_n|z|^n$ there and $\sum_{n=1}^{\infty} b_n^{-1/n} = \infty$ (e.g. $b_n = n!$), then $g \equiv 0$ on $R_{\epsilon,B}$.  This leads to a notion of strong asymptotic condition and an associated result of there being at most one function obeying that condition (and so a strong asymptotic series determines $E$) -- see Simon \cite{SimonSAC1, SimonSAC2} or Reed--Simon \cite[Section XII.4]{RS4}.

Algorithms for recovering a function from a possibly divergent series are called summability methods.  Hardy \cite{HardyDivSer} has a famous book on the subject.  Many methods, such as Abel summability (i.e. $\lim_{t \uparrow 1} \sum_{n=0}^{\infty} a_n t^n$) work only for barely divergent series like $a_n = (-1)^n$.  The series that arise in eigenvalue perturbation theory are usually badly divergent but, fortunately, there are some methods that work even in that case.  Two that have been shown to work for suitable eigenvalue problems are Pad\'{e} and Borel summability.

The ordinary approximates for a power series are by the polynomials obtained by truncating the power series.  If instead, one uses rational functions, one gets \emph{Pad\'{e}}, aka Hermite--Pad\'{e}, \emph{approximates} (they were formally introduced by Pad\'{e} \cite{Pade} in his thesis -- Hermite, who was Pad\'{e}'s advisor, introduced them earlier in the special case of the exponential function \cite{Hermite}).  Specifically, given a formal power series, $\sum_{n=0}^{\infty} a_n z^n$, the Pad\'{e} approximates, $f^{[N,M]}$, are given by
\begin{equation}\label{3.31}
  f^{[N,M]}(z)= \frac{P^{[N,M]}(z)}{Q^{[N,M]}(z)}; \quad \deg P^{[N,M]} = M, \quad \deg Q^{[N,M]} = N
\end{equation}
\begin{equation}\label{3.32}
  f^{[N,M]}(z) - \sum_{n=0}^{N+M} a_n z^n = \textrm{O}\left(z^{N+M+1}\right)
\end{equation}
In \eqref{3.32}, $f^{[N,M]}$ has $(N+1+M+1)-1$ parameters as does the sum.  Thus \eqref{3.31}/\eqref{3.32} is $(N+M+1)$ equations in the coefficients of $P$ and $Q$.  So long as certain determinants formed from $\{a_n\}_{n=0}^{N+M}$ are non--zero, there is a unique solution, $f^{[N,M]}(z)$.  For more on Pad\'{e} approximates, see Baker \cite{Baker, BakerBk, BakerGamel}.

The other method is called \emph{Borel summability}, introduced by Borel \cite{Borel}.  The method requires that
\begin{equation}\label{3.33}
  |a_n| \le AB^n n!
\end{equation}
for some $A, B$ and all $n$.  If that is so, one forms the Borel transform
\begin{equation}\label{3.34}
  g(w) = \sum_{n=0}^{\infty} \frac{a_n}{n!} w^n
\end{equation}
which defines an analytic function in $\{w \,|\, |w| < B^{-1}\}$.  One supposes that $g$ has an analytic continuation to a neighborhood of $[0,\infty)$ and defines for $z$ real and positive
\begin{equation}\label{3.35}
  f(z) = \int_{0}^{\infty} e^{-a} g(az) da
\end{equation}
Since $\int_{0}^{\infty} e^{-a} a^n da = n!$, formally $f(z)$ is $\sum_{n=0}^{\infty} a_n z^n$.  For this method to work, $g$ has to have an analytic continuation so that the integral in \eqref{3.35} converges.

As far as Pad\'{e} is concerned, a major result involves sequences, $\{a_n\}_{n=0}^\infty$, called \emph{series of Stieltjes} which have the form
\begin{equation}\label{3.35A}
  a_n = (-1)^n \int_{0}^{\infty} x^n d\mu (x)
\end{equation}
for some positive measure $d\mu$ on $[0,\infty)$ with all moments finite.  The associated Stieltjes transform of $\mu$
\begin{equation}\label{3.36}
  f(z) = \int_{0}^{\infty} \frac{d\mu (x)}{1+xz}
\end{equation}
is defined and analytic in $z \in \bbC\setminus (-\infty,0]$.  Expanding $(1+xz)^{-1}$ in a geometric series with remainder, one sees that in every sector $\{z \,|\, |\arg z| < \pi - \epsilon\}$ with $\epsilon > 0$, $\sum_{0}^{\infty} a_n z^n$ is an asymptotic series for $f$.  Here is the big theorem for such series:

\begin{theorem} \lb{T3.9}  If $\{a_n\}_{n=0}^\infty$ is a series of Stieltjes, then for each $j \in \bbZ$, the diagonal Pad\'{e} approximates, $f^{[N,N+j]}(z)$, converge as $N \to \infty$ for all $z \in \bbC\setminus [0,\infty)$ to a function $f_j(z)$ given by \eqref{3.36} with $\mu$ replaced by $\mu_j$ which obeys \eqref{3.35A} (with $\mu=\mu_j$).  The $f_j$ are either all equal or all different depending on whether \eqref{3.35A} has a unique solution, $\mu$, or not.
\end{theorem}

The result is due to Stieltjes \cite{Stie} who discussed solutions of the moment problem \eqref{3.35} but not Pad\'{e} approximates.  Rather following ideas of Jacobi, Chebyshev and Markov, he discussed continued fractions expansions
\begin{equation*}
  \cfrac{\alpha_1}{z+\beta_1+\cfrac{\alpha_2}{z+\beta_3+ \cfrac{\alpha_3}{\ddots}}}
\end{equation*}
for the Stieltjes transform.  These are the $f^{[N+1,N]}(z)$ and his convergence results imply the theorem.  For details, see Baker \cite{BakerBk} or Simon \cite[Section 7.7]{OT}.

It follows from results of Loeffel et al \cite{LM,LMSW} that if $E_m(\beta)$ is an eigenvalue of $p^2+x^2+\beta x^4$ for $\beta \in [0,\infty)$, then $E_m(\beta)$ has an analytic continuation to $\bbC\setminus [0,\infty)$ with a positive imaginary part in the upper half plane.  Results of Simon \cite{SimonAHO} imply that $|E_m(\beta)| \le C(1+|\beta|)^{1/3}$.  A Cauchy integral formula then implies that $(E_m(0)-E_m(\beta))/\beta$ has a representation of the form \eqref{3.36}.  Thus, by \cite{LMSW}, the diagonal Pad\'{e} approximates converge.  Moreover, it is a fact (related to the above mentioned theorem of Carleman) that if $\{a_n\}_{n=0}^\infty$ is the set of moments of a measure on $[0,\infty)$ with $|a_n| \le CD^n(kn)!$ with $k \le 2$, then the solution to the moment problem is unique \cite[Problem 5.6.2]{RA}.  This implies that for the $x^4$ anharmonic oscillator, the diagonal Pad\'{e} approximates converge to the eigenvalues.  The same is true for the $x^6$ oscillator but for the $x^8$ oscillator, it is known (Graffi-Grecchi \cite{GGx8}) that, while the diagonal Pad\'{e} approximates converge, they have different limits and none is the actual eigenvalue!

The key convergence result for Borel sums is a theorem of Watson \cite{Watson}; see Hardy \cite{HardyDivSer} for a proof:

\begin{theorem} \lb{T3.10} Let $\Theta \in \left(\tfrac{\pi}{2},\tfrac{3\pi}{2}\right)$ and $B > 0$.  Define
\begin{align}
  \Omega &= \{z\,|\,0 < |z| < B, |\arg z| < \Theta\} \lb{3.37} \\
  \wti{\Omega} &= \{z\,|\,0 < |z| < B, |\arg z| < \Theta - \tfrac{\pi}{2} \} \lb{3.38} \\
  \Lambda &= \{w\,|\,w \ne 0, |\arg w| < \Theta - \tfrac{\pi}{2} \} \lb{3.39}
\end{align}
Suppose that $\{a_n\}_{n=0}^\infty$ is given and that $f$ is analytic in $\Omega$ and obeys
\begin{equation}\label{3.40}
  \left|f(z) - \sum_{n=0}^{N} a_n z^n\right| \le A C^{N+1} (N+1)!
\end{equation}
on $\Omega$ for all $N$.  Define
\begin{equation}\label{3.41}
  g(w) = \sum_{n=0}^{\infty} \frac{a_n}{n!} w^n; \qquad |w| < C^{-1}
\end{equation}
Then $g(w)$ has an analytic continuation to $\Lambda$ and for all $z \in \wti{\Omega}$, we have that
\begin{equation}\label{3.42}
  f(z) =  \int_{0}^{\infty} e^{-a} g(az) da
\end{equation}
\end{theorem}

Graffi--Grecchi--Simon \cite{GGS} proved that this theorem is applicable to the $x^4$ anharmonic oscillator.  They did numeric calculations making an unjustified use of Pad\'{e} approximation to analytically continue $g$ to all of $[0,\infty)$ and found more rapid convergence than Pad\'{e} on the original series. By conformally mapping a subset of the union of $\bbD$ and $\Lambda$ containing $[0,\infty)$ onto the disk, one can do the analytic continuation by summing a mapped power series and so do numerics without an unjustified Pad\'{e}; see Hirsbrunner and Loeffel \cite{HL}

There is a higher order Borel summation where one picks $m=2,3,\dots$, $\Theta \in \left(\tfrac{m\pi}{2},\tfrac{3m\pi}{2}\right)$ and replaces $\Theta-\tfrac{\pi}{2}$ in \eqref{3.38} by $\Theta - \tfrac{m\pi}{2}$, $(N+1)!$ in \eqref{3.40} is replaced by $[m(N+1)]!$,  $n!$ in \eqref{3.41} by $(mn)!$ and \eqref{3.42} by
\begin{equation}\label{3.43}
  f(z) = \int_{0}^{\infty} e^{-a^{1/m}} g(za) a^{\left(\tfrac{1}{m} - 1\right)} da
\end{equation}
They showed \cite{GGS} that the $x^{2(m+1)}$ oscillator is modified $m$--Borel summable.

Avron--Herbst--Simon \cite[Part III]{AvronHerbstSimon} proved that for the Zeeman effect in arbitrary atoms, the perturbation series of the discrete eigenvalues is Borel summable.  The Schwinger functions of various quantum field theories have been proven to have Borel summable Feynman perturbation series: $P(\phi)_2$ \cite{EMS}, $\phi^4_3$ \cite{MSPhi43}, $Y_2$ \cite{Renou}, $Y_3$ \cite{MSY3}.

In general, Pad\'{e} summability is hard to prove because it requires global information, so it has been proven to work only in very limited situations (for example a higher dimensional quartic anharmonic oscillator is known to be Borel summable but nothing about Pad\'{e} is known).  Clearly, when it can be proven, Borel summability is an important improvement over the mere asymptotic series that concerned Kato.

Before leaving asymptotic perturbation theory, we mention a striking example of Herbst--Simon \cite{HerbS2}
\begin{equation*}
  A(\beta) = -\frac{d^2}{dx^2}+x^2-1+\beta^2 x^4 + 2\beta x^3 -2 \beta x
\end{equation*}

If $E_0(\beta)$ is the lowest eigenvalue, they prove that for all small, non--zero positive $\beta$
\begin{equation*}
  0 < E_0(\beta) < C \exp(-D\beta^{-2})
\end{equation*}
Thus $E_0(\beta)$ has $\sum_{n=0}^{\infty} a_n \beta^n$ as asymptotic series where $a_n \equiv 0$.  The asymptotic series converges but, since $E_0$ is strictly positive, it converges to the wrong answer!

%%%%%%%%%%%%%%%%%%%%%%%%%%%%%%%%%%%%%%%%%%%%%%%%%%%%%%%%%%%%%%
\section{Eigenvalue Perturbation Theory, III: Spectral Concentration} \lb{s4}
%%%%%%%%%%%%%%%%%%%%%%%%%%%%%%%%%%%%%%%%%%%%%%%%%%%%%%%%%%%%%%

Starting around 1950, Kato \cite{KatoThesis} and Titchmarsh \cite{TitPT, TitBk} considered what the perturbation series might mean for a problem like the Stark problem where a discrete eigenvalue is swallowed by continuous spectrum as soon as the perturbation is turned on.  Titchmarsh looked mainly at ODEs; in particular, he looked at what has come to be called the Titchmarsh problem, ($g \ge -\tfrac{1}{4}, z > 0$)
\begin{equation}\label{4.1}
  h(g,z,f) = -\frac{d^2}{dx^2}+\frac{g}{x^2}-\frac{z}{x}-fx
\end{equation}
(for some values of $g$, one needs a boundary condition at $x=0$).  Kato used operator theory techniques and studied Examples \ref{E3.3} and \ref{E3.5}.

Titchmarsh proved that the Green's kernel for $h$, originally defined for energies in $\bbC_+$, had a continuation onto the lower half plane with a pole near the discrete eigenvalues of $h(g,z,f)$ and he identified the real part of the pole with perturbation theory up to second order. He conjectured that the imaginary part of the pole was exponentially small in $1/f$.  He then showed in a certain sense that the spectrum of $h(g,z,f \ne 0)$ as $f \downarrow 0$ concentrated near the real parts of his poles \cite[Part V]{TitPT}.

Kato discussed things in terms of what he called pseudo--eigenvalues and pseudo--eigenvectors.  He later realized that these notions imply a concentration of spectrum like that used by Titchmarsh.  In his book \cite{KatoBk}, he emphasized what he formally defined as spectral concentration and linked the two approaches.  In this section, I'll begin by defining spectral concentration and then prove, following Kato, that it is implied by the existence of pseudo--eigenvectors.  Finally, I'll discuss the complex scaling theory of resonances and how it extends and illuminates the theory of spectral concentration.

Consider first the case where $A(\beta)$ converges to $A_0$ as $\beta \downarrow 0$ in srs and $E_0$ is a discrete simple eigenvalue of $A_0$.  Let $T$ be a closed interval with $\sigma(A_0) \cap T = \{E_0\}$.  By Theorem \ref{T3.7}, for any $\epsilon > 0$, we have that $P_{T\setminus (E_0-\epsilon,E_0+\epsilon)} (A(\beta)) \overset{s}{\rightarrow} 0$.  Thus, in a sense, the spectrum of $A(\beta)$ in $T$ is concentrated near $E_0$.  In the above, if we could replace $(E_0-\epsilon,E_0+\epsilon)$ by $(E_0+a_1\beta-\beta^{3/2},E_0+a_1\beta + \beta^{3/2})$, we'd be able to claim that the spectrum was concentrated near $E_0+a_1\beta$ in a way that would determine $a_1$.

Taking into account that we may want to also have $T$ shrink in cases like Example \ref{E3.2}, we make the following definition.  Let $T(\beta), S(\beta)$ be Borel sets in $\bbR$ given for $0 < \beta < B$ so that if $0 < \beta' < \beta$, then $T(\beta') \subset T(\beta), S(\beta') \subset S(\beta)$ and so that for all $\beta$, $S(\beta) \subset T(\beta)$.  We say that the spectrum of $A(\beta)$ in $T(\beta)$ is \emph{asymptotically concentrated} in $S(\beta)$ if and only if $P_{T(\beta)\setminus S(\beta)}  \overset{s}{\rightarrow} 0$.

If $E_0$ is a simple eigenvalue of $A_0$ and $\{a_j\}_{j=1}^N$ are real numbers, we say the spectrum near $E_0$ is asymptotically concentrated near $E_0+\sum_{j=0}^{N} a_j \beta^j$ if there exist positive functions $f$ and $g$ obeying $f(\beta) \to 0,\, f(\beta)/\beta \to \infty,\, g(\beta)/\beta^N \to 0$ as $\beta \downarrow 0$ so that the spectrum of $A(\beta)$ in $(E_0-f(\beta),E_0+f(\beta))$ is asymptotically concentrated in $(E_0+\sum_{j=0}^{N} a_j \beta^j -g(\beta), E_0+\sum_{j=0}^{N} a_j \beta^j + g(\beta))$.  It is easy to see if that happens, it determines the $a_j,\, j=1,\dots,n$.

Kato's thesis \cite{KatoThesis} introduced the notion of $N$th order \emph{pseudo--eigenvectors} and \emph{pseudo--eigenvalues}.  In later usage, this is a pair of functions, $\varphi(\beta)$ and $\lambda(\beta)$, on $(0,B)$ with values in $\calH$ and $\bbR$ so that
\begin{equation}\label{4.2}
  \varphi(\beta) \in D(A(\beta)), \qquad \norm{\varphi(\beta)}=1, \qquad \lambda(\beta) \to E_0
\end{equation}
\begin{equation}\label{4.3}
  \norm{(A(\beta)-\lambda(\beta))\varphi(\beta)} = \textrm{o}(\beta^N)
\end{equation}

Conley--Rejto \cite{CRejto} and Riddell \cite{Ridd} (Riddell was a student of Kato and this paper was based on his PhD. thesis) proved the following
\begin{theorem} \lb{T4.1} If $E_0$ is a simple isolated eigenvalue of $A_0$ and $(\varphi(\beta),\lambda(\beta))$ are an $N$th order pseudo--eigenvector and pseudo--eigenvalue so that as $\beta \downarrow 0$, we have that
\begin{equation}\label{4.3A}
  (1-P_{E_0}(A_0)) \varphi(\beta) \to 0
\end{equation}
Then there exists $g(\beta) = \textrm{o}(\beta^N)$ and $d>0$ so that the spectrum of $A(\beta)$ in $(E_0-d,E_0+d)$ is concentrated in $(\lambda(\beta)-g(\beta),\lambda(\beta)+g(\beta))$.
\end{theorem}

\begin{remarks} 1.  Riddell also has a converse.

2. Both papers consider the situation where $E_0$ has multiplicity $k < \infty$ and there are $k$ orthonormal pairs obeying \eqref{4.3} and they prove spectral concentration on a union of $k$ intervals of size $\textrm{o}(\beta^N)$ about the $\lambda_j$.

3.  The proof isn't hard.  One picks $g(\beta) = \textrm{o}(\beta^N)$ so that $\norm{(A(\beta)-\lambda(\beta))\varphi(\beta)}/g(\beta) \to 0$.  This implies that if $Q(\beta) = P_{(\lambda(\beta)-g(\beta),\lambda(\beta)+g(\beta))}(A(\beta))$, then $\norm{(1-Q(\beta))\varphi(\beta)} \to 0$.  By \eqref{4.3A}, this implies that
\begin{equation}\label{4.5}
  \norm{Q(\beta)-P_{E_0}(A_0)} \to 0
\end{equation}
If $d < \dist(E_0,\sigma(A)\setminus\{E_0\})$, Theorem \ref{T3.7} implies that $P_{(E_0-d,E_0+d)}(A(\beta))\psi \to P_{E_0}(A_0)\psi$ for any $\psi$.  Thus by \eqref{4.5}, $\left[P_{(E_0-d,E_0+d)}(A(\beta))-Q(\beta)\right]\psi \to 0$ which is the required spectral concentration
\end{remarks}

These ideas were used by Friedrichs and Rejto \cite{FRejto} to prove spectral concentration in Example \ref{E3.5}  (i.e. $A_0$ of rank 1 and $B$ multiplication by $x$).  They assumed the function $\psi(x)$ of \eqref{3.15} is strictly positive on $\bbR$ and H\"{o}lder continuous and prove that $A(\beta)$ has no point eigenvalues and has a weak spectral concentration (of order $\beta^p$ for some $0<p<1$).  Riddell \cite{Ridd} proved spectral concentration to all orders for the Stark effect for Hydrogen using pseudo--eigenvectors and Rejto \cite{Rejto1, Rejto2} proved the analog for Helium (see below for more on spectral concentration for the Stark effect).

Veseli\'{c} \cite{Ves} systematized and simplified the results in Theorem \ref{T4.1} and applied it to certain models (not linear in $\beta$) where $A_0$ has a discrete eigenvalue while $A(\beta)$ has no eigenvalue due to tunnelling through a barrier.  An example is $A(\beta) = -\tfrac{d^2}{dx^2}+V(x,\beta)$ where
\begin{equation}\label{4.6}
  V(x,\beta) = V_0(x)-(1-e^{-\beta x})
\end{equation}
$V_0$ goes to zero at infinity and is such that $A_0$ has a single negative eigenvalue at $-\tfrac{1}{2}$.  Thus $A(\beta)$ has essential spectrum $[-1,\infty)$ and instantaneously the discrete eigenvalue is swamped in continuous spectrum.  There is a barrier of size $\beta^{-1}$ trapping the initial bound state.  Veseli\'{c} proved spectral concentration.

As noted Titchmarsh related spectral concentration to second sheet poles of Green's functions for certain differential operators.  This theme was developed by James Howland, a student of Kato, in 5 papers \cite{Howland3, Howland5, Howland4, Howland1, Howland2}. Howland discussed two situations.  One was where $A_0$ was finite rank and whose non--zero eigenvalues are washed away much like Example \ref{E3.5}.  The other was where $A_0$ has eigenvalues embedded in continuous spectrum and $B$ is finite rank, so related to the Friedrichs model mentioned at the end of Example \ref{E3.2}.

In both cases, there is a finite dimensional space, $\calV$, where the finite rank operator lives and Howland considered ${\{\jap{\varphi,(A(\beta)-z)^{-1}\psi},|\, \varphi,\psi \in \calV\}}$ and proved (under suitable conditions) that these functions initially defined on $\bbC_+$ have meromorphic continuations through $\bbR$ into a neighborhood of $E_0$, a finite multiplicity eigenvalue of $A_0$.  These continuations had second sheet poles at $E_j(\beta)$ converging as $\beta\downarrow 0$ to $E_0$.  The number of poles is typically the multiplicity of $E_0$ as an eigenvalue of $A_0$.

In the case where $A_0$ has a discrete eigenvalue, Howland showed that $\textrm{Im}\,E(\beta) = \textrm{O}(\beta^\ell)$ for all $\ell$ and was able to use this to prove spectral concentration to all orders.  But in cases where $A_0$ had an embedded eigenvalue, it was typically true that $\textrm{Im}\, E(\beta) = a_k \beta^k + \textrm{o}(\beta^k)$ for some $k$ and some $a_k < 0$; indeed Howland often proved a Fermi golden rule with $a_2 \ne 0$.  In that case, he showed there was spectral concentration of order $k-1$ but not $k$ so spectral concentration couldn't specify a perturbation series to all orders.

Howland also discovered that even when $A_0$ and $B$ were self--adjoint, an eigenvalue could turn into a second order pole whose perturbation series could have non-trivial fractional power series in the asymptotic expression, i.e. Rellich's theorem fails for resonance energies.

Howland also introduced what I've called Howland's razor (see the discussion of Example \ref{E3.3}) and he gave one possible answer: it often happened that the embedded eigenvalue turned into a resonance, i.e. second sheet pole, for real values of $\beta$ but for suitable complex $\beta$, it was a pole in $\bbC_+$ and so a normal discrete eigenvalue of $A(\beta)$.  Thus the resonance energy could be interpreted as the analytic continuation of a perturbed eigenvalue.

Perhaps the most successful approach to the study of resonances, one that handles problems in atomic physics like Examples \ref{E3.2}} and \ref{E3.3}, is the method of complex scaling, initially called dilation or dilatation analyticity (the name change to complex scaling was by quantum chemists when they took up the method for numerical calculation of molecular resonances).  The idea appeared initially in a technical appendix of a never published note by  J. M. Combes who realized the potential of this idea and then published papers with coauthors: Aguilar--Combes \cite{AC} on the two body problem and Balslev--Combes \cite{BalCom} on $N$--body problems (Eric Balslev was Kato's first Berkeley student); see Simon \cite{SimonBC} for extensions and simplifications and \cite[Sections XIII.10 and XII.6]{RS4} for a textbook presentation.  Combes and collaborators knew that the formalism, which they used to prove the absence of singular continuous spectrum, provided a possible definition of a resonance.  It was Simon \cite{SimonTDPT} who realized that the formalism was ideal for studying eigenvalues embedded in the continuous spectrum like autoionizing states.  We will not discuss an extension needed for molecules in the limit of infinite nuclear masses where one uses exterior complex scaling or a close variant, see Simon \cite{SimonExtCS}, Hunziker \cite{HunCS} and G\'{e}rard \cite{GerCS}.

We begin with the two body case.  On $L^2(\bbR^\nu,d^\nu x)$, let $U(\theta),\,\theta\in\bbR$ be the set of real scalings:
\begin{equation}\label{4.7}
  (U(\theta)f)(\boldsymbol{r}) = e^{\nu\theta/2}f(e^\theta \boldsymbol{r})
\end{equation}
which defines a unitary group.  If $H = -\Delta+V(\boldsymbol{r})$, then
\begin{equation}\label{4.8}
   H(\theta) \equiv U(\theta)HU(\theta)^{-1} = -e^{-2\theta}\Delta+V(e^\theta \boldsymbol{r})
\end{equation}
The first term, $H_0(\theta)$, can be analytically continued and
\begin{equation}\label{4.9A}
  \sigma(H_0(\theta)) = \{z \in \bbC\setminus\{0\}\,|\, \arg z = -2\textrm{Im}\,\theta\} \cup \{0\} \equiv S_\theta
\end{equation}

Suppose that $\theta \mapsto V(e^\theta\boldsymbol{r})$ has an analytic continuation as a compact operator from $D(-\Delta)$ to $L^2(\bbR^\nu)$ for $|\textrm{Im}\,\theta| < \Theta_0$ as happens for $V(\boldsymbol{r})=r^{-\alpha}\,(0<\alpha<2$; including $\alpha = 1$, i.e. Coulomb) for all $\Theta_0$ or for $V(\boldsymbol{r})=e^{-\gamma r}$ for $\Theta_0=\tfrac{\pi}{2}$.  Such $V$'s are called dilation analytic.  Then $H(\theta)$ is a type (A) analytic family on the strip of width $2\Theta_0$ about $\bbR$.  For any $\theta$, the essential spectrum of $H(\theta)$ is $S_\theta$.

Discrete eigenvalues are given by analytic functions, $E_j(\theta)$.  Since changing $\textrm{Re}\,\theta$ provides unitarily  equivalent $H$'s, $E_j(\theta)$ is constant under changes of $\textrm{Re}\,\theta$, so constant by analyticity.  We conclude that so long as discrete eigenvalues avoid  $S_\theta$, they remain discrete eigenvalues of $H(\theta)$.  In particular, negative eigenvalues of $H$ are eigenvalues of $H(\theta)$ if $|\textrm{Im}\,\theta| < \tfrac{\pi}{2}$.  An additional argument shows that embedded positive eigenvalues become discrete eigenvalues of $H(\theta)$ for $\textrm{Im}\,\theta \in (0,\tfrac{\pi}{2})$.

By this persistence, $H(\theta)$ for $\theta$ with $\textrm{Im}\,\theta \in (0,\tfrac{\pi}{2})$, there can't be any eigenvalues in $\{z\,|\, \arg z \in (0,2\pi-2\textrm{Im}\,\theta)\setminus\{-\pi\}\}$ (for taking $\theta$ back to zero would result in  non--real eigenvalues of $H$) but there isn't any reason there can't be for $z$ with $\arg z \in (-2\textrm{Im}\,z,0)$.  That is, moving $\textrm{Im}\,\theta$ can uncover eigenvalues in $\bbC_-$ which we interpret as resonances (but see the discussion below).

Using techniques from $N$--body quantum theory (essentially the HVZ theorem to be discussed in Section \ref{s11}; we'll use notation from that section below), one can similarly analyze $N$--body Hamiltonians with center of mass removed when all the $V_{ij}$ are dilation analytic.  The spectrum of $H(\theta)$ with $\theta$ not real looks like that in Figure 1.

\begin{figure}[h]
\includegraphics[scale=0.4,clip=true]{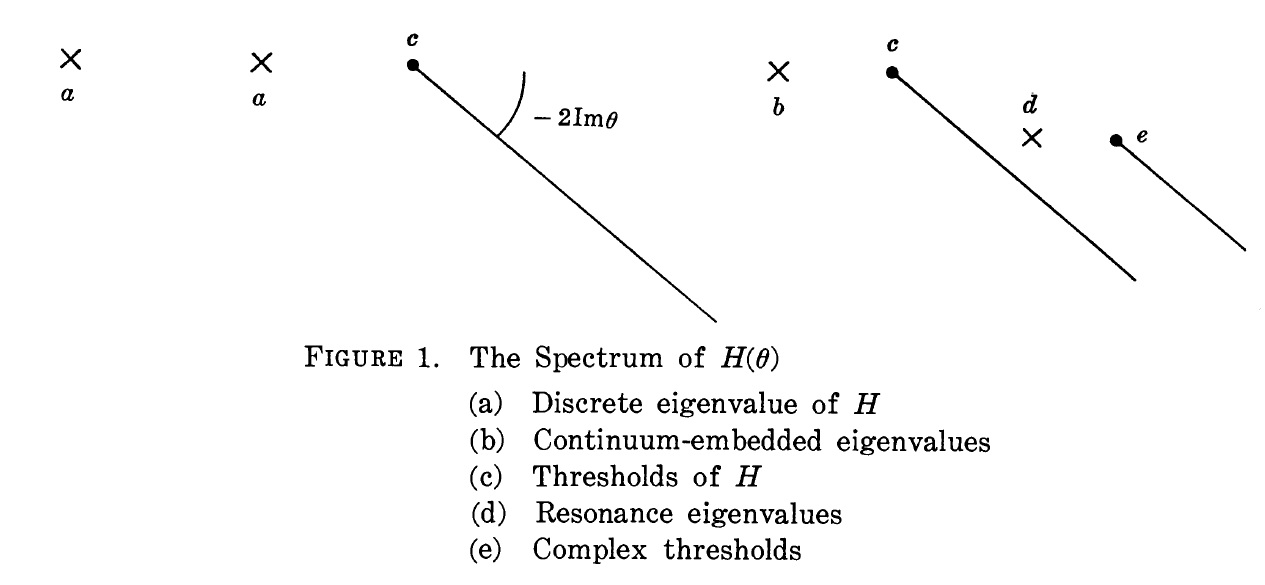}
\end{figure}

If $\calC$ is a non-trivial cluster decomposition of $\{1,\dots,N\}$, $\calC=\{C_1,\dots,C_k\}$ and $h(C_j)$ is the internal Hamiltonian of $C_j$, the set of $E_1+\dots+E_k$ where $E_j$ is an eigenvalue of $h(C_j)$ is called the set of thresholds (if some $C_\ell$ has one particle, then $h(C_\ell)$ is the zero operator on $\bbC$ and $E_\ell=0$).  It can be shown \cite{BalCom, SimonBC} that the set, $\Sigma$, of all thresholds (running over all non--trivial cluster decompositions) is a closed countable set and that for $0 < \textrm{Im}\,\theta<\Theta_0<\tfrac{\pi}{2}$, one has that
\begin{equation}\label{4.9}
  \sigma_{\rm{ess}}(H(\theta)) = \bigcup_{\lambda \in \Sigma(\theta)} \lambda+S_\theta
\end{equation}
Here $\Sigma(\theta)$ includes some complex $\lambda$ where the $E_j$ are resonance eigenvalues of $h(C_j,\theta)$.

\noindent \textbf{Example \ref{E3.2} revisited}. (following \cite{SimonTDPT}) The thresholds are $\left\{ -\tfrac{1}{4n^2}\right\}_{n=1}^\infty$ so the eigenvalue at $E_{2,2} = -\tfrac{1}{8}$ is not a threshold.  Thus it is an isolated eigenvalue of $A(1/Z,0,\theta)$ if $-i\theta \in (0,\tfrac{\pi}{2})$.  It follows that the Kato--Rellich theory applies so, for $1/Z$ small, there is an eigenvalue, $E_{2,2}(1/Z,\theta)$ independent of $\theta$ (although it is only an eigenvalue if $-\arg(E_{2,2}(1/Z)+\tfrac{1}{4}) < \textrm{Im}\,\theta$.  This first implies there is a convergent perturbation series (i.e. time--dependent perturbation theory, suitably defined, converges).  One can compute the perturbation coefficients which are $\theta$ independent for $-i\theta \in (0,\tfrac{\pi}{2})$ and then take $-i\theta$ to $0$.  One gets a suitable limit of $-(V\varphi,SV\varphi)$ where $S$ is a reduced resolvent.  Using the fact that the distribution limit of $1/(x+i\epsilon)$ is $\calP\left(\tfrac{1}{x}\right)-i\pi \delta(x)$, Simon \cite{SimonTDPT} computed $\textrm{Im}\, a_2$ as given by the Fermi golden rule.

For Stark Hamiltonians, the initial belief among mathematical physicists was that complex scaling couldn't work.  For let
\begin{equation}\label{4.10}
  H_0(\theta,F)=-e^{-2\theta}\Delta+Fe^\theta z
\end{equation}
on $L^2(\bbR^3)$. Since $H_0(\theta=0,F \ne 0)$ has no threshold (translating $z$ by a constant, adds a constant to the energy), there is no place for the spectrum $(-\infty,\infty)$ to go when $\theta$ is made imaginary.  So it was assumed the theory could not make sense.

In spite of this accepted wisdom, a quantum chemist, Bill Reinhardt, did calculations for the Stark problem using complex scaling \cite{Rein} and got sensible results.  Motivated by this, Herbst \cite{HerbstStark2} was able to define complex scaling for a class of two body Hamiltonians including the Hydrogen Stark problem.  He discovered that for $F \ne 0$, and $0 < \arg \theta < \pi/3$, $H_0(\theta,F)$ has empty spectrum (!), i.e. $(H_0(\theta,F) - z)$ is invertible for all $z$.  It is a theorem that elements in Banach algebras and, in particular, bounded operators on any Banach space, have non--empty spectrum but that is only for bounded operators.  In some sense, $H_0(\theta,F)$ has only $\infty$ in its spectrum -- specifically $\sigma[(H_0(\theta,F)-z)^{-1}] = \{0\}$ for all $z$.

\noindent \textbf{Example \ref{E3.3} revisited} With this in hand, Herbst \cite{HerbstStark2} considered \eqref{3.8} and defined $A(F,Z,\theta)$ by
\begin{equation}\label{4.11}
  A(F,Z,\theta) = -e^{-2\theta}\Delta - e^{-\theta}\frac{Z}{r} + e^\theta Fz
\end{equation}
and proved that for $0 < -i\theta<\tfrac{\pi}{3}$, and $F \ne 0$, $A(F,Z,\theta)$ has purely discrete spectrum and if $E_0 \in (-\infty,0)$ is an eigenvalue of $A(F=0,Z,\theta=0)$ of multiplicity $k$, then for $F$ small and $-i\theta \in (0.\pi/3)$, $A(F,Z,\theta)$ has at most $k$ eigenvalues near $E_0$ and their combined multiplicities is $k$.  The Rayleigh--Schr\"{o}dinger series can be proven to be asymptotic by the method of Theorem \ref{T3.8}.  Since its coefficients are real, Herbst showed that the width, $\Gamma(F)$, is $\textrm{o}(F^\ell)$ for all $\ell$ and, by Howland's method, this provided another proof of spectral concentration for all orders for the Stark problem.

Herbst--Simon \cite{HerbS} studied the analytic properties of $E(F,Z,\theta)$ and proved analyticity for $-F^2 \in \{z \,|\, |z| < R\} \cap (\bbC\setminus (-\infty,0])$ and used this to prove Borel summability that recovers $E(F,Z,\theta)$ directly for $\textrm{Re}\,(-F^2)>0$ (which doesn't include any real $F$).  The physical value is then determined by analytic continuation.  Graffi--Grecchi \cite{GG1} had proven Borel summability slightly earlier using very different methods.  Graffi--Grecchi \cite{GG8} and Herbst--Simon \cite{HerbS} also proved Borel summability for discrete eigenvalues of general atoms.

For Hydrogen, Herbst--Simon conjectured \eqref{3.11} noting that it was implied by their analyticity results and the then unproven Oppenheimer formula.  Shortly thereafter, Harrell--Simon \cite{HS} proved the Oppenheimer formula for the complex scaled defined Stark resonance and so also \eqref{3.11}.  They used similar arguments to prove the Bender--Wu formula for the anharmonic oscillator.  Later Helffer-Sj\"{o}strand \cite{HelfS2} proved Bender--Wu formulae for higher dimensional oscillators.

We have not discussed in detail various subtleties that are dealt with in the quoted papers: among them, Herbst \cite{HerbstStark2} showed that $A(F,Z,\theta)$ is of type(A) with domain $D(-\Delta)\cap D(z)$ on $\{(F,Z,\theta) \,|\, F>0, \textrm{Im}\,\theta \in (0,\pi/3)\}$ by proving a quadratic estimate.  The proof of stability of the eigenvalues of $A(F=0,Z,\theta)$ for $\textrm{Im}\,\theta \in (0,\pi/3)$ uses ideas from \cite[Part I]{AvronHerbstSimon}.  While the free Stark problem has scaled Hamiltonians with empty spectrum when there is one positive charge and $N$ particles of equal mass and equal negative charge, there are charges and masses, where the spectrum is not empty.

Sigal \cite{SigStark1, SigStark2, SigStark3, SigStark4} and Herbst--M{\o}ller--Skibsted \cite{HMSkib} have further studied Stark resonances in multi--electron atoms proving that the widths are strictly positive and exponentially small in $1/F$.

We end this discussion by noting that I have reason to believe that, at least at one time, Kato had severe doubts about the physical relevance of the complex scaling approach to resonances.  \cite{HS} was rejected by the first journal it was submitted to.  The editor told me that the world's recognized greatest expert on perturbation theory had recommended rejection so he had no choice.  I had some of the report quoted to me.  The referee said that the complex scaling definition of resonance was arbitrary and physically unmotivated with limited significance.

There is at least one missing point in a reply to this criticism: however it is defined, a resonance must correspond to a pole of the scattering amplitude.  While this is surely true for resonances defined via complex scaling, as of this day, it has not been proven for the models of greatest interest.  So far, resonance poles of scattering amplitudes in quantum systems have only been proven for two and three cluster scattering with potentials decaying faster (often much faster) than Coulomb and not for Stark scattering; see Babbitt--Balslev \cite{BB}, Balslev \cite{Bal1, Bal2, Bal3}, Hagedorn \cite{Hag}, Jensen \cite{JenScatt} and Sigal \cite{SigRes1, SigRes2}.  This is a technically difficult problem which hasn't drawn much attention.  That said, following \cite{HS} and others, we note the following in support of the notion that eigenvalues of $H(\theta)$ that lie in $\bbC_-$ are resonances:

(1) Going back to Titchmarsh \cite{TitPT, TitBk}, poles of the diagonal (i.e. $x=y$) Green's function (integral kernel, $G(x,y;z)$ of $(H-z)^{-1})$ are viewed as resonances for one dimensional problems.  In dimension $\nu \ge 2$, $G(x,y;z)$ diverges as $x \to y$ so it is natural to consider poles of $\jap{\varphi,(H-z)^{-1}\varphi}$.  Howland's razor implies that you can't look at all $\varphi \in L^2(\bbR^\nu, d^\nu x)$ but a special class of functions which are smooth in $x$ and $p$ space would be a reasonable replacement for $x=y$.  One can show (see \cite[Section XIII.10]{RS4}) that if $\varphi$ is a polynomial times a Gaussian, then $\jap{\varphi,(H-z)^{-1}\varphi}$ has a meromorphic continuation across $\bbR$ between thresholds with poles exactly at the eigenvalues of $H(\theta)$.

(2) In the autoionizing case, $E$ is an analytic function of $1/Z$ and in the Stark case, analytic for $-F^2$ in a cut disk about $0$.  For the physically relevant values, $1/Z$ real or $F$ real, $E$ has $\textrm{Im}\,E < 0$ and these resonances are on the second sheet and disappear at $\theta=0$. But for $1/Z$ or $F$ pure imaginary, the corresponding $E$ is in $\bbC_+$ and so persists when $\textrm{Im}\,\theta \downarrow 0$, i.e. $E$ for these unphysical values of the parameters is an eigenvalue of these corresponding $H$.  Thus resonances can be viewed as analytic continuations of actual eigenvalues from unphysical to physical values of the parameters.

(3)  It is connected to the sum or Borel sum of a suitable perturbation series, see \cite{GG6A, GG7}.

(4) It yields information on asymptotic series and spectral concentration in a particularly clean way and, in particular, a proof of a Bender--Wu type formula for the asymptotics of the perturbation coefficients in the Stark problem.

While we've focused on the complex scaling approach to resonances, there are other methods.  One called distortion analyticity works sometimes for potentials which are the sum of a dilation analytic potential and a potential with exponential decay (but not necessarily any $x$--space analyticity).  The basic papers include Jensen \cite{JenScatt}, Sigal \cite{SigalDistort}, Cycon \cite{CyconDistort}, and Nakamura \cite{NakaDistort1, NakaDistort2}.  Some approaches for non--analytic potentials include Cattaneo--Graf--Hunziker \cite{CGHRes}, Cancelier--Martinez--Ramond \cite{CMRRes} and Martinez--Ramond--Sj\"{o}strand \cite{MRSRes}.  There is an enormous literature on the theory of resonances from many points of view.  It would be difficult to attempt a comprehensive discussion of this literature and given that the subject is not central to Kato's work, I won't even try.  But I should mention a beautiful set of ideas about counting asymptotics of resonances starting with Zworski \cite{ZworskiOrig}; see Sj\"{o}strand \cite{SjRes} for unpublished lectures that include lots of references, a recent review of Zworski \cite{ZworskiReview} and forthcoming book of Dyatlov--Zworski \cite{DyaZwor}. The form of the Fermi Golden Rule at Thresholds is discussed in Jensen--Nenciu \cite{JNFermi} (see Section \ref{s16}).  A review of the occurrence of resonances in NR Quantum Electrodynamics and of the smooth Feshbach--Schur map is Sigal \cite{SigalNRQED} and a book on techniques relevant to some approaches to resonances is Martinez \cite{MarBk}.

%%%%%%%%%%%%%%%%%%%%%%%%%%%%%%%%%%%%%%%%%%%%%%%%%%%%%%%%%%%%%%
\section{Eigenvalue Perturbation Theory, IV: Pairs of Projections} \lb{s5}
%%%%%%%%%%%%%%%%%%%%%%%%%%%%%%%%%%%%%%%%%%%%%%%%%%%%%%%%%%%%%%

Recall \cite[Section 2.1]{OT} that a (bounded) projection on a Banach space, $X$, is a bounded operator with $P^2=P$.  If $Y = \ran(P)=\ker(1-P)$ and $Z = \ran(1-P)=\ker(P)$, then $Y$ and $Z$ are disjoint closed subspaces and $Y+Z=X$ and that $(y,z) \mapsto y+z$ is a Banach space linear homeomorphism of $Y \oplus Z$ and $X$.  There is a one-one correspondence between such direct sum decompositions and bounded projections.  We saw in Section \ref{s2} that the following is important in eigenvalue perturbation theory:

\begin{theorem} \lb{T5.1}  Fix a Banach space, $X$.  For any pair of bounded projections, $P, Q$ on $X$ with $\norm{P-Q} < 1$, there exists an invertible map, $U$ so that
\begin{equation}\label{5.1}
  UPU^{-1} = Q
\end{equation}
Moreover, $U$ can be chosen so that

{\textrm(a)} For $P$ fixed, $U(P,Q)$ is analytic in Q in that it is a norm limit, uniformly in each ball $\{Q\,|\,\norm{P-Q} < 1-\epsilon\}$, of polynomials in $Q$.

{\textrm (b)} If $X$ is a Hilbert space and $P,Q$ are self-adjoint projections, then $U$ is unitary.
\end{theorem}

\begin{remarks} 1. We don't require $U(P,P) = \bdone$ which might seem natural because, below, when $P$ and $Q$ are self--adjoint, we'll find a $U$ for which \eqref{5.15} holds and it can be shown that is inconsistent with $U(P,P) = \bdone$.  Of course, given any $U_0(P,Q)$ obeying \eqref{5.1}, $U(P,Q)=U_0(P,Q)U_0(P,P)^{-1}$ also obeys \eqref{5.1} and has $U(P,P)=\bdone$ so it is no great loss.  Both the $U$'s we construct below also obey $U(Q,P) = U(P,Q)^{-1}$.

2. $U$ is actually jointly analytic in $P,Q$ and the proof easily implies if $P$ is fixed and $\beta \mapsto Q(\beta)$ is analytic (resp. continuous, $C^k$, $C^\infty$) in $\beta$, then so is $U$.
\end{remarks}

A first guess for $U$ might be
\begin{equation}\label{5.2}
  W = QP+(1-Q)(1-P)
\end{equation}
which obeys
\begin{equation}\label{5.3}
  WP=QP=QW
\end{equation}
so if $W$ is invertible, we get \eqref{5.1}.  Of course \eqref{5.3} is also true of $W=QP$ but it is easy to see if $\ran\, P \ne X$, then $QP$ can't be invertible.  \eqref{5.2} isn't invertible for an arbitrary pair of projections, for if $\varphi \in (\ran\, P\cap\ker Q)\cup(\ker P\cap\ran Q)$, then $W\varphi=0$.  But when $\norm{P-Q}<1$, this space is trivial, so under the norm condition, $W$ might be (and as we'll see is) invertible.

Define
\begin{equation}\label{5.4}
  \wti{W} = PQ+(1-P)(1-Q)
\end{equation}
\begin{equation}\label{5.5}
  A = P-Q; \qquad B=1-P-Q
\end{equation}
The following easy algebraic calculations are basic to the rich structure of pairs of projections
\begin{equation}\label{5.7}
  A^2 + B^2 = \bdone; \qquad AB+BA=0
\end{equation}
(which Avron \cite{AvronNC} calls the anticommutative Pythagorean Theorem).  Moreover
\begin{equation}\label{5.6}
  PA^2 = P-PQP = A^2P
\end{equation}
so
\begin{equation}\label{5.8}
  [P,A^2]=[Q,A^2]=[P,B^2]=[Q,B^2]=0
\end{equation}
In addition
\begin{equation}\label{5.8}
  (PQ-QP) = BA; \qquad (PQ-QP)^2 = A^4-A^2
\end{equation}

Finally,
\begin{equation}\label{5.9}
  W\wti{W} = \wti{W}W = 1-A^2
\end{equation}
This means that $W$ is invertible if $\norm{A} < 1$, so for \eqref{5.1}, we could take $U=W$ but that won't be unitary when $X$ is a Hilbert space and the two projections are self--adjoint, so, following Kato, we make a slightly different choice

\begin{proof} [First Proof of Theorem \ref{T5.1}]  If $\norm{A} < 1$, we can define
\begin{equation}\label{5.10}
  (1-A^2)^{-1/2} = \sum_{n=0}^{\infty} (-1)^n \binom{-\tfrac{1}{2}}{n} A^{2n}
\end{equation}
where as usual
\begin{equation}\label{5.11}
  \binom{-\tfrac{1}{2}}{n} = \frac{(-\tfrac{1}{2})(-\tfrac{3}{2})\dots(\tfrac{1}{2}-n)}{n!}
\end{equation}
Since $j^{-1}|\tfrac{1}{2} -j| < 1$ for $j=1,2,\dots$, we have that $\sup_n |\binom{-1/2}{n}| < 1$, so if $\norm{A} < 1$, the series in \eqref{5.10} converges and series manipulation proves that
\begin{equation}\label{5.12}
  \left[(1-A^2)^{-1/2}\right]^2 = (1-A^2)^{-1}
\end{equation}
which in turn implies that if we define
\begin{equation}\label{5.13}
  U = W(1-A^2)^{-1/2} = (1-A^2)^{-1/2}W, \qquad \wti{U} = (1-A^2)^{-1/2}\wti{W}
\end{equation}
then, by \eqref{5.8}
\begin{equation}\label{5.13B}
  U\wti{U} = \wti{U}U = \bdone, \qquad UP=QU
\end{equation}
so $U$ is invertible and \eqref{5.1} holds.

Since $(1-A^2)^{-1/2}$ is a norm limit of polynomials in $P$ and $Q$, so is $U$ proving (a).  If $X$ is a Hilbert space and $P^*=P, Q^*=Q$, then $\wti{U} = U^*$, so by \eqref{5.13B} $U$ is unitary, proving (b).    \end{proof}

Theorem \ref{T5.1} for the self--adjoint Hilbert space case goes back to Sz--Nagy \cite{Nagy1947} who was interested in the result because of its application to the convergent perturbation theory of eigenvalues.  His formula for $U$ looks more involved than \eqref{5.2}/\eqref{5.13}.  Wolf \cite{Wolf} then extended the result to general Banach spaces but needed $\norm{P}^2\norm{P-Q} < 1$ and $\norm{1-P}^2\norm{P-Q} < 1$ which is a strictly stronger hypothesis.

In \cite{KatoPTClosed}, Kato proved that if $\beta \mapsto P(\beta)$ is a real analytic family of projections on a Banach space for $\beta \in [0,B]$, then there exists a real analytic family of invertible maps, $U(\beta)$ so that $U(\beta)P(\beta)U(\beta)^{-1}=P(0)$.  He did this using the same formalism he had developed for his treatment of the adiabatic theorem (Kato \cite{KatoAdi} and Section \ref{s17} below).  In 1955, in an unpublished report \cite{KatoPairs}, Kato presented all of the algebra above (except for $AB+BA=0$) and used it to prove Theorem \ref{T5.1} exactly as we do above.

After Avron et al \cite{ASS} found and exploited $AB+BA=0$ (see below), Kato told me that he had found this relation about 1972 but didn't have an application.  Because \cite{KatoPairs} isn't widely available, the standard reference for his approach to pairs of projections is his book \cite{KatoBk}.  In \cite{KatoPairs}, Kato noted that his expression was equal to the object found by Sz--Nagy \cite{Nagy1947} but in the Banach space case, one could get better estimates from his formula for the object.  In that note, he also remarked that when $\norm{P-Q} < 1$, one can find a smooth, one parameter family of projections, $P(t),\, 0 \le t \le 1$ with $P(0) = P$ and $P(1) = Q$ so that the $U$ obtained via his earlier method of solving a differential equation was identical to the $U$ of \eqref{5.2}/\eqref{5.13}.

While this concludes Kato's contribution to the subject of pairs of projections, I would be remiss if I didn't say more about the rich structure of this simple setting, especially when $\norm{P-Q} \ge 1$ (in the self--adjoint Hilbert space setting one has that $\norm{P-Q} \le 1$ but for non-self-adjoint projections and the general case of Banach spaces, one often has $\norm{P-Q} > 1$).  There are two approaches.  The one we'll discuss first is due to Avron--Seiler--Simon \cite{ASS} and uses algebraic relations, especially \eqref{5.7}.  Since $AB+BA=0$ is the signature of supersymmetry, we'll call this the supersymmetric approach.  Here is a typical use of this method:

\begin{theorem} [Avron et. al. \cite{ASS}] \lb{T5.2} Let $P$ and $Q$ be self--adjoint projections so that $P-Q$ is compact.  For $\lambda \in [-1,1]\setminus\{0\}$, let $P_\lambda$ be the projection onto the eigenspace $\calH_\lambda \equiv \{\varphi\,|\,A\varphi = \lambda\varphi\}$

(a) If $\lambda \ne \pm 1$, then
\begin{equation}\label{5.13A}
  V = (1-\lambda^2)^{-1/2}B \restriction \calH_\lambda
\end{equation}
is a unitary map of $\calH_\lambda$ onto $\calH_{-\lambda}$.

(b) For such $\lambda$, we have that
\begin{equation}\label{5.13C}
  \dim \calH_{-\lambda} = \dim \calH_{\lambda}
\end{equation}

(c) If $P-Q$ is trace class, then
\begin{equation}\label{5.13D}
  \tr(P-Q) \in \bbZ
\end{equation}

(d) If $\norm{P-Q} < 1$, then $U \equiv \mathrm{sgn}(B)$ is a unitary operator obeying \eqref{5.1}.  Indeed,
\begin{equation}\label{5.15}
  UPU^{-1} = Q, \qquad UQU^{-1} = P
\end{equation}
\end{theorem}

\begin{remarks} 1.  By $\mathrm{sgn}(B)$, we mean $f(B)$ defined by the functional calculus \cite[Section 5.1]{OT} where
\begin{equation*}
  f(x) = \left\{
           \begin{array}{ll}
             \null \,\,\,\, 1, & x > 0\\
              -1, & x<0 \\
             \null\,\,\,\, 0, & x=0
           \end{array}
         \right.
\end{equation*}
This is unitary because $\norm{A} < 1$ and $B^2 = 1-A^2$ implies that $\ker B = \{0\}$.  One can also write
\begin{equation}\label{5.14}
  U = B(1-A^2)^{-1/2}
\end{equation}

2. If we use \eqref{5.14} to define $U$ in the general Banach space case when $\norm{P-Q} < 1$, the same proof shows that we have \eqref{5.15}.  Indeed, since $[A^2,B]=0$, we have that $U^2=\bdone$ so \eqref{5.1} implies $UQU^{-1}=P$.  So we get another proof of Theorem \ref{T5.1} in the general Banach space case.  However if $P=Q$, then $B=1-2P$ and $A=0$ so by \eqref{5.14}
\begin{equation}\label{5.16}
  U = \bdone - 2P
\end{equation}
Thus, $U(P,P) \ne \bdone$ but see the remarks after Theorem \ref{T5.1}.

3.  That $\tr(P-Q) \in \bbZ$ was first proven by Effros \cite{Effros} and can also be proven using the Krein spectral shift \cite[Problem 5.9.1]{OT}.  It is also true if $P,Q$ are not necessarily self--adjoint projections in a Hilbert space and for suitable Banach space cases; see below.
\end{remarks}

\begin{proof} (a) If $A\varphi = \lambda\varphi$, then
\begin{equation}\label{5.17}
  AB\varphi = -BA\varphi = -\lambda B\varphi
\end{equation}
so $B$ maps $\calH_\lambda$ to $\calH_{-\lambda}$.  Since
\begin{equation*}
  \norm{B\varphi}^2=\jap{\varphi,B^2\varphi}=\jap{\varphi,(1-A^2)\varphi}=(1-\lambda^2)\norm{\varphi}^2
\end{equation*}
we see that $V$ is norm preserving.

If $\psi \in \calH_{-\lambda}$, then, by the above, $\varphi \equiv (1-\lambda^2)^{-1}B\psi \in \calH_\lambda$ and $B\varphi = \psi$ so $\ran B \restriction \calH_\lambda$ is all of $\calH_{-\lambda}$ and thus $V$ is unitary.

(b) is immediate from (a)

(c) Lidskii's Theorem for self--adjoint operators says that if $C$ is a self--adjoint trace class operator, and for any $\lambda \ne 0$, we define $\calH_\lambda = \{\varphi\,|\,C\varphi = \lambda\varphi\}$, then
\begin{equation}\label{5.18}
  \tr(C) = \sum_{\lambda \ne 0} \lambda \dim(\calH_\lambda)
\end{equation}
For the self--adjoint case this is easy since $\tr (C) = \sum_{n=0}^{\infty} \jap{\psi_n,C\psi_n}$ for any trace class operator and any orthonormal basis (see \cite[Theorem 3.6.7]{OT}) and any self--adjoint compact operator has an orthonormal basis of eigenvectors (see \cite[Theorem 3.2.1]{OT}).  By (b), the terms of \eqref{5.18} for $\lambda$ and $-\lambda$ when $C=A$ cancel so long as $\lambda \ne \pm 1$, so
\begin{equation}\label{5.19}
  \tr(A) = \dim \calH_1 - \dim \calH_{-1} \in \bbZ
\end{equation}

(d) Since $\norm{A} < 1$, $B^2 = \bdone - A^2 \ge \epsilon >0$ for $\epsilon = 1-\norm{A}^2$.  Thus, $|B|$ is invertible and
\begin{equation}\label{5.20}
  U=B|B|^{-1}
\end{equation}
is unitary since $U=U^*$ and $U^2=B^2|B|^{-2} = \bdone$.

Moreover, since $|B|$ commutes with $A$ and $B$ (since $[B^2,P]=[B^2,Q]=0$) and $B$ anticommutes with $A$, we see that
\begin{equation}\label{5.21}
  UBU^{-1} = B, \qquad UAU^{-1} = -A
\end{equation}

Since
\begin{equation}\label{5.22}
  P = \tfrac{1}{2}(A-B+\bdone), \qquad Q = \tfrac{1}{2}(-A-B+\bdone)
\end{equation}
\eqref{5.21} implies \eqref{5.15}.  \end{proof}

We can also say something about non--self--adjoint projections on Hilbert spaces and also about the general Banach space case.  The spectral theory of general compact operators, $A$, is more subtle than the self--adjoint case (\cite[Section 3.3]{OT}).  One has that $\sigma(A)\setminus\{0\}$ is discrete, a notion explained in Section \ref{s2}.  Thus, if we define for $\lambda \in \sigma(A)\setminus\{0\}$
\begin{equation}\label{5.23}
  P_\lambda = \frac{1}{2\pi i}\oint_{|z-\lambda|=\delta} \frac{dz}{z-A}
\end{equation}
for $\delta < \dist(\lambda,\sigma(A)\setminus\{\lambda\})$ and $\calH_\lambda = \ran\, P_\lambda$, then $\dim(\calH_\lambda) < \infty$ and is called the algebraic multiplicity of $\lambda$.  Also, as explained in Section \ref{s2},
\begin{equation}\label{5.24}
  AP_\lambda = \lambda P_\lambda + N
\end{equation}
where $N$ is nilpotent, indeed $N^{\dim(\calH_\lambda)} = 0$ so
\begin{equation}\label{5.25}
  \varphi \in \calH_\lambda \Rightarrow (A-\lambda)^{\dim(\calH_\lambda)}\varphi = 0
\end{equation}

Lidskii's Theorem says that for trace class Hilbert space operators, \eqref{5.18} still holds.  Its proof \cite[Section 3.12]{OT} is more subtle.  Lidskii's Theorem doesn't hold on all Banach spaces (where there is an analog of the trace on a class known as nuclear operators).  We say that an operator, $C$ on a Banach space, $X$, obeys Lidskii's Theorem if $C$ is nuclear and obeys \eqref{5.18} -- see \cite{Piet, John} for discussions on when this holds.

\begin{theorem} \lb{T5.3} Let $P,Q$ be two projections on a Banach space, $X$, so that $A=P-Q$ is compact.  Then

(a) $\lambda \in \sigma(A)\setminus\{1,-1\} \Rightarrow -\lambda \in \sigma(A)$

(b) For such $\lambda$, we have that
\begin{equation}\label{5.26A}
  \dim \calH_{\lambda} = \dim \calH_{-\lambda}
\end{equation}

(c) If $\pm 1 \notin \sigma(A)$, then there exists an invertible map $U$ so that \eqref{5.15} holds.

(d) If $A$ obeys Lidskii's theorem, then $\tr(P-Q) \in \bbZ$.
\end{theorem}

\begin{remark} (d) was proven by Kalton \cite{Kalt} using different methods.  The results (a)-(c) and the proof we give of (d) is new in the present paper.
\end{remark}

\begin{proof} (a),(b) For any $z \in \bbC$, we have that $B(A-z) = -(A+z)B$ so, if $z,-z \notin \sigma(A)$, we see that
\begin{equation}\label{5.26}
  B(A-z)^{-1} = -(A+z)^{-1}B
\end{equation}

Since $\sigma(A)\setminus\{0\}$ is a set of isolated points, for any $\lambda \ne 0$, we can find $\epsilon_\lambda > 0$ so that $\sigma(A) \cap \{z\,|\,0<|z-\lambda| \le \epsilon_\lambda\} = \emptyset$.  Taking into account that $z \mapsto -z$ reverses the direction of a contour, by picking $0<\delta<\min(\epsilon_\lambda,\epsilon_{-\lambda})$ in \eqref{5.23} and using \eqref{5.26}, we see that
\begin{equation}\label{5.27}
   BP_\lambda = P_{-\lambda}B
\end{equation}
where $P_\lambda$ is defined by \eqref{5.23} with $\delta$ small even if $\lambda \notin \sigma(A)$ (in which case $P_\lambda =0$).

Suppose $\lambda \ne \pm 1$.  Since $A$ leaves $\calH_\lambda$ invariant and $\sigma(A\restriction\calH_\lambda)=\{\lambda\}$, we have that $(1-A^2)=(1-A)(1+A)$ restricted to $\calH_\lambda$ has an inverse $R$.  Thus $RB$ is a left inverse to $B$ as a map of $\calH_\lambda \to \calH_{-\lambda}$ so $B$ as a map between those spaces is $1$--$1$.  This implies that $\dim \calH_\lambda \le \dim \calH_{-\lambda}$.  By interchanging $\lambda$ and $-\lambda$, we see that \eqref{5.26A} holds which implies (a) and (b).

(c) Since
\begin{equation*}
  BB=BB, \qquad BA=-AB
\end{equation*}
\eqref{5.22} implies that
\begin{equation}\label{5.28}
  BP=QB,\qquad BQ=PB
\end{equation}
We can take $U=B$ if we show that $B$ is invertible.  Since $\pm 1 \notin \sigma(A)$, we see that $(1-A)^{-1}(1+A)^{-1}B$ is a two sided inverse for B.

(d) From Lidskii's theorem and \eqref{5.26A}, we see that \eqref{5.19} holds.
\end{proof}

Our final result from the supersymmetric approach returns to the self--adjoint case.  We define for projections $P, Q$:
\begin{equation}\label{5.29}
  \calK_{P,Q} = \ran\, P \cap \ker Q = \{\varphi \,|\, P\varphi = \varphi, \, Q\varphi = 0\}
\end{equation}

\begin{theorem} \lb{T5.4} Let $P,Q$ be two self--adjoint projections on a Hilbert space, $\calH$.  Then there exists a unitary map, $U$, obeying \eqref{5.15} if and only if
\begin{equation}\label{5.30}
  \dim \calK_{P,Q} = \dim \calK_{1-P,1-Q}
\end{equation}
Moreover, if such a $U$ exists, one can choose it so that
\begin{equation}\label{5.31}
  U=U^*, \qquad U^2 = \bdone
\end{equation}
\end{theorem}

\begin{remarks} 1.  In \eqref{5.30}, both sides may be infinite.

2. If $\pm 1$ are isolated points of the spectrum of $A$ and are discrete eigenvalues, then $K:\ran\, P \to \ran Q$ by $K=QP\restriction\ran\, P$ is a Fredholm operator \cite[Section 3.15]{OT}, both sides of \eqref{5.30} are finite and their difference is the index of $K$.  So, in this case, the theorem says that $U$ obeying \eqref{5.15} exists if and only if index$(K)=0$.  This special case is in \cite{ASS}.

3.  The general case of this theorem is due to Wang, Du and Dou \cite{WDD} whose proof used the Halmos representation discussed below.  Our proof here is from Simon \cite{SiPairs}. Two recent papers \cite{classify1, classify2} classify all solutions of \eqref{5.15}

4.  Operators obeying \eqref{5.31} are called symmetries by Halmos--Kakutani \cite{HalKak}
\end{remarks}

\begin{proof} If $U$ exists, it is easy to see that $U$ must be a unitary map of $\calK_{P,Q}$ to $\calK_{1-P,1-Q}$, so \eqref{5.30} must hold.

For the converse, suppose that \eqref{5.30} holds.  Clearly, $P, Q$ leave both $\calK_{P,Q}$ and $\calK_{1-P,1-Q}$ invariant and so $\calH_1 = \calK_{P,Q} \oplus \calK_{1-P,1-Q}$.  Let $\calH_2 = \calH_1^\perp$ so $\calH=\calH_1\oplus\calH_2$.  Since \eqref{5.30} is assumed, there exists $W:\calK_{P,Q} \to \calK_{1-P,1-Q} $ unitary and onto.  Define on $\calH_1$ as a direct sum
\begin{equation*}
  U_1 = \left(
          \begin{array}{cc}
            0 & W \\
            W^* & 0 \\
          \end{array}
        \right)
\end{equation*}
Then $U_1^2=\bdone$ and $U_1^* = U_1$ and for the restrictions of $P,Q$ to $\calH_1$, we have that $U_1P_1U_1^{-1} = Q_1, \, U_1Q_1U_1^{-1} = P_1$.

So it suffices to prove the result for $\calH_2$, i.e. in the special case that $\calK_{P,Q} = \calK_{1-P,1-Q} = \{0\}$.  If that holds, we have that $\ker(1-A^2) = \{0\}$, so $\ker(B) = \{0\}$ and $U_2 \equiv \mathrm{sgn}(B)$ is unitary.  Since $U_2A_2=-A_2U_2,\, U_2B_2=B_2U_2$, we get that $U_2P_2U_2^{-1} = Q_2,\, U_2Q_2U_2^{-1} = P_2$ by \eqref{5.22}.  Clearly, also $U_2^2=\bdone,\, U_2^*=U_2$. \end{proof}

Our final big topic in this section concerns the Halmos representation.  As a first step, we note that

\begin{proposition} \lb{P5.5}  Let $P,Q$ be two orthogonal projections on a Hilbert space, $\calH$ and let $A, B$ be given by \eqref{5.5}.  Then:

(a) $\calK_{P,Q} = \{\varphi\,|\,A\varphi=\varphi\}, \qquad  \calK_{1-P,1-Q} = \{\varphi\,|\,A\varphi=-\varphi\}$

(b)  $\calK_{P,1-Q} = \{\varphi\,|\,B\varphi=-\varphi\},\qquad \calK_{1-Q,P} = \{\varphi\,|\,B\varphi=\varphi\}$

(c) $\calK_{P,1-Q}\oplus\calK_{1-Q,P}= \{\varphi\,|\,A\varphi=0\}$ \\
\null\qquad\quad $\calK_{P,Q}\oplus\calK_{1-P,1-Q}= \{\varphi\,|\,B\varphi=0\}$

(d) These four spaces are mutually orthogonal.

(e) All four spaces are $\{0\}$ if and only if $\ker A=\ker B = \{0\}$.
\end{proposition}

\begin{proof} (a) $P \le \bdone,\,Q \ge 0$ so $A\varphi=\varphi \Rightarrow  \norm{\varphi}^2 \ge \jap{\varphi,P\varphi} = \norm{\varphi}^2+\jap{\varphi,Q\varphi} \Rightarrow \jap{\varphi,Q\varphi} = 0 \Rightarrow \jap{Q\varphi,Q\varphi} = 0 \Rightarrow Q\varphi = 0 \Rightarrow$ (since $(P-Q)\varphi=\varphi$) $P\varphi=\varphi \Rightarrow \varphi \in \calK_{P,Q}$. Conversely, $\varphi \in \calK_{P,Q} \Rightarrow P\varphi=\varphi \textrm{ \& } Q\varphi=0 \Rightarrow A\varphi=\varphi$.  The proof of the second statement is similar.

(b) Similar to (a) using $B=(1-P)-Q$.

(c) The two spaces in the first statement are orthonormal by (b) and the mutual orthogonality of eigenspaces.  Since $A^2\varphi=(1-B^2)\varphi$, that direct sum is $\ker A^2=\ker A$.  Conversely, if $A\varphi = 0$, then $(1-B^2)\varphi=A^2\varphi=0$.  If $\varphi_\pm = \tfrac{1}{2}(1 \mp B)\varphi$, then $\varphi_\pm \in \ker(1\pm B)$ and $\varphi=\varphi_+ + \varphi_-$, so by (b), $\varphi \in \calK_{P,1-Q}\oplus\calK_{1-Q,P}$.  The second relation has a similar proof.

(d) Immediate from the orthogonality of different eigenspaces of a self--adjoint operator.

(e) Immediate from (c).
\end{proof}

We say that two orthogonal projections are in \emph{generic position} if $\ker A=\ker B = \{0\}$, equivalently if $\calK_{P,Q}, \calK_{1-P,1-Q}, \calK_{P,1-Q}, \calK_{1-Q,P}$  are all $\{0\}$.  The Halmos two projection theorem says

\begin{theorem} [Halmos Two Projection Theorem] \lb{T5.6}  Let $P,Q$ be self--adjoint projections on a Hilbert space, $\calH$ which are in generic position.  Let $\calB_1=\ran\, P,\,\calB_2 = \ran(1-P)$.  Then there exists a unitary map $W$ from $\calB_1$ onto $\calB_2$ and self--adjoint operators $C>0, \, S>0$ on $\calB_1$ with
\begin{equation}\label{5.32}
  C^2+S^2 = \bdone, \qquad [C,S] = 0
\end{equation}
so that under $\calH = \calB_1\oplus\calB_2$,
\begin{equation}\label{5.34}
  P = \left(
        \begin{array}{cc}
          \bdone & 0 \\
          0 & 0 \\
        \end{array}
      \right)
\end{equation}
\begin{equation}\label{5.33}
  Q = \left(
        \begin{array}{cc}
          C^2 & CSW^{-1} \\
          WCS & WS^2W^{-1} \\
        \end{array}
      \right)
\end{equation}
\end{theorem}

\begin{remarks} 1.  There are alternate ways that this theorem is often expressed.  Rather than state it for pairs with generic position, the theorem says that the space is a direct sum of six spaces, two of the form just given and the other four simultaneous eigenspaces with $A\varphi = \lambda\varphi,\, B\varphi=\kappa\varphi$  with $\lambda,\kappa \in \{0,1\}$.  Sometimes, \eqref{5.33} is written:
\begin{equation*}
  Q = \left(
        \begin{array}{cc}
          \bdone & 0 \\
          0 & W \\
        \end{array}
      \right)                   \left(
                                  \begin{array}{cc}
                                    C^2 & CS \\
                                    CS & S^2 \\
                                  \end{array}
                                \right)                  \left(
                                                           \begin{array}{cc}
                                                             \bdone & 0 \\
                                                             0 & W \\
                                                           \end{array}
                                                         \right)^{-1}
\end{equation*}
where the first factor maps $\calB_1\oplus\calB_1$ to $\calB_1\oplus\calB_2$ and the middle factor is an operator on $\calB_1\oplus\calB_1$.  Some authors even implicitly use the first matrix above to identify $\calH$ with $\calB_1\oplus\calB_1$ and only write the middle factor above.

2.  $C$ and $S$ stand, of course, for $\textrm{cosine}$ and $\textrm{sine}$.  One often defines an operator, $\Theta$ with spectrum in $[0,\pi/2]$ so that $C=\cos(\Theta),\,S=\sin(\Theta)$.  While $0$ and/or $1$ may lie in the spectrum of $\Theta$, they cannot be eigenvalues.

3.  This result is due to Halmos \cite{Hal}.  There were earlier related results by Krein et. al. \cite{Krein}, Dixmier \cite{Dix} and Davis \cite{Davis1958}.  The proof we give here is due to Amrein--Sinha \cite{AS}.
\end{remarks}

\begin{proof} By the above
\begin{equation}\label{5.35}
  \ker A = \ker B = \{0\}
\end{equation}
Write the polar decompositions \cite[Section 2.4]{OT}
\begin{equation}\label{5.36}
  A = U_A|A|,\qquad B=U_B|B|
\end{equation}
By \eqref{5.35}, $U_A$ and $U_B$ are unitary and as functions of $A$ and $B$ respectively, they commute with $A$ and $B$ respectively.  It also holds that they each commute with both $|A|$ and $|B|$ (since, for example, $|B|$ commutes with $A$ and so $|A|$ and so $U_A = s-\lim A(|A|+\epsilon)^{-1}$).  Multiplying $AB+BA$ by $(|A|+\epsilon)^{-1}$ and $(|B|+\epsilon)^{-1}$ and taking $\epsilon$ to zero, we see that
\begin{equation}\label{5.37}
  U_AU_B=-U_BU_A \Rightarrow (U_AU_B)^2 = -\bdone
\end{equation}

We've already seen that $\cdot \mapsto U_A \cdot U_A^{-1}$ interchanges $P$ and $Q$. Since $B$ is the $A$ when $P$ is replaced by $1-P$, we see that $\cdot \mapsto U_B \cdot U_B^{-1}$ interchanges $Q$ and $1-P$ and similarly, it interchanges $1-Q$ and $P$.

Let $U=U_AU_B$.  Then we have that
\begin{align}
  UPU^{-1}=(1-P), \qquad & U(1-P)U^{-1}=P \nonumber \\
  UQU^{-1}=(1-Q), \qquad & U(1-Q)U^{-1}=Q \lb{5.39}
\end{align}
which, in particular, implies that $U[\calB_1]$ is all of $\calB_2$ (so they have the same dimension).

Define $W=U\restriction \calB_1$ which we've just seen is a unitary map from $\calB_1$ onto $\calB_2$.  In the $\calB_1\oplus\calB_2$ decomposition, \eqref{5.34} is obvious.  Moreover the decomposition of $Q$ is
\begin{equation}\label{5.40}
  Q=\left(
      \begin{array}{cc}
        PQP\restriction\calB_1 & PQ(1-P)\restriction\calB_2 \\
        (1-P)QP\restriction\calB_1 & (1-P)Q(1-P)\restriction\calB_2 \\
      \end{array}
    \right)
\end{equation}

By the formula for $B$, $BP=-QP$, so $P|B|^2P=PB^2P=PQP$.  Similarly $(1-P)|A|^2(1-P) = (1-P)Q(1-P)$, $PBA(1-P)=PQ(1-P)$ and $(1-P)ABP=(1-P)QP$.

$P|B|^2P$ is already an operator on $\calB_1$.  Using $[U,|A|^2] = 0$, we can write
\begin{equation*}
  (1-P)|A|^2(1-P)=UPU^{-1}|A|^2UPU^{-1}=UP|A|^2PU^{-1}
\end{equation*}
Next note that $UP|A|^2U^{-1}\restriction\calB_2 = W(|A|^2\restriction\calB_1)W^{-1}$.  If we define
\begin{equation}\label{5.41}
  C=|B|\restriction\calB_1,\qquad S=|A|\restriction\calB_1
\end{equation}
then the above calculation and similar calculations on the off--diagonal piece implies \eqref{5.33}.  \end{proof}

B\"{o}ttcher--Spitkovsky \cite{Bott2010} is a review article on lots of applications of the Halmos representation.  We mention also Lenard \cite{Len} who computes the joint numerical range (i.e. $\{(\jap{\varphi,P\varphi},\jap{\varphi,Q\varphi})\,|\,\norm{\varphi}=1\}$) for pairs of projections in terms of the operator $\Theta$ of remark 2 to Theorem \ref{T5.6}.  This range is a union of certain ellipses.

Finally, we mention one result that Kato proved in 1960 \cite{KatoSpectralRadius} that turns out to be connected to pairs of self--adjoint projections, although Kato didn't himself mention or exploit this connection.

\begin{theorem} \lb{T5.7} Let $\Pi$ be a general (i.e. not necessarily self--adjoint) projection in a Hilbert space, $\calH$.  Suppose that $\Pi \ne 0,\bdone$.  Then
\begin{equation}\label{5.43}
  \norm{\Pi}=\norm{\bdone-\Pi}
\end{equation}
\end{theorem}

Kato has this as a Lemma in a technical appendix to \cite{KatoSpectralRadius}, but it is now regarded as a significant enough result that Szyld \cite{Szy} wrote an article to advertise it and explain myriad proofs (\cite{Bott2010} also discusses proofs).  Del Pasqua \cite{DP} and Ljance \cite{Lj} found proofs slightly before Kato but the methods are different and independent; indeed, for many years, no user of the result seemed to know of more than one of these three papers.

Ljance's proof \cite{Lj} shows a close connection to pairs of projections.  Let $P$ be the orthogonal projection on $\ran(\Pi)$ and $Q$ the orthogonal projection onto $\ran(\bdone-\Pi)$  ($P$ and $Q$ must obey $\ker(P)\cap\ker(Q) = \ker(1-P)\cap\ker(1-Q)=\{0\}$ and every such pair of orthogonal projections corresponds to an oblique projection $\Pi$).  Then one can show Ljance's formula (see \cite{Bott2010})
\begin{equation}\label{5.44}
  \norm{\Pi} = \frac{1}{(1-\norm{PQ}^2)^{1/2}}
\end{equation}
so that \eqref{5.43} follows from $\norm{QP}=\norm{(QP)^*}=\norm{PQ}$.

Del Pasqua \cite{DP} noted that \eqref{5.43} might fail in general Banach spaces -- indeed, it is now known \cite{Gur} that if \eqref{5.43} holds for all projections in a Banach space, X, then its norm comes from an inner product.

%%%%%%%%%%%%%%%%%%%%%%%%%%%%%%%%%%%%%%%%%%%%%%%%%%%%%%%%%%%%%%
\section{Eigenvalue Perturbation Theory, V: Temple--Kato Inequalities} \lb{s6}
%%%%%%%%%%%%%%%%%%%%%%%%%%%%%%%%%%%%%%%%%%%%%%%%%%%%%%%%%%%%%%

While strictly speaking the central material in this section is not so much about perturbation theory as variational methods, the subjects are related as Kato mentioned in several places, so we put it here.  In fact, following Kato, we'll see the inequalities proven here can be used to prove certain irregular perturbations yield asymptotic perturbation series. Kato also had several other papers about variational methods for scattering phase shifts \cite{KatoVar1,KatoVar2,KatoVar3} and for an aspect of Thomas--Fermi theory \cite{KatoTF} (not the energy variational principle central to TF theory but one concerning a technical issue connected to the density at the nucleus).  But none of these other papers had the impact of the work we discuss in this review, so we will not discuss them further.

Let $A$ be a self--adjoint operator bounded from below and $\norm{\varphi}=1$ with $\varphi \in D(A)$.  Then Rayleigh's principle says that
\begin{equation}\label{6.1}
  \lambda \equiv \inf \sigma(A) \le \jap{\varphi,A\varphi}\equiv \eta_\varphi
\end{equation}

In 1928, Temple \cite{Temple1928a,Temple1928b} proved a complementary lower bound in case
\begin{equation}\label{6.2}
  \sigma(A) \subset \{\lambda\} \cup [\mu,\infty)
\end{equation}
with $\mu > \lambda$ and $\lambda$ a simple eigenvalue.  So long as
\begin{equation}\label{6.3}
  \eta_\varphi < \mu
\end{equation}
we have \emph{Temple's inequality}
\begin{equation}\label{6.4}
  \lambda \ge \eta_\varphi - \frac{\epsilon_\varphi^2}{\mu-\eta_\varphi}
\end{equation}
where $\epsilon_\varphi \ge 0$ and
\begin{equation}\label{6.5}
   \epsilon_\varphi^2 \equiv \norm{(A-\eta_\varphi)\varphi}^2 = \jap{\varphi,A^2\varphi} -\jap{\varphi,A\varphi}^2
\end{equation}

Temple's inequality had historical importance.  Before the advent of modern computers, variational calculations were difficult and estimating their accuracy was important.  If $\mu^* \le \mu$ (i.e. if one had a possibly crude lower bound on the second eigenvalue), then \eqref{6.2}/\eqref{6.4} $\Rightarrow |\lambda-\eta_\varphi| \le \epsilon_\varphi^2(\mu^*-\eta_\varphi)^{-1}$ so long as $\eta_\varphi < \mu^*$.  One of the early success of perturbation theoretic quantum electrodynamics was the calculation of the Lamb shift in Hydrogen.  That was possible because the unshifted Hydrogen ground state was known precisely.  To check the Lamb shift in Helium, one needed to know its ground state to very high order (the Lamb shift is about one hundred thousandth of that binding energy).  The necessary calculations were done by Kinoshita \cite{Kino} and Pekeris \cite{Pek1,Pek2,Pek3} using variational calculations which in Pekeris' case involved 1078 parameter trial functions.  They used Temple's inequality to estimate how accurately they had computed this ground state energy.  In fact, Kinoshita sketched a proof of Temple's inequality in his paper using Kato's method (he quoted Kato's paper).   The result was the verification of the Lamb shift in Helium to within experimental error.

In 1949, Kato \cite{KatoTemple1} (with an announcement in Physical Review \cite{KatoTempleAnoun}) in one of his little gems found a simple proof of Temple's inequality and also extended the result to any eigenvalue.  Here is his theorem:

\begin{theorem} [Temple--Kato inequality] \lb{T6.1}  Let $A$ be any self--adjoint operator and let $\varphi \in D(A)$.  Let $(\alpha,\zeta) \subset \bbR$ so that
\begin{equation}\label{6.5}
  \alpha < \eta_\varphi < \zeta
\end{equation}
and so that
\begin{equation}\label{6.6}
  \epsilon_\varphi^2 < (\eta_\varphi - \alpha)(\zeta-\eta_\varphi)
\end{equation}
Then:
\begin{equation*}
  \null\hspace{-3.5cm} \textrm{(a)} \qquad\qquad   \sigma(A) \cap (\alpha,\zeta) \ne \emptyset
\end{equation*}
If $\sigma(A) \cap (\alpha,\zeta)$ contains only a single point, $\lambda$, then
\begin{equation}\label{6.7}
  \null\hspace{-1cm}\textrm{(b)} \qquad\qquad\eta_\varphi-\frac{\epsilon_\varphi^2}{\zeta-\eta_\varphi} \le \lambda \le \eta_\varphi + \frac{\epsilon_\varphi^2}{\eta_\varphi-\alpha}
\end{equation}
If, in addition, $\lambda$ is a simple eigenvalue with associated eigenvector, $\psi$, with $\norm{\psi}  = 1$ and $\jap{\psi,\varphi} \ge 0$ and if $\epsilon_\varphi < \delta \equiv \min(\eta_\varphi-\alpha,\zeta-\eta_\varphi)$, then
\begin{equation}\label{6.8}
  \textrm{(c)} \qquad\qquad \norm{\varphi-\psi} \le \left[2-2\left(1-\frac{\epsilon_\varphi^2}{\delta^2}\right)^{1/2}\right]^{1/2}
\end{equation}
\end{theorem}

\begin{remarks} 1. As we'll see, a version of \eqref{6.7} holds even if we don't suppose there is only one point in $\sigma(A) \cap (\alpha,\zeta)$, namely if
\begin{equation}\label{6.8a}
  \gamma_0 = \eta_\varphi - \frac{\epsilon_\varphi^2}{\zeta-\eta_\varphi}; \qquad \kappa_0 = \eta_\varphi + \frac{\epsilon_\varphi^2}{\eta_\varphi-\alpha}
\end{equation}
then $\sigma(A) \cap (\alpha,\kappa_0] \ne \emptyset$ and $\sigma(A) \cap [\gamma_0,\zeta) \ne \emptyset$

2. If we take $\alpha \to -\infty$ and $\zeta = \eta_\varphi+1$, the upper bound in \eqref{6.7} is just the Rayleigh bound \eqref{6.1} and if we take $\zeta=\mu$, then the lower bound in \eqref{6.7} is just Temple's inequality \eqref{6.4}.

3. If $0<\alpha < 1$, then
\begin{align*}
  2-2(1-\alpha^2)^{1/2} &= \left[\frac{4-4(1-\alpha^2)}{2+2(1-\alpha^2)^{1/2}}\right] \\
                        &\le \frac{4\alpha^2}{4(1-\alpha^2)^{1/2}} = \left[\frac{\alpha}{(1-\alpha^2)^{1/4}}\right]^2
\end{align*}
so \eqref{6.8} implies that
\begin{equation}\label{6.9}
  \norm{\varphi-\psi} \le \frac{\epsilon}{\delta}\left(1-\frac{\epsilon^2}{\delta^2}\right)^{-1/4}
\end{equation}
which is how Kato writes it in Kato \cite{KatoTemple2} (see Knyazev \cite{Kny} for refined versions of these types of estimates).
\end{remarks}

The proof we'll give follows Kato's approach (see also Harrell \cite{Harrell1978}).  The key to this proof is what Temple \cite{Temple1955} calls Kato's Lemma:

\begin{lemma} \lb{L6.2} Let $A$ be a self--adjoint operator and $\varphi \in D(A)$ with $\norm{\varphi} = 1$.  Then
\begin{equation}\label{6.10}
  \sigma(A) \cap (\alpha,\zeta) = \emptyset \Rightarrow \jap{\varphi,(A-\alpha)(A-\zeta)\varphi} \ge 0
\end{equation}
\end{lemma}

\begin{proof}  The spectral theorem (see \cite[Chapter V and Section 7.2]{OT}) says that $A$ is a direct sum of multiplications by $x$ on $L^2(\bbR\setminus (\alpha,\zeta), d\mu(x))$.  Since $(x-\alpha)(x-\zeta) \ge 0$ for $x \in \bbR\setminus (\alpha,\zeta)$, we see that $(A-\alpha)(A-\zeta) \ge 0$.
\end{proof}

\begin{remark} While we use the Spectral Theorem (as Kato did), all we need is a spectral mapping theorem, i.e. if $f(x) = (x-\alpha)(x-\zeta)$, then $\sigma(f(A))=f[\sigma(A)]$ and the fact that an operator with spectrum in $[0,\infty)$ is positive.  The spectral mapping theorem for polynomials holds for elements of any Banach algebra and the proof in \cite[Theorem 2.2.6]{OT} extends to unbounded operators.  That this lemma follows from considerations of resolvents only was noted by Temple \cite{Temple1955}.
\end{remark}

Taking contrapositives in \eqref{6.10}, we get the following Corollary (if Lemmas are allowed to have Corollaries):

\begin{corollary}\lb{C6.3} Let $A$ be a self--adjoint operator and $\varphi \in D(A)$ with $\norm{\varphi} = 1$.  Then
\begin{equation}\label{6.11}
  \jap{\varphi,(A-\alpha)(A-\zeta)\varphi} < 0 \Rightarrow \sigma(A) \cap (\alpha,\zeta) \ne \emptyset
\end{equation}
\end{corollary}

The final preliminary of the proof is

\begin{lemma} \lb{L6.4} Suppose that $A$ is self--adjoint and that $\lambda \in \bbR$  is an isolated simple eigenvalue with $A\psi = \lambda\psi$ and $\norm{\psi}=1$.  If $\varphi \in D(A)$ with $\norm{\varphi} = 1$ and
\begin{equation}\label{6.12}
  \epsilon_\varphi < \delta \equiv \textrm{\dist}(\eta_\varphi,\sigma(A)\setminus \{\lambda\})
\end{equation}
and if the phase of $\psi$ is changed so that $\jap{\varphi,\psi} \ge 0$, then
\begin{equation}\label{6.13}
  \norm{\varphi-\psi}^2 \le 2 - 2\left(1-\frac{\epsilon_\varphi^2}{\delta^2}\right)^{1/2}
\end{equation}
\end{lemma}

\begin{proof}  Let $P$ be the projection onto multiples of $\psi$.  Since $(A-\eta_\varphi)^2 \ge \delta^2$ on the $A$-invariant subspace $\ran(1-P)$ (by the spectral theorem as in the proof of Lemma \ref{L6.2}), we have that
\begin{equation}\label{6.14}
   \epsilon_\varphi^2 = \norm{(A-\eta_\varphi)\varphi}^2 \ge \delta^2 \norm{(1-P)\varphi}^2
\end{equation}
so
\begin{equation}\label{6.15}
  \norm{(1-P)\varphi}^2 \le \epsilon_\varphi^2/\delta^2 < 1
\end{equation}
by \eqref{6.12}.  Since $\norm{(1-P)\varphi}^2+\norm{P\varphi}^2=1$, we see that (if $\jap{\psi,\varphi} \ge 0$)
\begin{equation}\label{6.16}
  \jap{\psi,\varphi} = \norm{P\varphi} \ge \left(1-\frac{\epsilon_\varphi^2}{\delta^2}\right)^{1/2}
\end{equation}
Since $\norm{\psi-\varphi}^2=2-2\jap{\psi,\varphi}$, \eqref{6.13} is immediate.
\end{proof}

\begin{proof} [Proof of Theorem \ref{T6.1}] (a) We have that
\begin{align}
  \jap{\varphi,(A-\alpha)(A-\zeta)\varphi} &= \jap{\varphi,(A-\eta_\varphi)^2\varphi}+\jap{\varphi,\left[\eta_\varphi^2+\alpha\zeta-(\alpha+\zeta)A\right]\varphi} \nonumber\\
                                           &= \epsilon_\varphi^2-(\eta_\varphi-\alpha)(\zeta-\eta_\varphi) < 0 \lb{6.17}
\end{align}
by \eqref{6.6}.  By Corollary \ref{C6.3}, we see that $\sigma(A) \cap (\alpha,\zeta) \ne \emptyset$.

(b) As in the proof of \eqref{6.17}, for any $\gamma, \kappa$, we have that
\begin{equation}\label{6.18}
  \jap{\varphi,(A-\gamma)(A-\kappa)\varphi} = \epsilon_\varphi^2-(\eta_\varphi-\gamma)(\kappa-\eta_\varphi)
\end{equation}
Fix $\kappa=\zeta$.  Then, using $\zeta > \eta_\varphi$:
\begin{equation}\label{6.19}
  \textrm{RHS of \eqref{6.18}} < 0 \iff \gamma < \gamma_0
\end{equation}
(with $\gamma_0$ given by \eqref{6.8a}) so by Corollary \ref{C6.3},
\begin{equation}\label{6.20}
  \gamma < \gamma_0 \Rightarrow \sigma(A) \cap (\gamma,\zeta) \ne \emptyset
\end{equation}
Since $\sigma(A)$ is closed, this implies that
\begin{equation}\label{6.21}
  \sigma(A) \cap [\gamma_0,\zeta) \ne \emptyset
\end{equation}

Similarly,
\begin{equation}\label{6.23}
  \sigma(A) \cap (\alpha,\kappa_0] \ne \emptyset
\end{equation}

In particular, if there is a single point, $\lambda$, in $(\alpha,\zeta)$, we must have that $\lambda \in (\alpha,\kappa_0] \cap [\gamma_0,\zeta) = [\gamma_0,\kappa_0]$ which is \eqref{6.7}.

(c) This is Lemma \ref{L6.4}.
\end{proof}

Kato exploited what are now called the Temple--Kato inequalities in his thesis to prove results on asymptotic perturbation theory.  Below are two typical results whose proofs are very much in the spirit of this work of Kato -- see Section \ref{s3} for what it means for an eigenvalue to be stable.

\begin{theorem} \lb{T6.5} Let $A_0$ be a self--adjoint operator on a Hilbert space, $\calH$.  Let $B$ be a symmetric operator with $D(A_0) \cap D(B) \equiv \calD$ dense in $\calH$ and a core for $A_0$.  For each $\beta >0$ (perhaps only for sufficiently small such $\beta$), let $A(\beta)$ be a self--adjoint extension of $A_0+\beta B \restriction \calD$.  Let $E_0$ be a simple, discrete eigenvalue for $A_0$ which is stable for $A(\beta)$.  Let $\varphi \in D(A_0), \, \norm{\varphi} = 1$ and $A_0\varphi=E_0\varphi$.  Suppose that $\varphi \in D(B)$.  Then the eigenvalue, $E(\beta)$ of $A(\beta)$ near $E_0$ obeys
\begin{equation}\label{6.24}
  E(\beta) = E_0 + \beta \jap{\varphi,B\varphi}+\textrm{O}(\beta^2)
\end{equation}
\end{theorem}

\begin{proof} Since $\calD$ is a core and for $\eta \in \calD, \, z \in \bbC\setminus\bbR$, $[(A(\beta)-z)^{-1}-(A_0-z)^{-1}](A_0-z)\eta=-\beta(A(\beta)-z)^{-1}B\eta$ we see that $A(\beta) \to A_0$ in strong resolvent sense as $\beta\downarrow 0$.  By the definition of stability, there is an interval $(\alpha,\zeta)$ containing $E_0$, so that for small $\beta$, $A(\beta)$ has a unique eigenvalue, $E(\beta)$, in $(\alpha,\zeta)$.  Showing the operator involved in a superscript, we see that
\begin{equation*}
  \eta_\varphi^{A(\beta)} = E_0+\beta \jap{\varphi,B\varphi} \to E_0
\end{equation*}

Since $(A(\beta)-\eta_\varphi^{A(\beta)})\varphi=\beta(B-\jap{\varphi,B\varphi})\varphi$, we see that
\begin{equation*}
  \left(\epsilon_\varphi^{A(\beta)}\right)^2 = \beta^2(\norm{B\varphi}-\jap{\varphi,B\varphi}^2)=\textrm{O}(\beta^2)
\end{equation*}
so, by the Temple--Kato inequalities, $E(\beta)-\eta_\varphi^{A(\beta)}  = \textrm{O}(\beta^2)$ which is \eqref{6.24}
\end{proof}

To go to the next order, we need the reduced resolvent, $S$, of $A_0$ at $E_0$, defined in Section \ref{s2} (see \eqref{2.4D}). In his thesis, Kato realized that contour integrals of $B(A_0-z)^{-1}\dots B(A_0-z)^{-1}\varphi$ could be expressed in terms of $S$.  In particular, the first order formal eigenvector for $A(\beta)$ is
\begin{equation}\label{6.26}
  \psi_1(\beta)=\varphi-\beta SB\varphi
\end{equation}
Since $\ran S \subset \ran(1-P)$ is orthogonal to $\varphi$, we see that
\begin{equation}\label{6.27}
  \norm{\psi_1(\beta)}^2=1+\beta^2\norm{SB\varphi}^2
\end{equation}

For $\psi_1(\beta)$ to be in $D(B)$, we will need to suppose that
\begin{equation}\label{6.28}
  \varphi\in D(B), \qquad SB\varphi \in D(B)
\end{equation}
We can also write down the first three perturbation coefficients for the energy  (see for example \cite[pg 7]{RS4}):
\begin{equation}\label{6.29}
  E_1=\jap{\varphi,B\varphi}, \qquad E_2=-\jap{B\varphi,SB\varphi}
\end{equation}
\begin{equation}\label{6.30}
  E_3=E_1 E_2 + \jap{B\varphi,SBSB\varphi}
\end{equation}

Straightforward calculations show that
\begin{align*}
  (A_0-E_0)\psi_1(\beta) &= -\beta(1-P)B\varphi \\
  (A(\beta)-E_0)\psi_1(\beta) &= \beta E_1\varphi-\beta^2BSB\varphi
\end{align*}%
since $\beta PB\varphi=\beta E_1\varphi$.  Thus:
\begin{equation*}
  (A(\beta)-E_0-\beta E_1)\psi_1(\beta) = \beta^2 E_1 SB\varphi - \beta^2BSB\varphi
\end{equation*}

From this, using \eqref{6.27}, one sees easily that
\begin{equation}\label{6.31}
  \jap{\psi_1(\beta),A(\beta)\psi_1(\beta)}=(E_0+\beta E_1+ \beta^2 E_2+ \beta^3 E_3)\norm{\psi_1(\beta)}^2 + \textrm{O}(\beta^4)
\end{equation}
\begin{equation}\label{6.32}
  \norm{\left[A(\beta)-(E_0+\beta E_1+ \beta^2 E_2+ \beta^3 E_3)\right]\psi_1(\beta)}^2 = \textrm{O}(\beta^4)
\end{equation}
Thus, we have, using $\psi_1(\beta)/\norm{\psi_1(\beta)}$ as a trial vector

\begin{theorem} \lb{T6.6} Under the hypotheses of Theorem \ref{T6.5} if also \eqref{6.28} holds, then
\begin{equation}\label{6.33}
  E(\beta) = E_0+\beta E_1+ \beta^2 E_2+ \beta^3 E_3 + \textrm{O}(\beta^4)
\end{equation}
\begin{equation}\label{6.34}
  \norm{\varphi(\beta) - \psi_1(\beta)} = \textrm{O}(\beta^2)
\end{equation}
where $\varphi(\beta)$ is the normalized eigenvector for $A(\beta)$ chosen so that for small $\beta$, $\jap{\varphi(\beta),\varphi} > 0$.
\end{theorem}

As Kato noted in his thesis, this idea shows if all the terms for the $n$th order formal series for the eigenvector lie in $\calH$, then one gets asymptotic series for the energy with errors of order O$(\beta^{2n})$, i.e. the $2n$ coefficients $E_0,\dots,E_{2n-1}$ but the method doesn't handle odd powers.  Indeed in \cite{KatoThesis}, he said: ``\emph{However, there has been a serious gap in the series of these conditions; for all of them had in common the property that they give the expansion of the eigenvalues up to even orders of approximation, and there was no corresponding theorem giving an expansion up to an odd order}.''  Personally, I think ``serious'' is a bit strong given that he handles the case of infinite order (for me the most important) and first order results but it shows he was frustrated by a problem he tried to solve without initial success.  But in \cite{KatoPT2}, he put in a \emph{Note Added in Proof} announcing he had solved the problem!  The solution appeared in \cite{KatoAsymPT}.  For example, if $A_0 \ge 0, B \ge 0$, he proved that if $\varphi \in Q(B)$, then $E(\beta) = E_0+E_1\beta+\textrm{o}(\beta)$ and if $B^{1/2}\varphi \in Q(B^{1/2}A_0^{-1}B^{1/2})$, he proved that $E(\beta) = E_0+E_1\beta+E_2\beta^2+\textrm{o}(\beta^2)$.  Not surprisingly, in addition to estimates of Temple--Kato type, the proofs use a variant of quadratic form methods.  I note that Kato did not put any of these results in his book where his discussion of asymptotic series applies to general Banach space settings and not just positive operators and the ideas are closer to what we put in Section \ref{s3}.

Besides the original short paper on Temple--Kato inequalities, Kato returned to the subject several times.  In two papers \cite{KatoTemple2,KatoTemple3}, he considered the fact that in some applications of interest, the natural trial vector has $\varphi \in Q(A)$, not $D(A)$.  Trial functions only in $Q(A)$ are fine for the Rayleigh upper bound but if $\varphi \notin D(A)$, then $\epsilon_\varphi^A = \infty$, so $\varphi$ cannot be used for Temple's inequality or the Temple--Kato inequality.  Of course, one could look at the Temple--Kato inequality for $\sqrt{A}$ if $A \ge 0$ but calculation of $\jap{\varphi,\sqrt{A}\varphi}$ may not be easy for, say, a second order differential operator where $\sqrt{A}$ is a pseudo-differential operator.  But such operators can often be written $A=T^*T$ where $T$ is a first order differential operator.  Variants of the Temple--Kato inequality for operators of this form are the subject of two papers of Kato \cite{KatoTemple2,KatoTemple3}.  Kato et al. \cite{KatoEtAlTemple} studies an application of these ideas.

Interesting enough, while Kato's work was 20 years after Temple, Temple was young when he did that work and was still active in 1949 and he reacted to Kato's paper with two of his own \cite{Temple1952,Temple1955}.  George Frederick James Temple (1901-1992) was a mathematician  with a keen interest in physics -- he wrote two early books on quantum mechanics in 1931 and 1934. He spent much of his career at King's College, London although for the last fifteen years of it, he held the prestigious Sedleian Chair of Natural Philosophy at Oxford, the chair going back to 1621.  He was best known in British circles for a way of discussing distributions as equivalence classes of approximating smooth functions, an idea that was popular because the old guard didn't want to think about the theory of topological vector spaces central to Schwartz' earlier approach.  His other honors include a knighthood (CBE, for War work), a fellowship in and the Sylvester Medal of the Royal Society.  At age 82, he became a benedictine monk and spent the last years of his life in a monastery on the Isle of Wright. The long biographical note of his life written for the Royal Society \cite{TempleBio} doesn't even mention Temple's inequality!

Davis \cite{Davis1952} extended what he calls ``the ingenious method of Kato'' by replacing the single interval $(\alpha,\zeta)$ by a finite union of intervals.  Thirring \cite{Thirr3} has discussed Temple's inequality as a consequence of the Feshbach \cite{Fesh} projection method (which mathematicians call the method of Schur \cite{Schur} complements).  Turner \cite{Turner1969} and Harrell \cite{Harrell1978} have extensions to the case where $A$ is normal rather than self--adjoint and Kuroda \cite{Kuroda2007} to $n$ commuting self-adjoint operators (and so including the normal case).  Cape et al. \cite{Cape2016} apply Temple--Kato inequalities to graph Laplacians.  Golub--van der Vost \cite{Golub2000} have a long review on eigenvalue values bounds mentioning that by the time of their review in 2000, Temple--Kato inequalities had become a standard part of linear algebra.

%%%%%%%%%%%%%%%%%%%%%%%%%%%%%%%%%%%%%%%%%%%%%%%%%%%%%%%%%%%%%%
\section{Self--Adjointness, I: Kato's Theorem} \lb{s7}
%%%%%%%%%%%%%%%%%%%%%%%%%%%%%%%%%%%%%%%%%%%%%%%%%%%%%%%%%%%%%%

This is the first of four sections on self-adjointness issues.  We assume the reader knows the basic notions, including what an operator closure and an operator core are and the meaning of essential self-adjointness.  A reference for these things is \cite[Section 7.1]{OT}.

This section concerns the Kato--Rellich theorem and its application to prove the essential self--adjointness of atomic and molecular Hamiltonians. The quantum mechanical Hamiltonians typically treated by this method are bounded from below.  Section \ref{s8} discusses cases where $V(x) \ge -cx^2-d$ like Stark Hamiltonians.  Section \ref{s9} discusses Kato's contribution to the realization that the positive part of $V$ can be more singular than the negative part without destroying essential self--adjointness and Section \ref{s10} turns to Kato's contribution to the theory of quadratic forms.  To save ink, in this article, I'll use ``esa'' as an abbreviation for ``essentially self-adjoint'' or ``essential self-adjointness'' and ``esa--$\nu$'' for ``essentially self--joint on $C_0^\infty(\bbR^\nu)$.''.

As we've mentioned, Kato's 1951 paper \cite{KatoHisThm} is a pathbreaking contribution of great significance.  He considered \emph{$N$--body Hamiltonians} on $L^2(\bbR^{\nu N})$ of the formal form
\begin{equation}\label{7.1}
  H = -\sum_{j=1}^{N}\frac{1}{2m_j}\Delta_j + \sum_{i<j} V_{ij}(x_i-x_j)
\end{equation}
where $x \in \bbR^{\nu N}$ is written $\boldsymbol{x} = (x_1,\dots,x_N)$ with $x_j \in \bbR^\nu$, $\Delta_j$ is the $\nu$--dimensional Laplacian in $x_j$ and each $V_{ij}$ is a real valued function on $\bbR^\nu$.  In 1951, Kato considered only the physically relevant case $\nu=3$.

If there are $N+k$ particles in the limit where the masses of particles $N+1,\dots,N+k$ are infinite, one considers an operator like $H$ but adds terms
\begin{equation}\label{7.2}
  \sum_{j=1}^{N} V_j(x_j), \qquad V_j(x) = \sum_{\ell=N+1}^{N+k} V_{j\ell}(x-x_\ell)
\end{equation}
where $x_{N+1},\dots,x_{N+k}$ are fixed points in $\bbR^\nu$.

More generally, one wants to consider (as Kato did) Hamiltonians with the center of mass removed.  We discuss the kinematics of such removal in Section \ref{s11}.  We note that the self--adjointness results on the Hamiltonians of the form \eqref{7.1} easily imply results on Hamiltonians (on $L^2(\bbR^{(N-1)\nu})$) with the center of mass motion removed.  Of especial interest is the Hamiltonian of the form \eqref{7.2} with $N=1$, i.e.
\begin{equation}\label{7.3}
  H=-\Delta+W(x)
\end{equation}
on $L^2(\bbR^\nu)$ which we'll call \emph{reduced two body Hamiltonians} (since, except for a factor of $(2\mu)^{-1}$ in front of $-\Delta$, it is the two body Hamiltonian with the center of mass removed).

Kato's big 1951 result was

\begin{theorem} [Kato's Theorem \cite{KatoHisThm}, First Form] \lb{T7.1} Let $\nu=3$.  Let each $V_{ij}$ in \eqref{7.1} lie in $L^2(\bbR^3)+L^\infty(\bbR^3)$.  Then the Hamiltonian of \eqref{7.1} is self--adjoint on $D(H) = D(-\Delta)$ and esa--$(3N)$.
\end{theorem}

\begin{remarks} 1. The same results holds with the terms in \eqref{7.2} added so long as each $V_j$ lies in $L^2(\bbR^3)+L^\infty(\bbR^3)$.

2. Kato also notes the exact description of $D(-\Delta)$ on $L^2(\bbR^\nu)$ in terms of the Fourier transform (see \cite[Chapter 6]{RA}) $\hat{\varphi}(k) = (2\pi)^{-\nu/2} \int e^{-ik\cdot x} \varphi(x) d^\nu x$:
\begin{equation}\label{7.3a}
  D(-\Delta) = \{ \varphi \in L^2(\bbR^\nu) \,|\, \int (1+k^2)^2 |\hat{\varphi}(k)|^2 d^\nu k < \infty\}
\end{equation}

3. The proof shows that the graph norms of $H$ and $-\Delta$ on $D(-\Delta)$ are equivalent, so any operator core for $-\Delta$ is a core for $H$. Since it is easy to see that $C_0^\infty(\bbR^{3N})$ is a core for $-\Delta$, the esa result follows from the self-adjointness claim, so we concentrate on the latter.

4.  Kato didn't assume that $V \in L^2(\bbR^3)+L^\infty(\bbR^3)$ but rather the stronger hypothesis that for some $R < \infty$, one has that $\int_{|x| < R} |V(x)|^2 d^3x < \infty$ and $\sup_{|x| \ge R} |V(x)| < \infty$, but his proof extends to $L^2(\bbR^3)+L^\infty(\bbR^3)$.

5.  Kato didn't state that $C_0^\infty(\bbR^{3N})$ is a core but rather that $\psi$'s of the form $P(x) e^{-\tfrac{1}{2}x^2}$ with $P$ a polynomial in the coordinates of $x$ is a core (He included the $\tfrac{1}{2}$ so the set was invariant under Fourier transform.) His result is now usually stated in terms of $C_0^\infty$.
\end{remarks}

If $v(x) = 1/|x|$ on $\bbR^3$, then $v \in L^2(\bbR^3)+L^\infty(\bbR^3)$, so Theorem \ref{T7.1} has the important Corollary, which includes the Hamiltonians of atoms and molecules:

\begin{theorem} [Kato's Theorem \cite{KatoHisThm}, Second Form] \lb{T7.2} The Hamiltonian, $H$, of \eqref{7.1} with $\nu = 3$ and each
\begin{equation}\label{7.4}
  V_{ij}(x) = \frac{z_{ij}}{|x|}
\end{equation}
and this Hamiltonian with terms of the form \eqref{7.2} where
\begin{equation}\label{7.5}
  V_j(x) = \sum_{\ell = N+1}^{N+k} \frac{z_{j\ell}}{|x-x_\ell|}
\end{equation}
are self--adjoint on $D(-\Delta)$ and esa--$3N$
\end{theorem}

\begin{remark} This result assures that the time dependent Schr\"{o}dinger equation $\dot \psi_t = -iH\psi_t$ has solutions (since self--adjointness means that $e^{-itH}$ exists as a unitary operator).  The analogous problem for Coulomb Newton's equation (i.e. solvability for a.e. initial condition) is open for $N \ge 5$!
\end{remark}

As Kato remarks in \cite{WP}, ``the proof turned out to be rather easy''.  It has three steps:

(1)  The Kato--Rellich theorem which reduces the proof to showing that each $V_{ij}$ is relatively bounded for Laplacian on $\bbR^3$ with relative bound $0$.

(2)  A proof that any function in $L^2(\bbR^3)+L^\infty(\bbR^3)$ is $-\Delta$--bounded with relative bound $0$.  This relies on a simple Sobolev estimate.

(3) A piece of simple kinematics that says that the two body estimate in step 2 extends to one for $v_{ij}(x_i-x_j)$ as an operator on $L^2(\bbR^{3N})$.

\textbf{Step 1.} The needed result (recall that $A$--bounded is defined in \eqref{2.7}):

\begin{theorem} [Kato--Rellich Theorem] \lb{T7.3} Let $A$ be self--adjoint, $B$ symmetric and let $B$ be $A$--bounded with relative bound $a<1$, i.e. $D(A) \subset D(B)$ and for some fixed $b$ and all $\varphi \in D(A)$
\begin{equation}\label{7.6}
  \norm{B\varphi} \le a\norm{A\varphi}+b\norm{\varphi}
\end{equation}
Then $A+B$ is self--adjoint on $D(A)$ and any operator core for $A$ is one for $A+B$.
\end{theorem}

\begin{remarks}  1.  This result is due to Rellich \cite[Part III]{RellichPT}.  Kato found it in 1944, when he was unaware of Rellich's work, so it is independently his.

2. The proof uses von Neumann's criteria: a closed symmetric operator, $C$, on $D(C)$  is self--adjoint if and only if for some $\kappa \in (0,\infty)$, one has that $\ran(C\pm i\kappa)=\calH$.  For $C$ closed implies that $\ran(C\pm i\kappa)$ are closed subspaces with $\ran(C\pm i\kappa)^\perp = \ker(C^*\mp i\kappa)$.  Thus, if $C$ is self--adjoint, then $\ker(C^*\mp i\kappa) = \{0\}$ proving one direction.  For the other direction, suppose that $\ran(C\pm i\kappa) = \calH$.  Given $\psi \in D(C^*)$, find $\varphi \in D(C)$ with $(C+i\kappa)\varphi = (C^*+i\kappa)\psi$ (since $\ran(C+i\kappa)=\calH$).  Thus $(C^*+i\kappa)(\varphi-\psi) = 0$.  Since $\ran(C-i\kappa) = \calH = \ker(C^*+i\kappa)^\perp$, we have that $\varphi-\psi = 0$.  Thus $D(C^*) = D(C)$ and $C$ is self--adjoint.

3. For the rest of the proof, use $\norm{(C+i\kappa)\varphi}^2 = \norm{C\varphi}^2+|\kappa|^2\norm{\varphi}^2$ to see that
\begin{equation}\label{7.7}
  \norm{C(C\pm i\kappa)^{-1}} \le 1, \qquad \norm{(C\pm i\kappa)^{-1}} \le |\kappa|^{-1}
\end{equation}
It follows from this (with $C=A$) that when \eqref{7.6} holds, one has that
\begin{equation}\label{7.8}
  \norm{B(A\pm i\kappa)^{-1}} \le a + b|\kappa|^{-1}
\end{equation}
Since $a < 1$, we can be sure that if $|\kappa|$ is very large, then ${\norm{B(A\pm i\kappa)^{-1}} < 1}$ so using a geometric series, we have that ${1 + B(A\pm i\kappa)^{-1}}$ is invertible which implies that it maps $\calH$ onto $\calH$.  Since $(A\pm i\kappa)$ maps $D(A)$ onto $\calH$, we see that
\begin{equation}\label{7.9}
  (A+B\pm i\kappa) =(1+B(A\pm i\kappa)^{-1})(A\pm i\kappa)
\end{equation}
maps $D(A)$ onto $\calH$.  Thus by von Neumann's criterion, $A+B$ is self--adjoint on $D(A)$.  By a simple argument, $\norm{A\cdot}+\norm{\cdot}$ is an equivalent norm to $\norm{(A+B)\cdot}+\norm{\cdot}$ which proves the esa result.

4. The case $B=-A$ shows that one can't conclude self-adjointness of $A+B$ on $D(A)$ if \eqref{7.6} holds with $a=1$ but Kato \cite{KatoBk} proved that $A+B$ is esa on $D(A)$ in that case and W\"{u}st \cite{WustESA} proved the stronger result of esa on $D(A)$ if one has for all $\varphi \in D(A)$
\begin{equation}\label{7.9a}
  \norm{B\varphi}^2 \le \norm{A\varphi}^2 + b\norm{\varphi}^2
\end{equation}

5.  In some of my early papers, I called $B$ Kato small if $B$ was $A-bounded$ with relative bound less than $1$ and Kato tiny if the relative bound was $0$.  I am pleased to say that while many of my names (hypercontractive, almost Mathieu, Berry's phase, Kato class,...) have stuck, this one has not!
\end{remarks}

\textbf{Step 2.} Kato began by considering $\varphi \in L^2(\bbR^3)$ with $\varphi \in D(-\Delta)$, i.e. $\int (1+k^2)^2|\hat{\varphi}(k)|^2 d^3k < \infty$.  He noted that this implied that
\begin{align}
  \int |\hat{\varphi}(k)| d^3k &= \int (1+k^2)^{-1}(1+k^2)|\hat{\varphi}(k)| d^3k \nonumber \\
                               & \le \norm{(1+k^2)^{-1}}_2 \norm{(1-\Delta)\varphi}_2 \lb{7.10}
\end{align}
by the Schwarz inequality and Plancherel theorem.  Thus
\begin{align}
  \norm{\varphi}_\infty &\le (2\pi)^{-3/2} \int |\hat{\varphi}(k)| d^3k \lb{7.10a}  \\
                        &\le C\left( \norm{\Delta\varphi}_2 + \norm{\varphi}_2\right) \lb{7.11}
\end{align}

It follows that if $V=V_1+V_2$ with $V_1 \in L^2(\bbR^3), V_2 \in L^\infty(\bbR^3)$, then as operators on $L^2(\bbR^3)$
\begin{align}
  \norm{V\varphi}_2 &\le \norm{V_1\varphi}_2+ \norm{V_2\varphi}_2 \nonumber \\
                    &\le \norm{V_1}_2 \norm{\varphi}_\infty + \norm{V_2}_\infty \norm{\varphi}_2 \nonumber \\
                    &\le C\norm{V_1}_2 \norm{\Delta\varphi}_2 + \left(C\norm{V_1}_2+\norm{V_2}_\infty\right) \norm{\varphi}_2 \lb{7.12}
\end{align}

If $f \in L^2$ and
\begin{equation}\label{7.13}
  f^{(n)}(x) = \left\{
                 \begin{array}{ll}
                   f(x), & \hbox{ if } |f(x)| > n\\
                   0, & \hbox{ if } |f(x)| \le n
                 \end{array}
               \right.
\end{equation}
then $\norm{f^{(n)}}_2 \to 0$ as $n \to \infty$ by the dominated convergence theorem and for all $n$, $\norm{f-f^{(n)}}_\infty < \infty$.  It follows from \eqref{7.12} that any $V \in L^2(\bbR^3)+L^\infty(\bbR^3)$ is $-\Delta$--bounded with relative bound zero as operators on $L^2(\bbR^3)$.

\textbf{Step 3.}  In modern language, one shows that if $\calH = \calH_1\otimes\calH_2$ (tensor products are defined, for example, in \cite[Section 3.8]{RA}) and \eqref{7.6} holds, then
\begin{equation}\label{7.14}
  \norm{(B\otimes\bdone)\varphi} \le a\norm{(A\otimes\bdone)\varphi}+ b\norm{\varphi}
\end{equation}
Thus, if $V$ is a function of $x_1$ alone, $V(x_1,\dots,x_N) = v(x_1),\, v \in L^2(\bbR^3)+L^\infty(\bbR^3)$ so that \eqref{7.6} holds for $v$ on $L^2(\bbR^3)$, then it also holds for $B = V(x)$ and $A=-\Delta_1$ on $L^2(\bbR^{3N})$.  Since $|k_1|^2 \le |k|^2$, we conclude that $V$ is $-\Delta$--bounded with relative bound zero on $L^2(\bbR^{3N})$.  By a coordinate change, the same is true for $v(x_i-x_j)$.

Rather than talk about tensor products, Kato used iterated Fourier transforms and states inequalities like
\begin{equation}\label{7.15}
  \sup_{x_1} \left[\int |\varphi(x_1,\dots,x_N)|^2 d^3x_2\dots d^3x_N\right] \le C \int (1+k_1^2)^2 |\hat{\varphi}(k)|^2 d^{3N}k
\end{equation}
which is equivalent to the tensor product results.  This concludes our sketch of Kato's proof of his great theorem.

Kato states in the paper that he had found the results by 1944.  Kato originally submitted the paper to Physical Review.  Physical Review transferred the manuscript to the Transactions of the AMS where it eventually appeared.  They had trouble finding a referee and in the process the manuscript was lost (a serious problem in pre-Xerox days!).  Eventually, von Neumann got involved and helped get the paper accepted.  I've always thought that given how important he knew the paper was, von Neumann should have suggested Annals of Mathematics and used his influence to get it published there. The receipt date of October 15, 1948 on the version published in the Transactions shows a long lag compared to the other papers in the same issue of the Transactions which have receipt dates of Dec., 1949 through June, 1950. Recently after Kato's widow died and left his papers to some mathematicians (see the end of Section \ref{s1}) and some fascinating correspondence of Kato with Kemble and von~Neumann came to light.  There are plans to publish an edited version \cite{KatoLetters}.

It is a puzzle why it took so long for this theorem to be found.  One factor may have been von Neumann's attitude.  Bargmann told me of a conversation several young mathematicians had with von Neumann around 1948 in which von Neumann told them that self--adjointness for atomic Hamiltonians was an impossibly hard problem and that even for the Hydrogen atom, the problem was difficult and open.  This is a little strange since, using spherical symmetry, Hydrogen can be reduced to a direct sum of one dimensional problems.  For such ODEs, there is a powerful limit point--limit circle method named after Weyl and Titchmarsh (although it was Stone, in his 1932 book, who first made it explicit).  Using this, it is easy to see (there is one subtlety for $\ell=0$ since the operator is limit point at 0) that the Hydrogen Hamiltonian is self--adjoint and this appears at least as early as Rellich \cite{RellichLP}.  Of course, this method doesn't work for multielectron atoms.  In any event, it is possible that von Neumann's attitude may have discouraged some from working on the problem.

Still it is surprising that neither Friedrichs nor Rellich found this result.  In exploring this, it is worth noting that there is an alternate to step 2:

\textbf{Step 2$'$}.  On $\bbR^3$, there is the well known operator inequality (discussed further in Section \ref{s10} and in \cite[Section 6.2]{HA}) known as Hardy's inequality ($A \le B$ for positive operators is discussed in Section \ref{s10} and \cite[Section 7.5]{OT}; for this case, it means $\jap{\varphi,A\varphi} \le \jap{\varphi,B\varphi}$ for all $\varphi \in C_0^\infty(\bbR^3)$):
\begin{equation}\label{7.16}
  \frac{1}{4r^2} \le -\Delta
\end{equation}
Since $x \le \epsilon x^2 + \tfrac{1}{4} \epsilon^{-1}$ for $x \in (0,\infty)$, the spectral theorem implies that for any positive, self--adjoint operator, $C$, we have that
\begin{equation}\label{7.16a}
  C \le \epsilon C^2 + \tfrac{1}{4} \epsilon^{-1}
\end{equation}
so using this for $C=-\Delta$, \eqref{7.16} implies that
\begin{equation}\label{7.17}
  \frac{1}{4r^2} \le \epsilon (-\Delta)^2 + \frac{1}{4} \epsilon^{-1}
\end{equation}
equivalently, for $\varphi \in C_0^\infty(\bbR^3)$
\begin{equation}\label{7.18}
  \norm{r^{-1}\varphi}^2 \le 4\epsilon \norm{-\Delta\varphi}^2 + \epsilon^{-1} \norm{\varphi}^2
\end{equation}
which implies that $r^{-1}$ is $-\Delta$--bounded with relative bound zero.

Rellich used Hardy's inequality in his perturbation theory papers \cite{RellichPT} in a closely related context.  Namely he used \eqref{7.16} and \eqref{7.16a} for $C=r^{-1}$ to show that $r^{-1} \le 4\epsilon (-\Delta) + \tfrac{1}{4} \epsilon^{-1}$ to note the semiboundedness of the Hydrogen Hamiltonian.  Since Rellich certainly knew the Kato--Rellich theorem, it appears that he knew steps 1 and 2$'$.

In a sense, it is pointless to speculate why Rellich didn't find Theorem \ref{T7.2}, but it is difficult to resist.  It is possible that he never considered the problem of esa of atomic Hamiltonians, settling for a presumption that using the Friedrichs extension suffices (as Kato suggests in \cite{WP}) but I think that unlikely.  It is possible that he thought about the problem but dismissed it as too difficult and never thought hard about it.  Perhaps the most likely explanation involves Step 3: once you understand it, it is trivial, but until you conceive that it might be true, it might elude you.

Kato's original paper required that the $L^2$ piece have compact support (in the relevant variables).  While it is easy to accommodate global $L^2$, it is true that it is enough to be uniformly locally $L^2$, i.e.
\begin{equation}\label{7.18a}
  \sup_x \int_{|x-y| \le 1} |V(y)|^2 d^3y
\end{equation}
denoted $L^2_{unif}(\bbR^3)$.  It was Stummel \cite{Stummel} who first realized this.  There are general localization techniques, originally developed for form estimates by Ismagilov \cite{IMSI}, Morgan \cite{IMSM} and Sigal \cite{IMSS} (and discussed as the IMS localization formula in \cite[Section 3.1]{CFKS}) which have operator versions.  For a recent paper on these techniques, see Gesztesy et. al. \cite{GMNT}.  For example, \cite[Problem 7.1.9]{OT} proves:

\begin{theorem} \lb{T7.4} For each $\alpha \in \bbZ^\nu$, let $\Delta_\alpha$ be the cube of side $3$ centered at $\alpha$ and $\chi_\alpha$ its characteristic function.  Let $V$ be a measurable function on $\bbR^\nu$ so that for some positive $a,b$ and all $\alpha$ and all $\varphi \in C_0^\infty(\bbR^\nu)$
\begin{equation}\label{7.19}
  \norm{V\chi_\alpha\varphi}_2 \le a\norm{-\Delta\varphi}_2+b\norm{\varphi}_2
\end{equation}
Then for any $\epsilon > 0$, there is a $b_\epsilon$ so that for all $\varphi \in C_0^\infty(\bbR^\nu)$, we have that
\begin{equation}\label{7.20}
  \norm{V\varphi}_2 \le (a+\epsilon)\norm{-\Delta\varphi}+b_\epsilon\norm{\varphi}_2
\end{equation}
In particular, any $V \in L^2_{unif}(\bbR^3)$ is $-\Delta$--bounded on $L^2(\bbR^3)$ with relative bound $0$.
\end{theorem}

In exploring extensions of Theorem \ref{T7.1}, it is very useful to have simple self--adjointness criteria for $-\tfrac{d^2}{dx^2}+q(x)$ on $L^2(0,\infty)$ which then translate to criteria for $-\Delta+V(x)$ if $V(x) = q(|x|)$ is a spherically symmetric potential.  If $q \in L^2_{loc}(0,\infty)$, for each $z \in \bbC$, the set of solutions of $-u''+qu=zu$ (in the sense that $u$ is $C^1$, $u'$ is absolutely continuous and $u''$ is its $L^1_{loc}$ derivative) is two dimensional.  If all solutions are $L^2$ at $\infty$ (resp. $0$), we say that $-\tfrac{d^2}{dx^2}+q(x)$ is limit circle at $\infty$ (resp. $0$).  If it is not limit circle, we say it is limit point.  It is a theorem that whether one is limit point or limit circle is independent of $z$.  However, in the limit point case, whether the set of $L^2$ solutions near infinity is $0$ or $1$ dimensional can be $z$ dependent.  One has the basic

\begin{theorem} [Weyl limit point--limit circle theorem] \lb{T7.5} Let $q \in L^2_{loc}(0,\infty)$.  Then $-\tfrac{d^2}{dx^2}+q(x)$ is esa on $C_0^\infty(0,\infty)$ if and only if $-\tfrac{d^2}{dx^2}+q(x)$ is limit point at both $0$ and $\infty$.
\end{theorem}

\begin{remarks} 1.  This result holds for any interval $(a,b) \subset \bbR$ where $a$ can be $-\infty$ and/or $b$ can be $\infty$.

2.  If it is limit point at only one of $0$ and $\infty$ and limit circle at the other point, the deficiency indices (see \cite[Section 7.1]{OT}  for definitions) are $(1,1)$ and if it is limit circle at both $0$ and $\infty$, they are $(2,2)$.  In particular if it is limit point at $\infty$ and $\int_{0}^{1} |V(x)| dx < \infty$, then the deficiency indices are (1,1) and the extensions are described by boundary conditions $\cos\theta\, u'(0)+\sin\theta\, u(0) = 0$.

3. The ideas behind much of the theorem go back to Weyl \cite{WeylLP1, WeylLP2} in 1910 and predate the notion of self--adjointness.  It was Stone \cite{StoneBk} who first realized the implications for self--adjointness and proved Theorem \ref{T7.5}.  \cite[Thm 7.4.12]{OT} has a succinct proof.  Titchmarsh \cite{TitBk1} reworked the theory so much that it is sometimes called Weyl--Titchmarsh theory.  For additional literature, see \cite{CoddLev, EastKalf, LevitSars}.
\end{remarks}

\begin{example} \lb{E7.6} ($x^{-2}$ on $(0,\infty)$)  Let $q(x) = \beta x^{-2}$.  Trying $x^\alpha$ in $-u''+\beta x^{-2} u = 0$, one finds that $\alpha(\alpha-1)=\beta$ is solved by $\alpha_\pm = \tfrac{1}{2}(1 \pm \sqrt{1+4\beta})$.  For $\beta \ne -\tfrac{1}{4}$, this yields two linearly independent solutions, so a basis.  The larger solution (and sometimes both) is not $L^2$ at infinity, so it is always limit point there.

For $\alpha \ge -\tfrac{1}{4}$, there is a positive solution which implies that $H_\beta \equiv -\tfrac{d^2}{dx^2}+\beta x^{-2} \ge 0$.  If $\alpha < -\tfrac{1}{4}$, the solutions oscillate and the real solutions have infinitely many zeros which implies that the operator is not positive (see \cite[Section 7.4]{OT}).  Thus
\begin{equation}\label{7.20}
  -\tfrac{d^2}{dx^2}+\beta x^{-2} \ge 0 \textrm{ on } C_0^\infty(0,\infty)  \iff  \beta \ge -\tfrac{1}{4}
\end{equation}
This is Hardy's inequality on $L^2(0,\infty)$.

$x^\alpha \notin L^2(0,1) \iff \alpha \le -\tfrac{1}{2}$.  At $\beta = \tfrac{3}{4},\, \alpha_- = -\tfrac{1}{2}$.  Thus $H_\beta$ is always limit point at $\infty$ and is limit point at $0$ if and only if $\beta \ge \tfrac{3}{4}$, i.e.
\begin{equation}\label{7.21}
  -\tfrac{d^2}{dx^2}+\beta x^{-2} \textrm{ is esa on } C_0^\infty(0,\infty)  \iff  \beta \ge \tfrac{3}{4}
\end{equation}
A comparison theorem shows that if
\begin{equation}\label{7.22}
  q(x) \ge \tfrac{3}{4} x^{-2} - c
\end{equation}
for some real $c$, then $-\tfrac{d^2}{dx^2}+q(x)$ is esa on $C_0^\infty(0,\infty)$.
\end{example}

On $\bbR^\nu$, one defines spherical harmonics (see \cite[Section 3.5]{HA}), $\{Y_{\ell m}\}_{m=1;\ell=0,1,\dots}^{D(\nu,\ell)}$ on $S^{\nu - 1}$, the unit sphere in $\bbR^\nu$, to be the restriction to the unit sphere of harmonic polynomials of degree $\ell$.  These polynomials are a vector space of dimension $D(\nu,\ell) = \tfrac{\ell+\nu-2}{\nu-2}\binom{\nu-3+\ell}{\nu-3}$ and $Y_{\ell m}$ are a convenient orthonormal basis.  Any function $f \in \calS(\bbR^\nu)$ can be expanded in the form ($r \in (0,\infty),\, \omega \in S^{\nu - 1}$)
\begin{equation}\label{7.23}
  f(r\omega) = \sum_{\ell,m} r^{-(\nu - 1)/2} f_{\ell m}(r) Y_{\ell m}(\omega)
\end{equation}
(where for $\nu \ge 2$, $f_{\ell m}$ vanishes so rapidly at $r=0$ that $r^{-(\nu - 1)/2}f_{\ell m}(r)$ has a limit as $r \downarrow 0$ which must be zero unless $(\ell m) = (0 1)$).  Moreover, if $\sigma_\nu$ is the area of the unit sphere, then
\begin{equation}\label{7.24}
  \norm{f}^2_{L^2(\bbR^\nu,d^\nu x)} = \sigma_\nu \sum_{\ell, m} \norm{f_{\ell m}}^2_{L^2(\bbR,dr)}
\end{equation}
and
\begin{equation}\label{7.25}
  (\Delta f)_{\ell m} = \left[\frac{d^2}{dr^2}-\frac{(\nu-1)(\nu-3)}{4r^2}-\frac{\ell(\ell+\nu-2)}{r^2}\right]f_{\ell m}
\end{equation}

If $V(\boldsymbol{r}) = q(r)$, then $-\Delta+V$ is a direct sum of operators of the form
\begin{align}\label{7.26}
  H_{\ell m}(V) &= -\frac{d^2}{dr^2}+q_\ell(r) \\
  q_{\ell}(x) &= \frac{(\nu-1)(\nu-3)}{4x^2}+\frac{\ell(\ell+\nu-2)}{x^2}+q(x)
\end{align}
It is easy to see that such direct sums are bounded from below (resp. esa) on $C_{00}^\infty(\bbR^\nu) \equiv C_0^\infty(\bbR^\nu\setminus\{0\})$ if and only if each $H_{\ell m}$ is bounded from below (resp. esa) on $C_0^\infty(0,\infty)$  We conclude that

\begin{proposition} \lb{P7.7} On $\bbR^\nu$, $H_\beta^{(\nu)} \equiv -\Delta+\beta |x|^{-2}$ on $C_{00}^\infty(\bbR^\nu)$ is

(1) Bounded from below
\begin{equation}\label{7.27}
H_\beta^{(\nu)} \ge 0 \iff \beta \ge -\frac{(\nu-2)^2}{4}
\end{equation}

(2) $H_\beta^{(\nu)}$ is esa on $C_{00}^\infty(\bbR^\nu)$ if and only if
\begin{equation}\label{7.28}
  \beta \ge -\frac{\nu(\nu-4)}{4}
\end{equation}
\end{proposition}

\begin{remarks}  1. This uses $-\tfrac{(\nu-1)(\nu-3)}{4}-\tfrac{1}{4} = -\tfrac{(\nu-2)^2}{4}$ and $-\tfrac{(\nu-1)(\nu-3)}{4}+\tfrac{3}{4} = -\tfrac{\nu(\nu-4)}{4}$.

2. \eqref{7.27} is the $\nu$--dimensional Hardy inequality with optimal constant (see Section \ref{s10} below).

3.  By \eqref{7.22}, if $\nu \ge 4$ and $V$ is spherically symmetric and obeys $V(x) \ge -\tfrac{\nu(\nu-4)}{4|x|^2}$, then $-\Delta+V$ is esa--$\nu$ (discussed further in Section 9).

4.  In particular, $C_{00}^\infty(\bbR^\nu)$ is an operator core for $-\Delta$ if and only if $\nu \ge 4$ and a form core for $-\Delta$ if and only if $\nu \ge 2$.

5. By \eqref{7.22}, if $\gamma > 2$, then $-\Delta+\lambda |x|^{-\gamma}$ ($\lambda>0$) is esa on $C_{00}^\infty(\bbR^\nu)$.  If $\nu \ge 5$ and $2 < \gamma < \nu/2$, we have that $|x|^{-\gamma} \in L^2(\bbR^\nu)+L^\infty(\bbR^\nu)$, so one can define $T \equiv -\Delta+\lambda |x|^{-\gamma}$ on $C_{0}^\infty(\bbR^\nu)$ and it is easy to see that $T$ is symmetric.  It follows by general principles \cite[Section 7.1]{OT} that $T$ is esa--$\nu$.

6.  There is an intuition to explain why one loses self--adjointness of $-\Delta-|x|^{-\gamma}$ when $\gamma > 2$.  If $\gamma <2$, in classical mechanics there is an $\tfrac{\ell^2}{|x|^2}$ barrier which dominates the $-|x|^{-\gamma}$, so for almost every initial condition, the classical particle avoids the singularity at the origin. But when $\gamma > 2$, every negative energy  initial condition will fall into the origin in finite time so in classical mechanics, one needs to supplement with a rule about what happens when the particle is captured by the singularity.  The quantum analog is the loss of esa.  There is of course a difference at $\gamma = 2$ where classically there is a problem no matter the coupling but not in quantum mechanics.  This is associated with the uncertainty principle.  In the next section, we'll see that this intuition is also useful to understand what happens with $V$'s going to $-\infty$ at spatial infinity.
\end{remarks}

We summarize in

\begin{example} \lb{E7.8} ($|x|^{-2}$ in $\bbR^\nu$; $\nu \ge 5$)  Rellich's Inequality \cite{RellichIneq} (see also \cite[Problem 7.4.10]{OT}, Section \ref{s10} below, Gesztesy--Littlejohn \cite{GLRellich} or Robinson \cite{Rob17} for a proof of Rellich's inequality via a double commutator estimate like the one before \eqref{3.28Q} and Hardy's inequality; this proof is a variant of one of Schmincke \cite{SchmSing}) says that on $\bbR^\nu, \, \nu \ge 5$, one has
\begin{equation}\label{7.29}
  \frac{\nu(\nu-4)}{4} \norm{|x|^{-2} \varphi} \le \norm{\Delta\varphi}
\end{equation}
for any $\varphi \in C_{0}^\infty(\bbR^\nu)$.  (Of course, this also hold if $\nu \le 4$ since the left side is negative or $0$ (maybe $-\infty$) in that case.)  This says that $B=-|x|^{-2}$ is $-\Delta$--bounded if $\nu \ge 5$.  When $B$ is $A$--bounded with $A$ positive, there are three natural values of $\lambda$, call them $\lambda_1, \lambda_2, \lambda_3$ with $0<\lambda_1\le\lambda_2\le\lambda_3$ so that

$\lambda B$ is A-bounded with relative bound $<1$ if and only if $0 \le \lambda < \lambda_1$.

$A+\lambda B$ is esa on $D(A)$ if $0 \le \lambda < \lambda_2$ and not if $\lambda > \lambda_2$.

$A+\lambda B$ is bounded from below if $0 \le \lambda < \lambda_3$ and not if $\lambda > \lambda_3$

\noindent By \eqref{7.27}, \eqref{7.28} and \eqref{7.29}, we see that
\begin{equation}\label{7.30}
  \lambda_1(\nu) = \lambda_2(\nu) = \frac{\nu(\nu-4)}{4}, \qquad \lambda_3(\nu)=\frac{(\nu-2)^2}{4}
\end{equation}

There is no reason that $\lambda_1$ has to equal $\lambda_2$, i.e. esa can persist past the point where the relative bound is $1$.  For example, if $A=-\Delta$ on $\bbR^\nu$ with $\nu \ge 5$ and
\begin{equation*}
  B = -|x|^{-2}+ 2|x-e|^{-2}
\end{equation*}
for $e$ some fixed, non--zero vector, then one can prove that $\lambda_1 = \tfrac{\nu(\nu-4)}{8},\,\lambda_2= \tfrac{\nu(\nu-4)}{2}$ and $\lambda_3=\tfrac{(\nu-2)^2}{4}$.
\end{example}

We turn now to the extensions of Theorem \ref{T7.1} to $\nu \ne 3$.  The first results are due to Stummel which we'll discuss later.  In 1959, Brownell \cite{Brownell} proved any $V \in L^p(\bbR^\nu)+L^\infty(\bbR^\nu)$ is $-\Delta$--bounded with relative bound zero (see also Nilsson \cite{Nilsson}) if
\begin{equation}\label{7.31}
  p=2\, (\nu \le 3), \qquad p > \nu/2 \, (\nu \ge 4)
\end{equation}
Since $|x|^{-2} \in L^p+L^\infty$ for any $p < \nu/2$, we see that \eqref{7.31} is optimal, except perhaps for the borderline case $p=\nu/2$ (see remark 2 below).   Brownell mimicked Kato's proof, except that \eqref{7.11} is replaced by
\begin{equation}\label{7.32}
  \norm{\varphi}_r \le C(\norm{\Delta\varphi}_2 + \norm{\varphi}_2)
\end{equation}
for any $r>r_0$ where $r_0^{-1} = \tfrac{1}{2} -\tfrac{2}{\nu}$ when $\nu \ge 4$.  In place of \eqref{7.10a}, Brownell used a Hausdorff--Young inequality (see \cite[Theorem 6.6.2]{RA}).

\eqref{7.32} is what is known as an inhomogeneous Sobolev inequality.  There are now (and even then, but not so widely known) sharper inequalities than \eqref{7.32}.  Recall that $L^p_w(\bbR^\nu)$, the weak $L^p$ space is defined as the measurable functions for which $\norm{f}_{p,w}^*$ is finite where
\begin{equation}\label{7.33}
  |\{x\,|\,|f(x)|>t\}| \le \frac{(\norm{f}_{p,w}^*)^p}{t^p}
\end{equation}
$\norm{f}_{p,w}^*$ is defined to be the minimal constant so that \eqref{7.33} holds.  It is not a norm but, for $p > 1$, it is equivalent to one -- see \cite[Section 2.2]{HA}.  One has that $L^p(\bbR^\nu) \subset L^p_w(\bbR^\nu)$ but for $f \in L^p_w$, one can have $\int |f(x)|^p d^\nu x$ logarithmically divergent, for example $f(x) = |x|^{-\nu/p}$ is in $L^p_w$ but not $L^p$.

We call $p$, \emph{$\nu$--canonical} if $p=2$ for $\nu \le 3$, $p > 2$ if $\nu=4$ and $p=\nu/2$ if $p \ge 5$. The optimal $L^p$ extension of Theorem \ref{T7.1} is

\begin{theorem} \lb{T7.9}  Let $p$ be $\nu$--canonical.  Then $V \in L^p(\bbR^\nu)+L^\infty(\bbR^\nu)$ is $-\Delta$--bounded with relative bound zero.  If $\nu \ge 5$, $V \in L^p_w(\bbR^\nu)+L^\infty(\bbR^\nu)$ is $-\Delta$--bounded on $L^2(\bbR^\nu)$.
\end{theorem}

\begin{remarks}  1.  In the $L^p_w$ case, the relative bound may not be zero; for example $V(x) = |x|^{-2}$ as discussed above.  Since any $L^p$ function can be written as the sum of a bounded function and an $L^p_w$ function of arbitrarily small $\norm{\cdot}_{p, w}^*$, the second sentence implies the first.

2. One proof of the $\nu \ge 5$ result uses a theorem of Stein--Weiss \cite{StW} (see \cite[Section 6.2]{HA}) that if $f \in L^{\nu/2}_w(\bbR^\nu)$ and $g \in L^{\nu/(\nu-2)}_w(\bbR^\nu)$, then $h \mapsto g*(fh)$ maps $L^2$ to $L^2$.  Another proof uses Rellich's inequality and Brascamp--Lieb--Luttinger inequalities (see \cite{BLL, Rogers1, Rogers2, Rogers3} or \cite{SimonConv}).

3.  That one can use $p = \nu/2$ rather than $p > \nu/2$ when $\nu \ge 5$ was noted first by Faris \cite{FarisTP}.

I'm not sure to whom to attribute the use of sharp Sobolev and Stein--Weiss inequalities.  I learned it in about 1968 from a course of lectures of Ed Nelson and it was popularized by Reed--Simon \cite{RS2}.
\end{remarks}

When Brownell did his work, he was unaware that his results were a consequence of a different approach of Stummel \cite{Stummel} (Brownell thanks the referee for telling him about Stummel's work).  Stummel considered the class, $S_{\nu,\alpha}$, of functions, $V(x)$, on $\bbR^\nu$ obeying
\begin{equation}\label{7.34}
  \norm{f}_{\nu,\alpha} = \sup_x \int_{|x-y| \le 1} |V(y)|\, |x-y|^{-(\nu-4+\alpha)} d^\nu y
\end{equation}
is finite.  Here $\alpha > 0$ and $\alpha \ge (4-\nu)$, so if $\nu \le 3$, one has that $S_{\nu,4-\nu} = L^2_{unif}$.  Stummel \cite{Stummel} proves that if $V \in S_{\nu,\alpha}$ with $\alpha$ as above, then $V$ is $-\Delta$--bounded with relative bound $0$.  This has several advantages over the Kato--Brownell approach:

(a) Since $\int_{|w| \le 1} |w|^{-\beta+\nu} dw_{\kappa+1}\dots dw_\nu \sim |(w_1,\dots.w_\kappa)|^{-\beta+\kappa}$ where the tilde means comparable in terms of upper and lower bounds, extra variables go through directly and there is no need for step 3 in Kato's proof.

(b) As we've seen, it is uniformly local, i.e. to be in a Stummel class rather than $L^p$, one only needs $L^p_{unif}$.

(c) By Young's inequality \cite[Theorem 6.6.3]{RA}, the Brownell $L^p$ condition implies Stummel's condition, so Stummel's result is stronger.

Stummel's proof relies on the fact that $((-\Delta)^2+1)^{-1}$ has an integral kernel diverging as $|x-y|^{-(\nu-4)}$ for $|x-y|$ small and decaying exponentially for $|x-y|$ large.  As with Brownell's paper, Stummel's $\alpha > 0$ condition isn't needed if $\nu \ge 5$.  The issue is that instead of using Young's inequality, one needs to use the stronger Hardy--Littlewood--Sobolev inequalities \cite[Theorem 6.2.1]{HA} which were not well known in the 1950s.  Motivated by Kato's introduction of the class $K_\nu$ (see Section \ref{s9}), in \cite{CFKS}, I introduced the class $S_\nu$ which I defined as those measurable $V$ on $\bbR^\nu$ with
\begin{equation}\label{7.35}
  \left\{
    \begin{array}{ll}
      \sup_x \int_{|x-y| \le 1} |V(y)|^2 d^\nu y < \infty, & \hbox{ if }  \nu \le 3 \\
      \lim_{\alpha\downarrow 0} \sup_x \int_{|x-y| \le \alpha} \log(|x-y|^{-1}) |V(y)|^2 d^\nu y = 0 , & \hbox{ if } \nu=4\\
      \lim_{\alpha\downarrow 0} \sup_x \int_{|x-y| \le \alpha} |x-y|^{-(\nu-4)} |V(y)|^2 d^\nu y = 0, & \hbox{ if } \nu \ge 5
    \end{array}
  \right.
\end{equation}
Then one has

\begin{theorem} [\cite{CFKS}; Section 1.2] \lb{T7.10} A multiplication operator, $V \in S_\nu$ is $-\Delta$--bounded with relative bound zero.  Conversely, if $V$ is a multiplication operator so that for some $a,b>0$ and some $\delta \in (0,1)$ and for all $\epsilon \in (0,1)$ and $\varphi \in D(-\Delta)$, one has that
\begin{equation}\label{7.36}
  \norm{V\varphi}_2^2 \le \epsilon\norm{\Delta\varphi}_2^2 + a \exp(b\epsilon^{-\delta})\norm{\varphi}_2^2
\end{equation}
then $V \in S_\nu$.
\end{theorem}

One key to the proof is a simple necessary and sufficient condition

\begin{theorem} [\cite{CFKS}; Section 1.2] \lb{T7.11} A multiplication operator, $V$, is in $S_\nu$ if and only if $\lim_{E \to \infty} \norm{(-\Delta+E)^{-2}|V|^2}_{\infty,\infty} = 0$ where $\norm{\cdot}_{p,p}$ is the operator norm from $L^p(\bbR^\nu)$ to itself.
\end{theorem}

For example, to get the boundedness, one uses duality and interpolation to see that

\begin{align*}
  \lim_{E \to \infty} \norm{(-\Delta+E)^{-2}|V|^2}_{\infty,\infty} = 0 &\Rightarrow \lim_{E \to \infty} \norm{|V|(-\Delta+E)^{-2}|V|}_{2,2}=0 \\
                                                                       &\iff \lim_{E \to \infty} \norm{|V|(-\Delta+E)^{-1}}_{2,2} = 0
\end{align*}

This concludes what I want to say about uses of the Kato--Rellich theorem to study esa of Schr\"{o}dinger operators.  In \cite{KatoHisThm}, Kato also remarks on self--adjointness of Dirac Coulomb Hamiltonians, an issue he returned to several times as we'll see in Section \ref{s10}.

Let $\alpha_1, \alpha_2, \alpha_3, \alpha_4=\beta$ be four $4\times 4$ matrices obeying
\begin{equation}\label{7.37}
  \alpha_i\alpha_j+\alpha_j\alpha_i=2\delta_{ij} \bdone; \qquad i,j=1,\dots,4
\end{equation}
Our Hilbert space is $\calH=L^2(\bbR^3;\bbC^4,d^3x)$ of $\bbC^4$ valued $L^2$ functions.  The free Dirac operator is
\begin{equation}\label{7.38}
  T_0 = \sum_{j=1}^{3} \alpha_j p_j + m \beta; \qquad p_j=\frac{1}{i}\frac{\partial}{\partial x}
\end{equation}

One has, using \eqref{7.37}, that formally
\begin{equation}\label{7.39}
  T_0^2=-\Delta+m^2
\end{equation}
Using Fourier transform, one can prove that $T_0$ is esa on $C^\infty_0(\bbR^3;\bbC^4)$ with the domain of the closure being $\{\varphi\,|\, \int (1+p^2)|\hat{\varphi}(p)|^2\,d^3p< \infty\}$.  The Dirac Coulomb operator is
\begin{equation}\label{7.40}
  T=T_0+\frac{\mu}{|x|}
\end{equation}
In terms of the nuclear charge, $Z$, one has that $\mu=Z\alpha$ where $\alpha$ is the fine structure constant, $\alpha^{-1}=137.035999139\dots$, so a given $\mu$ corresponds to $Z\sim 137\mu$.  In \cite{KatoHisThm}, Kato notes without proof that his method proves esa of Dirac Coulomb operators on $C^\infty_0(\bbR^3;\bbC^4)$ for $Z \le 68$.  Clearly he had the result for $\mu < \tfrac{1}{2}$ and $68$ is the integral part of $\tfrac{1}{2}\alpha^{-1}$.  Raised as a physicist, Kato thought of integral $Z$.

In fact
\begin{equation*}
  \left\lVert\frac{\mu}{r}\varphi\right\rVert^2 \le \norm{T_0\varphi}^2+c\norm{\varphi}^2 \iff \frac{\mu^2}{r^2}\le T_0^2+c
\end{equation*}
Hardy's inequality says that on $\bbR^3$, $(4r^2)^{-1} \le p^2$ (with no larger constant).  This and \eqref{7.39} shows that $r^{-1}$ is $T_0$--bounded with precise relative bound $2$, so the Kato--Rellich Theorem implies self--adjointness if and only if $\mu < \tfrac{1}{2}$.  But \eqref{7.40} can be essentially self--adjoint on $C^\infty_0(\bbR^3;\bbC^4)$ even though the Kato--Rellich theorem doesn't work -- in the language of Example \ref{E7.8} it can happen that $\lambda_1$ is strictly smaller than $\lambda_2$.  Indeed, it is known that

\begin{theorem} \lb{T7.12} \eqref{7.40} is esa on $C^\infty_0(\bbR^3;\bbC^4)$ if and only if
\begin{equation}\label{7.41}
  |\mu| \le \tfrac{1}{2}\sqrt{3}
\end{equation}
\end{theorem}

This result is essentially due to Rellich \cite{RellichLP} in 1943.  He proved it using spherical symmetry and applying the Weyl limit--limit circle theory (Theorem \ref{T7.5}).  We say ``essentially'' because at the time he did this, the Weyl theory had not been proven for systems and \eqref{7.40} is a system.  This theory for systems was established by Kodaira \cite{KodairaLP} in 1951 (see also Weidmann \cite{WeidLP}) so Theorem \ref{T7.12} should be regarded as due to Rellich--Kodaira.  Interestingly enough, Kato seems to have been unaware of this result when he wrote his book (second edition was 1976).

One can also consider $T_0+V$ where $V$ is not necessarily spherically symmetric and $V$ obeys
\begin{equation}\label{7.42}
  |V(x)| \le \frac{\mu}{|x|}
\end{equation}
By Kato's argument, one can use the Kato--Rellich theorem to get esa when $|\mu| < \tfrac{1}{2}$.  Schmincke \cite{SchmDirac1} proved

\begin{theorem} \lb{T7.13} Let $V$ obey \eqref{7.42} where $\mu <\tfrac{1}{2}\sqrt{3}$.  Then $T_0+V$ is esa on $C^\infty_0(\bbR^3;\bbC^4)$.
\end{theorem}

We'll return to Dirac operators in Section 8 and at the end of section 10.  Having mentioned a result of Schmincke, I should mention that in the 1970s and early 1980s there was a lively school founded by G\"{u}nter Hellwig that produced a cornucopia of results on esa questions for Schr\"{o}dinger and Dirac operators.  Among the group were H. Cycon, H. Kalf, U.-W. Schmincke, R. W\"{u}st and J. Walter.

This said, there is a sense in which Kato's critical value $\mu=\tfrac{1}{2}$ is connected to loss of esa.  Arai \cite{AraiDirac1, AraiDirac2} has shown that for any $\mu > \tfrac{1}{2}$ there is a symmetric matrix valued potential $Q(x)$ with $\norm{Q(x)}=\mu |x|^{-1}$ for all $x$ so that $T_0+Q$ is not esa on $C^\infty_0(\bbR^3;\bbC^4)$, so Theorems \ref{T7.12} and \ref{T7.13} depend on scalar potentials.

%%%%%%%%%%%%%%%%%%%%%%%%%%%%%%%%%%%%%%%%%%%%%%%%%%%%%%%%%%%%%%
\section{Self--Adjointness, II: The Kato--Ikebe Paper} \lb{s8}
%%%%%%%%%%%%%%%%%%%%%%%%%%%%%%%%%%%%%%%%%%%%%%%%%%%%%%%%%%%%%%

Kato was clearly aware that his great 1951 paper didn't include the Stark Hamiltonian where $H$ isn't bounded from below, and in fact $\eta(x) \equiv \int_{|x-y| \le 1} |\min(V(y),0)|^2 dy \to \infty$ if one takes $x \to \infty$ in a suitable direction.  For esa, one needs restrictions on the growth of $\eta$ at infinity (whereas, we'll see in Section \ref{s9}, if $|\min(V(y),0)|$ is replaced by $\max(V(y),0)$, no restriction is needed).  To understand this, it is useful to first consider one dimension.  Suppose that $V(x) \to -\infty$ as $x \to \infty$.  In classical mechanics, if a particle of mass $m$ starts at $x=c$ with zero speed, $V(c) = 0$ and $V'(x) < 0$ on $(c,\infty)$, the particle will move to the right.  By conservation of energy, the speed when the particle is at point $x > c$ will be $v(x) = \sqrt{-V(x)}$ if $\tfrac{1}{2}m = 1$.  The time to get from $c$ to $x_0 > c$ is thus $\int_{c}^{x_0} \tfrac{dx}{\sqrt{-V(x)}}$.  Thus the key issue is whether $\int_{c}^{\infty} \tfrac{dx}{\sqrt{-V(x)}}$ is finite or not.  If it is finite, the particle gets to infinity in finite time and the motion is incomplete.  One expects that the quantum mechanical equivalent is that $-\tfrac{d^2}{dx^2}+V(x)$ is esa if and only if $\int_{c}^{\infty} \tfrac{dx}{\sqrt{-V(x)}} =\infty$.  In particular, if $V(x) = -\lambda|x|^\alpha$, this suggests esa if and only if $\alpha \le 2$.

The classical/quantum intuition can fail if $V(x)$ has severe oscillations or interspersed high bumps (see Rauch--Reed \cite{RR} or Sears \cite{Sears} for examples).  These esa results for ODEs were studied in the late 1940s using limit point--limit circle methods.  Under the non--oscillation assumption (and $V(x) < 0$)
\begin{equation*}
  \int_{c}^{\infty} \left(\frac{[\sqrt{-V}]'}{(-V)^{3/2}}\right)' (-V)^{-1/4} \, dx < \infty
\end{equation*}
(if $V(x) = -x^\alpha$, the integrand is $x^{-(5\alpha+8)/4}$, so there have to be severe oscillations for this to fail),  Wintner \cite{Wintner1947} proved in 1947 that $-\tfrac{d^2}{dx^2}+V(x)$  is limit point at $\infty$ if and only if $\int_{c}^{\infty} \tfrac{dx}{\sqrt{-V(x)}} = \infty$.  Slightly later, in 1949, Levinson \cite{Lev} proved that it is limit point at infinity if there is a positive comparison function, $M(x)$, so that $V(x) > -M(x)$ near infinity, $M'(x)/M(x)^{3/2}$ bounded and $\int_{c}^{\infty} \tfrac{dx}{\sqrt{M(x)}} = \infty$.  For proofs, see \cite{RS2}.

This suggests that a good condition for esa--$\nu$ of $-\Delta+V(x)$ should be
\begin{equation}\label{8.1}
  V(x) \ge -c|x|^2-d
\end{equation}
Indeed, in 1959, Nilsson \cite{Nilsson} and Wienholtz \cite{Wienholtz2} independently proved that

\begin{theorem} [Nilsson--Wienholtz] \lb{T8.1}  If $V(x)$ is a continuous function of $\bbR^\nu$ obeying \eqref{8.1}, then $-\Delta+V$ is esa--$\nu$.
\end{theorem}

Further developments (all later than the Ikebe--Kato paper discussed below) are due to Hellwig \cite{H1, H2, H3}, Rohde \cite{Ro1, Ro2} and Walter \cite{Walt79}.  In 1962, Kato and his former student Ikebe \cite{IK} studied operators of the form
\begin{equation}\label{8.2}
  -\sum_{j,k=1}^{\nu} c_{jk}(x)\left(\frac{\partial}{\partial x_j}-ia_j\right) \left(\frac{\partial}{\partial x_k}-ia_k\right) + V(x)
\end{equation}
where $c_{jk}(x)$ and $a_j(x)$ are $C^2$ functions and for each $x$, $c_{jk}(x)$ is a strictly positive matrix.  For quantum mechanics, one only considers $c_{ij}(x) = \delta_{ij}$ (at least if one ignores quantum mechanics on curved manifolds) and our discussion will be limited to that case.

Wienholtz had also considered first order terms but didn't write it in the form \eqref{8.2} which is the right form for quantum physics; $a_j(x)$ is the vector potential, i.e. B=da is the magnetic field.  Ikebe--Kato had the important realization that one needs no global hypothesis on $a$, i.e. any growth at $\infty$ of $a$ is allowed.  While they had too strong a local hypothesis on local behavior of $a$ (see Section \ref{s9}), their discovery on behavior at $\infty$ was important.

For $V$, they supposed that $V=V_1+V_2$ where $V_2$ is in a Stummel space, $S_{\nu,\alpha},\, \alpha >0$ and $V_1(x) \ge -q(|x|)$ where $q$ is increasing and obeys $\int^{\infty} q(r)^{-1/2} dr = \infty$.  Unlike Wienholtz, they could allow local singularities such as atoms in Stark fields.

\begin{wrapfigure}{l}{0.7\textwidth}
   \centering
   \vspace{-0.3cm}
   \includegraphics[width=0.7\textwidth]{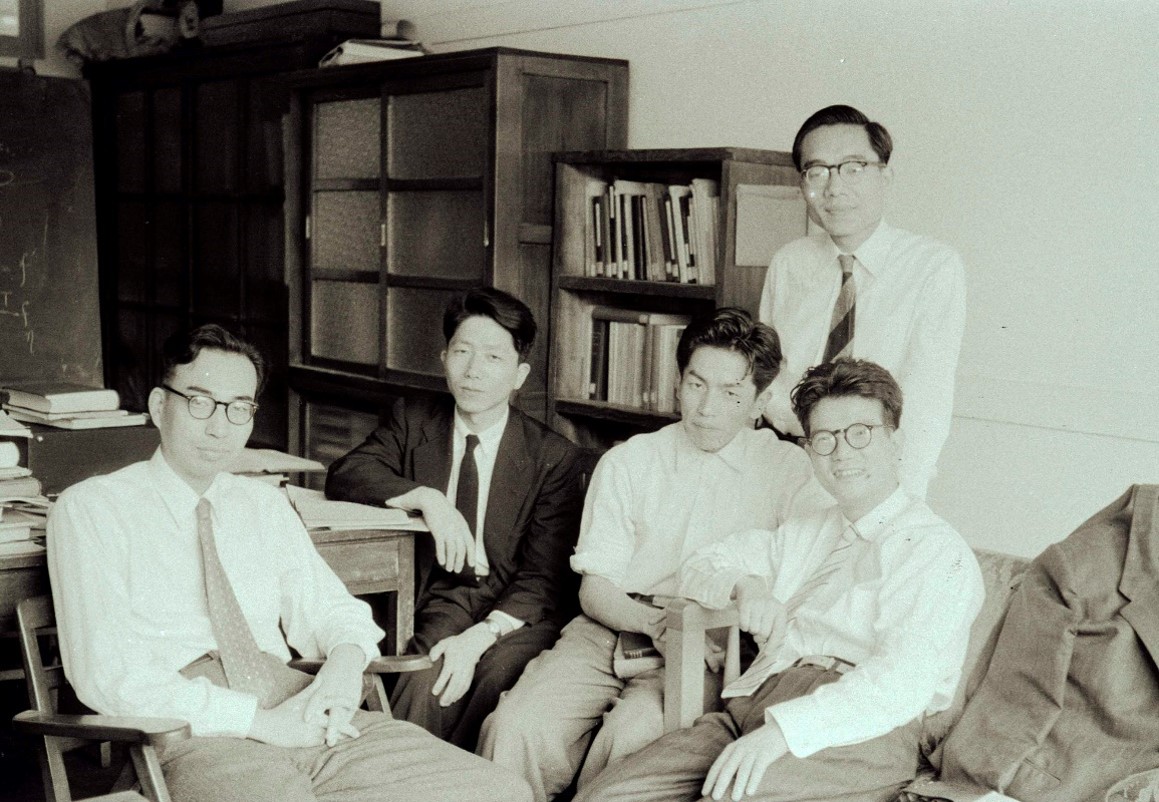}
The Kato group, late 1950s. \\
S.T. Kuroda (standing), T. Kato, T. Ikebe, H.~Fujita, Y. Nakata \lb{figure3}
\vspace{-0.2cm}
\end{wrapfigure}
Rather than discuss their techniques, I want to sketch two approaches to Wienholtz's result which allow local singularities and are of especial elegance.  For one of them, Kato made an important contribution. The first approach is due to Chernoff \cite{Chernoff1, Chernoff2} as modified by Kato \cite{KatoChernoff} and the second approach is due to Faris--Lavine \cite{FL}.  Interesting enough, each utilizes a self--adjointness criterion of Ed Nelson but two different criteria that he developed in different contexts.  Here is the criteria for Chernoff's method (which Nelson developed in his study of the relation between unitary group representations and their infinitesimal generators).

\begin{theorem} [Chernoff--Nelson Theorem] \lb{T8.2} Let $A$ be a self-adjoint operator and $U_t=e^{itA},\,t \in \bbR$, the induced unitary group. Suppose that $\calD$ is a dense subspace of $\calH$ with $\calD \subset D(A^\ell)$ for some $\ell = 1,2,\dots$ and suppose that for all $t$, we have that $U_t[\calD] \subset \calD$.  Then $\calD$ is a core for $A,A^2,\dots,A^\ell$.
\end{theorem}

\begin{remarks}  1.  Recall that Stone's theorem \cite[Theorem 7.3.1]{OT} says there is a one--one correspondence between one--parameter unitary groups and self--adjoint operators, via $U_t=e^{itA},\,t \in \bbR$.

2. Chernoff considers the case $\calD \subset D^\infty(A) \equiv \cap_\ell D^\ell(A)$ in which case $\calD$ is a core for $A^\ell$ for all $\ell$.

3. Nelson \cite{NelsonAnal} did the case $\ell=1$ and Chernoff \cite{Chernoff1} noted his argument can be used for general $\ell$.

4. The argument is simple.  Let $B=A^k\restriction \calD$ for some $k=1,\dots,\ell$.  Suppose that $B^*\psi = i\psi$.  Let $\varphi \in \calD$ and let $f(t) = \jap{\psi,U_t\varphi}$.  Then since $U_t\varphi \in D(A^k)$, we have that $f$ is a $C^k$ function and
\begin{align}
  f^{(k)}(t) &= \jap{\psi,(iA)^kU_t\varphi} = i^k \jap{\psi,BU_t\varphi}\nonumber \\
             &= i^k \jap{B^*\psi,U_t\varphi} = -i^{k+1} f(t) \lb{8.3}
\end{align}
If $g(t) = e^{i\alpha t}$, then $g$ solves \eqref{8.3} if and only if $(i\alpha)^k = -i^{k+1}$, i.e. $\alpha^k = -i$.  No solution of this is real, so $g$ is a linear combination of exponentials which grow at different rates at either $+\infty$ or $-\infty$, so the only bounded solution is $0$. Since $|f(t)| \le \norm{\psi}\norm{\varphi}$, we conclude that $f(0)=0$ so $\psi \perp \calD$.  Since $\calD$ is dense, $\psi = 0$, i.e. $\ker(B^*-i)=\{0\}$. Similarly, $\ker(B^*+i)=\{0\}$, so B is esa.
\end{remarks}

Kato proved his famous self--adjointness result to be able to solve the time dependent Schr\"{o}dinger equation, $\dot{\psi}_t = -iH\psi_t$.  Chernoff turned this argument around!  If one can solve the equation $\dot{\psi}_t = -iA\psi_t$ for a dense set $\calD$ in $D^\infty(A)$ and prove that $\psi_{t=0} \in \calD \Rightarrow \psi_t \in \calD$, then by Theorem \ref{T8.2}, all powers of A are esa on $\calD$.  He combined this with existence and smoothness results of Friedrichs \cite{FriedHyper} and Lax \cite{LaxHyper} for hyperbolic equations plus finite propagation speed to show that if $A$ is a hyperbolic equation, then the solution map takes $C_0^\infty$ to itself.

In particular, since the Dirac equation is hyperbolic, Chernoff proved

\begin{theorem} [Chernoff \cite{Chernoff1}] \lb{T8.3} If $T_0$ is the free Dirac operator, \eqref{7.38}, and $V$ is a $C^\infty(\bbR^3)$ function, then $T = T_0+V$ and all its powers are esa on $C_0^\infty(\bbR^3;\bbC^4)$.
\end{theorem}

\begin{wrapfigure}{r}{0.4\textwidth}
   \centering
   \vspace{-0.4cm}
   \includegraphics[width=0.4\textwidth]{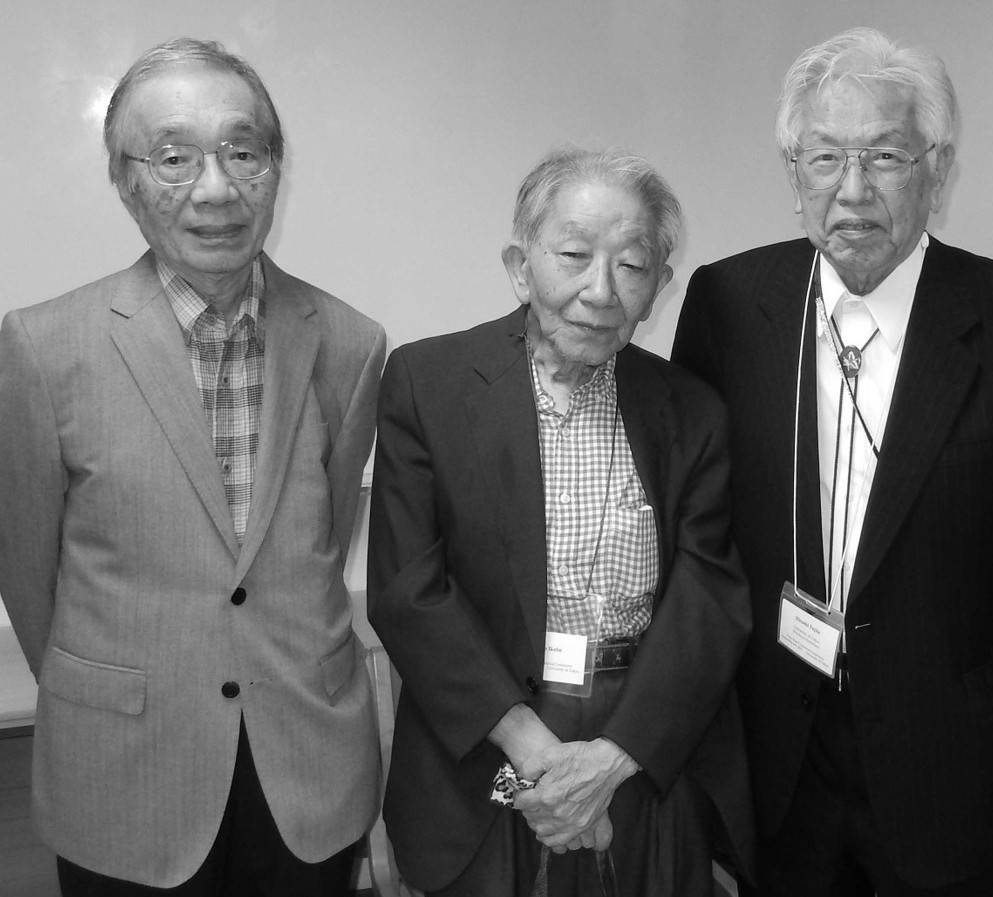}
S. Kuroda, T. Ikebe, H. Fujita recently \lb{figure4}
\vspace{-0.4cm}
\end{wrapfigure}
Notice that there are no restrictions on the growth of $V$ at $\infty$.  This is an expression of the fact that for the Dirac equation, no boundary condition is needed at infinity -- intuitively, this is because the particle cannot get to infinity in finite time because speeds are bounded by the speed of light!  Several years after his initial paper, Chernoff \cite{Chernoff2} used results on solutions of singular hyperbolic equations and proved the following version of the fact that Dirac equations have no boundary condition at infinity:

\begin{theorem} \lb{T8.4} Let $T_0$ be the free Dirac equation and $V \in L^2_{loc}(\bbR^3)$ (so $T_0+V$ is defined on $C_0^\infty(\bbR^3;\bbC^4)$).  Suppose for each $x_0 \in \bbR^3$, there is a $V^{(x_0)}$ equal to $V$ in a neighborhood of $x_0$ and so that $T_0+V^{(x_0)}$ is esa on $C_0^\infty(\bbR^3;\bbC^4)$.  Then $T_0+V$ is esa on $C_0^\infty(\bbR^3;\bbC^4)$.
\end{theorem}

Combining this with Schmincke's result (Theorem \ref{T7.13}) one gets

\begin{corollary} \lb{C8.5} Let $T_0$ be a free Dirac operator.  Let $V$ be a measurable function so that for some sequence $\{x_j\}_{j=1}^N$ (with $N$ finite or infinite) with no finite limit point, we have that

(a) There are constants $\mu_j < \sqrt{3}/2$ and $C_j$ so that for $x$ near $x_j$, say $x$ obeys $|x-x_j| \le \tfrac{1}{2} \min_{k\ne j} |x_j-x_k|$, one has that
\begin{equation}\label{8.4}
  |V(x)| \le \mu_j |x-x_j|^{-1} + C_j
\end{equation}

(b) V is locally bounded near any $x \notin \{x_j\}_{j=1}^N$.

Then $T_0+V$ is esa on $C_0^\infty(\bbR^3;\bbC^4)$.
\end{corollary}

Other results on esa for Dirac operators which are finite sums of Coulomb potentials include \cite{JorgDirac, Karn, NenciuDirac, LRejto1, LRejto2, KlausDirac}.

At first sight, this lovely idea seems to have nothing to do with Schr\"{o}dinger operators since that equation isn't hyperbolic; after all it has infinite propagation speed and even for the free case, the dynamical unitary group doesn't leave the $C_0^\infty$ functions fixed.  But the wave equation
\begin{equation}\label{8.5}
  \frac{\partial^2 u}{\partial t^2}=(\Delta - V) u
\end{equation}
is hyperbolic (and has finite propagation speed, namely 1).  It is second order in $t$ but can be written as a first order equation:
\begin{equation}\label{8.5}
  v=\frac{\partial u}{\partial t}, \qquad \frac{\partial v}{\partial t} = -Bu, \qquad B=-\Delta+V
\end{equation}
or equivalently
\begin{equation}\label{8.6}
  \frac{\partial}{\partial t}\left(
                               \begin{array}{c}
                                 u \\
                                 v \\
                               \end{array}
                             \right) =-iA \left(
                                            \begin{array}{c}
                                              u \\
                                              u \\
                                            \end{array}
                                          \right); \qquad -iA = \left(
                                                                  \begin{array}{cc}
                                                                    0 & \bdone \\
                                                                    -B & 0 \\
                                                                  \end{array}
                                                                \right)
\end{equation}
If $V$ is in $C^\infty(\bbR^\nu)$, one can use hyperbolic theory to prove solutions exist for $(u(0),v(0)) \in C_0^\infty(\bbR^\nu) \times C_0^\infty(\bbR^\nu)$ and the solution remains in this space.  To apply Theorem \ref{T8.2}, we need this dynamics to be unitary.  The energy
\begin{equation}\label{8.7}
  E(u,v) = \jap{v,v}+\jap{u,Bu}
\end{equation}
is formally conserved, so it is natural to use $E$ as the square of a Hilbert space norm.  For this to work, one needs that $B \ge c\bdone$ with $c > 0$.  Actually, so long as $B$ is bounded from below we can add a constant to $B$ so that $B \ge \bdone$ which we'll assume.  When this is so, one can prove that on the Hilbert space $L^2(\bbR^\nu)\oplus Q(-\Delta+V)$ (where $Q$ is the quadratic form of the Friedrichs extension as discussed in Section \ref{s10}), $e^{-itA}$ with $A$ given by \eqref{8.6} is a unitary group which leaves $\calD = C_0^\infty(\bbR^\nu)\oplus C_0^\infty(\bbR^\nu)$ invariant and with $\calD \subset D^\infty(A)$.  We note that
\begin{equation}\label{8.8}
  A^2 = -(iA)^2 = \left(
                    \begin{array}{cc}
                      B & 0 \\
                      0 & B \\
                    \end{array}
                  \right)
\end{equation}
on $\calD$.  We have thus related the Schr\"{o}dinger equation to the square of a hyperbolic equation so we can use Chernoff's idea to conclude that

\begin{theorem} \lb{T8.6} If $V$ is $C^\infty(\bbR^\nu)$, so that $-\Delta+V$ is bounded from below on $C_0^\infty(\bbR^\nu)$, i.e. for some $c$ and all $u \in C_0^\infty(\bbR^\nu)$
\begin{equation}\label{8.9}
  \jap{u,(-\Delta+V)u} \ge c\jap{u,u}
\end{equation}
then $-\Delta+V$ is esa--$\nu$.
\end{theorem}

\begin{remarks}  1.  This proof of the result appeared in Chernoff \cite{Chernoff1}, but the result itself appeared earlier in Povzner \cite{Povz1} and Wienholtz \cite{Wienholtz1}.

2.  In his second paper, Chernoff \cite{Chernoff2} handled singular $V$'s and also used the idea of Kato we'll describe shortly and also Kato's inequality ideas (see Section \ref{s9}).  He proved that $-\Delta+V$ is esa--$\nu$ if $V=U-W$ with $U, W \ge 0$, $U \in L^2_{loc}(\bbR^\nu), W \in L^p_{loc}(\bbR^\nu)$ (with  $p$ $\nu$--canonical) and $-\Delta+V+cx^2$ bounded from below for some $c>0$.
\end{remarks}

In \cite{KatoChernoff}, Kato showed how to modify Chernoff's argument to extend Theorem \ref{T8.6} to replace the condition that $-\Delta+V$ is bounded from below by the condition that for some $c>0$, one has that $-\Delta+V+cx^2$ is bounded from below (and thereby gets a Wienholtz--Ikebe--Kato type of result).  Kato's idea (when $c=1$) was to solve $\tfrac{\partial^2 u}{\partial t^2} = (\Delta-V)u-4t^2u$.  He was able to prove that $\norm{u(t)}_2$ (which is bounded in the case $-\Delta+V$  is bounded below) doesn't grow worse than $|t|^3$ and then push through a variant of the Chernoff--Nelson argument (since a $|t|^3$ bound can eliminate exponential growth).

This completes our discussion of the Chernoff approach.  The underlying self--adjointness criterion of Nelson needed for the Faris--Lavine approach is

\begin{theorem} [Nelson's Commutator Theorem \cite{NelsonCommutator}] \lb{T8.7} Let $A, N$ be two symmetric operators so that $N$ is self--adjoint with $N \ge 1$.  Suppose that $D(N) \subset D(A)$ and there are constants $c_1$ and $c_2$ so that for all $\varphi,\psi \in D(N)$ we have that
\begin{equation}\label{8.9a}
  |\jap{\varphi,A\varphi}| \le c_1 \jap{\varphi,N\varphi}
\end{equation}
\begin{equation}\label{8.10}
  |\jap{A\varphi,N\varphi} - \jap{N\varphi,A\varphi}| \le c_2\jap{\varphi,N\varphi}
\end{equation}
Then $A$ is esa on any core for $N$.
\end{theorem}

\begin{remarks} 1.  The name comes from the fact that $\jap{A\psi,N\varphi} - \jap{N\psi,A\varphi} = \jap{\psi, [N,A] \varphi}$ if $N\varphi \in D(A)$ and $A\varphi \in D(N)$.

2.  Nelson \cite{NelsonCommutator} was motivated by Glimm--Jaffe \cite{GJComm} which also required bounds on $[N,[N,A]]$ which would not apply to the Faris--Lavine choices without extra conditions on $V$.

3.  For a proof, see Nelson \cite{NelsonCommutator} or Reed--Simon \cite[Theorem X.36]{RS2}.
\end{remarks}

To illustrate the use of this theorem, here is a special case of the Faris--Lavine theorem (see Faris--Lavine \cite{FL} or Reed--Simon \cite[Theorem X.38]{RS2} for the full theorem) that gives a $V(x) \ge -x^2$ type of result:

\begin{theorem} [Faris--Lavine \cite{FL}] \lb{T8.8}  Let $V(x)\in L^2_{loc}(\bbR^\nu)$ and obey:
\begin{equation}\label{8.11}
  V(x) \ge -cx^2-d
\end{equation}
Then $-\Delta+V$ is esa--$\nu$.
\end{theorem}

\begin{proof} By a simple argument, we can assume $c=1, d=0$.  Let $N=-\Delta+V+2x^2$ by which we mean the closure of that sum on $C_0^\infty(\bbR^\nu)$.  Let $A$ be the operator closure of  $-\Delta+V\restriction C_0^\infty(\bbR^\nu)$.  By Theorem \ref{T9.1} below, $N$ is self--adjoint.  $N-A = 2x^2 \ge 0$ while $N+A = -\Delta + (2V(x)+2x^2) \ge 0$ so $\pm A \le N$ which is \eqref{8.9}.

The same method that proved \eqref{3.28Q} implies an estimate $\norm{x^2\varphi} \le a\norm{N\varphi}$ on $C_0^\infty(\bbR^\nu)$ so $\varphi \in D(N) \Rightarrow x^2\varphi \in L^2 \Rightarrow \varphi \in D(A)$.  Thus $D(N) \subset D(A)$.

By \eqref{8.11} $N \ge -\Delta + x^2 \ge \pm(x\cdot p+p\cdot x)$ (by completing the square).  Note that
\begin{align*}
  i[N,-\Delta+V] &= i[2x^2,-\Delta+V] \\
                 &= 2i[x^2,p^2] \\
                 &= -4(x\cdot p+p\cdot x)
\end{align*}
so $|\jap{N\varphi,A\varphi}-\jap{A\varphi,N\varphi}| \le c\jap{\varphi,N\varphi}$.  We can apply Theorem \ref{T8.7} to see that $-\Delta+V$ is esa--$\nu$.
\end{proof}

%%%%%%%%%%%%%%%%%%%%%%%%%%%%%%%%%%%%%%%%%%%%%%%%%%%%%%%%%%%%%%%
\section{Self--Adjointness, III: Kato's Inequality} \lb{s9}
%%%%%%%%%%%%%%%%%%%%%%%%%%%%%%%%%%%%%%%%%%%%%%%%%%%%%%%%%%%%%%

This section will discuss a self--adjointness method that appeared in Kato \cite{KI1} based on a remarkable distributional inequality.  Its consequences is a subject to which Kato returned to often with at least seven additional papers \cite{KI2, KI3, KIComp1, KIComp2, KI4, KIComp3, KILp}.  It is also his work that most intersected my own -- I motivated his initial paper and it, in turn, motivated several of my later papers.  Throughout this section, we'll use quadratic form ideas that we'll only formally discuss in Section \ref{s10} (see \cite[Section 7.5]{OT}).

To explain the background, recall that in Section \ref{s7}, we defined $p$ to be $\nu$--canonical ($\nu$ is dimension) if $p=2$  for $\nu \le 3$, $p > 2$ for $\nu = 4$ and $p = \nu/2$ for $\nu \ge 5$.  For now, we focus on $\nu \ge 5$ so that $p=\nu/2$.  As we saw, if $V \in L^p(\bbR^\nu)+L^\infty(\bbR^\nu)$, then $-\Delta+V$ is esa--$\nu$.  The example $V(x) - \lambda |x|^{-2}$ for $\lambda$ sufficiently large shows that $p=\nu/2$ is sharp.  That is, for any $2 \le q \le \nu/2$, there is a $V \in L^q(\bbR^\nu)+L^\infty(\bbR^\nu)$, so that $-\Delta+V$ is defined on but not esa on $C_0^\infty(\bbR^\nu)$.

In these counterexamples, though, $V$ is negative.  It was known since the late 1950s (see Section \ref{s8}) that while the negative part of $V$ requires some global hypothesis for esa--$\nu$, the positive part does not (e.g. $-\Delta-x^4$ is not esa--$\nu$ while $-\Delta+x^4$ is esa--$\nu$).  But when I started looking at these issues around 1970, there was presumption that for local singularities, there was no difference between the positive and negative parts.  In retrospect, this shouldn't have been the belief!  After all, as we've seen (see the Remarks after Proposition \ref{P7.7}), limit point--limit circle methods show that if $V(x) = |x|^{-\alpha}$ with $\alpha < \nu/2$ (to make $V \in L^2_{loc}$ so that $-\Delta+V$ is defined on $C_0^\infty(\bbR^\nu$)) then $-\Delta+V$ is esa--$\nu$ although $-\Delta-V$ is not.  (Limit point--limit circle methods apply for $-\Delta+V$ for any $\alpha$ if we look at $C_{00}^\infty(\bbR^\nu)$ but then only when $\alpha < \nu/2$, we can extend the conclusion to $C_0^\infty(\bbR^\nu)$.)  This example shows that the conventional wisdom was faulty but people didn't think about separate local conditions on
\begin{equation}\label{9.1}
  V_+(x) \equiv \max(V(x),0); \qquad V_-(x) = \max(-V(x),0)
\end{equation}

Kato's result shattered the then conventional wisdom:

\begin{theorem} [Kato \cite{KI1}]  \lb{T9.1} If $V \ge 0$ and $V \in L^2_{loc}(\bbR^\nu)$, then $-\Delta+V$ is esa--$\nu$.
\end{theorem}

\begin{remark}  As we'll see later, this extends, for example, to $V_+\in L^2_{loc}, V_-\in L^p_{unif}$ with $p$ $\nu$--canonical
\end{remark}

Kato's result was actually a conjecture that I made on the basis of a slightly weaker result that I had proven:

\begin{theorem} [Simon \cite{SimonPosPot}] \lb{T9.2} If $V \ge 0$ and $V \in L^2(\bbR^\nu,e^{-cx^2}\,d^\nu x)$ for some $c > 0$, then $-\Delta+V$ is esa--$\nu$.
\end{theorem}

Of course this covers pretty wild growth at infinity but Theorem \ref{T9.1} is the definitive result since one needs that $V \in L^2_{loc}(\bbR^\nu)$ for $-\Delta+V$ to be defined on all functions in $C_0^\infty(\bbR^\nu)$.

I found Theorem \ref{T9.2} because I was also working at the time in constructive quantum field theory which was then studying the simplest interacting field models $\varphi^4_2$ and $P(\varphi)_2$ (the subscript $2$ means two space--time dimensions).  To start with, one wanted to define $H_0+V$ where $H_0$ was a positive mass free quantum field Hamiltonian and $V$ a spatially cutoff interaction.  Nelson \cite{NelsonHyperC1} realized that one could view $H_0$ as an infinite sum of independent harmonic oscillators (shifted to have ground state energy $0$) which he analyzed as follows:  For a single variable oscillator on $L^2(\bbR,dx)$, there is a unit vector $\Omega_0$ with $H_0\Omega_0=0$.  The map $Uf \mapsto f\Omega_0^{-1}$ maps $L^2(\bbR,dx)$ unitarily to $L^2(\bbR,\Omega_0^2 \, dx)$ and Nelson analyzed $A_0=UH_0U^{-1}$ on $L^2(\bbR,\Omega^2\,dx)$ and found (with  $d\mu=\Omega^2\,dx$ a probability measure on $X=\bbR$) that
\begin{align}
  \norm{e^{-tA_0}\varphi}_p &\le \norm{\varphi}_p \qquad \textrm{all } \varphi \in L^p(X,d\mu), \textrm{ all } t>0 \lb{9.2} \\
  \norm{e^{-TA_0}\varphi}_4 &\le B\norm{\varphi}_2 \qquad T \textrm{ large enough} \lb{9.3}
\end{align}
By taking products, he got similar bounds on the infinite dimensional spaces of the field theory (he was restricted to a field theory with a periodic boundary condition but Glimm \cite{GlimmHyperC} did the full theory).  Eventually, semigroups, $e^{-tA_0}$, obeying \eqref{9.2}/\eqref{9.3} were called hypercontractive semigroups.  \cite[Section 6.6]{HA} has a lot on the general theory and the history.

Nelson also proved that the $V$ of the cutoff field theory wasn't bounded below but it did obey
\begin{equation}\label{9.4}
     V \in L^p(X,d\mu),\, p<\infty \textrm{ and } e^{-sV} \in L^1(X,d\mu), \textrm{ all } s>0
\end{equation}
He also showed that \eqref{9.2}, \eqref{9.3}, \eqref{9.4} $\Rightarrow A_0+V$ is bounded from below on $D(A_0)\cap D(V)$.

Segal \cite{SegalHyperC} then proved that these same hypotheses imply that $A_0+V$ is esa on $D(A_0)\cap D(V)$ (for the field theory case Glimm--Jaffe \cite{GJHyperC} and Rosen \cite{RosenHyperC} using Nelson's estimates but additional properties had earlier proven esa for this specific situation).

Simon--{H{\o}egh} Krohn \cite{SHK} systematized these results and showed that if $V \ge 0$, one can replace $V \in L^p$ for some $p>2$ by $V \in L^2(X,d\mu)$.  The Simon--{H{\o}egh Krohn paper was written in 1970.  In 1972, I realized that by looking at $-\Delta+x^2$ on $L^2(\bbR^\nu)$, one could prove that if $V \ge 0$ and $V \in L^2(\bbR^\nu,e^{-x^2}\,dx)$, then $-\Delta+V+x^2$ is esa--$\nu$.  Arguments like those that proved \eqref{3.28Q}, using that $[x_i,[x_i,-\Delta+V+x^2]]$ is a constant, show that one has that
\begin{equation}\label{9.5}
  \norm{x^2\varphi}^2 \le \norm{(-\Delta+V+x^2)\varphi}^2 + b \norm{\varphi}^2
\end{equation}
so by W\"{u}st's theorem (see the discussion around \eqref{7.11}), one sees that $-\Delta+V=-\Delta+V+x^2 - x^2$ is esa--$\nu$.  This idea of adding an operator $C$ to $A+B$ so that $C$ is $A+C+B$ bounded with relative bound one so one can use W\"{u}st theorem is called Konrady's trick after Konrady \cite{Konrady}

Within a few weeks of my sending out a preprint with Theorem \ref{T9.2} and the conjecture of Theorem \ref{T9.1}, I received a letter from Kato proving the conjecture by what appeared to be a totally different method.  Over the next few years, I spent some effort understanding the connection between Kato's work and semigroups.  I will begin the discussion here by sketching a semigroup proof of Theorem \ref{T9.1}, then give Kato's proof of this theorem, then discuss semigroup aspects of Kato's inequality and finally discuss some other aspects of Kato's paper \cite{KI1}.

After the smoke cleared, it was apparent that my failure to get the full Theorem \ref{T9.1} in 1972 was due to my focusing on $L^p$ properties of semigroups on probability measure spaces rather than on $L^p(\bbR^\nu,d^\nu x)$.  As a warmup to the semigroup proof of Theorem \ref{T9.1}, we prove (we use quadratic form ideas only discussed in Section \ref{s10})

\begin{theorem} [Simon \cite{SimonMaxMin}] \lb{T9.3} Let $V \ge 0$ be in $L^1_{loc}(\bbR^\nu,d^\nu x)$ and let $a \in L^2_{loc}(\bbR^\nu,d^\nu x)$ be an $\bbR^\nu$ valued function.  Let $Q(D_j^2) = \{\varphi \in L^2(\bbR^\nu,d^\nu x)\,|\, (\nabla_j-ia_j)\varphi \in L^2(\bbR^\nu,d^\nu x)\}$ with quadratic form $\jap{\varphi,-D_j^2\varphi}= \norm{(\nabla_j-ia_j)\varphi}^2$.  Let $h$ be the closed form sum $\sum_{j=1}^{\nu} -D_j^2+V$.  Then $C_0^\infty(\bbR^\nu)$ is a form core for $h$.
\end{theorem}

\begin{remarks}  1.  For $a=0$, this result was first proven by Kato \cite{KI2}, although \cite{OT} mistakenly attributes it to Simon.

2.  Kato \cite{KI3} proved this result if $a \in L^2_{loc}$ is replaced by $a \in L^\nu_{loc}$ and he conjectured this theorem.

3.  Since $a_j \in L^2_{loc}$, we have that $a_j\varphi \in L^1_{loc}$ so $(\nabla_j-ia_j)\varphi$ is a well defined distribution and it makes sense to say that it is in $L^2$.

4.  Just as $V \in L^2_\loc$ is necessary for $H\varphi$ to lie in $L^2$ for all $\varphi \in C_0^\infty(\bbR^\nu)$, $V \in L^1_{loc}$ and $a \in L^2_{loc}$ are necessary for $C_0^\infty \subset V_h$.

5.  There is an analog of Theorem \ref{T9.1} with magnetic field.  If $V \ge 0$, one needs to have $V \in L^2_{loc}, \, a \in L^4_{loc}$ and $\nabla\cdot\overrightarrow{a} \in L^2_{loc}$ for $H$ to be defined as an operator on $C_0^\infty$.  It is a theorem of Leinfelder--Simader \cite{LS} that this is also sufficient for esa--$\nu$ (see \cite[Section 1.4]{CFKS} for a proof along the lines discussed below for the current theorem).

6.  Kato \cite{KI2} has a lovely way of interpreting that $C_0^\infty$ is a form core.  A natural maximal operator domain for the operator associated with $h$ is $H_{max}$ defined on (here $V_h=Q(V)\cap\bigcap_{j=1}^\nu Q(D_j^\nu)$)
\begin{equation}\label{9.6}
  D(H_{max}) = V_h \cap \{\varphi \,|\, \sum_{j=1}^{\nu} -D_j^2\varphi+V\varphi \in L^2(\bbR^\nu)\}
\end{equation}
Since $\varphi \in V_h$, we have that $D_j\varphi \in L^2$ which implies that $a_jD_j\varphi \in L^1_{loc}$ and $\nabla_j D_j\varphi$ makes sense as a distribution.  Also $\varphi \in V_h \Rightarrow V^{1/2}\varphi \in L^2 \Rightarrow V\varphi = V^{1/2}(V^{1/2}\varphi) \in L^1_{loc}$ so $-D_j^2\varphi+V\varphi$ is a well defined distribution.  What Kato shows is that if $H$ is the operator associated to the closed form, $h$, then $H_{max}$ symmetric $\iff H_{max} = H \iff C_0^\infty $ is a form core for $h$.
\end{remarks}

Here is a sketch of a proof of Theorem \ref{T9.3} following \cite{SimonMaxMin}

\textbf{Step 1.} Use Kato's ultimate Trotter product formula of Section \ref{s18} (for $\nu + 1$ rather than $2$ operators, so one needs the result of Kato--Masuda \cite{KM}; we note these results weren't available in 1972 but they are only needed for the case $a \ne 0$) to see that
\begin{equation}\label{9.7}
  |(e^{-tH}\varphi)(x)| \le \left(|e^{t\Delta}|\varphi|\right)(x)
\end{equation}
which is implied by
\begin{align}
  |(e^{-tV}\varphi)(x)| &\le |\varphi|(x) \lb{9.8} \\
  |(e^{tD_j^2}\varphi)(x)| &\le \left(|e^{t\partial^2_j}|\varphi|\right)(x) \lb{9.9}
\end{align}
(We note that \eqref{9.7} is called a diamagnetic inequality; we'll say more about its history below.)

\textbf{Step 2.} This step proves \eqref{9.9}. Since $V \ge 0$, \eqref{9.8} is trivial.  Define
\begin{equation*}
  \lambda_j(x) = \int_{0}^{x_j} a_j(x_1,\dots,x_{j-1},s,x_{j+1},\dots,x_\nu) \, ds
\end{equation*}
so $\partial_j \lambda_j = a_j$ in distributional sense.  One proves that $D_j=e^{i\lambda_j}\partial_j e^{-i\lambda_j}$ in the sense that $\varphi \mapsto e^{-i\lambda_j}$ maps $D(D_j)$ to $D(\partial_j)$ and the unitary map $U:\varphi \mapsto e^{-i\lambda_j}\varphi$ obeys $e^{tD_j^2} = U e^{t\partial_j^2} U^{-1}$.  From this and the fact that $e^{t\partial_j^2}$ is positivity preserving, \eqref{9.9} follows.  From the point of view of physics, we exploit the fact that $1D$ magnetic fields can be ``gauged away''.

\textbf{Step 3.}  Let $g \in C_0^\infty(\bbR^\nu)$.  Then $\varphi \mapsto g\varphi$ maps $Q(H)$ to itself. Moreover, if $g(x) = 1$ for $|x| \le 1$ and $g_n(x)=g(x/n)$, then for any $\varphi \in Q(H)$ we have that $g_n\varphi \to \varphi$ in the form norm of $H$.  Since $V^{1/2}\varphi \in L^2 \Rightarrow gV^{1/2}\varphi \in L^2$ and $\norm{(g_n-1)V^{1/2}\varphi}_2 \to 0$, we see that the $V$ pieces behave as claimed.  Moreover, $D_j(g\varphi) = gD_j\varphi+(\partial_jg)\varphi$ as distributions, so $D_j\varphi,\varphi \in L^2 \Rightarrow D_j(g\varphi), g\varphi \in L^2$ and since $\norm{\partial_j g_n}_\infty \le Cn^{-1}$, we get the required convergence.

\textbf{Step 4.} Since $e^{t\Delta}$ maps $L^2$ to $L^\infty$, by \eqref{9.7}, we have that $e^{-H}[L^2]$, which is a form core for $H$, lies in $L^\infty$.  We conclude by step 3 that $\{\varphi \in Q(H)\,|\, \varphi \in L^\infty$ and $\varphi$ has compact support$\}$ is a core for $H$.

\textbf{Step 5.} We haven't yet used $V \in L^1_{loc}$ in that the above arguments work, for example, if $V(x) = |x|^{-\beta}$ for any $\beta > 0$.  We now want to look at $k*\varphi$ for $k \in C_0^\infty(\bbR^\nu)$ and for $\beta > \nu$ it is easy to see that $\varphi \mapsto k*\varphi$ does not leave $Q(|x|^{-\beta})$ invariant (since such functions must vanish at $x=0$).

If $\varphi$ is bounded with compact support and $V \in L^1_{loc}$ it is easy to see that for $k \in C_0^\infty(\bbR^\nu)$, we have that $V^{1/2}(k*\varphi) \in L^2$ and if $k_n$ is an approximate identity, that $\norm{V^{1/2}(k_n*\varphi)-V^{1/2}\varphi}\to 0$.  Similarly, if $(\partial_j-ia_j)\varphi \in L^2$ and $\varphi$ bounded with compact support, then $\partial_j\varphi \in L^2$ so $D_j(k*\varphi) \in L^2$ and if $k_n$ is an approximate identity, then ${\norm{D_j(k_n*\varphi)-D_j\varphi} \to 0}$.  It follows that $C_0^\infty(\bbR^\nu)$ is a form core concluding this sketch of the proof of Theorem \ref{T9.3}.

Next, we provide our first proof of Theorem \ref{T9.1} following \cite{SimonMaxMin}.  So we have, $V \ge 0$, $V \in L^2_{loc}$ and $a=0$.  By the just proven Theorem \ref{T9.3} and Remark 5 after the statement of the theorem:
\begin{equation}\label{9.10}
  D(H)=\{\varphi \in L^2\,|\, \nabla\varphi \in L^2, V^{1/2}\varphi \in L^2, -\Delta\varphi+V\varphi \in L^2\}
\end{equation}
where $-\Delta\varphi+V\varphi$ is viewed as a sum of distributions.  If $g \in C_0^\infty(\bbR^\nu)$ and $\varphi \in D(H)$, then
\begin{equation*}
  H(g\varphi) = g(H\varphi) - 2\nabla g \cdot \nabla \varphi - (\Delta g) \varphi
\end{equation*}
so $\varphi \mapsto g\varphi$ maps $D(H)$ to itself with $g_n\varphi \to \varphi$ ($g_n(x) = g(x/n); g(x) \equiv 1$ for $x$ near $0$) in graph norm for any $\varphi \in D(H)$.  Moreover, as above, $e^{-tH}[L^2] \subset L^\infty$ and is an operator core for $H$.  It follows that the set of bounded, compact support functions in $D(H)$ is a core.  For any such function, it is easy to see that if $h_n$ is an approximate identity, then $h_n*\varphi \to \varphi$ in graph norm so we conclude esa--$\nu$ completing the first proof of Theorem \ref{T9.1}.

We next turn to Kato's original approach to proving his theorem, Theorem \ref{T9.1}.  He proved

\begin{theorem} [Kato's inequality]  \lb{T9.4} Let $u \in L^1_{loc}(\bbR^\nu)$ be such that its distributional Laplacian, $\Delta u$ is also in $L^1_{loc}(\bbR^\nu)$.  Define
\begin{equation}\label{9.11}
  \mathrm{sgn}(u)(x) = \left\{
                         \begin{array}{ll}
                           \overline{u(x)}/|u(x|), & \hbox{ if } u(x) \ne 0 \\
                           0, & \hbox{ if } u(x) = 0
                         \end{array}
                       \right.
\end{equation}
(so $u\,\mathrm{sgn}(u) = |u|$). Then as distributions
\begin{equation}\label{9.12}
  \Delta |u| \ge \Real\left[\mathrm{sgn}(u) \Delta u\right]
\end{equation}
\end{theorem}

\begin{remarks} 1. What we call $\mathrm{sgn}(u)$, Kato calls $\mathrm{sgn}(\bar{u})$.

2.  We should pause to emphasize what a surprise this was.  Kato was a long established master of operator theory.  He was 55 years old.  Seemingly from left field, he pulled a distributional inequality out of his hat.  It is true, like other analysts, that he'd been introduced to distributional ideas in the study of PDEs, but no one had ever used them in this way.  Truly a remarkable discovery.
\end{remarks}

The proof is not hard.  By replacing $u$ by $u*h_n$ with $h_n$ a smooth approximate identity and taking limits (using $\mathrm{sgn}(u*h_n)(x) \to \mathrm{sgn}(u)(x)$ for a.e. $x$ and using a suitable dominated convergence theorem), we can suppose that $u$ is a $C^\infty$ function.  In that case, for $\epsilon > 0$, let $u_\epsilon = (\bar{u}u + \epsilon^2)^{1/2}$.  From $u_\epsilon^2=\bar{u}u+\epsilon^2$, we get that
\begin{equation}\label{9.13}
  2u_\epsilon \overrightarrow{\nabla} u_\epsilon = 2\Real(\bar{u}\overrightarrow{\nabla} u)
\end{equation}
which implies (since $|\bar{u}| \le u_\epsilon$) that
\begin{equation}\label{9.14}
  |\overrightarrow{\nabla}u_\epsilon| \le |\overrightarrow{\nabla}u|
\end{equation}
Applying $\tfrac{1}{2}\overrightarrow{\nabla}\cdot$ to \eqref{9.13}, we get that
\begin{equation}\label{9.15}
  u_\epsilon \Delta u_\epsilon + |\overrightarrow{\nabla}u_\epsilon|^2 = \Real(\bar{u}\Delta(u)) + |\overrightarrow{\nabla}u|^2
\end{equation}
Using \eqref{9.14} and letting $\mathrm{sgn}_\epsilon(u) = \bar{u}/u_\epsilon$, we get that
\begin{equation}\label{9.16}
   \Delta u_\epsilon \ge \Real(\mathrm{sgn}_\epsilon(u)\Delta u)
\end{equation}
Taking $\epsilon \downarrow 0$ yields \eqref{9.12}.

Once we have \eqref{9.12}, here is Kato's proof of Theorem \ref{T9.1} (the second proof that we sketch).  Consider $T$, the operator closure of $-\Delta+V$ on $C_0^\infty(\bbR^\nu)$.  $T \ge 0$, so, by a simple argument (\cite[Corollary to Theorem X.1]{RS2}), it suffices to show that $\ran(T+\bdone)=\calH$ or equivalently, that $T^*u=-u \Rightarrow u=0$.  So suppose that $u \in L^2(\bbR^\nu)$ and that
\begin{equation}\label{9.17}
  T^*u = -u
\end{equation}
Since $T^*$ is defined via distributions, \eqref{9.17} implies that
\begin{equation}\label{9.18}
  \Delta u = (V+1)u
\end{equation}
Since $u$ and $V+1$ are both in $L^2_{loc}$, we conclude that $\Delta u \in L^1_{loc}$ so by Kato's inequality
\begin{equation}\label{9.19}
  \Delta |u| \ge (\mathrm{sgn}(u))(V+1)u = |u|(V+1) \ge |u|
\end{equation}
Convolution with non--negative functions preserves positivity of distributions, so for any non-negative $h \in C_0^\infty(\bbR^\nu)$, we have that
\begin{equation}\label{9.20}
  \Delta(h*u) = h*\Delta |u| \ge h*|u|
\end{equation}
Since $u \in L^2$, $h*u$ is a $C^\infty$ function with classical Laplacian in $L^2$, so $h*u \in D(-\Delta)$.  $(-\Delta+1)^{-1}$ has a positive integral kernel, so \eqref{9.20}$\Rightarrow (-\Delta+1)(h*|u|) \le 0 \Rightarrow h*|u| \le 0 \Rightarrow h*|u| = 0$.  Taking $h_n$ to be an approximate identity, we have that $h_n*u \to u$ in $L^2$, so $u=0$ completing the proof.

At first sight, Kato's proof seems to have nothing to do with the semigroup ideas used in the proof of Theorem \ref{T9.2} and our first proof of Theorem \ref{T9.1}.  But in trying to understand Kato's work, I found the following abstract result:

\begin{theorem}[Simon \cite{SimonKI1}]  \lb{T9.5} Let $A$ be a positive self--adjoint operator on $L^2(M,d\mu)$ for a $\sigma$--finite, separable measure space $(M,\Sigma,d\mu)$.  Then the following are equivalent:

(a) ($e^{-tA}$ is positivity preserving)
\begin{equation*}
         \forall u \in L^2,\, u\ge 0, t \ge 0 \Rightarrow e^{-tA}u \ge 0
\end{equation*}

(b) (Beurling--Deny criterion) $u \in Q(A) \Rightarrow |u| \in Q(A)$ and
\begin{equation}\label{9.22}
  q_A(|u|) \le q_A(u)
\end{equation}

(c) (Abstract Kato Inequality) $u \in D(A) \Rightarrow |u| \in Q(A)$ and for all $\varphi \in Q(A)$ with $\varphi \ge 0$, one has that
\begin{equation}\label{9.23}
  \jap{A^{1/2}\varphi,A^{1/2}|u|} \ge \Real\jap{\varphi,\mathrm{sgn}(u) Au}
\end{equation}
\end{theorem}

The equivalence of (a) and (b) for $M$ a finite set (so $A$ is a matrix) is due to Beurling--Deny \cite{BD}.  For a proof of the full theorem (which is not hard), see Simon \cite{SimonKI1} or \cite[Theorem 7.6.4]{OT}.

In his original paper, Kato \cite{KI1} proved more than \eqref{9.12}.  He showed that
\begin{equation}\label{9.24}
  \Delta |u| \ge \Real\left[\mathrm{sgn}(u)(\overrightarrow{\nabla}-i\overrightarrow{a})^2u\right]
\end{equation}
In \cite{KI1}, he required that $\overrightarrow{a}$ to be $C^1(\bbR^\nu)$ but he implicitly considered less regular $\overrightarrow{a}$'s in \cite{KI3}.  For smooth $a$'s, one gets \eqref{9.24} as we got \eqref{9.12}.  Since $\Real(\bar{u}(-ia)u) = 0$, \eqref{9.13}, with $D=\nabla-ia$ implies that
\begin{equation}\label{9.25}
  u_\epsilon\nabla u_\epsilon = \Real(\bar{u}Du)
\end{equation}
which implies that
\begin{equation}\label{9.26}
  |\nabla u_\epsilon| \le |Du|
\end{equation}

Note next that
\begin{equation*}
  \nabla_j(\bar{u}D_j u) = \left[(\nabla_j+ia_j)\bar{u}\right]D_ju+\bar{u}D_j^2u
\end{equation*}
since $ia_j\bar{u}D_ju+\bar{u}(-ia_j)D_ju=0$.  Thus applying $\overrightarrow{\nabla}$ to \eqref{9.25} yields
\begin{equation}\label{9.27}
  u_\epsilon\Delta u_\epsilon + |\nabla u_\epsilon|^2 = |D u|^2+\Real(\bar{u}D^2u)
\end{equation}
By \eqref{9.26}, we get \eqref{9.24}.

In \cite{KI1}, Kato followed his arguments to get Theorem \ref{T9.1} with $-\Delta+V$ replaced by $-(\nabla-ia)^2+V$ when $a \in C^1(\bbR^\nu), V \in L^2_{loc}(\bbR^\nu), V \ge 0$.  But there was a more important consequence of \eqref{9.24} than a self--adjointness result.  In \cite{SimonDiamagnetic}, I noted that \eqref{9.24} implies by approximating $|u|$ by positive $\varphi \in C_0^\infty(\bbR^\nu)$, that
\begin{equation*}
  \jap{|u|,\Delta|u|} \ge \jap{u,D^2u}
\end{equation*}
which implies that
\begin{equation}\label{9.28}
  \jap{u,(-D^2+V)u} \ge \jap{|u|,(-\Delta+V)|u|}
\end{equation}
This in turn implies that turning on a magnetic field always increases the ground state energy (for spinless bosons), something I called \emph{universal diamagnetism}.

If one thinks of this as a zero temperature result, it is natural to expect a finite temperature result (that is, for, say, finite matrices, one has that $\lim_{\beta \to \infty} -\beta^{-1}\tr(e^{-\beta A}) = \inf \sigma(A)$ which in statistical mechanical terms is saying that as the temperature goes to zero, the free energy approaches a ground state energy).
\begin{equation}\label{9.29}
  \tr(e^{-tH(a,V)}) \le \tr(e^{-tH(a=0,V)})
\end{equation}
where
\begin{equation}\label{9.30}
  H(a,V) = -(\nabla-ia)^2+V
\end{equation}

This suggested to me the inequality
\begin{equation}\label{9.31}
  |e^{-tH(a,V)}\varphi| \le e^{-tH(a=0,V)}|\varphi|
\end{equation}
I mentioned this conjecture at a brown bag lunch seminar when I was in Princeton.  Ed Nelson remarked that formally, it followed from the Feynman--Kac--Ito formula for semigroups in magnetic fields which says that adding a magnetic field with gauge, $\overrightarrow{a}$, adds a factor $\exp(i\int \overrightarrow{a}(\omega(s))\cdot d\omega)$ to the Feynman--Kac formula (the integral is an Ito stochastic integral).  \eqref{9.31} is immediate from $|\exp(i\int \overrightarrow{a}(\omega(s))\cdot d\omega)| = 1$ and the positivity of the rest of the Feynman--Kac integrand.  Some have called \eqref{9.31} the Nelson--Simon inequality but the name I gave it, namely \emph{diamagnetic inequality}, has stuck.

The issue with Nelson's proof is that at the time, the Feynman--Kac--Ito was only known for smooth $a$'s.  One can obtain the Feynman--Kac--Ito for more general $a$'s by independently proving a suitable core result.  Simon \cite{SimonKI1} and then Kato \cite{KI3} obtained results for more and more singular $a$'s until Simon \cite{SimonMaxMin} proved

\begin{theorem}[Simon \cite{SimonMaxMin}] \lb{T9.6} \eqref{9.31} holds for $V \ge 0$, $V \in L^1_{loc}(\bbR^\nu)$ and $\overrightarrow{a} \in L^2_{loc}$.
\end{theorem}

Indeed, our proof of \eqref{9.7} above implies this if we don't use \eqref{9.8} but keep $e^{-tV}$ (equivalently, if we just use \eqref{9.9}).

As with Theorem \ref{T9.5}, there is an abstract two operator Kato inequality result (originally conjectured in Simon \cite{SimonKI1}):

\begin{theorem}[Hess--Schrader--Uhlenbrock \cite{HessKI}, Simon \cite{SimonKI2}]  \lb{T9.7} Let $A$ and $B$ be two positive self--adjoint operators on $L^2(M,d\mu)$ where $(M,\Sigma,d\mu)$ is a $\sigma$--finite, separable measure space.  Suppose that $\varphi \ge 0 \Rightarrow e^{-tA}\varphi \ge 0$.  Then the following are equivalent:

(a) For all $\varphi \in L^2$ and all $t \ge 0$, we have that
\begin{equation*}
  |e^{-tB}\varphi| \le e^{-tA}|\varphi|
\end{equation*}

(b) $\psi \in D(B) \Rightarrow |\psi| \in Q(A)$ and for all $\varphi \in Q(A)$ with $\varphi \ge 0$ and all $\psi \in D(B)$ we have that
\begin{equation}\label{9.32}
  \jap{A^{1/2}\varphi,A^{1/2}|\psi|} \le \Real\jap{\varphi,\mathrm{sgn(\psi) B\psi}}
\end{equation}
\end{theorem}

For a proof, see the original papers or \cite[Theorem 7.6.7]{OT}.

As one might expect, the ideas in Kato \cite{KI1} have generated an enormous literature.  Going back to the original paper are two kinds of extensions: replace $\Delta$ by $\sum_{i,j=1}^{\nu} \partial_i a_{ij}(x) \partial_j$ and allowing $q(x) \to -\infty$ as $|x| \to \infty$ with lower bounds of the Wienholtz--Ikebe--Kato type as discussed in Section \ref{s8}.  Some papers on these ideas include Devinatz \cite{Devan}, Eastham et al \cite{EEM}, Evans \cite{Evans}, Frehse \cite{Fre}, G\"{u}neysu--Post \cite{GP}, Kalf \cite{KalfGauss}, Knowles \cite{Know1, Know2, Know3}, Milatovic \cite{Mila} and Shubin \cite{Shubin}. There is a review of Kato \cite{KI4}.  For applications to higher order elliptic operators, see Davies--Hinz \cite{DH}, Deng et al \cite{DDY} and Zheng--Yao \cite{ZY}. There are papers on $V$'s obeying $V(x) \ge -\nu(\nu-4)|x|^{-2}; \, \nu \ge 5$, some using Kato's inequality by Kalf--Walter \cite{KalfWalter1}, Schmincke \cite{SchmSing}, Kalf \cite{KalfSing}, Simon \cite{SimonSing}, Kalf--Walter \cite{KalfWalter2} and Kalf et. al. \cite{KSWW}.

Kato himself applied these ideas to complex valued potentials in three papers \cite{KIComp1, KIComp2, KIComp3}.  In particular, Br\'{e}zis--Kato \cite{KIComp3} has been used extensively in the nonlinear equation literature as part of a proof of $L^p$ regularity of eigenfunctions.

There is one final aspect of \cite{KI1} which should be mentioned.  In it, Kato introduced a condition on the negative part of the potential that I dubbed Kato's class and denoted $K_\nu$ and which has since been used extensively.  Earlier, Schechter \cite{Schechter} had introduced a family of spaces with several parameters which agrees with $K_\nu$ for one choice of parameters but he didn't single it out.  A function, $V$ on $\bbR^\nu$ is said to lie in $K_\nu$ if and only if
\begin{equation}\label{9.33}
 \left\{
               \begin{array}{ll}
                 \lim_{\alpha \downarrow 0} \left[\sup_x \int_{|x-y| \le \alpha} |x-y|^{2-\nu} |V(y)|\, d^\nu y\right] =0, & \hbox{ if } \nu >2\\
                 \lim_{\alpha \downarrow 0} \left[\sup_x \int_{|x-y| \le \alpha}  \log(|x-y|^{-1}) |V(y)|\, d^\nu y\right] =0, & \hbox{ if } \nu =2\\
                      \sup_x \int_{|x-y| \le 1} |V(y)| \, dy < \infty, & \hbox{ if } \nu=1
                \end{array}
                   \right.
\end{equation}
$K_\nu^{loc}$ is those where we demand \eqref{9.33} not for $\sup_x$ but rather, for each $x_0$ for $\sup_{|x-x_0|}
\le 1$.  Note that the class $S_\nu$ of Section 7 is an operator analog of this and was motivated by Kato's definition.  There are analogs of Theorem \ref{T7.10} and \ref{T7.11} for $K_\nu$, see \cite[Section 1.2]{CFKS}.

Kato used $K_\nu$ to discuss local (and global) singularities of the negative part of $V$.  Ironically, $K_\nu$ is not maximal for such considerations.  If $\nu \ge 3$ and $V(x) = |x|^{-2} \log(|x|^{-1})^{-\delta}$ (for $|x| < \tfrac{1}{2})$, then $V \in K_\nu \iff \delta > 1$ but $V$ is form bounded if and only if $\delta > 0$.  However, Aizenman--Simon \cite{AizSimon} have proven the following showing the naturalness of Kato's class for semigroup considerations:

\begin{theorem} [Aizenman--Simon \cite{AizSimon}] \lb{T9.8} Let $V \le 0$ have compact support.  Then $V \in K_\nu$ if and only if $e^{-tH},\, (H=-\Delta+V)$ maps $L^\infty(\bbR^\nu)$ to itself for all $t > 0$ and
\begin{equation}\label{9.34}
  \lim_{t \downarrow 0}\norm{e^{-tH}}_{\infty,\infty} = 1
\end{equation}
\end{theorem}

For more on this theme, see \cite{AizSimon, SimonSmgp}.

%%%%%%%%%%%%%%%%%%%%%%%%%%%%%%%%%%%%%%%%%%%%%%%%%%%%%%%%%%%%%%
\section{Self--Adjointness, IV: Quadratic Forms} \lb{s10}
%%%%%%%%%%%%%%%%%%%%%%%%%%%%%%%%%%%%%%%%%%%%%%%%%%%%%%%%%%%%%%

Hilbert, around 1905, originally discussed operators on inner product spaces in terms of (bounded) quadratic forms, not surprising given Hilbert's background in number theory.  F. Riesz emphasized the operator theory point of view starting in 1913 and von Neumann's approach to unbounded operators in 1929 also emphasized the operator point of view which has dominated most of the discussion since.  In the 1930s and 1940s, there was work in which the quadratic form point of view was implicit but it was only in the 1950s that forms became explicitly discussed objects and Kato was a major player in this development.   In this section, we'll first describe the basic theory and give a Kato--centric history and then discuss two special aspects in which Kato had seminal contributions: first, the theory of monotone convergence for forms and secondly, the theory of pseudo--Friedrichs extensions and its application to the Dirac Coulomb problem, as well as some other work of Kato on the Dirac Coulomb problem.

In his delightful reminisces of Kato, Cordes \cite{Cordes} quotes Kato as saying ``there is no decent Banach space, except Hilbert space.''  While this ironic given Kato's development of eigenvalue perturbation theory and semigroup theory in general Banach spaces, it is likely he had in mind the spectral theorem and the subject of this section.

Let $\calH$ be a (complex, separable) Hilbert space.  A \emph{quadratic form} is a map $q:\calH \to [0,\infty]$ with $\infty$ an allowed value that is quadratic and obeys the parallelogram law, i.e.
\begin{align}
  q(z\varphi) &= |z|^2 q(\varphi),\quad \textrm{ all } \varphi \in \calH, z \in \bbC \lb{10.1} \\
  q(\varphi+\psi)+q(\varphi-\psi) &= 2q(\varphi)+2q(\psi) \lb{10.2}
\end{align}
where $a\infty=\infty$ (for $a>0$), $=0$ for $a=0$ and $\infty+a=a+\infty=\infty$ for any $a \in [0,\infty]$.  The \emph{form domain} of q is
\begin{equation}\label{10.3}
  V_q = \{\varphi\,|\, q(\varphi) < \infty\}
\end{equation}

A \emph{sesquilinear form} is a pair $(V,Q)$ of a subspace $V \subset \calH$ ($V$ is not necessarily a closed and/or dense subspace.  Typically $V$ is dense in $\calH$, but as we'll see in Section \ref{s18}, there are very interesting cases where $V$ is not dense.) and a map $Q:V \times V \to \bbC$ obeying
\begin{align}
  \forall \psi\in V,\quad &\varphi \mapsto Q(\psi,\varphi) \textrm{ is linear} \lb{10.4} \\
  \forall \varphi,\psi \in V, \quad &Q(\psi,\varphi) = \overline{Q(\varphi,\psi)}
\end{align}
which imply that $\forall \psi \in V, \varphi \in V \mapsto Q(\varphi,\psi)$ is antilinear.  $Q$ is called \emph{positive} if and only if $\forall \varphi \in V$ one has that $Q(\varphi,\varphi) \ge 0$.

An elementary fact is:

\begin{theorem} \lb{T10.1} There is a one--one correspondence between quadratic forms and positive sesquilinear forms given by

(a) If $(V,Q)$ is a sesquilinear form, define a quadratic form, $q$, by
\begin{equation}\label{10.6}
  q(\varphi) = \left\{
                 \begin{array}{ll}
                   Q(\varphi,\varphi) & \hbox{ if } \varphi \in V \\
                   \infty, & \hbox{ if } \varphi \notin V
                 \end{array}
               \right.
\end{equation}
(so $V_q=V$).

(b)  If $q$ is a quadratic form, take $V = V_q$ and define a map, $Q$ on $V \times V$ by
\begin{equation}\label{10.7}
  Q(\varphi,\psi) = \tfrac{1}{4}[q(\varphi+\psi)-q(\varphi-\psi)+i q(\varphi-i\psi) -i q(\varphi+i\psi)]
\end{equation}
\end{theorem}

If $q:\calH \to (-\infty,\infty]$ so that there is an $\alpha$ so that $\wti{q}(\varphi)=q(\varphi)+\alpha \norm{\varphi}^2$ is a (positive) quadratic form, we say that $q$ is a \emph{semibounded quadratic form}.  Theorem \ref{T10.1} extends and we speak of semibounded sesquilinear forms (where $Q(\varphi,\varphi) \ge 0$ is replaced by $Q(\varphi,\varphi) \ge -\alpha\norm{\varphi}^2$).  For any semibounded sesquilinear form, we define $\beta=\inf_{\varphi \in V,\varphi \ne 0} Q(\varphi,\varphi)/\norm{\varphi}^2$ to be the lower bound of $Q$.

Given two quadratic forms, $q_1$ and $q_2$, we write
\begin{equation}\label{10.9}
  q_1 \le q_2 \iff \forall \varphi \in \calH, \quad q_1(\varphi) \le q_2(\varphi)
\end{equation}
If in addition
\begin{equation}\label{10.10}
  q_2(\varphi) < \infty \Rightarrow q_1(\varphi) = q_2(\varphi)
\end{equation}
we say that $q_1$ is an \emph{extension} of $q_2$.  The name comes from the fact that \eqref{10.9}/\eqref{10.10} is equivalent to $V_{q_2} \subset V_{q_1}$ and $Q_2 = Q_1 \restriction V_{q_2}\times V_{q_2}$.

Given a (positive) quadratic form, $q$, one defines a norm, $\norm{\cdot}_{+1}$ on $V_q$ by
\begin{equation}\label{10.8}
  \norm{\varphi}_{+1}^2 = q(\varphi)+\norm{\varphi}^2
\end{equation}
$\norm{\cdot}_{+1}$ is a norm (because of the $\norm{\varphi}^2$, we have that $\norm{\varphi}_{+1} \ne 0$ if $\varphi \ne 0$ even if $q(\varphi)=0$) which also obeys the parallelogram law so $\norm{\cdot} _{+1}$ comes from an inner product \cite[Theorem 3.1.6]{RA}.  We say that $q$ is a \emph{closed quadratic form} if and only if $V$ is complete in $\norm{\cdot}_{+1}$ (see Theorem \ref{T10.13} below for an important characterization of closed forms).  A subspace $W \subset V$ is called a \emph{form core} for $q$ if $W$ is dense in $V$ in $\norm{\cdot} _{+1}$.

We say that a quadratic form, $q$, is \emph{closable} if and only if $q$ has a closed extension. One can show that there is then a smallest closed extension, $\bar{q}$ (in that if $t$ is another closed extension of $q$, it is also an extension of $\bar{q}$).

\begin{example} \lb{E10.2}  Let $\calH=L^2(\bbR,dx)$.  Define $q$ with $V_q = C_0^\infty(\bbR)$ and for $\varphi \in V_q$
\begin{equation}\label{10.11}
  q(\varphi) = |\varphi(0)|^2
\end{equation}
For obvious reasons, we write $q=\delta(x)$, the Dirac delta function.  One can show \cite[Example 7.5.17]{OT} that this form is not closable  (see also the Remark after Theorem \ref{T10.13} below).
\end{example}

\begin{example} \lb{E10.3}  Let $\calK \subset \calH$ be a closed subspace, so $\calK$ is a Hilbert space.  Let $A$ be a self--adjoint operator on $\calK$.  We recall that the spectral theorem \cite[Chapters 5 and Section 7.2]{OT} lets one define $f(A)$ as an operator on $\calK$ for any real valued measurable function, $f$, from the spectrum of $A$ to $[0,\infty)$.  $f(A)$ is self--adjoint with domain $\{\varphi \,|\, \int |f(x)|^2 d\mu^A_\varphi(x) < \infty\}$ where $d\mu^A_\varphi$ is the spectral measure, defined, for example by $\jap{\varphi,(A-z)^{-1}\varphi} = \int (x-z)^{-1} d\mu^A_\varphi(x)$ for all $z \in \bbC\setminus\bbR$.  In particular, if $A$ is a positive self--adjoint operator on $\calK$, we can define a positive, self--adjoint operator, $A^{1/2}$ on $\calK$.  We  define the quadratic form $q_A$ on $\calH$ by
\begin{equation}\label{10.6A}
  q_A(\varphi) = \left\{
                   \begin{array}{ll}
                     \norm{A^{1/2}\varphi}^2, & \hbox{ if } \varphi\in\calK \textrm{ and } \varphi \in D(A^{1/2}) \\
                     \infty, & \hbox{ otherwise}
                   \end{array}
                 \right.
\end{equation}
This definition is basic even when $\calK=\calH$.  It is not hard to prove that this quadratic form is closed.  We call $V_q$ the form domain of $A$ and denote it by $Q(A)$.
\end{example}

\begin{example} \lb{E10.4} Given $A$ as in the last example and $g:\sigma(A) \to [0,\infty)$ which is continuous and bounded and obeys $\lim_{t \to \infty} g(t)=0$, we define $g(A)$ on $\calH$ by setting it to the spectral theorem $g(A)$ on $\calK$ and to $0$ on $\calK^\perp$.  If $A=0$ on $\calK$ (and in some sense $\infty$ on $\calK^\perp$), then for any $t>0$, we have that $e^{-tA}$ is the orthogonal projection onto $\calK$.
\end{example}

What makes quadratic forms so powerful is that, in a sense, Example \ref{E10.3} has a converse.  Here are two versions of this result:

\begin{theorem} \lb{T10.5} Let $q$ be a closed quadratic form.  Let $\calK=\overline{V_q}$.  Then there is a unique positive self--adjoint operator, $A$, on $\calK$ so that $q=q_A$.
\end{theorem}

\begin{remark}  The closure in $\overline{V_q}$ means closure in the Hilbert space topology (which in many cases is the entire Hilbert space).
\end{remark}

\begin{theorem} \lb{T10.6} Let $q$ be a closed quadratic form with $V_q$ dense in $\calH$.  Then, there is a unique self--adjoint operator, $A$, on $\calH$ so that:

(a) $D(A) \subset V_q$

(b) If $\varphi \in D(A), \psi \in V_q$, then
\begin{equation}\label{10.7A}
  Q_q(\psi,\varphi) = \jap{\psi,A\varphi}
\end{equation}
Moreover, $D(A)$ is a form core for $A$.
\end{theorem}

\begin{remarks}  1.  In his book \cite{KatoBk}, Kato calls Theorem \ref{T10.6} the first representation theorem and Theorem \ref{T10.5} the second representation theorem.  He puts Theorem \ref{T10.6} first because it is the version going back to the 1930s (see below).  I put Theorem \ref{T10.5} first because I think that it is the fundamental result -- indeed, it is the only variant in Reed--Simon \cite{RS1}  and Simon \cite{OT}.

2.  For proofs, see Kato \cite{KatoBk}, Reed--Simon \cite[Theorem VIII.15]{RS1} or Simon \cite[Theorem 7.5.5]{OT}.
\end{remarks}

\begin{example} \lb{E10.7} Let $B$ be a densely defined symmetric operator on $\calH$ with $\jap{\varphi,B\varphi} \ge 0$ for all $\varphi\in D(B)$.  $B$ might not be self--adjoint.  Define a quadratic form, $\wti{q_B}$, (which differs from $q_B$ if $B$ is self--adjoint!) by
\begin{equation}\label{10.7b}
  \wti{q_B}(\varphi) = \left\{
                         \begin{array}{ll}
                           \jap{\varphi,B\varphi}, & \hbox{ if } \varphi \in D(B) \\
                           \infty, & \hbox{ if } \varphi\notin D(B)
                         \end{array}
                       \right.
\end{equation}
\end{example}
If $B$ is not bounded, one can show that $\wti{q_B}$ is never closed but one can prove \cite[Theorem 7.5.19]{OT} that it is always closable.  If $q^\#$ is its closure, there is a self--adjoint $A$ with $q^\#=q_A$.  One can show (it is immediate from Theorem \ref{T10.6}) that $A$ is an operator extension of $B$ so $B$ has a natural self--adjoint extension.  It is called the \emph{Friedrichs extension}, $B_F$.  Unless $B$ is esa, there are lots of other self--adjoint extensions as we'll see.  It can happen (but usually doesn't) that $B$ is not esa but has a unique positive self--adjoint extension.

There is a form analog of the Kato--Rellich theorem:

\begin{theorem} [KLMN theorem] \lb{T10.8} Let $q$ be a closed quadratic form.  Let $(V_R,R)$ be a (not necessarily positive or even bounded from below) sesquilinear form with $V_q \subset V_R$ so that for some $a \in (0,1)$ and $b > 0$ and all $\varphi \in V_q$, we have that
\begin{equation}\label{10.8a}
  |R(\varphi,\varphi)| \le aq(\varphi)+b\norm{\varphi}^2
\end{equation}
Define a quadratic form, $s$, with $V_s=V_q$ so that for $\varphi \in V_q$, we have that
\begin{equation}\label{10.9a}
  s(\varphi) = q(\varphi) + R(\varphi,\varphi)+b\norm{\varphi}^2
\end{equation}
Then $s$ is a positive, closed quadratic form.
\end{theorem}

\begin{remarks}  1.  The name comes from Kato \cite{KatoQF}, Lax--Milgram \cite{LaxMilgramKLMN}, Lions \cite{LionsKLMN} and Nelson \cite{NelsonKLMN}.

2.  If formally $q(\varphi)=\jap{\varphi,A\varphi}, R(\psi,\varphi)=\jap{\psi,C\varphi}$, then since $s$ is closed, we have that $s=q_D$.  Then $D-b\bdone$ gives a self--adjoint meaning to the formal sum $A+C$.  It is called the form sum.

3.  The proof is really simple.  If $\norm{\cdot} _{+1,q}$ and $\norm{\cdot} _{+1,s}$ are the $\norm{\cdot} _{+1}$ for $q$ and $s$, then \eqref{10.8a} implies that the two norms are equivalent so one is complete if and only if the other one is.
\end{remarks}

\begin{example} \lb{E10.8a} Let $q$ be the quadratic form, $q_A$, for $A=-\tfrac{d^2}{dx^2}$ on $L^2(\bbR,dx)$.  The same argument that we used to prove \eqref{7.10} shows that any $\varphi \in V_q$ is a continuous function and for some $C$ and all $\epsilon > 0$ and all $\varphi \in V_q$:
\begin{equation}\label{10.9b}
  |\varphi(0)|^2 \le C\left[\epsilon q(\varphi)+\epsilon^{-1}\norm{\varphi}^2\right]
\end{equation}
Thus, by the KLMN theorem, we can define $A=-\tfrac{d^2}{dx^2}+\lambda \delta(x)$ for any $\lambda \in \bbR$ as the quadratic form $q_\lambda$ with $V_{q_\lambda}=V_q$ and, for all $\varphi \in V_q$:
\begin{equation}\label{10.9c}
  q_\lambda(\varphi) = q(\varphi)+\lambda |\varphi(0)|^2
\end{equation}
\end{example}

The following is elementary to prove but useful

\begin{theorem} \lb{T10.9} The sum of two closed quadratic forms is closed
\end{theorem}

\begin{remarks} 1.  This allows a definition of a self--adjoint sum of any two positive self--adjoint operators.

2.  It is obvious that $V_{q_1+q_2} = V_{q_1} \cap V_{q_2}$.

3.  There is a similar result for $n$ arbitrary closed forms.

4.  The simplest proof is to use the Davies--Kato characterization (below) that closedness is equivalent to lower semicontinuity.
\end{remarks}

We end our discussion of the general theory by noting some distinctions between forms and symmetric operators.

$\circled{1}$. There are closed symmetric operators which are not self--adjoint but every closed quadratic form is the form of a self--adjoint operator.

$\circled{2}$.  Every symmetric operator has a smallest closed extension but there exist quadratic forms with no closed extensions.

$\circled{3}$. If $A$ and $B$ are self--adjoint operators and $B$ is an extension of $A$ (i.e. $D(A) \subset D(B)$ and $B \restriction D(A) = A$), then $A=B$.  But there exist closed quadratic forms $q_1$ and $q_2$ where $q_2$ is an extension of $q_1$ but $q_1 \ne q_2$.  For example, let $\calH = L^2([0,1],dx)$ and $q_0$ given by
\begin{equation*}
  q_0(\varphi) = \left\{
                   \begin{array}{ll}
                     \int_{0}^{1} |\varphi'(x)|^2\, dx, & \hbox{ if } \varphi \in C^\infty([0,1]) \\
                     \infty, & \hbox{ otherwise}
                   \end{array}
                 \right.
\end{equation*}
Here $C^\infty([0,1])$ means the functions infinitely differentiable on $[0,1]$ with one sided derivatives at the end points. Let $q_1$ be the closure of the restriction of $q_0$ to $C_0^\infty(0,1)$ and $q_2$ the closure of $q_0$.  Then $q_1$ is the quadratic form of $-\tfrac{d^2}{dx^2}$ with Dirichlet boundary conditions and $q_2$ the quadratic form of $-\tfrac{d^2}{dx^2}$ with Neumann boundary conditions (see \cite[Examples 7.5.25 and 7.5.26]{OT}) and $q_2$ is an extension of $q_1$.

Having completed our discussion of the general theory, we turn to a brief indication of its history.  In his original paper on self--adjoint operators \cite{vNSA}, von Neumann noted that if $A$ was a closed symmetric operator with
\begin{equation}\label{10.10a}
  \jap{\varphi,A\varphi} \ge \epsilon\norm{\varphi}^2
\end{equation}
for some $\epsilon > 0$ and all $\varphi\in D(A)$, $A^*\restriction D(A)+\ker(A^*)$ is a self-adjoint extension $A_{KvN}$ of $A$.  By looking at $(A-\epsilon_1\bdone)_{KvN}+\epsilon_1\bdone$ for any $\epsilon_1 < \epsilon$, we get self--adjoint extensions, $B_{\epsilon_1} \ge \epsilon_1\bdone$.  von Neumann conjectured there were self--adjoint extensions with lower bound exactly $\epsilon$.  Many years later, Krein \cite{KreinHisExt} (see also Ando--Nishio \cite{AndoN}) proved that $\lim_{\epsilon_1\uparrow\epsilon}B_{\epsilon_1}$ exists (this follows from the monotone convergence theorem below).  Put differently, given $A \ge 0$ symmetric, there is the Krein--von Neumann extension $A_{KvN} \equiv \lim_{\epsilon_2\downarrow 0}\left[(A+\epsilon_2\bdone)_{KvN}-\epsilon_2\bdone\right]$ which is a positive self--adjoint extension.  (The full theory of positive self--adjoint extensions \cite[Theorem 7.5.20]{OT} shows the set of such extensions is all positive self--adjoint operators, $B$ with $A_{KvN} \le B \le A_F$.)

Friedrichs \cite{FriedHisExt} (long before Krein) provided the first proof of von Neumann's conjecture (Stone \cite{StoneBk} had a proof at about the same time) by a construction related to the method behind Theorem \ref{T10.6}.  A follow--up paper of Freudenthal \cite{Freud} did Friedrichs extension in something close to form language.  In the 1950s, work on parabolic PDEs and NRQM by Kato \cite{KatoQF}, Lax--Milgram \cite{LaxMilgramKLMN}, Lions \cite{LionsKLMN} and Nelson \cite{NelsonKLMN} led to a systematic general theory.  In particular, Kato's lecture notes \cite{KatoQF} had considerable impact.

Next, we turn to a discussion of monotone convergence of quadratic forms.  Given a closed form, $q$, with $\calK$ the closure of $V_q$, define for $z \in \bbC\setminus\bbR$
\begin{equation}\label{10.11b}
  (\tilde{A}-z)^{-1} \equiv (A-z)^{-1}P_\calK
\end{equation}
i.e. under $\calH=\calK\oplus\calK^\perp$, $(\tilde{A}-z)^{-1} = (A-z)^{-1}\oplus 0$, consistent with how we said to define $f(A)$.

We will need the following result of Simon \cite{SimonMonotone} (see also \cite[Theorem 7.5.15]{OT})

\begin{theorem} \lb{T10.10}  Any quadratic form $q$ has an associated closed quadratic form, $q_r$, which is the largest closed form less than $q$, i.e. $q_r \le q$ and if $t$ is closed with $t \le q$, then $t \le q_r$.
\end{theorem}

\begin{remarks}  1.  One defines $q_s=q-q_r$. More precisely, $V_{q_s} = V_q$ and for $\varphi\in V_q$ we have that $q_s(\varphi)=q(\varphi)-q_r(\varphi)$. ``r'' is for regular and ``s'' for singular.

2.  Let $\mu$ and $\nu$ be two probability measures on a compact space, X, and $d\nu = fd\mu+d\nu_s$ with $d\nu_s$ singular wrt $d\mu$ the Lebesgue decomposition (see \cite[Theorem 4.7.3]{RA}).  If $\calH=L^2(X,d\mu)$ and if $q_\nu$ is defined with $V_{q_\nu}=C(X)$ and for $\varphi \in C(X)$
\begin{equation}\label{10.12}
  q_\nu(\varphi) = \int |\varphi(x)|^2 d\nu(x)
\end{equation}
then \cite[Problem 7.5.7]{OT} $(q_\nu)_r$ is the closure of the form (on $C(X)$)
\begin{equation}\label{10.13}
  \varphi \mapsto \int f(x)|\varphi(x)|^2 d\mu
\end{equation}
whose associated operator is multiplication by $f(x)$ (on the operator domain of those $\varphi$ with $\int f(x)^2|\varphi(x)|^2 d\mu < \infty$). $V_{q_s} = C(X)$. For $\varphi \in C(X)$, $q_s$ is given by \eqref{10.12} with $d\nu$ replaced by $d\nu_s$. In particular, if $q$ is the form of \eqref{10.11}, then $q_r=0$.
\end{remarks}

The two monotone convergence theorems for (positive) quadratic forms are

\begin{theorem} \lb{T10.11} Let $\{q_n\}_{n=1}^\infty$ be an increasing family of positive closed quadratic forms.  Define
\begin{equation}\label{10.14}
  q_\infty(\varphi) = \lim_{n \to \infty} q_n(\varphi) = \sup_n q_n(\varphi)
\end{equation}
Then $q_\infty$ is a closed form.  If $\calK_n$ (resp. $\calK_\infty$) is the closure of $V_{q_n}$ (resp. $V_{q_n}$) and $A_n$ (resp. $A_\infty$) the associated self--adjoint operators on $\calK_n$ (resp. $\calK_\infty$), then for any $z \in \bbC\setminus\bbR$, we have that
\begin{equation}\label{10.15}
  (\wti{A_n}-z)^{-1} \overset{s}{\to}(\wti{A_\infty}-z)^{-1}
\end{equation}
where $\tilde{B}$ is given by \eqref{10.11b}.
\end{theorem}

\begin{theorem} \lb{T10.12}  Let $\{q_n\}_{n=1}^\infty$ be a decreasing family of positive closed quadratic forms.  Define
\begin{equation}\label{10.16}
  q_\infty(\varphi) = \lim_{n \to \infty} q_n(\varphi) = \inf_n q_n(\varphi)
\end{equation}
Let $A_\infty$ be the self--adjoint operator on $\calK_\infty$, the closure of $V_{(q_\infty)_r}$ associated to $(q_\infty)_r$.  Let $A_n$ be as in the last theorem.  Then \eqref{10.15} holds.
\end{theorem}

\begin{remarks} 1. For proofs, see \cite[Theorem 7.5.18]{OT}.

2.  Let $q_n$ be the form of $-\tfrac{1}{n}\tfrac{d^2}{dx^2}+\delta(x)$ as defined in Example \ref{E10.8a}.  Then $q_n$ is decreasing and $q_\infty$ is the form $\delta(x)$ so that $(q_\infty)_r=0$.  This shows that in the decreasing case, the limit need not be closed or even closable.
\end{remarks}

Theorems of this genre appeared first in Kato's book \cite{KatoBk} (already in the first edition).  He only considered cases where all $V_{q_n}$ are dense.  In the increasing case, he assumed there was a $\tilde{q}$ with $V_{\tilde{q}}$ dense so that for all $n$, one has that $q_n \le \tilde{q}$.  In both cases, he proved there was a self--adjoint operator, $A_\infty$, with $A_n$ converging to $A_\infty$ in srs.  He considered the form $q_\infty(\varphi) = \lim_{n} q_n(\varphi)$.  In the decreasing case, he proved that if $q_\infty$ is closable, its closure is the form of $A_\infty$.  In the increasing case, he said it was an open question whether $q_\infty$ was the form of $A_\infty$.  This material from the 1966 first edition was unchanged from the 1976 second edition.

In 1971, Robinson \cite{RobMonotone} proved Theorem \ref{T10.11}.  He noted that $q_\infty$ was closed by writing $q_n = \sum_{j=1}^{n}s_j$ where $s_1=q_1, s_j=q_j-q_{j-1}$ if $j \ge 2$.  Then $q_\infty=\sum_{j=1}^{\infty} s_j$ and he says that the proof that $q_\infty$ is closed is the same as the proof that an infinite direct sum of Hilbert spaces is complete; see Bratteli--Robinson \cite[Lemma 5.2.13]{BR2} for a detailed exposition of the proof.  In 1975, Davies \cite{DaviesMonotone} also proved this theorem.  His proof relied on lower semicontinuity being equivalent to $q$ being closed (see below).  Robinson seems to have been aware of the results in Kato's book. While Davies quotes Kato's book for background on quadratic forms, he may have been unaware of the monotone convergence results which are in a later chapter (Chapter VIII) than the basic material on forms (Chapter VI).  When Kato published his second edition, he was clearly unaware of their work.

The lower semicontinuity fits in nicely with even then well known work on variational problems that used the weak lower semicontinuity of Banach space norms so it was not surprising.  Indeed Davies mentions it in passing in his paper without proof.  To add to the historical confusion, in his 1980 book \cite{DaviesBk}, when Davies quoted this result, he seems to have forgotten that it appeared first explicitly in his paper and attributes it to the 1966 first edition of Kato \cite{KatoBk} where it doesn't appear!

Shortly after this second edition, I wrote and published \cite{SimonMonotone} which had the notion of $(q)_r$ and the full versions of Theorems \ref{T10.11} and \ref{T10.12}.  I noted that these extended and complemented what was in Kato's book.  At the time I wrote the preprint, I was unaware of the relevant work of Davies and Robinson although I knew each of them personally.  In response to my preprint, Kato wrote to me that he had an alternate proof that in the increasing case, $q_\infty$ was always closed.  He stated a lovely result.

\begin{theorem} \lb{T10.13}  A quadratic form is closed if and only if it is lower semicontinuous as a function from $\calH$ to $[0,\infty]$.
\end{theorem}

\begin{remarks} 1.  For a proof, see \cite[Theorem 7.5.2]{OT}

2. This theorem provides a quick proof that $\delta(x)$ is not closable.  It is easy to find a $C_0^\infty(\bbR)$ function $\varphi$ with $\varphi(0)=1$ and a sequence $\varphi_n \in C_0^\infty$ with $\varphi_n(0)=0,\,\varphi_n \le \varphi$ and $\varphi_n \to \varphi$ in $L^2$.  Given this convergent sequence with $\lim \delta(\varphi_n) =0 < \delta(\varphi) = 1$, there cannot be a lower semicontinuous function that agrees with $\delta$ on $C_0^\infty$.
\end{remarks}

Given the theorem, it is immediate that $q_\infty$ is closed in the increasing case, since an increasing limit of lower semicontinuous functions is lower semicontinuous.  I note that in precisely this context, Theorem \ref{T10.13} was also found by Davies \cite{DaviesMonotone}.  Kato told me that he had no plans to publish his remark and approved my writing \cite{SimonLSC} that explores consequences of Theorem \ref{T10.13}.  However, in 1980, Springer published an ``enlarged and corrected'' printing of the second edition of Kato's book and one of the few changes was a completely reworked discussion of monotone convergence theorems!  In particular, he had the full Theorem \ref{T10.11} using Theorem \ref{T10.13}.  In the Supplemental Notes, he quotes \cite{SimonMonotone} and \cite{SimonLSC} but neither of the papers of Davies and Robinson, despite the fact that in response to their writing to me after the preprint, I added a Note Added in Proof to \cite{SimonMonotone} referencing their work.

The final topic of this section concerns pseudo--Friedrichs extensions and form definitions of the Dirac Coulomb operator. Recall that in Section \ref{s7} we discussed the free Dirac operator $T_0=\alpha\cdot(-i\nabla)+m\beta$ and the formal sum, \eqref{7.34}:
\begin{equation}\label{10.17}
  T=T_0+\frac{\mu}{|x|}
\end{equation}
As we saw in Section 7, Kato proved that \eqref{10.17} is esa--$3$ (where for the rest of the section, this means on $C_0^\infty(\bbR^3;\bbC^4)$) so long as $|\mu| < \tfrac{1}{2}$.  Moreover, one can prove esa--$3$ if and only if $|\mu| \le \tfrac{1}{2}\sqrt{3}$.  In his book, \cite[Sections V.5 and VII.3]{KatoBk}, Kato attempted to show that the $T$ of \eqref{10.17} had a natural self--adjoint extension for suitable $\mu \in (\tfrac{1}{2},1)$.  He found an extension of the KLMN theorem to cover cases where the unperturbed operator is not semibounded.  He proved the following result:

\begin{theorem} \lb{T10.14} Let $A$ be a self--adjoint operator and $B$ a symmetric operator with $D(B) \subset D(A)$ and so that $D(B)$ is a core for $|A|^{1/2}$.  Suppose that for some $a \in (0,1)$ and $b \ge 0$ and all $\varphi\in D(B)$ we have that
\begin{equation}\label{10.18}
  |\jap{\varphi,B\varphi}| \le a \jap{\varphi,|A|\varphi}+b\norm{\varphi}^2
\end{equation}
Then there is a unique self--adjoint operator, $C$, extending $A+B$ on $D(B)$ which also obeys
\begin{equation}\label{10.19}
  D(C) \subset D(|A|^{1/2})
\end{equation}
\end{theorem}

Kato called $C$ the \emph{pseudo--Friedrichs extension}.  Kato remarked that this had little to do with quadratic forms (which for him were positive) but the constructions shared elements of Friedrichs' construction of his extension.  Faris \cite{FarisBk} has a presentation that uses sesquilinear forms and makes this closer to the KLMN theorem.

In applying this to Dirac operators, Kato \cite{KatoBk} states without proof, that for each $\varphi \in C_0^\infty(\bbR^3)$, one has:
\begin{equation}\label{10.20}
  \jap{\varphi,|x|^{-1}\varphi} \le \tfrac{\pi}{2} \jap{\varphi,|p|\varphi}
\end{equation}
in the sense that
\begin{equation}\label{10.21}
  \int \frac{|\varphi(x)|^2}{x} d^3x \le \frac{\pi}{2} \int |k| |\hat{\varphi}(k)|^2 d^3k
\end{equation}

Like Hardy's and Rellich's inequality, this is scale invariant.  And Kato implies (but doesn't explicitly state) that $\tfrac{\pi}{2}$ is the optimal constant.  This is often called Kato's inequality (of course, it has no connection to what we called Kato's inequality in Section \ref{s9}).  In his book, Kato states this inequality with its optimal constant and then says that it is equivalent to $|p|^{-1/2}|x|^{-1}|p|^{-1/2}$ as an operator on $L^2$ having norm $\tfrac{\pi}{2}$.  He then notes that since $|x|^{-1}$ has a Fourier space kernel $(2\pi^2)^{-1}|k-k'|^{-2}$, one has to compute the norm of the integral operator with kernel $(2\pi^2)^{-1}(|k|\,|k'|)^{-1/2}|k-k'|^{-2}$ but he doesn't tell the reader how to actually compute this norm.  However, Kato's proof can be found in the appendix at the end of this paper.

So while the book is given as the source for the inequality, the standard place given for the proof is a lovely paper of Herbst \cite{HerbstKatoIneq} who computes the norm of $|x|^{-\alpha}|p|^{-\alpha}$ as an operator on $L^p(\bbR^\nu)$ when $1 < p < \nu\alpha^{-1}$  (that the operator is bounded on $L^p$ is a theorem of Stein--Weiss \cite{StW}).  This has as special cases the optimal constants for Kato's, Hardy's and Rellich's inequalities.  Herbst notes that this operator commutes with scaling, so after applying the Mellin transform, it commutes with translations and so, it is a convolution operator in Mellin transform space.  The function it is convolution with is positive function so the norm is related to the computable integral of this explicit function.  Five later publications on the optimal constant are Beckner \cite{BeckHardy}, Yafaev \cite{YafaevHardy}, Frank--Lieb--Seiringer \cite{FLSHardy},Frank--Seiringer \cite{FSHardy} and Balinsky--Evans \cite[pgs 48-50]{BEvans}.

In his book, Kato \cite{KatoBk} noted that by combining his definition of the pseudo--Friedrichs extension and his inequality, one can define a natural self--adjoint extension of \eqref{10.17} for $\tfrac{1}{2} \le \mu < \tfrac{2}{\pi}$.  But note that $\tfrac{2}{\pi} = 0.6366\dots$ while $\tfrac{1}{2}\sqrt{3}=0.866\dots$ so
\begin{equation}\label{10.22}
  \frac{2}{\pi} < \frac{\sqrt{3}}{2}
\end{equation}
and the regime that Kato was able to treat in his book was a subset of the region where Kato--Rellich fails but one can still prove esa--$3$ by other means!

That said, Kato's ideas stimulated later work which picked out a natural extension for all $\mu$ with $|\mu| < 1$.  Among the papers on the subject are Schmincke \cite{SchmDirac2}, W\"{u}st \cite{Wust1, Wust2, Wust3}, Nenciu \cite{NenciuDirac}, Kalf et. al. \cite{KSWW}, Estaban--Loss \cite{EL} and Estaban--Lewin--S\'{e}r\'{e} \cite{ELS}.  Domain conditions motivated by Kato's pseudo--Friedrichs extension are common.  Typical is the following result of Nenciu \cite{NenciuDirac} (which is a variant of Schmincke \cite{SchmDirac2}):

\begin{theorem} \lb{T10.15} For any $\mu$ with $|\mu|<1$, there exists a unique self--adjoint operator, $T$, with $D(T) \subset D(|T_0|^{1/2})$ so that for all $\varphi \in D(T), \psi \in D(T_0^{1/2})$ we have that
\begin{equation}\label{10.23}
  \jap{\psi,T\varphi} = \jap{|T_0|^{1/2}\psi,(T_0|T_0|^{-1/2})\varphi} + \mu \jap{r^{-1/2}\psi,r^{-1/2}\varphi}
\end{equation}
\end{theorem}

\eqref{10.23} uses the fact that, by the above mentioned inequality of Kato, if $\psi \in D(|T_0|^{1/2})$, then $\psi \in D(r^{-1/2})$.

In 1983, Kato wrote a further paper on the Dirac Coulomb problem \cite{KatoDirac} (see also \cite{KatoHolo}) which seems to be little known (I only learned of it while preparing this article).  To understand Kato's idea, return to $-\Delta-\beta r^{-2}$ on $L^2(\bbR^\nu), \nu \ge 5$ as discussed in Proposition \ref{P7.7} above.  If $0 < \beta \le \tfrac{\nu(\nu-4)}{4}$, then $H(\beta)$ can be defined as the operator closure of the operator on $C_0^\infty(\bbR^\nu)$.  It is self--adjoint and except at the upper end, we know the domain is that of $-\Delta$.  For $\tfrac{\nu(\nu-4)}{4} < \beta \le \tfrac{(\nu-2)^2}{4}$, there is a Friedrichs extension since $-\Delta-\beta r^{-2} \ge 0$ on $C_0^\infty(\bbR^\nu)$.  Kato notes that the Friedrichs extension is natural from the following point of view: $H(\beta)$ is an analytic family of operators for $0 < \beta < \tfrac{(\nu-2)^2}{4}$ and is the unique analytic family from the esa region -- it is type (A) if $\beta \in (0,\tfrac{\nu(\nu-4)}{4})$ and type (B) if $\beta \in (0,\tfrac{(\nu-2)^2}{4})$.  (In fact, it can proven that as a holomorphic family, there is a square root singularity at $\beta = \tfrac{(\nu-2)^2}{4}$ and in the variable $m = \sqrt{\beta-\tfrac{(\nu-2)^2}{4}}$, one has a holomorphic family in $\mbox{Re}(m) > -1$; see Bruneau--Derezi\'{n}ski--Georgescu \cite{BDG}).

In the same way, Kato showed that the distinguished self--adjoint extension of the Dirac operator in \eqref{10.17} found by others for $|\mu| < 1$ is an analytic family for $\mu \in (-1,1)$ and is the unique analytic continuation from the Kato--Rellich region $\mu \in (-\tfrac{1}{2},\tfrac{1}{2})$.

%%%%% %%%%%%%%%%%%%%%%%%%%%%%%%%%%%%%%%%%%%%%%%%%%%%%%%%%%%%%%%
\section{Eigenvalues, I: Bound State of Atoms} \lb{s11}
%%%%%%%%%%%%%%%%%%%%%%%%%%%%%%%%%%%%%%%%%%%%%%%%%%%%%%%%%%%%%%

In a short companion paper \cite{KatoHe} to his famous 1951 paper \cite{KatoHisThm}, Kato proved that

\begin{theorem} [Kato \cite{KatoHe}] \lb{T11.1} The non--relativistic Helium atom with infinite nuclear mass has infinitely many bound states.  With the physical masses, it has at least 25,585 bound states.
\end{theorem}

The number 25,585 seems unusual but it is just $\sum_{j=1}^{42} j^2$ corresponding to the number of bound states in the first 42 complete shells of a Hydrogenic atom.

An operator like the Helium atom Hamiltonian typically has an essential spectrum, $[\Sigma,\infty)$ (for an arbitrary self--adjoint operator, $A$,  we define $\Sigma(A) = \inf \{\lambda\,|\, \lambda \in \sigma_{ess}(A)\}$ where, we recall, $\sigma_{ess}(A) =\sigma(A)\setminus\sigma_d(A)$ and $\sigma_d(A)$, the discrete spectrum, is the isolated points of $\sigma(A)$, the spectrum, for which the spectral projection is finite dimensional (see Section \ref{s2}).

There may be one or more eigenvalues of $A$ below $\Sigma$, i.e., counting multiplicity, $\{E_k\}_{k=1}^N,\, N \in \{0,1,2,\dots\}\cup\{\infty\}$ where $E_{j-1} \le E_j <\Sigma$.  If $N=\infty$, then $\lim_{k \to \infty} E_k = \Sigma$

Most modern approaches to results like Theorem \ref{T11.1} rely on the min--max principle \cite[Theorem 3.14.5]{OT} which says that if $A$ is self--adjoint and bounded from below, and if one defines
\begin{equation}\label{11.1}
  \mu_n(A) = \sup_{\psi_1,\dots,\psi_{n-1}} \left(\inf_{\substack{\varphi \in D(A),\, \norm{\varphi}=1 \\ \varphi\perp\psi_1,\dots,\psi_{n-1}}} \jap{\varphi,A\varphi}\right)
\end{equation}
then $\mu_j(A) = E_j(A)$ for $j \le N$ and if $N < \infty$, then for $j > N$, $\mu_j(A) = \Sigma(A)$.  Instead, Kato notes the following

\begin{lemma} \lb{L11.2} Let $A$ be a self--adjoint operator which is bounded from below and $W \subset D(A)$ a subspace of dimension $k$ so that
\begin{equation}\label{11.3}
  \sup_{\varphi \in W,\,\norm{\varphi}=1} \jap{\varphi,A\varphi} = J
\end{equation}
then
\begin{equation}\label{11.4}
  \dim \ran\, P_{(-\infty,J]}(A)  \ge k
\end{equation}
\end{lemma}

\begin{remarks}  1.  $P_\Omega(A)$ are the spectral projections of $A$, see \cite[Section 5.1]{OT}.

2.  While Kato uses this lemma instead of the min-max principle, it should be emphasized that this lemma can be used to prove that principle!
\end{remarks}

\begin{proof}  Suppose that $\dim \ran\, P_{(-\infty,J]}(A) < k$.  Then we can find $\varphi \in W$ so $\varphi\perp\ran\, P_{(-\infty,J]}(A)$. Thus, by the spectral theorem $\jap{\varphi,A\varphi} > J$ contrary to \eqref{11.3}
\end{proof}

For Kato, $\Sigma$ is defined not in terms of essential spectrum but by
\begin{equation}\label{11.5}
  \Sigma = \inf \{\lambda\,|\,\dim \ran\, P_{(-\infty,\lambda)}(A) = \infty\}
\end{equation}
although it is the same.  His strategy is simple.

(1)  Get a lower bound, $\Sigma_0$, on $\Sigma$.

(2) Find a $k$--dimensional subspace, $W$, and a $J$ given by \eqref{11.3} which obeys $J < \Sigma_0$.  By \eqref{11.5}, $\dim \ran\, P_{(\infty,J]}(A) < \infty$ and by the lemma, it is at least $k$ so there must be at least $k$ discrete eigenvalues, counting multiplicity in $(-\infty,J]$.

Let's discuss first the case where the nuclear mass is infinite.  The Hamiltonian in suitable units is
\begin{equation}\label{11.6}
  H=-\Delta_1-\Delta_2 - \frac{2}{r_1}-\frac{2}{r_2}+\frac{1}{|\boldsymbol{r_1}-\boldsymbol{r_2}|}
\end{equation}
on $L^2(\bbR^6,d^6 x)$ where $x=(\boldsymbol{r_1},\boldsymbol{r_2}),\,\boldsymbol{r_j} \in \bbR^3$.  Kato then considers
\begin{equation}\label{11.7}
  \tilde{H} = H - \frac{1}{|\boldsymbol{r_1}-\boldsymbol{r_2}|} = h\otimes\bdone+\bdone\otimes h
\end{equation}
where
\begin{equation}\label{11.8}
  h=-\Delta-\frac{2}{r}
\end{equation}
He talks about ``two independent Hydrogen like atoms'' rather than tensor products, but it is the same thing.  Thus the spectrum of $\tilde{H} $ is $\{\lambda_1+\lambda_2\,|\, \lambda_1,\lambda_2 \in \sigma(h)\}$.  Since $\sigma(h) = \{-1/n^2\}_{n=1}^\infty \cup [0,\infty)$, we see that $\Sigma(\tilde{H}) = -1$.  Since $H \ge \tilde{H}$, we conclude that
\begin{equation}\label{11.9}
  \Sigma(H) \ge -1 \equiv \Sigma_0
\end{equation}
(we'll eventually see that this is actually equality).  This concludes step 1 in this infinite nuclear mass case.

Kato next picked the subspace, $W$, of trial functions.  Let $\varphi_0$ be the ground state of $h$, i.e.
\begin{equation}\label{11.10}
  h\varphi_0 = - \varphi_0
\end{equation}
Kato notes the explicit formula, $\varphi_0(\boldsymbol{x}) = \pi^{-1/2}e^{-|x|}$ but other than that it is spherically symmetric, the exact formula plays no role.  He picks $W=\{\varphi_0\otimes\eta \,|\, \eta \in W_1\}$ where $W_1$ will be a suitable subspace of $L^2(\bbR^3)$, i.e. $\varphi(\boldsymbol{x_1} ,\boldsymbol{x_2}) = \varphi_0(\boldsymbol{x_1})\eta(\boldsymbol{x_2})$

One easily computes that
\begin{equation}\label{11.11}
  \jap{\varphi,H\varphi} = -1+ \jap{\eta,(-\Delta+Q(x))\eta}
\end{equation}
where
\begin{equation}\label{11.12}
  Q(x) = -\frac{2}{|x|}+\int |\varphi_0(y)|^2 \frac{1}{|x-y|}\,d^3y
\end{equation}
The second term in \eqref{11.12} is the gravitational potential of a spherically symmetric ``mass distribution'' $|\varphi_0(y)|^2d^3y$ and this has been computed by Newton who showed that
\begin{equation}\label{11.13}
  \int_{S^2} \frac{d\omega}{|r\omega-\boldsymbol{x}|} = \frac{1}{\max(|x|,r)}
\end{equation}
(where $d\omega$ is normalized measure on the unit 2-sphere).  Thus
%\
\begin{align}
  Q(x) &= -\frac{2}{|x|} + \int |\varphi_0(y)|^2 \frac{1}{\max(|x|,|y|)}\,d^3y \nonumber  \\
       & \le -\frac{1}{|x| \lb{11.14}}
\end{align}
since $\max(|x|,|y|) \ge |x|$.  Thus
\begin{equation}\label{11.15}
  \jap{\varphi,H\varphi} \le -1+\jap{\eta,(-\Delta-1/r)\eta}
\end{equation}
Picking $\eta$ in the space of dimension $\tfrac{1}{6}k(k+1)(2k+1)$ of linear combinations of eigenfunctions of $-\Delta-1/r$ of energies $\{-\tfrac{1}{4j^2}\}_{j=1}^k$, we see that the $J$ of \eqref{11.3} is $-1-(1/4k^2) < \Sigma_0$, so there are infinitely many eigenvalues below $\Sigma_0$ (which also shows that $\Sigma=\Sigma_0$).

If one now considers a nucleus of mass $M$ and electrons of mass $m$, the Hamiltonian with the center of mass motion removed becomes (instead of \eqref{11.6})
\begin{equation}\label{11.16}
  H=-\Delta_1-\Delta_2 - 2\alpha\boldsymbol{\nabla_1}\cdot\boldsymbol{\nabla_2}- \frac{2}{r_1}-\frac{2}{r_2}+\frac{1}{|\boldsymbol{r_1}-\boldsymbol{r_2}|}
\end{equation}
where
\begin{equation}\label{11.17}
  \alpha = \frac{m}{M+m}
\end{equation}
The extra $2\alpha\boldsymbol{\nabla_1}\cdot\boldsymbol{\nabla_2}$ term, called the Hughes--Eckart term (after \cite{HE}), is present if one uses atomic coordinates, $\mathbf{r}_j=\mathbf{x}_j-\mathbf{x}_3;\,j=1,2$, where $\mathbf{x}_j$ is the coordinate of electron $j$ and $\mathbf{r}_3$ is the nuclear position (we'll say a lot about such $N$--body kinematics below).

The second step in the proof is unchanged.  Since $\jap{\varphi_0,\boldsymbol{\nabla}\varphi_0}=0$ (by either the reality of $\varphi$ or its spherical symmetry), the Hughes--Eckart terms contribute nothing to the calculation of $\jap{\varphi,H\varphi}$ and we get a subspace of trial functions of dimension $\tfrac{1}{6}k(k+1)(2k+1)$ with $J_k=-1-1/4k^2$.

Here is how Kato estimated $\Sigma$ in this case.  With $p_j=-i\nabla_j$, one can write:
\begin{equation}\label{1.18}
  \boldsymbol{p}_1^2+\boldsymbol{p}_2^2+2\alpha\boldsymbol{p}_1\cdot\boldsymbol{p}_2=\alpha(\boldsymbol{p}_1+\boldsymbol{p}_2)^2+(1-\alpha)(\boldsymbol{p}_1^2+\boldsymbol{p}_2^2)
\end{equation}
Since $|\boldsymbol{r}_1-\boldsymbol{r}_2|^{-1} \ge 0$ and $\alpha(\boldsymbol{p}_1+\boldsymbol{p}_2)^2 \ge 0$, we see that
\begin{equation}\label{11.19}
  H \ge (1-\alpha)(-\Delta_1-\Delta_2)-\frac{2}{r_1}-\frac{2}{r_2} \equiv H_{Kato}
\end{equation}
As in the infinite mass case, $H_{Kato}$ is a sum of independent Hydrogen like atoms, so one finds that
\begin{equation}\label{11.20}
  \Sigma \ge \Sigma_0 = \Sigma(H_{Kato})=-\frac{1}{1-\alpha}
\end{equation}
Putting in the physical value of $\alpha$ (i.e. \eqref{11.17} with $M=$Helium nuclear mass and $m=$electron mass), one finds that
\begin{equation}\label{11.21}
  \Sigma_0 \ge -1-1/4k^2 \textrm{ if } k \le 42
\end{equation}
so Kato concluded there were at least 42 shells and got the number 25,585 of Theorem \ref{T11.1}.

\begin{remarks}  1.  As Kato emphasized, before his work, it wasn't proven that the Helium Hamiltonian had any bound states!

2.  Kato ignored both spin (the Hamiltonian is spin--independent but each electron has two spin states, so on $L^2(\bbR^{3N};\bbC^2\otimes\bbC^2,d^{3N}x)$ there are 4 times as many states) and statistics (the Pauli principle, which, as interpreted by Fermi and Dirac, says the total wave function is antisymmetric under interchange of a pair of particles in both spin and space).  $H$ is symmetric under interchange of the two electrons in space alone, so its eigenfunctions can be chosen to be either symmetric or antisymmetric under spatial interchange.  Kato's trial functions are neither but the lower bound, $N_{Kato}$ that he gets provides a lower bound on $N_S+N_A$, the sum of the spatially symmetric and spatially antisymmetric functions.  To get a state totally antisymmetric under interchange of space and spin, each spatially symmetric wave function is multiplied by a spin 0 state (multiplicity 1) and each spatially antisymmetric state is multiplied by a spin 1 state (multiplicity 3).  So taking into account both spin and statistics, the total number of states is $N_S+3N_A$ so
\begin{equation}\label{11.22}
  N_S+N_A \le N_S+3N_A \le 3(N_S+N_A)
\end{equation}
In particular, $N_{Kato}$ is a lower bound on $N_S+3N_A$, so Kato's estimates are lower bounds even if one properly takes into account spin and statistics.

3.  Even in the infinite mass case, Kato's method doesn't work for three electron atoms.  The problem is with his estimate of $\Sigma$.  If one drops the repulsion of electron 3 from both 1 and 2, one gets an independent sum of an ion and a charge 3 Hydrogen like atom.  The bottom of the essential spectrum of such a system is actually twice the ground state energy of two of the charge 3 Hydrogen like atoms which is below the energy of the ion where one expects (and we actually know) the bottom of the essential spectrum really is.
\end{remarks}

This completes our description of Kato's paper.  To go beyond it, one realizes the weak point of his analysis (as seen in Remark 3 above) is no efficient way of estimating the bottom of the continuous spectrum.  As a preliminary to discussing this bottom, we pause to present some $N$--body kinematics, an issue that already entered when we discussed the Hughes--Eckart term above.  We'll be more expansive than absolutely necessary, in part, because we'll need this when we briefly turn to $N$--body scattering in Sections \ref{s13}-\ref{s15} and, in part, because the elegant formalism, which I learned from Sigalov--Sigal \cite{SigSig} (see also Hunziker--Sigal \cite{HunzSig}), deserves to be better known.

Given $N$ particles $(\boldsymbol{r}_1,\dots,\boldsymbol{r}_N)$ with masses $m_1,\dots,m_N$, we consider the inner product
\begin{equation}\label{11.23}
  \jap{r^{(1)},r^{(2)}} = \sum_{j=1}^{N} m_j \boldsymbol{r^{(1)}}_j \cdot \boldsymbol{r^{(2)}}_j
\end{equation}
on $x$--space, $X = \bbR^{\nu N}$.  This is natural because the free Hamiltonian
\begin{equation}\label{11.24}
  H_0 =  -\sum_{j=1}^{N} (2m_j)^{-1}\Delta_{\boldsymbol{r}_j}
\end{equation}
is precisely one half the Laplace--Beltrami operator for the Riemann metric associated to \eqref{11.23}.

We let $X^*$ be the dual to $X$, which we think of as momentum space.  If $\boldsymbol{p} \in X^*$ and $\boldsymbol{x} \in X$, they are paired as
\begin{equation}\label{11.25}
  \jap{\boldsymbol{p},\boldsymbol{x}} = \sum_{j=1}^{N} \boldsymbol{p}_j \cdot \boldsymbol{x}_j
\end{equation}
as occurs in the Fourier transform.  This induces an inner product on $X^*$
\begin{equation}\label{11.26A}
  \jap{p^{(1)},p^{(2)}}_{X^*} = \sum_{j=1}^{N} (m_j)^{-1} \boldsymbol{p^{(1)}}_j \cdot \boldsymbol{p^{(2)}}_j
\end{equation}
consistent with \eqref{11.24}

A coordinate change is associated to a linear basis, $e_1,\dots,e_N$ of $\bbR^{N}$ via
\begin{equation}\label{11.26}
  \boldsymbol{\rho}_j(\boldsymbol{x}_1,\dots,\boldsymbol{x}_N)= \sum_{r=1}^{N} e_{jr}\boldsymbol{x}_r
\end{equation}
(the $e_{jr} \in \bbR$ and $\boldsymbol{x}_r \in \bbR^\nu$.)

To be a trifle pedantic, we note that $X$ and $X^*$ depend on $N$ and $\nu$.  We'll use $Y$ for the case $\nu=1$ so that $X = Y \otimes \bbR^\nu$ and the $X$ inner product is the tensor product of the $Y$ inner product and the Euclidean inner product on $\bbR^\nu$ which we denoted with $\cdot$ in \eqref{11.23} and \eqref{11.26}.  Since the $e$'s act on $Y$, we think of them as lying in $Y^*$ (acting isotropically on the $\bbR^\nu$ piece).  The dual basis $f_j$ is defined by
\begin{equation}\label{11.27}
  \jap{f_j,e_\ell} = \delta_{j\ell} \textrm{, i.e. } \sum_{r=1}^{N} f_{jr} e_{\ell r} = \delta_{j \ell}
\end{equation}
If we think of $E, F$ as the $N\times N$ matrices with $F_{jr}=(f_j)_r, \, E_{jr} = (e_j)_r$, then \eqref{11.27} says that $FE^T=\bdone$.  Since $\bdone^T=\bdone$ and for finite matrices $AB=\bdone \Rightarrow BA=\bdone$, we conclude that $EF^T=E^TF=F^TE=\bdone$, i.e.
\begin{equation}\label{11.28}
  \sum_j f_{rj}e_{sj} =  \sum_j f_{jr}e_{js} =  \sum_j f_{sj}e_{rj} =  \sum_j f_{js}e_{jr} = \delta_{rs}
\end{equation}
First this implies that if
\begin{equation}\label{11.29}
  \boldsymbol{k}_j(\boldsymbol{p}_1,\dots,\boldsymbol{p}_N) = \sum_{q=1}^{N} f_{jq}\boldsymbol{p}_q
\end{equation}
then by \eqref{11.28}
\begin{align}
  \sum_{j=1}^{N} \boldsymbol{k}_j\cdot\boldsymbol{\rho}_j &= \sum_{j=1}^{N}\sum_{q=1}^{N}\sum_{r=1}^{N} f_{jq}e_{jr}\boldsymbol{p}_q\cdot\boldsymbol{x}_r \nonumber \\
                                                  &= \sum_{q=1}^{N} \boldsymbol{p}_q\cdot\boldsymbol{r}_q \lb{11.29A}
\end{align}
so the $k$'s are the Fourier duals to the $\rho$'s and \eqref{11.28} describes the transformation of momenta.

Moreover, we claim that $\jap{e_j,e_k}_{Y^*}$ and $\jap{f_j,f_k}_Y$ are inverse matrices to each other, i.e.
\begin{equation}\label{11.30}
  \jap{e,e}_{Y^*}\jap{f,f}_Y = \bdone
\end{equation}
If $e_j^{(0)} = \delta_j$, then $f_j^{0} = \delta_j$ and $\jap{e_j^{(0)},e_k^{(0)}}_{Y^*} = m_j^{-1} \delta_{jk}$ is indeed the inverse to $\jap{f_j^{(0)},f_k^{(0)}}_Y = m_j\delta_{jk}$. Since $e_r = \sum_{q=1}^{N}E_{rq}e_q^{(0)}$ and $f_j = \sum_{k=1}^{N}F_{jk}f_K^{(0)}$, we see that $\jap{e,e}_{Y^*}\jap{f,f}_Y = E^T\jap{e^{(0)},e^{(0)}}_{Y^*}EF^T\jap{f^{(0)},f^{(0)}}_Y F = \bdone$ by \eqref{11.28} and \eqref{11.27} for the $e^{(0)},f^{(0)}$ special case just proven.

Finally by \eqref{11.26} and \eqref{11.27}, we see that
\begin{align}
  \sum_{j=1}^{N} m_j \boldsymbol{r}_j^2 &= \sum_{r,s=1}^{N} \jap{f_r,f_s}_Y \boldsymbol{\rho}_r \cdot \boldsymbol{\rho}_s \lb{11.31} \\
  \sum_{j=1}^{N} m_j^{-1} \boldsymbol{p}_j^2 &= \sum_{r,s=1}^{N} \jap{e_r,e_s}_{Y^*} \boldsymbol{k}_r \cdot \boldsymbol{k}_s \lb{11.32}
\end{align}
\begin{example} [Removing the center of mass] \lb{E11.3}  First consider $N=2$.  Since we have $V(\boldsymbol{r}_1-\boldsymbol{r}_2)$, we want $\boldsymbol{r}_1-\boldsymbol{r}_2$ to be one coordinate, i.e. $e_1=(1,-1)$.  The natural second coordinate should be orthogonal in $Y^*$, i.e. $\tfrac{1}{m_1}e_{21}-\tfrac{1}{m_2}e_{22} = 0$ so $(m_1,m_2)$ will work but it is more usual to take $e_2 = \tfrac{1}{M}(m_1,m_2),\, M=m_1+m_2$ the total mass.  That is, the second coordinate is $(m_1\boldsymbol{r}_1+m_2\boldsymbol{r}_2)/M$, the center of mass.  One computes
\begin{equation*}
  \jap{e_1,e_1}_{Y^*}=\frac{1}{m_1}+\frac{1}{m_2}\equiv \frac{1}{\mu} \qquad \jap{e_1,e_2}_{Y^*}= 0
\end{equation*}
\begin{equation}\label{11.33}
   \jap{e_2,e_2}_{Y^*} =\frac{1}{M^2}\left(\frac{m_1^2}{m_1}+\frac{m_2^2}{m_2}\right)=\frac{1}{M}
\end{equation}
We compute
\begin{equation}\label{11.34}
  f_1=\left(\frac{m_2}{M},-\frac{m_2}{M}\right), \qquad f_2=(1,1)
\end{equation}
By either direct calculation or \eqref{11.30}
\begin{equation*}
  \jap{f_1,f_1}_{Y}=\frac{m_1m_2^2+m_1^2m_2}{M^2} = \frac{m_1m_2}{M} = \mu \qquad \jap{f_1,f_2}_{Y}= 0
\end{equation*}
\begin{equation}\label{11.35}
  \jap{f_2,f_2}_Y = m_1+m_2=M
\end{equation}
Thus
\begin{align}
  \boldsymbol{r}_{12} = \boldsymbol{r}_1-\boldsymbol{r}_2 \qquad & \boldsymbol{R} = \frac{1}{M}(m_1\boldsymbol{r}_1+m_2\boldsymbol{r}_2) \\
  \boldsymbol{k}_{12} = \frac{m_2\boldsymbol{p}_1-m_1\boldsymbol{p}_2}{M} \qquad & \boldsymbol{K} = \boldsymbol{p}_1+\boldsymbol{p}_2
\end{align}
and we see that
\begin{equation}\label{11.37}
  m_1\boldsymbol{r}_1^2+m_2\boldsymbol{r}_2^2 = \mu\boldsymbol{r}_{12}^2+M\boldsymbol{R}^2;  \qquad H_0 = -\frac{1}{2M}\Delta_{\boldsymbol{R}}-\frac{1}{2\mu}\Delta_{\boldsymbol{r}_{12}}
\end{equation}

For $N$ bodies, motivated by the above, we want to take $f_N=(1,\dots,1)$ and $f_1,\dots,f_{N-1}$ all orthogonal to it.  Then $\jap{f,f}_Y$ will be the direct sum of an $(N-1)\times(N-1)$ matrix and $\jap{f_N,f_N}_Y=M$.  Thus $\jap{e,e}_{Y^*}$ with be the direct sum of an $(N-1)\times(N-1)$ matrix and $\jap{e_N,e_N}_{Y^*}=1/M$.  Moreover, we claim that
\begin{equation}\label{11.38}
  \jap{e_N,f} = \jap{f_N,f}/\jap{f_N,f_N}
\end{equation}
since this holds for each $f_j$.  Putting $f=\delta_j$ in, we conclude that $e_N=M^{-1}(m_1,\dots,m_N)$.  We summarize in this Proposition
\end{example}

\begin{proposition} \lb{P11.4} In any coordinate system, $\boldsymbol{\rho}_1,\dots,\boldsymbol{\rho}_N$ where $\boldsymbol{\rho}_j,\,j=1,\dots,N-1$ is a linear combination of $\boldsymbol{r}_k-\boldsymbol{r}_\ell$ and
\begin{equation}\label{11.39}
  \boldsymbol{\rho}_N=\frac{1}{M}\sum_{j=1}^{N} m_j\boldsymbol{r}_j
\end{equation}
we have that
\begin{equation}\label{11.40}
  H_0 = -\sum_{j=1}^{N} \frac{1}{2m_j}\Delta_{\boldsymbol{r}_j} = h_0\otimes\bdone + \bdone\otimes T_0
\end{equation}
where $h_0=-(2M)^{-1}\Delta_{\boldsymbol{\rho}_N}$ and $T_0$ is a quadratic form in $-i\boldsymbol{\nabla}_{\boldsymbol{\rho}_j}, \, j=1,\dots,N-1$.
\end{proposition}

\begin{example} [Atomic Coordinates] \lb{E11.5} This is named for the natural coordinates when there is a heavy nucleus, $\boldsymbol{r}_N$ and $N-1$ electrons.  We take (with $m_j=m$ for $j=1,\dots,N-1$)
\begin{equation}\label{11.41}
     \boldsymbol{\rho}_j = \boldsymbol{r}_j-\boldsymbol{r}_N,\, j=1,\dots,N-1; \qquad   \boldsymbol{\rho}_N=\frac{1}{M}\sum_{j=1}^{N} m_j\boldsymbol{r}_j
\end{equation}
Thus, by \eqref{11.26}
\begin{equation}\label{11.42}
  e_j=\delta_j-\delta_N; \qquad e_N=\frac{1}{M}
\end{equation}
Since $\jap{a,a}_{Y^*}=\sum_{j=1}^{N} m_j^{-1}a_j^2$, we see that
\begin{equation}\label{11.43}
  \jap{e_N,e_j}_{Y^*} = M^{-1} \delta_{Nj}
\end{equation}
\begin{equation}\label{11.44}
  \jap{e_j,e_j}_{Y^*} = \frac{1}{m}+\frac{1}{m_N} \equiv \frac{1}{\mu} \qquad j=1,\dots,N-1
\end{equation}
\begin{equation}\label{11.44A}
  \jap{e_j,e_k} = \frac{1}{m_N} \qquad 1\le j,k\le N-1;\, j\ne k
\end{equation}
Thus, by \eqref{11.32}
\begin{align}
  T_0 &= -\sum_{j,k=1}^{N-1} \frac{1}{2} \jap{e_j,e_k}_{Y^*} \boldsymbol{\nabla}_j\cdot\boldsymbol{\nabla}_k \nonumber    \\
      &= -\sum_{j=1}^{N-1} \frac{1}{2\mu}\Delta_j - \frac{1}{m_N} \sum_{j<k} \boldsymbol{\nabla}_j\cdot\boldsymbol{\nabla}_k  \lb{11.44B}
\end{align}
(there is no 2 in front of $m_N$ because we have changed from a sum over $j \ne k$ to $j \le k$.)  Noting that
\begin{equation*}
  \frac{\mu}{m_N}=\frac{m\,m_n}{m+m_n}\frac{1}{m_n}=\frac{m}{m_n+m}
\end{equation*}
which is \eqref{11.16}/\eqref{11.17} (taking into account a changed meaning for the symbol $M$ there and here!).
\end{example}
\emph{}
\begin{example} [Jacobi Coordinates]  These coordinate changes go back to classical mechanics.  Jacobi noted one could avoid cross terms in the kinetic energy changing first from $r_1$ and $r_2$ to $r_{1,2}$ and the center of mass, $R_{12}$, of the first two particles.  Then one goes from $R_{12}$ and $r_3$ to $r_3-R_{12}$ and the center of mass of the first three particles.  After $N-1$ steps, one has $R$, the total center of mass as one of the coordinates, and $N-1$ ``internal'' coordinates.
\end{example}

\begin{example} [Clustered Jacobi Coordinates] \lb{E11.7} Given $\{1,\dots,N\}$, a \emph{cluster decomposition} or clustering, $\calC = \{C_\ell\}_{\ell=1}^k$, is a partition, i.e. a family of disjoint subsets whose union is $\{1,\dots,N\}$.  We set $\#(C_\ell)$ to be the number of particles in $C_\ell$.  A coordinate, $\boldsymbol{\rho}$, is said to be internal to $C_\ell$ if it is a function only of $\{\boldsymbol{r}_m\}_{m \in C_\ell}$ and is invariant under $\boldsymbol{r}_m \to \boldsymbol{r}_m+\boldsymbol{a}$, , equivalently, it is a linear combination of $\{\boldsymbol{r}_m-\boldsymbol{r}_q\}_{m,q \in C_\ell}$.  A \emph{clustered Jacobi coordinate system} is a set of $\#(C_\ell)-1$ independent internal coordinates for each cluster together with $\boldsymbol{R}_\ell = (\sum_{q \in C_\ell} m_q\boldsymbol{r}_q)/(\sum_{q \in C_\ell} m_q)$,  If we write $\calH(C_\ell)$ to be $L^2$ of the internal coordinates and $\calH^{(\calC)}$ to be $L^2$ of the internal coordinates then
\begin{equation}\label{11.46}
  \calH = \calH^{(\calC)} \otimes \bigotimes_{\ell=1}^k \calH(C_\ell)
\end{equation}
\begin{equation}\label{11.47}
  H_0 = \wti{T}^{(\calC)}\otimes\bdone\dots\otimes\bdone+ \sum_{\ell=1}^{k} \bdone\otimes\dots\otimes T(C_\ell) \otimes\dots\otimes\bdone
\end{equation}
where $\wti{T}^{(\calC)} = -\sum_{\ell=1}^{k} (2M(C_\ell))^{-1}\Delta_{\boldsymbol{R}_\ell}$ and $T(C_\ell)$ is a quadratic form in the derivatives of the internal coordinates.
\end{example}

As noted, the big limitation in Kato's work on Helium bound states concerns his estimate of $\Sigma$, the bottom of the essential spectrum of $H$. We turn to understanding that.  In the two body case, $H=-\Delta+V$, one expects that $\sigma_{ess}(H)=[0,\infty)$.  This requires that $V$ go to zero at spatial infinity in some sense.  If one is looking at $V$'s for which $D(H)=D(-\Delta)$, the natural condition is that $V(-\Delta+1)^{-1}$ is a compact operator (see \cite[Section 3.14]{OT}).  To be explicit, we introduce $L^p(\bbR^\nu)+L^\infty(\bbR^\nu)_\epsilon$ to be the set of $V$ so that for any $\epsilon>0$, one can decompose $V=V_{1,\epsilon}+V_{2,\epsilon}$ with $V_{1,\epsilon} \in L^p(\bbR^\nu)$ and $\norm{V_{2,\epsilon}}_\infty \le \epsilon$.  If $p$ is $\nu$--canonical, one can prove that if $V \in L^p(\bbR^\nu)+L^\infty(\bbR^\nu)_\epsilon$, then $V(-\Delta+1)^{-1}$ is compact and $\sigma_{ess}(H)=[0,\infty)$.  If one wishes, there are Stummel--type conditions to replace this but we'll make such $L^p$ assumptions below for simplicity of exposition.

We also want to remove the total center of mass motion if all masses are finite.  That is we let $\boldsymbol{R}=\left(\sum_{j=1}^{N} m_j\boldsymbol{r}_j\right)/\left(\sum_{j=1}^{N}m_j\right)$ and pick some set of internal coordinates so that $\calH^{full} = \calH_{CM}\otimes\calH,\, \calH^{full} = L^2(\bbR^{\nu N}), \calH_{CM}=$ functions of $\boldsymbol{R}$, $\calH=$ functions of the internal coordinates.  If $H^{full}=H_0+\sum_{j<k} V_{jk}$, then under this tensor product decomposition
\begin{equation}\label{11.48A}
  H^{full}=H_{0,CM}\otimes\bdone+\bdone\otimes H
\end{equation}
where $H_{0,CM}=-(2\sum_{j=1}^{N}m_j)^{-1}\Delta_{\boldsymbol{R}}$.  We'll consider $H$ below.

In \eqref{11.47}, the operator $\wti{T}^{(\calC)}$ has a decomposition like \eqref{11.48A} where $\calH$ is replaced by $\calH^{(\calC)}$, the functions of the differences of the centers of mass of the $C_j$.  We write
\begin{equation}\label{11.48B}
  \wti{T}^{(\calC)} = H_{0,CM}\otimes\bdone+\bdone\otimes T^{(\calC)}
\end{equation}

Given a cluster decomposition, $\calC=\{C_\ell\}_{\ell=1}^k$, we write $(jq) \subset \calC$ if $j$ and $q$ are in the same cluster of $\calC$ and $(jq) \not\subset \calC$ if they are in different clusters.  We define
\begin{align}
  V(C_\ell)  &= \sum_{\substack{j,q\in C_\ell \\ j<q }} V_{jq} \lb{11.49} \\
  V(\calC)   &= \sum_{\ell=1}^{k} V(C_\ell) = \sum_{\substack{ (jq) \subset \calC \\ j<q}} V_{jq} \lb{11.50} \\
  I(\calC)   &= \sum_{j<q} V_{jq} - V(\calC) = \sum_{\substack{ (jq) \not\subset \calC \\ j<q}} V_{jq} \lb{11.51}
\end{align}
$V(\calC)$ is the \emph{intracluster interaction} and $I(\calC)$ \emph{the intercluster interaction}.  We define on $\calH(C_\ell)$
\begin{equation}\label{11.52}
  h(C_\ell) = T(C_\ell) + V(C_\ell)
\end{equation}
\begin{align}
  H(\calC)  &= T^{(\calC)}\otimes\bdone\dots\otimes\bdone+\sum_{\ell=1}^{k} \bdone\otimes\dots\otimes h(C_\ell)\otimes\dots\otimes\bdone \lb{11.53} \\
            &= H-I(\calC) \nonumber \\
   \Sigma(\calC) &= \sum_{\ell=1}^{k} \inf\sigma(H(C_\ell)) \lb{11.54}
\end{align}
We let $\calC_{min}$ be the one cluster decomposition of $\{1,\dots,N\}$ so $H(\calC_{min}) = H$.  We note that
\begin{equation}\label{11.54A}
    \calC \ne \calC_{min} \Rightarrow \sigma(T^{(\calC)}) = [0,\infty)
\end{equation}

By \eqref{11.53}, we have that $\sigma(H(\calC))=\sigma(T^{(\calC)})+\sigma(H(C_1))+\dots+\sigma(H(C_k))$.  By \eqref{11.54A}
\begin{equation}\label{11.55}
  \calC \ne \calC_{min} \Rightarrow\sigma(H(\calC)) = [\Sigma(\calC),\infty)
\end{equation}
When we discuss $N$--body spectral and scattering theory briefly in Sections \ref{s12}--\ref{s14}, we'll be interested in thresholds.  A \emph{threshold}, $t$, is a decomposition $\calC=\{C_\ell\}_{\ell=1}^k \ne \calC_{min}$ and an eigenvalue, $E_\ell$ of $h(C_\ell)$ for each $\ell=1.\dots,k$.  The \emph{threshold energy} is $E(t)=\sum_{\ell=1}^{k} E_\ell$.  Of course, $E(t) \ge \Sigma(\calC)$.

Fix $\calC \ne \calC_{min}$.  Pick distinct vectors, $X_1,\dots,X_k \in \bbR^\nu$.  For $\lambda \in \bbR$, let $U(\lambda)$ be the unitary implementing $x_j \mapsto x_j+\lambda X_p$ if $j \in C_q$.  It is easy to see that $U(\lambda)H(\calC)U(\lambda)^{-1} = H(\calC)$ and if each $V_{jq} \in L^p(\bbR^\nu)+L^\infty(\bbR^\nu)_\epsilon$, then for all $\varphi \in D(-\Delta)$ one has that
\begin{equation}\label{11.56}
  \lim_{\lambda \to \infty}[U(\lambda)H U(\lambda)^{-1} - H(\calC)]\varphi =0
\end{equation}
which implies \cite[Problem 3.14.5]{OT} that $\sigma(H(\calC)) = [\Sigma(\calC),\infty) \subset \sigma(H)$.  In particular, if
\begin{equation}\label{11.57}
  \Sigma = \inf_{\calC \ne \calC_{min}} \Sigma(\calC)
\end{equation}
then
\begin{equation}\label{11.58}
  [\Sigma,\infty) \subset \sigma(H)
\end{equation}
The celebrated HVZ theorem says that

\begin{theorem} [HVZ Theorem] \lb{T11.8} For $N$--body Hamiltonians with $V_{jq} \in L^p(\bbR^\nu)+L^\infty(\bbR^\nu)_\epsilon$ (with $p$ $\nu$--canonical) one has that
\begin{equation}\label{11.59}
  \sigma_{ess}(H) = [\Sigma,\infty)
\end{equation}
\end{theorem}

\begin{remarks}  1.  There is a variant where there are infinite mass particles, i.e. some $V_j$ terms, and the center of mass isn't removed.  Decompositions are now of $\{0,1,\dots,N\}$.  One says that $(j) \subset \calC$ if $0$ and $j$ are in the same cluster.

2.  The result is named after Hunziker \cite{HunzHVZ}, van Winter \cite{vWHVZ} and Zhislin \cite{ZhisInfinite}.

3.  There are essentially three generations of proofs of this theorem.  The initial proofs of Hunziker and van Winter relied on integral equations (what are now called the Weinberg--van Winter equations).  van Winter restricted her work to $L^2(\bbR^3)$ potentials since she only considered Hilbert--Schmidt operators while Hunziker's independent work handled the general case above.  This work was independent of the earlier work of Zhislin who only considered and proved results for atomic Hamiltonians.  His methods were geometric.

4.  The second wave concerns geometric proofs by Enss \cite{EnssHVZ}, Simon \cite{SimonHVZ}, Agmon \cite{AgmonBk}, G\r{a}rding \cite{GardingHVZ} and Sigal \cite{IMSS}.  In one variant, the key is a geometric fact that there exists a partition of unity $\{J_\calC\}_{\calC \ne \calC_{min}}$ indexed non--minimal partitions so that $\sum_{\calC} J_\calC = \bdone$ and so that on $\supp J_\calC \cap \{x\,|\, |x| > 1\}$, one has that, for some $Q>0$, $|x_j-x_k| \ge Q|x|$ if $(jk) \not\subset \calC$.  One proves that $[f(H)-f(H(\calC))]J_\calC$ is a compact operator for continuous functions, $f$ of compact support.  This, in turn, implies that when $\supp f \subset (-\infty,\Sigma)$, then $f(H)$ is compact.  For details, see \cite[Section 3.3]{CFKS}.  Agmon's version \cite{AgmonBk} looks at limits as one translates in an arbitrary direction and is especially intuitive.  In this regard, Agmon considered a class of potentials that generalize $N$--body systems. $\{\pi_j\}$ is a family of non-trivial projections in $\bbR^{\nu N}$ and $V=\sum V_j(\pi_j x)$ where $V_j$ is a functions on $\bbR^{\dim \ran \pi_j}$.  This setup has been used by many authors since.

5.  The third generation works in cases where $\sigma_{ess}(A)$ can have gaps.  This approach appeared (more or less independently) in Chandler--Wilde--Lindner \cite{CWL1, CWL2}, Georgescu--Iftimovici \cite{GI}, Last--Simon \cite{LS1, LS2}, M\v{a}ntoiu \cite{Mant} and Rabinovich \cite{Rabin}.  Perhaps the cleanest result from \cite{LS2} defines the notion of right limits and proves that $\sigma_{ess}(H)$ is the union over all right limits of $\sigma(H_r)$.  See also \cite[Section 7.2]{SimonSz}.
\end{remarks}

With the HVZ theorem in hand, one can easily carry Kato's argument to its logical conclusion

\begin{theorem} [Simon \cite{SimonInfinite}] \lb{T11.9} Let $H$ be an $N$--body Hamiltonian with center of mass removed.  Suppose that $\Sigma$ is a two--body threshold, i.e. there is a cluster decomposition, $\calC = \{C_1,C_2\}$ and vectors, $\varphi_j \in \calH(C_j),\, j=1,2$ so that $H(C_j)\varphi_j = E_j\varphi_j$, $\norm{\varphi_j}=1$ and $E_1+E_2 = \Sigma$.  Define $W$ on $\bbR^\nu$ as follows: $y \in \bbR^\nu$ is the difference of the centers of mass of $C_1$ and $C_2$ and let $x_k(y,\zeta_1,\zeta_2)$ be the position of particle $k$ in terms of $y$ and the internal coordinates $\zeta_j$ of $C_j$.  Then
\begin{equation}\label{11.60}
  W(y) = \sum_{\substack{q \in C_1 \\ k \in C_2}} \int V_{qk}(x_q(y,\zeta_j)-x_k(y,\zeta_j)) |\varphi_1(\zeta_1)|^2 |\varphi_2(\zeta_2)|^2 d\zeta_1 d\zeta_2
\end{equation}
Let $\mu$ be the reduced mass of the two clusters and suppose that
\begin{equation}\label{11.61}
  -(2\mu)^{-1}\Delta_y + W(y)
\end{equation}
has an infinite number of eigenvalues below $0$ as an operator on $L^2(\bbR^\nu)$.  Then $H$ has an infinite number of eigenvalues below $\Sigma$.
\end{theorem}

\begin{remarks}  1. Thus, with $M(C_j) = \sum_{k \in C_j} m_k$, we have that $\mu^{-1}=M(C_1)^{-1}+M(C_2)^{-1}$

2.  One might think that if $j \in C_1$, then $x_j(y,\zeta_1,\zeta_2)$ is independent of $\zeta_2$ but that's wrong for the total center of mass, $\boldsymbol{R}$, enters in $x_j$ and that causes a $\zeta_2$ dependence.

3.  The proof is essentially unchanged from the ideas in Kato \cite{KatoHe}.  If $\psi(y,\zeta_1,\zeta_2) = \varphi_1(\zeta_1) \varphi_2(\zeta_2)\eta(y)$, then $\jap{\psi,H\psi} = \Sigma+{\jap{\eta,(-(2\mu)^{-1}\Delta+W)\eta}}$.

4.  This result is from Simon \cite{SimonInfinite} who revisited Kato's paper after the discovery of the HVZ theorem.
\end{remarks}

Now fix $Z,N > 0$.  $N$ is an integer but $Z$ need not be.  We define on $L^2(\bbR^{3N})$:
\begin{equation}\label{11.62}
  H(Z,N) = \sum_{j=1}^{N} \left(-\Delta_j-\frac{Z}{|x_j|}\right) + \sum_{1 \le j,k \le N} \frac{1}{|x_j-x_k|}
\end{equation}
\begin{equation}\label{11.63}
  E(Z,N) = \inf\sigma(H(Z,N))
\end{equation}
One can accommodate Hughes Eckart terms in much of the discussion but we won't include them.

By the arguments before \eqref{11.55}, $\sigma(H(Z,N-1)) \subset \sigma(H(Z,N))$ so the HVZ theorem implies that
\begin{equation}\label{11.64}
   \Sigma(H(Z,N))=E(Z,N-1)
\end{equation}
so we are interested in
\begin{equation}\label{11.65}
  \delta(Z,N) = -E(Z,N)+E(Z,N-1)
\end{equation}
the ionization energy to remove electron $N$ from a nucleus of charge $Z$.  Put differently, $\delta \ge 0$ and $\delta > 0$ if and only if $N$ electrons bind to a charge $Z$ nucleus.

\begin{corollary} [Zhislin \cite{ZhisInfinite}] \lb{C11.10}  If $Z > N-1$, then $H(Z,N)$ has infinitely many bound states below $\Sigma$.  In particular, $\delta(Z,N) > 0$.
\end{corollary}

\begin{remarks} 1.  This is because by induction, $\Sigma$ is determined by a two cluster breakup into $N-1$ particles (in the same cluster as $0$) and one particle and then that $W(y) = [Z-(N-1)]|y|^{-1} + \textrm{o}(1/|y|)$ and such a potential has infinitely many bound states.

2.  This result was first proven by Zhislin using arguments somewhat more involved than Kato's argument (and before Simon noted that Kato's arguments work).
\end{remarks}

This completes the summary of the direct extensions of Kato's work.  We will end this section with a brief discussion of results on bound states of $H(Z,N)$ which are a direct descendent of Kato's consideration.  There is an enormous literature not only on this subject but also on bounds on the number of bound states when finite and on moments of the eigenvalues.  We refer the reader to the forthcoming book of Frank, Laptev and Weidl \cite{FLW}.

The other side of Corollary \ref{C11.10} is

\begin{theorem} \lb{11.11} If $Z \le N-1$, then $H(Z,N)$ has only finitely many bound states.
\end{theorem}

\begin{remarks} 1.  This theorem is due to Zhislin \cite{ZhisFinite}.  There were earlier results of Uchiyama \cite{UchiFinite} (for $N=2, Z < 1$),  and by Vugal'ter--Zhislin \cite{VugZhis} and Yafaev \cite{Yafaev3Ann, Yafaev3Full} (for $Z=N-1$).

2.  The intuition is that the left over Coulomb repulsion (if $Z < N-1$) or residual Coulomb attraction (if $Z=N-1$) is such that an effective $-\Delta+W$ has only finitely many states.  Of course, one needs techniques to conclude that when an effective two body problem has that property, the full $N$--body does -- one of the most effective methods is due to Sigal \cite{IMSS}.  I note in passing that there are three particle systems with short range interactions that surprisingly have an infinite number of bound states, $\{E_j\}_{j=1}^\infty$ with asymptotic geometric sequence placement, i.e. $E_{j+1}/E_j \to \alpha < 1$.  At least two of the three two body clusters must have zero energy resonances (what this means is discussed in Section \ref{s16}) so the bottom of the essential spectrum is 0.  The discovery on a formal level is due to Efimov \cite{Efimov} after whom the effect is named.  For mathematical proofs see Yafaev \cite{YafaevEfimov}, Tamura \cite{Tam1, Tam2}, Sobolev \cite{SoboEff} and Ovchinnikov--Sigal \cite{OvSigal}.  Wang \cite{Wang1, Wang2} discussed this for $N$--body systems.  For popular science treatments of experimental verification of the geometric progression (even for small $j$!) see Ouellette \cite{Oue} and Wolchover \cite{Wolch}.

3. This theorem is stated for systems with no statistics.  For $Z < N-1$, the result extends without much trouble to Fermi statistics \cite{ZhisFinite}.  For $Z=N-1$, one needs to assume that there is not an atomic ground state with a dipole moment (for there to be such a state, there would need to be a degeneracy of states with different parity) because $-\Delta+\lambda \hat{e}\cdot\boldsymbol{r}/(1+r)^3$ has an infinity of bound states when $\lambda$ is large enough.  In fact, in \cite{SimonHVZ}, it is claimed (quoting Lieb) that a molecule with two centers, $Z_1=1/3, Z_2=2/3, N=2$ (so $Z=N-1$) and $|\boldsymbol{R_1} - \boldsymbol{R_2}|$ large will have an infinity of bound states (although a proof has never been published to my knowledge).  In any event, under an assumption about no atomic ground state with dipole moment, the theorem does extend to $N=Z+1$ \cite{VugZhis}.
\end{remarks}

For most of the discussion below, we look at $E(Z,N)$ with Fermi statistics.  One might expect that for $Z$ fixed, one has that $\delta(Z,N) = 0$ for all sufficiently large $N$, i.e. there is an $N_c(Z)$ so that $\delta(Z,N)=0$ if $N \ge N_c(Z)$ and so that $\delta(Z,N_c(Z)-1) > 0$.  Ruskai \cite{RuskaiRS}  and Sigal \cite{IMSS, SigalRS} proved that for every $Z$, there is a such an $N_c(Z)$ and Lieb \cite{Lieb2Zplus1} found a simple, elegant argument that $N_c(Z) \le 2Z+1$ which, in particular, implies that $H^{--}$ does not exist although $H^{-}$ does.

In nature, there is no known example for $\delta(Z,N) > 0$ if $N \ge Z+2$, that is, there are once negatively charged ions in nature, but no twice negatively charged ions.  So it might even be that $N_c(Z)$ is always bounded by $Z+1$.  In any event, there is a conjecture \cite{Si21Cent} that $N_c(Z) \le Z+k$ for some finite $k$.  It is known (Lieb et al \cite{LSST}}) that for fermion electrons one has that $\lim_{Z \to \infty} N_c(Z)/Z =1$ but Benguria-Lieb \cite{BengLieb} have proven that the $\liminf$ is strictly bigger than 1 for bosonic electrons.  There is considerable literature since these two basic papers, but since this is already removed from Kato's work, we won't try to summarize it.

%%%%%%%%%%%%%%%%%%%%%%%%%%%%%%%%%%%%%%%%%%%%%%%%%%%%%%%%%%%%%%
\section{Eigenvalues, II: Lack of Embedded Eigenvalues} \lb{s12}
%%%%%%%%%%%%%%%%%%%%%%%%%%%%%%%%%%%%%%%%%%%%%%%%%%%%%%%%%%%%%%

Consider on $\bbR^\nu$, the equation $(-\Delta+V)\varphi=\lambda\varphi$ with $V(x) \to 0$ as $|x| \to \infty$ and $\lambda > 0$.  Naively, one might expect that no solution, $\varphi$, can be in $L^2(\bbR^\nu, d^\nu x)$.  The intuition is clear: classically, if the particle is in the region $\{x\,|\,|x| > R\}$ where $R$ is picked so large that $|x| > R \Rightarrow V(x) < \lambda/2$ and if the velocity is pointing outwards, the particle is not captured and so not bound.  Due to tunnelling, in quantum theory, a particle will always reach this region so there shouldn't be positive energy bound states.  This intuition of no embedded eigenvalues is incomplete due to the fact that bumps can cause reflections even when the bumps are smaller than the energy, so an infinite number of small bumps which don't decay too rapidly might be able to trap a particle.  Indeed, in 1929, near the birth of modern quantum theory, von Neumann--Wigner \cite{vNW} presented an example with an embedded eigenvalue of energy 1 (in fact they picked $V(x) \to -1$ at infinity and $\lambda=0$; we'll shift energies by $1$ and also pick their arbitrary constant $A$ to be $1$).  They had the idea of guessing the wave function, $\psi$, and setting $V(x) = 1+\psi^{-1}\Delta\psi (x)$.  They picked $\psi$ so that it had oscillations that cancelled the $+1$ at infinity.  Their choice as a function of $r=|x|$ in three dimensions was
\begin{equation}\label{12.1}
  \psi(x) = \frac{\sin r}{r}[1+g(r)^2]^{-1}; \qquad g(r) = 2r-2\sin(2r)
\end{equation}
and they claimed that (where $\tilde{g}(r) = 2r + 2\sin(2r)$)
\begin{equation}\label{12.2}
  V(x) = -32 \cos^4r\frac{1-3\tilde{g}(r)^2}{[1+\tilde{g}(r)^2]^2}
\end{equation}

With slow enough decay, one can have much more than a single embedded eigenvalue.  It is known (see Simon \cite{SimonRandDecay} and Kotani-Ushiroya \cite{KURandDecay}) that if $0<\beta<1/2$ and $q_\omega(x)$ is a random potential in one dimension with uniformly spaced independent, identically distributed random bumps, then $-\tfrac{d^2}{dx^2}+(1+x^2)^{-\beta/2} q_\omega(x)$ has only dense pure point spectrum, i.e. the essential spectrum is $[0,\infty)$ and there is a complete orthonormal set of $L^2$ eigenvectors!

In 1959, Kato proved the first strong result on the non-existence of positive eigenvalues:

\begin{theorem} [Kato \cite{KPE}, announced in \cite{KPEAnon}] \lb{T12.1} Let $V(x)$ be continuous on $\bbR^\nu$ and obey
\begin{equation}\label{12.3}
  \lim_{r \to \infty} r \sup_{|y| > r} |V(y)| = 0
\end{equation}
Then $(-\Delta+V)\varphi=\lambda\varphi$ with $\lambda > 0$ has no (non--zero) $L^2$ solutions.
\end{theorem}

\begin{remarks}  1.  ODE techniques easily prove in one dimension and in arbitrary dimension if $V$ is spherically symmetric, that there are no positive eigenvalues if $\int_{1}^{\infty} |V(r)| \, dr < \infty$.  This goes back at least to Weyl \cite{WeylLP2} who quotes results of Kneser \cite{Kne}.  In modern parlance, it follows from the existence of Jost solutions.

2.  Earlier, Brownell \cite[Theorem 6.7]{BrownPE} proved the absence of such eigenvalues under bounds of the form $|V(x)| \le C_1 \exp(-C_2 |x|)$.

3.  There is both earlier and illuminating later work in the one dimensional (equivalently spherical symmetric) case.  Let
\begin{equation}\label{12.4}
  K \equiv \limsup_{|x| \to \infty} [|x||V(x)|]
\end{equation}
Kato proved in general dimension that there are no eigenvalues, $E$, with $E \ge K^2$.  The (corrected) Wigner--von Neumann example has $K=8, E=1$ so one knows from that one can't do better than $K^2/64$ and it is easy to modify this example to show one can't do better than $K^2/4$.  In 1948, Wallach \cite{Wallach} proved the $E \le K^2$ in one dimension (extended by Borg \cite{Borg} and Eastham \cite{EasthamWvN}) and provided an example showing one couldn't do better than $K^2/4$. A breakthrough in this one dimensional case was made by Atkinson--Everitt \cite{AE} who proved there is no eigenvalue if $E \ge 4K^2/\pi^2$ and that there are examples with eigenvalues arbitrarily close to this bound.  Note that $4/\pi^2 = .405...$ lies in $(1/4,1)$.  Their example is a relative of the Wigner--von Neumann example but uses $\textrm{sgn}(\sin(r))$ in place of $\sin(r)$.  Their method using Pr\"{u}fer transforms is very one dimensional.  Eastham--Kalf \cite{EastKalf} give a textbook presentation of this work and mention that Halvorsen (unpublished) also found the optimal $4K^2/\pi^2$.  Remling \cite{Rem} extended the Atkinson--Everitt result to prove no singular continuous spectrum in $[4K^2/\pi^2,\infty)$.

4.  Kato proved results about more than $L^2$ solutions.  For example, he proved that if $|V(x)| \le (1+|x|)^{-\alpha}$ near infinity with $\alpha>1$, and if $(-\Delta+V)\varphi=\lambda\varphi$ with $\lambda > 0$ with $\varphi(x) \to 0$ as $x \to \infty$, then $\varphi$ vanishes near infinity (and depending on the structure of the singularities of $V$, one can often use unique continuation (see below) to conclude that $\varphi \equiv 0$).  This will be useful in Section \ref{s15}.
\end{remarks}

The observant reader may have noticed that since $g(r)/r \to 1$ as $r \to \infty$, the potential, $V(x)$, given by \eqref{12.2} is $\textrm{O}(r^{-2})$ so it seems to be a counterexample to Theorem \ref{T12.1}!  In fact, von Neumann--Wigner had a calculational error: in the middle they used $\cos r/\sin r = \tan r$ (!) and this error produces a remarkable cancellation.  Doing the calculation correctly yields
\begin{equation}\label{12.5}
  V(r) = -32 \sin r \frac{g(r)^3\cos r -3 g(r)^2 \sin^3r +g(r) \cos r+\sin ^3 r}{[1+g(r)^2]^2}
\end{equation}
so that $V(r) = -8 \sin(2r)/r + \textrm{O}(r^{-2})$ consistent with Kato's theorem.  I once pointed out this error to Wigner, who thought for a moment and then said to me: ``Oh, Johnny did that calculation.''

Kato proved some differential inequalities on $M(r) = r^{\nu-1}\int |\varphi(r\omega)|^2 d\omega$ (where $d\omega$ is surface measure on the unit sphere) and used them to prove that if $\int^{\infty}  M(r) dr < \infty$ (i.e. $\varphi \in L^2(\bbR^\nu))$, then $M(r) = 0$ for $r > R_0$ for some $R_0$. The final step in his proof needs a result that any solution of $(-\Delta+W)\varphi=0$ that vanishes on an open set is identically zero.  This is called a unique continuation theorem (we note the analog fails for hyperbolic equations).  Such theorems go back to Carleman \cite{CarlUC} in 1939.  He only treated $\nu=2$ and required that $V \in L^\infty$.  The kind of estimates he used, now called Carleman estimates, have been a staple, not only of later work on unique continuation, but for many other topics in the theory of elliptic PDEs.  Unique continuation when $V \in L^\infty$ and $\nu \ge 3$ was proven by M\"{u}ller \cite{Mull} in 1954 (see also Aronszajn \cite{AronUC}).  So when Kato did his work, there was only unique continuation for bounded $V's$.  Thus, in the final step, one needs to know there is a compact set, $S$, of measure zero so that $\bbR^\nu\setminus S$ is connected and so that $V$ is locally bounded on this connected set.

Starting in 1980, there were a number of unique continuation results with $L^p_{loc}$ conditions on $V$ culminating in the classic 1985 paper of Jerison--Kenig \cite{JKUC} who require (for $\nu \ge 3$; for $\nu=2$, the condition is more complicated) that $V \in L^{\nu/2}_{loc}$ which is known to be optimal.

In fact, one only needs something weaker than unique continuation, namely that there are no eigenfunctions of compact support.  We will discuss this shortly.

Ikebe--Uchiyama \cite{IU} extended Kato's result to allow magnetic fields which are $\textrm{o}(x^{-1})$ at infinity and Roze \cite{Roze} allowed suitable non--constant coefficient second order elliptic term.

In \cite{FHHOHO1}, Froese et al. proved a variant of Kato's result.  They first proved that if $V$ is $-\Delta$--bounded and $(-\Delta+1)^{-1/2}(|x|V)(-\Delta+1)^{-1}$ is a compact operator, and if $(-\Delta+V)\varphi=\lambda\varphi,\,\varphi \in D(H)$ and $\lambda>0$, then $e^{\alpha|x|} \varphi \in L^2$ for all $\alpha > 0$.  They then prove (and this also shows no compact support eigenfunctions) that if $V(-\Delta+1)^{-3/4}$ is bounded, $\lim_{\gamma \to \infty},\, \norm{V(-\Delta+\gamma)^{-3/4}} = 0$ and ${\lim_{R \to \infty} \norm{\chi_R(1+|x|)V(-\Delta+1)^{-3/4}} =0}$ (where $\chi_R$ is the characteristic function of $\{x\,|\, |x|>R\})$, then  $(-\Delta+V)\varphi=\lambda\varphi$ and $e^{\alpha|x|}\varphi \in L^2$ for all $\alpha >0 \Rightarrow \varphi =0$.  This provides a proof of a variant of Kato's theorem without a need for pointwise bounds on $V$.

A very interesting alternate proof to a theorem very close to Kato is due to Vakulenko \cite{Vaku}.  While Vakulenko and Yafaev \cite{Yafaev2} (who has a clear exposition of Vakulenko's work) say that he recovers Kato's result, instead he has a condition for a class of $V$'s with lots of overlap to, but distinct from, Kato's condition \eqref{12.3}.  A {\emph{Vakulenko bounding function}, $\eta(r)$, is a function on $(0,\infty)$ obeying:
\begin{equation}\label{12.6}
  \forall_{r \in (0,\infty)} \eta(r) >0; \qquad \lim_{r \downarrow 0} r\eta(r) = 0; \qquad \int_{0}^{\infty} \eta(r) dr < \infty
\end{equation}
A \emph{Vakulenko potential}, $V(x)$, on $\bbR^\nu$ is a measurable function for which there exists a Vakulenko bounding function, $\eta(r)$, with
\begin{equation}\label{12.7}
  |V(x)| \le \eta(|x|)
\end{equation}

If $\eta(x) = (1+|x|)^{-1-\epsilon}$ and $V$ obeys \eqref{12.7}, then $V$ obeys both Vakulenko's condition and Kato's \eqref{12.3}.  If we consider $V(x) = (1+|x|)^{-1} [\log(2+|x|)]^{-\alpha}$, then $V$ obeys \eqref{12.3} if $\alpha > 0$ but is only a Vakulenko potential if $\alpha >1$.  On the other hand, if
\begin{equation}\label{12.8}
  V(x) = \left\{
           \begin{array}{ll}
             |x|^{-\beta}, & \hbox{ if for some $n = 1,2,\dots$ } n^2<|x|<n^2+1\\
             0, & \hbox{ otherwise}
           \end{array}
         \right.
\end{equation}
then $V(x)$ obeys Kato's \eqref{12.3} only if $\beta >1$ but is a Vakulenko potential if $\beta > 1/2$.  So neither class is contained in the other, although they are very close.  There is, of course, a connection to his condition and the fact that in one dimension, it has been long known that if the potential is in $L^1$, then the positive spectrum is purely absolutely continuous (as mentioned in Remark 1 after Theorem \ref{T12.1}).

\begin{theorem} [Vakulenko \cite{Vaku}] \lb{T12.2} Let $V(x)$ be a Vakulenko potential with \eqref{12.7} for some $\eta$.  Let $H = -\Delta+V$ and let $B$ be multiplication by $\sqrt{\eta}$.  Then for any $0 < a <b < \infty$, there is a relatively $H$--bounded operator, $A$, so that for all $\lambda \in [a,b]$ and all $\varphi \in D(H)$, we have that
\begin{equation}\label{12.9}
  \textrm{Re} \jap{(H-\lambda)\varphi,A\varphi} \ge \norm{B\varphi}^2
\end{equation}
\end{theorem}

In Section \ref{s15}, we'll see that \eqref{12.9} has implications for local smoothness of $B$ and implies strong spectral properties of $H$.  We'll also prove the theorem when $\nu=1$ and say something about the proof for general $\nu$.  For now, we note that

\begin{corollary} [Vakulenko \cite{Vaku}] \lb{C12.3} If $V$ is a Vakulenko potential and $H=-\Delta+V$, then $H$ has no positive eigenvalues.
\end{corollary}

\begin{proof} Let $\lambda > 0$.  Pick $a,b$ with $0<a<\lambda<b<\infty$.  If $H\varphi=\lambda\varphi$ for $\varphi \in D(H)$, by \eqref{12.9}, we have that $\norm{B\varphi}=0$.  Since $\eta$ is everywhere non--vanishing, we conclude that $\varphi=0$.
\end{proof}

The Wigner--von Neumann example has oscillations and one expects that if such oscillations are absent, then there should also be no positive eigenvalues.  For example, if $V(x)$ looks like $r^{-\alpha}, \, 0<\alpha \le 1$, one expects that there should also be no positive eigenvalues.  Odeh \cite{OdehNPE} proved that if $\boldsymbol{x}\cdot\boldsymbol{\nabla} V \le 0$ for all large $x$, then Kato's method could be modified to show there are no positive eigenvalues.  Shortly thereafter, Agmon \cite{AgmonNPE} and Simon \cite{SimonNPE}, using Kato's methods, independently proved (with enough local regularity to apply a unique continuation theorem) that there are no positive eigenvalues if $V(x) = V_1(x)+V_2(x)$ so long as when $x \to \infty$, one has that $|x||V_1(x)| \to 0$, $V_2(x) \to 0$ and $\boldsymbol{x}\cdot\boldsymbol{\nabla}V_2(x) \to 0$.  Most later works and, in particular, both Froese et al \cite{FHHOHO1} and Vakulenko \cite{Vaku}, also considered such sums.  Khosrovshahi--Levine--Payne \cite{KLPNPE} and Kalf--Krishna Kumar \cite{KalfK} allow a third highly oscillatory piece and prove no positive eigenvalues (so for example, they allow $r^{-1} \sin(r^\beta)$ for $\beta > 1$ and Agmon--Simon allow $\beta < 1$).

Another way of extending Odeh's result proves absence of positive eigenvalues using the virial theorem as discussed below (see also the discussion of Lavine's work in Section \ref{s15}).

Before discussing more results on the absence of positive energy eigenvalues, we pause for some other examples, motivated by the Wigner--von Neumann example, where there are positive energy eigenvalues.  By taking suitable sums of $b_j \sin(\alpha_j r)/r$ (cutoff away from infinity), Naboko \cite{NabokoWvN} and Simon \cite{SimonWvN} constructed, for each $\delta>0$, $V(x)$, bounded by $r^{-1+\delta}$ near infinity with dense point spectrum.  Here is one such result (taken from \cite{SimonWvN}):

\begin{theorem} \lb{T12.4} For any countable subset $\{E_k\}_{k=1}^\infty$ of $(0,\infty)$ and any $\epsilon,\delta > 0$, there is $V(x)$ on $(0,\infty)$ so that $-\tfrac{d^2}{dx^2}+V(x)$ on $L^2(0,\infty;dx)$ with $\varphi(0)=0$ boundary conditions has $\varphi_k \in L^2\cap C^2(0,\infty)$, so $\varphi_k(0)=0$ and $-\varphi_k''+V\varphi_k = E_k\varphi_k$ and so that
\begin{equation}\label{12.10}
  |V(x)| \le \epsilon(1+|x|)^{-1+\delta}
\end{equation}
\end{theorem}

\begin{remark} If $0 < \delta < 1/2$, it is known (\cite{CK, Rem, DeiftKillip, KS2}) that $-\tfrac{d^2}{dx^2}+V(x)$ has a.c. spectrum on all of $[0,\infty)$ so this is point spectrum embedded in continuous spectrum.  As noted already, if $\delta > 1/2$, one can find $V$'s with only point spectrum.
\end{remark}

The Wigner--von Neumann and Naboko--Simon examples are spherically symmetric.  Ionescu--Jerison \cite{IJ} found examples where the slow $\textrm{O}(r^{-1})$ decay is only in a parabolic tube about a single direction:

\begin{theorem} [Ionescu--Jerison \cite{IJ}] \lb{T12.5}   Fix $\nu \ge 2$.  There exists $C>0$ and for each $n=1,2,\dots$ a potential obeying
\begin{equation}\label{12.11}
  |V(x_1,\dots,x_\nu)| \le \frac{C}{n+|x_1|+|x_2|^2+\dots+|x_\nu|^2}
\end{equation}
and so that $(-\Delta+V)\varphi=\varphi$ has a non--zero $L^2$ solution.
\end{theorem}

Frank--Simon \cite{FrSimon} have simplified the Ionescu--Jerison construction by hewing more closely to the Wigner--von Neumann method.  They use the wave function
\begin{equation}\label{12.12}
  \varphi_n(x) = \sin x_1 (n^2+g(x_1)^2+(x_2^2+\dots+x_\nu^2)^2)^{-\alpha}
\end{equation}
where $\alpha > \nu/4$ (which implies that $\psi_n \in L^2$) and $g$ is given by \eqref{12.1}.  $V_n$ is then defined by
\begin{equation}\label{12.13}
  V_n(x) = \frac{\Delta\psi_n+\psi_n}{\psi_n}
\end{equation}
which is seen to obey \eqref{12.11}.  \cite{FrSimon} also has versions of the central Wigner--von Neumann potentials for dimensions different from $1$ and $3$.

Notice that \eqref{12.11} implies that $V_n \in L^p(\bbR^\nu)$ for any $p > \tfrac{1}{2}(\nu+1)$.  That says that the value of $p$ in the following is optimal:

\begin{theorem} [Koch--Tataru \cite{KTCMP}] \lb{TT12.6} Let $\nu \ge 2$.  If $V \in L^{p_1}(\bbR^\nu)+L^{p_2}(\bbR^\nu)$ where $p_1=\tfrac{1}{2}\nu< p_2=\tfrac{1}{2}(\nu+1)$ (if $\nu =2$, one needs to take $p_1>1$), then $-\Delta+V$ has no eigenvalues in $(0,\infty)$.
\end{theorem}

\begin{remarks} 1.  Earlier Ionescu--Jerison \cite{IJ} proved the weaker result where $p_2=\tfrac{1}{2}(\nu+1)$ is replaced by $p_2=\tfrac{1}{2}\nu$.

2.  As we noted above, by Theorem \ref{T12.5}, $p_2=\tfrac{1}{2}(\nu+1)$ is optimal.  The lower bound on $p$ is needed to assure esa--$\nu$.

3.  The proof relies on $L^p$ Carleman estimates and the machinery of \cite{KTCPAM}.
\end{remarks}

In many ways, the most subtle results on the absence of positive eigenvalues concern $N$--body systems.  After all, we saw in Sections \ref{s3} and \ref{s4} (Example 3.2 and Example 3.2 revisited) that $N$--body systems can have eigenvalues embedded in negative continua without carefully tuned potentials due to either non--interacting clusters or due to an eigenvalue of one symmetry embedded in a continuum of another symmetry.  The earliest $N$--body results involve the Virial Theorem and showed no positive eigenvalues under specialized circumstances, for example repulsive potentials and also $V$'s homogeneous of degree $\beta$ (i.e. $V(\lambda \overrightarrow{x}) = \lambda^\beta V(\overrightarrow{x}), \, 0>\beta > -2$) which includes the physically important Coulomb case.  This is discussed in Weidmann \cite{WeidVirial}, Albeverio \cite{AlbVirial} and Kalf \cite{KalfVirial} (or \cite[Theorems XIII.59 and XIII.60]{RS4}).

Undoubtedly, the deepest results on lack of positive eigenvalues for $N$--body systems are in Froese--Herbst \cite{FHPE}.  They assume that the $V_{ij}(r) = v_{ij}(r_i-r_j)$ where $v_{ij}$ as functions on $\bbR^\nu$ obey $v_{ij}(-\Delta+1)^{-1}$ and $(-\Delta+1)^{-1}(y\cdot\nabla_y v_{ij})(y) (-\Delta+1)^{-1}$ are compact (here $\Delta$ is the Laplacian and all operators act on $\bbR^\nu$). These hypotheses are made so that Mourre theory applies (see Mourre \cite{Mourre}, Perry--Sigal--Simon \cite{PSS}, Froese--Herbst \cite{FHM}, Amrein--Boutet de Monvel--Georgescu \cite{ABG} and Sahbani \cite{Sab}).

One takes $N$ particles $(x_1,\dots,x_N), x_j \in \bbR^\nu$ and defines
\begin{equation}\label{12.14}
  |x| = (2\sum_{j=1}^{N} m_j |x_j-R|^2)^{1/2}
\end{equation}
where $R = \left(\sum_{j=1}^{N} m_j\right)^{-1}\left(\sum_{j=1}^{N} m_j x_j\right)$. If we are looking at a Hamiltonian on $L^2(\bbR^{\nu(N-1)})$ with center of mass motion removed or if we have some $v_j$ representing interactions with infinite mass particles, then we act on $L^2(\bbR^{\nu N})$, and set $R=0$. What Froese---Herbst found is

\begin{theorem} [Froese--Herbst \cite{FHPE}] \lb{T12.7} Under the above hypotheses, if $H\psi = \lambda\psi,\,\psi \in L^2(\bbR^\kappa)$, then
\begin{equation}\label{12.15}
  \beta \equiv \sup_{\alpha \ge 0} \{\alpha^2+\lambda \,|\, e^{\alpha|x|}\psi \in L^2\} \in \calT \cup \{\infty\}
\end{equation}
where $\calT$ is the set of thresholds of the system (see Section \ref{s11} for a discussion of thresholds).
\end{theorem}

If there are no positive thresholds (which one can prove inductively if there is a way to prove no positive eigenvalues), then if $\lambda>0$, the $\beta$ in \eqref{12.15} must be $\infty$.  For suitable two body systems, we saw above that eigenfunctions can't obey $e^{\alpha|x|}\psi \in L^2$ for all $\alpha >0$.  Froese et al \cite{FHHOHO2} proved the same for suitable $N$--body systems (see the paper for precise conditions); see also \cite[Theorem C.3.8]{SimonSmgp}.  In this way, one proves certain $N$--body systems have no positive eigenvalues.

The above touched on $L^2$ isotropic exponential bounds (and as we'll see in Section \ref{s19}, that implies pointwise exponential bounds).  There is a huge and beautiful literature on this subject and on non--isotropic bounds.  We refer the reader to the book of Agmon \cite{AgmonBk} and the review article of Simon \cite{SimonSmgp} which contains many references.

%%%%%%%%%%%%%%%%%%%%%%%%%%%%%%%%%%%%%%%%%%%%%%%%%%%%%%%%%%%%%%
\section{Scattering and Spectral Theory, I: Trace Class Perturbations} \lb{s13}
%%%%%%%%%%%%%%%%%%%%%%%%%%%%%%%%%%%%%%%%%%%%%%%%%%%%%%%%%%%%%%

This is the first of four sections on spectral and scattering theory.  For the 15 years between 1957 and 1972, this area was a major focus of Kato. When Kato was invited to give a plenary lecture at the 1970 International Congress of Mathematicians, his talk \cite{KatoICM} was entitled ``Scattering Theory and Perturbation of Continuous Spectra'' (interestingly enough, Agmon and Kuroda gave invited talks at the same congress and spoke on closely related subjects).  This section and the next two have brief introductory remarks introducing this subject.  This section's introduction has much of the background we'll give on scattering theory, the next section discusses the basics of spectral theory and something about the connection between time--independent and time--dependent scattering theory and Section \ref{s15} will say more about the background behind the time--independent approach.

\begin{figure}[t]
\includegraphics[width=12cm]{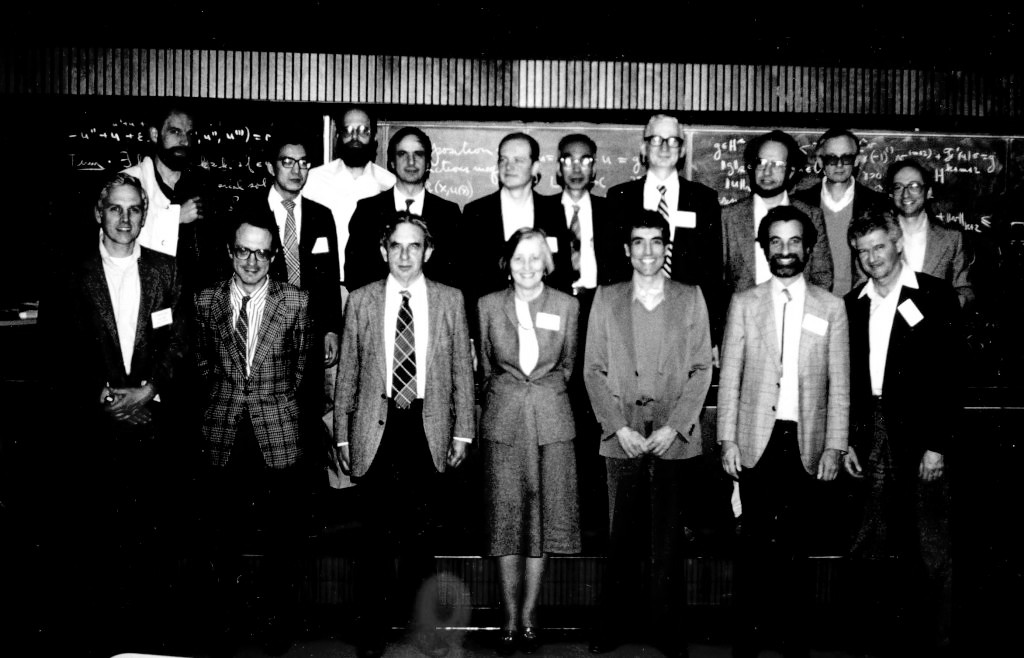}
\centering
\\Birmingham, AL, Meeting on Differential Equations,
1983. \\ Back row: Fr\"{o}hlich, Yajima, Simon, Temam, Enss,
Kato, Schechter, Brezis, Carroll, Rabinowitz. \\ Front row:
Crandall, Ekeland, Agmon, Morawetz, Smoller, Lieb,
Lax.  \lb{figure5}
\end{figure}

Starting with Rutherford's 1911 discovery of the atomic nucleus, scattering has been a central tool in fundamental physics, so it isn't surprising that one of the first papers in the new quantum theory was by Born \cite{Born} on scattering.  At its root, scattering is a time--dependent phenomenon: something comes in, interacts and moves off.  But since it relied on eigenfunctions, Born's work used time--independent objects.  He assumed that one could construct non--$L^2$ eigenfunctions, ${(-\Delta+V)\varphi=k^2\varphi,}\,(\overrightarrow{k} \in \bbR^3, k=|\overrightarrow{k}|)$ which as $r \to \infty$ looks like
\begin{equation}\label{13.1}
  \varphi(\overrightarrow{x}) \sim e^{i\overrightarrow{k}\cdot\overrightarrow{x}}+f(\theta)\frac{e^{ikr}}{r};\quad r=|\overrightarrow{x}|\quad \overrightarrow{k}\cdot\overrightarrow{x} = kr \cos(\theta)
\end{equation}

The time dependence gives $e^{-itH}\varphi(x)$ a $e^{i\overrightarrow{k}\cdot(\overrightarrow{x}-\overrightarrow{k}t)}$ term which is a usual plane wave with velocity $\overrightarrow{k}$ and a scattered wave $f(\theta)r^{-1}e^{ik(r-kt)}$.  One expects such a term to live near points where $r=kt$.  So if $t<0$ that term should not contribute (since $r>0$) while for $t$ positive and large we have an outgoing spherical wave representing the scattering.  We'll say a little more about making mathematical sense of this formal argument in Section \ref{s15}.  $|f(\theta)|^2$ was then interpreted as a scattering differential cross section.  Born also found a leading order perturbation formula for $f(\theta)$:
\begin{equation}\label{13.2}
  f(\theta) = -(2\pi) \int e^{i(\overrightarrow{k'}-\overrightarrow{k})\cdot\overrightarrow{x}}V(\overrightarrow{x})\,d\overrightarrow{x}
\end{equation}
where $k'=k$ and $\overrightarrow{k'}\cdot\overrightarrow{k}=k^2\cos\theta$.  This \emph{Born approximation} turns out to be leading order not only in $V$ but also, for $V$ fixed, as $k \to \infty$.

In the early 1940s, the theoretical physics community first considered time dependent approaches to scattering.  Wheeler \cite{Wheel} and Heisenberg \cite{HeisenS} defined the $S$--matrix and M\o ller \cite{Moller} introduced wave operators as limits (with no precision as to what kind of limit).

It was Friedrichs in a prescient 1948 paper \cite{FriedCont} who first considered the invariance of the absolutely continuous spectrum under sufficiently regular perturbations.  Friedrichs was Rellich's slightly older contemporary.  Both were students of Courant at G\"{o}ttingen in the late 1920s (in 1925 and 1929 respectively).  By 1948, Friedrichs was a professor at Courant's institute at NYU.  Friedrichs considered two classes of examples in this paper.  One was the model mentioned in Example \ref{E3.1} of a perturbation of an embedded point eigenvalue.  The other was $H=H_0+\lambda K$ where $H_0$ is multiplication by $x$ on $L^2([0,1],dx)$ and $K$ is a Hermitian integral operator with an integral kernel $K(x,y)$ assumed to vanish on the boundary (i.e. if $x$ or $y$ is $0$ or $1$) and to be H\"{o}lder continuous in $x$ and $y$.  Using what we'd call time--independent methods, Friedrichs constructed unitary operators, $U_\lambda$, for $\lambda$ sufficiently small, so that
\begin{equation}\label{13.3}
  H_0+\lambda K = U_\lambda H_0 U_\lambda^{-1}
\end{equation}

While Friedrichs neither quoted M\o ller nor ever wrote down the explicit formulae
\begin{equation}\label{13.4}
  \Omega^\pm(H,H_0) = {\textrm{s}-\lim}_{t \to \mp \infty} e^{itH}e^{-itH_0}
\end{equation}
(we remind the reader that the strange $\pm$ vs. $\mp$ convention that we use is universal in the theoretical physics community and uncommon among mathematicians and is not the convention that Kato used), he did prove something equivalent to showing that the limit $\Omega^+$ existed and was $U_\lambda$ and that the limit $\Omega^-$ existed and was equal to $S_\lambda\Omega^+$. Here $S_\lambda$ was an operator he constructed and identified with the $S$--matrix (although it differs slightly with what is currently called the $S$--matrix).

Motivated in part by Friedrichs, in 1957, Kato published two papers \cite{KatoFR, KatoTC} that set out the basics of the theory we will discuss in this section.  In the first, he had the important idea of defining
\begin{equation}\label{13.5}
  \Omega^\pm(A,B) = {\textrm{s}-\lim}_{t \to \mp \infty} e^{itA}e^{-itB} P_{ac}(B)
\end{equation}
where $P_{ac}(B)$ is the projection onto $\calH_{ac}(B)$, the set of all $\varphi \in \calH$ for which the spectral measure of $B$ and $\varphi$ is absolutely continuous with respect to Lebesgue measure (see \cite[Section 5.1]{OT} or the discussion at the start of Section \ref{s14}).  If these strong limits exist, we say that the wave operators $\Omega^\pm(A,B)$ exist.

By replacing $t$ by $t+s$, one sees that if $\Omega^\pm(A,B)$ exist then $e^{isA}\Omega^\pm=\Omega^\pm e^{isB}$.  Since $\Omega^\pm$ are unitary maps, $U^\pm$, of $\calH_{ac}(B)$ to their ranges, we see that $U^\pm B\restriction\calH_{ac}(B)(U^\pm)^{-1} = A \restriction \ran \Omega^\pm$.  In particular, $\ran \Omega^\pm$ are invariant subspaces for $A$ and lie in $\calH_{ac}(A)$.  It is thus natural to define: $\Omega^\pm(A,B)$ are said to be \emph{complete} if
\begin{equation}\label{13.6}
  \ran\, \Omega^+(A,B) = \ran\, \Omega^-(A,B) = \calH_{ac}(A)
\end{equation}
\begin{remarks} 1.  Kato also noted the relation
\begin{equation}\label{13.7}
  \Omega^\pm(A,B)\Omega^\pm(B,C)=\Omega^\pm(A,C)
\end{equation}
in that if both wave operators on the left exist, so does the one on the right and one has the equality.

2.  The wisdom of taking $P_{ac}(B)$ in the definition of wave operator is shown by the fact that it follows from results of Aronszajn \cite{AronAD} and Donoghue \cite{DonogAD} (see also Simon \cite{SimonTI}) that if $A-B=\jap{\varphi,\cdot}\varphi$ with $\varphi$ a cyclic vector for $B$ then $e^{itA}e^{-itB}\psi$ has a limit if and only if $\psi \in \calH_{ac}(B)$.
\end{remarks}

In \cite{KatoFR}, Kato proved the following

\begin{theorem} [Kato \cite{KatoFR}] \lb{T13.1} Let $\Omega^\pm(A,B)$ exist.  Then they are complete if and only if $\Omega^\pm(B,A)$ exist.
\end{theorem}

The proof is almost trivial.  It depends on noting that
\begin{equation}\label{13.8}
  \psi = \lim_{t \to \infty} e^{iAt}e^{-itB}\varphi \iff \varphi = \lim_{t \to \infty} e^{itB}e^{-itA}\psi
\end{equation}
since
\begin{equation}\label{13.9}
  \norm{\psi-e^{iAt}e^{-itB}\varphi} = \norm{e^{itB}e^{-itA}\psi-\varphi}
\end{equation}
That said, it is a critical realization because it reduces a completeness result to an existence theorem.  In particular, it implies that symmetric conditions which imply existence also imply completeness.  We'll say more about this below.

To show the importance of this idea, motivated by it in \cite{DeiftSimonWO}, Deift and Simon proved that completeness of multichannel scattering for $N$--body scattering was equivalent to the existence (using the $N$--body language of Section \ref{s11}) of $\textrm{s}-\lim_{t \to \pm\infty} e^{itH}(\calC)J_\calC e^{-itH}P_{ac}(H)$ for the partition of unity $\{J_\calC\}_{\calC \ne \calC_{min}}$ discussed in Remark 4 after Theorem \ref{T11.8}.  All proofs of asymptotic completeness for $N$--body systems prove it by showing the existence of these Deift--Simon wave operators in support of Kato's Theorem \ref{T13.1}.

In \cite{KatoFR}, Kato proved

\begin{theorem} [Kato \cite{KatoFR}] \lb{T13.2} Let $H_0$ be a self--adjoint operator and $V$ a (bounded) finite rank operator.  Then $H=H_0+V$ is a self--adjoint operator and the wave operators $\Omega^\pm(H,H_0)$ exist and are complete.
\end{theorem}

This implies the unitary equivalence of $H_0\restriction\calH_{ac}(H_0)$ and $H\restriction\calH_{ac}(H)$.  Remarkably, in the same year Aronszajn \cite{AronAD} proved that this invariance holds for finite rank perturbations of boundary conditions for Sturm--Liouville operators (extended later using similar ideas by Donoghue \cite{DonogAD} to general finite rank perturbations).  Their methods are totally different from Kato's and do not involve wave operators.

Later in 1957, Kato \cite{KatoTC} proved

\begin{theorem} [Kato--Rosenblum Theorem] \lb{T13.3} The conclusions of Theorem \ref{T13.2} remain true if $V$ is a (bounded) trace class operator.
\end{theorem}

In a sense this theorem is optimal.  It is a result of Weyl--von Neumann \cite{WeylvN, vNWeylvN} (see \cite[Theorem 5.9.2]{OT}) that if $A$ is a self--adjoint operator, one can find a Hilbert--Schmidt operator, C, so that $B=A+C$ has only pure point spectrum.  Kato's student, S. T. Kuroda \cite{KurodaWvN}, shortly after Kato proved Theorem \ref{T13.3}, extended this result of Weyl--von Neumann to any trace ideal strictly bigger than trace class.  So within trace ideal perturbations, one cannot do better than Theorem \ref{T13.3}.

The name given to this theorem comes from the fact that before Kato proved Theorem \ref{T13.3}, Rosenblum \cite{RosenTrace} proved a special case that motivated Kato: namely, if $A$ and $B$ have purely a.c. spectrum and $A-B$ is trace class, then $\Omega^\pm(A,B)$ exist and are unitary (so complete).

I'd always assumed that Rosenblum's paper was a rapid reaction to Kato's finite rank paper which, in turn, motivated Kato's trace class paper.  But I recently learned that this assumption is not correct.   Rosenblum was a graduate student of Wolf at Berkeley who submitted his thesis in March 1955.  It contained his trace class result with some additional technical hypotheses; a Dec. 1955 Berkeley technical report had the result as eventually published without the extra technical assumption.  Rosenblum submitted a paper to the American Journal of Mathematics which took a long time refereeing it before rejecting it.   In April 1956, Rosenblum submitted a revised paper to the Pacific Journal in which it eventually appeared (this version dropped the technical condition; I've no idea what the original journal submission had).

Kato's finite rank paper was submitted to J. Math. Soc. Japan on March 15, 1957 and was published in the issue dated April, 1957(!).  The full trace class result was submitted to Proc. Japan Acad. on May 15, 1957.  Kato's first paper quotes an abstract of a talk Rosenblum gave to an A.M.S. meeting but I don't think that abstract contained many details.   This finite rank paper has a note added in proof thanking Rosenblum for sending the technical report to Kato, quoting its main result and saying that Kato had found the full trace class results (``Details will be published elsewhere.'').   That second paper used some technical ideas from Rosenblum's paper.

I've heard that Rosenblum always felt that he'd not received sufficient credit for his trace class paper.  There is some justice to this.    The realization that trace class is the natural class is important.    As I've discussed, trace class is maximal in a certain sense.  Kato was at Berkeley in 1954 when Rosenblum was a student (albeit some time before his thesis was completed)  and Kato was in contact with Wolf.  However, there is no indication that Kato knew anything about Rosenblum's work until shortly before he wrote up his finite rank paper when he became aware of Rosenblum's abstract.   My surmise is that both, motivated by Friedrichs, independently became interested in scattering.

It should be emphasized that 1956-1957 was a year that (time--dependent) scattering theory seemed to be in the air.  J. Cook \cite{Cook} found a simple, later often used, method for proving that $\Omega^\pm(A,B)$ exists: if $\int_{-\infty}^{\infty} \norm{(A-B)e^{-iuB}\varphi}\, du < \infty$, then by integrating a derivative
\begin{equation}\label{13.10}
  \limsup_{\substack{t,s \to \infty\\ \textrm{or }t,s \to -\infty}} \norm{e^{itA}e^{-itB}\varphi-{e^{isA}e^{-isB}\varphi}} \le \lim \int_{s}^{t} \norm{(A-B)e^{-iuB}\varphi}\, du = 0
\end{equation}
so it suffices that
\begin{equation}\label{13.11}
  \int_{-\infty}^{\infty} \norm{(A-B)e^{-iuB}\varphi}\, du < \infty
\end{equation}
for a dense set of $\varphi$ for $\Omega^\pm(A,B)$ to exist.  Cook applied this to $B=-\Delta;\,A=-\Delta+V;\, V \in L^2(\bbR^3)$ (which translates to $\textrm{O}(|x|^{-3/2-\epsilon})$ decay).  Hack \cite{HackNC1} and Kuroda \cite{KurodaNC} extended this to allow $\textrm{O}(|x|^{-1-\epsilon})$ decay.

Since, for the free dynamics, $x \sim ct$, one expects and can prove that if $\alpha \le 1$, then $\int_{-\infty}^{\infty} \norm{(1+|x|)^{-\alpha}
e^{iu\Delta}\varphi}\, du = \infty$ for all $\varphi$.  Indeed, Dollard \cite{Dollard} showed that one needs modified wave operators for Coulomb potentials (again, there is a large literature on the subject of Coulomb or slower decay of which we mention Christ--Kiselev \cite{CK} and J.~Derezi\'{n}ski and C.~G\'{e}rard \cite{DerGer}).

Extensions of Cook's ideas and other scattering theory notions to quadratic form perturbations can be found in Kuroda \cite{KurodaJA2}, Schechter \cite{SchechQFScatt}, Simon \cite{SimonQFScatt} and Kato \cite{KatoQFScatt}.  Kato states his results in a two Hilbert space setting (see below).  $J$ is a bounded linear operator from $\calH_1$ to $\calH_2$ and $H_j$ are self--adjoint operators on $\calH_j; \, j=1,2$.  For $z \in \bbC\setminus\bbR$, let $C(z) = (H_2-z)^{-1}J-J(H_1-z)^{-1}$.  Kato proves that if for some $z$ and $\varphi \in \calH_1$, one has that
\begin{equation}\label{13.12}
  \int_{0}^{\infty} \norm{C(z)e^{-itH_1}\varphi}_2\, dt < \infty
\end{equation}
then
\begin{equation}\label{13.12A}
  \lim_{t \to \infty} e^{itH_2}Je^{-itH_1}\varphi \textrm{ exists}
\end{equation}
He then shows that this allows some cases where $H_2$ is only defined as a quadratic form, e.g. $H_1=-\Delta, H_2=-\Delta+V$ with $V \ge 0,\,V \in L^1(\bbR^3,(1+|x|)^{1-\epsilon}\, dx)$.

For many years, it was thought that this simple idea of Cook was limited to existence but not useful for completeness or spectral theory.  This was overturned by a brilliant paper of Enss \cite{Enss} (see also Perry \cite{PerryEnss}, Reed--Simon \cite[Section XI.17]{RS3} or Simon \cite{SimonEnss}), a subject we will not pursue here.

In 1958-59, there were also several influential papers by Jauch \cite{Jauch, JauchZinnes} that discussed scattering in a general framework.

In considering extensions of the Kato--Rosenblum, I begin with four issues that involve work by Kato himself.  First, we discuss proofs.  Like Friedrichs, both Kato and Rosenblum proved a time-dependent limit exists by first constructing objects with time--independent methods which they prove is the required limit.  The first fully time--dependent proof of Theorem \ref{T13.3} is in a Japanese language paper by Kato \cite{KatoTraceJap}  also published in 1957.  His argument was repeated with permission in a paper by his student S. T. Kuroda \cite{KurodaJA1}.  The slickest version of this time--dependent proof is in Kato's 1966 book \cite{KatoBk}.  It is a variant of this argument that Pearson used in his proof of Theorem \ref{T13.4} below.

The second concerns Kato's paper \cite{KatoInv} on what is called the invariance principle: for suitable functions $\Phi$, one shows that $A-B$ trace class $\Rightarrow \Omega^\pm(\Phi(A),\Phi(B))$ exist and are complete.  In case that $\Phi$ is strictly monotone increasing (respectively decreasing), one has that $\Omega^\pm(\Phi(A),\Phi(B))=\Omega^\pm(A,B)$ (resp $\Omega^\mp(A,B)$).  The first examples of this phenomenon are due to Birman \cite{BirmanRes1, BirmanRes2}.  Kato focused on the general form of the principle.  There is a considerable literature on non--trace class versions of an invariance principle; see \cite[Notes to Section XI.3]{RS3} for references.

The third involves two Hilbert space scattering theory \cite{KatoTwoHS}.  This came out of a set of concrete problems.  In Section \ref{s8} (see the discussion beginning with \eqref{8.5}), we saw that the equation $\tfrac{\partial^2 u}{\partial t^2} = (\Delta-V)u$ had a unitary propagation in the norm $\left[\norm{\dot{u}}_2^2+\jap{u,(-\Delta+V)u}\right]^{1/2}$.  This means to compare solutions of this equation to, say, the one with $V=0$, one needs to consider two different Hilbert space norms.  If for some $0 < \alpha <\beta < \infty$ one has for all $x$ that $\alpha \le V(x) \le \beta$, then there is a natural map, $J$ between the two spaces so that $J$ is bounded with bounded inverse which takes $\varphi$ viewed as an element of one Hilbert space into itself but viewed in the other Hilbert space.  One is interested in the limit in \eqref{13.12A} (and also the limit as $t \to -\infty$).  A similar setup applies to other hyperbolic systems, especially to the physically significant Maxwell's equation. Long after Kato's work on the subject, Isozaki--Kitada \cite{IsoKit} discovered one could use a $J$ operator to discuss long range scattering where ordinary wave operators do not exist. Before \cite{KatoTwoHS}, several authors (Schmidt \cite{Schmidt}, Shenk \cite{Shenk}, Thoe \cite{Thoe}, Wilcox \cite{Wilcox}) discussed scattering theory for some concrete examples of such systems.  Kato \cite{KatoTwoHS} looked at the theory systematically, focusing, for example, on $J$'s with $\textrm{s--}\lim_{t \to \pm\infty} (J^*J-\bdone)e^{-itH_1}=0$ which implies that the wave operators are isometries if they exist.  Under certain invertibility hypotheses on $J$, Kato could carry over the usual trace class scattering theory to get some two Hilbert space results.  Stronger results were subsequently obtained by Belopol'skii--Birman \cite{BBelop}, Birman \cite{BirmanTwoHS} and then Pearson \cite{PearsonTC} who proved

\begin{theorem} [Pearson's Theorem \cite{PearsonTC}] \lb{T13.4} Let $A, B$ be self--adjoint operators on Hilbert spaces $\calH_1$ and $\calH_2$.  Let $J$ be a bounded operator from $\calH_1$ to $\calH_2$ so that $C=AJ-JB$ is trace class (in the sense that there is a bounded operator $C$ from $\calH_1$ to $\calH_2$ with $\sqrt{C^*C}$ trace class and for $\varphi \in D(B)$ and $\psi \in D(A)$ we have that $\jap{A\psi,J\varphi}-\jap{\psi,JB\varphi}=\jap{\psi,C\varphi}$).  Then
\begin{equation}\label{13.13}
   \Omega^\pm(A,B;J)={\textrm{s--}\lim}_{t \to \mp\infty} e^{itA}Je^{-itB}P_{ac}(B)
\end{equation}
exists
\end{theorem}

No completeness is claimed (e.g., consider $J=0$) but one can sometimes get completeness.  For example, if $\calH_1=\calH_2=\calH$ and $A, B \ge 0$ are two positive operators on $\calH$ so that $(A+1)^{-1}-(B+1)^{-1}$ is trace class, then one can pick $J=(A+1)^{-1}(B+1)^{-1}$.  $C$ is trace class, so $\Omega^\pm(A,B;J)$ exist.  Apply this to $(B+1)\varphi$ to see that $\Omega^\pm(A,B;(A+1)^{-1})$ exists.  Since $(A+1)^{-1}-(B+1)^{-1}$ is compact, the Riemann--Lebesgue lemma shows that $\Omega^\pm(A,B;(A+1)^{-1}-(B+1)^{-1})=0$.  It follows that $\Omega^\pm(A,B;(B+1)^{-1})$ exists.  Applying this to $(B+1)\varphi$, we see that $\Omega^\pm(A,B)$ exists.  By symmetry, it is complete.  We thus recover Birman's result (see below) that $(A+1)^{-1}-(B+1)^{-1}$ trace class implies that $\Omega^\pm(A,B)$ exists and is complete.  Pearson's proof is a clever variant of Kato's time--dependent proof from \cite{KatoBk}; see \cite[pp 33-38]{RS3} for details and further applications.

\begin{example} \lb{E13.5}  The fourth of Kato's applications/extensions of the trace class theory is an example in a joint paper with Kuroda \cite{KKCheat}.  They consider three Hamiltonians on $L^2(\bbR^2,d^2 x)$:
\begin{equation}\label{13.14}
  H_0=-\frac{\partial^2}{\partial x_1^2}-\frac{\partial^2}{\partial x_2^2}; \quad H_1=H_0+V(x_2);\quad H=H_1+K
\end{equation}
where $V \in L^1(\bbR)\cap L^2(\bbR)$ and $K$ is a rank 1 operator, $Ku = c\jap{\varphi,u} \varphi$ with $\varphi$ a norm 1 function in $L^2(\bbR^2)$ and $c$ is a constant.  Moreover, they pick $V$ so that $h_1 = -\frac{d^2}{dx^2}+V(x)$, as an operator on $L^2(\bbR)$, has exactly one eigenvalue in $(-\infty,0]$.

Let $h_0=-\frac{d^2}{dx^2}$.  By results of Kuroda \cite{KurodaNC}, using the trace class theory, $\Omega^\pm(h_1,h_0)$ exist and are complete.  Since $H_0, H_1$ are of the form $H_j=\bdone\otimes h_j+h_0\otimes\bdone$, one sees that $\Omega^\pm(H_1,H_0)$ exist with $\ran\,\Omega^+(H_1,H_0)=\ran\,\Omega^-(H_1,H_0)$.  But they are not complete because $\calH_{ac}(H_1)$ has vectors of the form $\psi\otimes\varphi_0$ where $\psi \in L^2(\bbR)$ and $\varphi_0$ is the bound state of $h_1$.

Since $K$ is rank 1, $\Omega^\pm(H,H_1)$ exist and so by the chain rule $\Omega^\pm(H,H_0)$ exist.  But by a calculation, $K$ links the two parts of the a.c. spectrum of $H_1$, at least for $c$ small.  Thus they claim that, for $c$ small, $\ran\,\Omega^+(H,H_0) \ne\ran\,\Omega^-(H,H_0)$ and the S--matrix is non--unitary.  Hence the title of their paper ``A Remark on the Unitarity Property of the Scattering Operator''.

However, as Kuroda \cite{KurodaCheat} subsequently noted, this analysis leaves something out.  The S--matrix is unitary if one looks at the right S--matrix!  This is a multichannel system and if one includes also the channel for $\{\psi\otimes\varphi_0\}$, the arguments do imply unitarity.  So rather than find a non--unitary S--matrix, they found the first example of a multichannel scattering system with asymptotic completeness!
\end{example}

We conclude this section with some brief remarks on developments in the trace class scattering theory subsequent to Kato's original work.  Many of the significant results are due to M. S. Birman, so much so that the theory has taken the name Kato--Birman theory.

(1) A first key issue was making the theory apply to Schr\"{o}dinger operators, $H_0=-\Delta, H=-\Delta+V$ on $L^2(\bbR^\nu)$.  The pioneer was Kato's student, Kuroda, who first proved an extension of the Kato--Rosenblum theorem.  If $V$ is $H_0$--bounded with relative bound less than 1 and $|V|^{1/2}(H_0+1)^{-1}$ is Hilbert--Schmidt, then Kuroda proved that $\Omega^\pm(H,H_0)$ exist and are complete.  He used this to prove existence and completeness if $\nu \le 3$ and $V \in L^1(\bbR^\nu)\cap L^2(\bbR^\nu)$.  In terms of V's with
\begin{equation}\label{13.15}
  |V(x)| \le C(1+|x|)^{-\alpha}
\end{equation}
this requires $\alpha>\nu$ whereas existence by Cook's method only needs $\alpha>1$, so for $\nu \ge 2$, there is a gap that we'll discuss much more in the next two sections. Kuroda also noted that if $V(\overrightarrow{x})=V(|\overrightarrow{x}|)$ is a central potential, then, for any $\nu$ one can do a partial wave expansion (see \cite[Theorem 3.5.8]{HA}) and reduce the problem to half--line problems.  Since it is known that when \eqref{13.15} holds for any $\alpha>0$, that the essential spectrum for the half--line problem is $[0,\infty)$ and the spectrum is simple, one can see that existence implies completeness without needing the trace class theory.

(2) Birman is responsible for a wide variety of extensions and applications of the trace class theory.  First, he proved with Krein \cite{BKrein} an extension to the situation where $U$ and $V$ are two unitaries for which $V-U$ is trace class.  In that case, $\textrm{s--}\lim_{n \to \pm\infty}(V^*)^nU^n P_{ac}(U)$ exists, has range $\ran\,P_{ac}(V)$ and is a unitary equivalence of the a.c. parts of $U$ and $V$.  Secondly \cite{BirmanRes1, BirmanRes2}, he proved that if $A, B$ are self--adjoint and $(A-z)^{-1}-(B-z)^{-1}$ is trace class for some $z \notin \sigma(A)\cup\sigma(B)$, then $\Omega^\pm(A,B)$ exist and are complete (deBranges \cite{deBTC} proved the same result).  Kuroda's result on $|V|^{1/2}(H_0+1)^{-1}$ Hilbert Schmidt follows from this.  Later Birman \cite{BirmanLocal} proved that if $P_I(A)(A-B)P_I(B)$ is trace class for all bounded intervals, $I$, and if a technical condition called mutual subordinancy holds, then $\Omega^\pm(A,B)$ exist and are complete.  His proof was involved but using Pearson's Theorem (Theorem \ref{T13.4}), one can easily prove this result of  Birman (see \cite[Theorem XI.10]{RS3}).  With this result, one can prove existence and completeness of $\Omega^\pm(H,H_0)$ for $H_0=-\Delta,\,H=-\Delta+V$ on $L^2(\bbR^\nu)$ if $V \in L^{\nu/2}(\bbR^\nu)\cap L^1(\bbR^\nu)$, so $\alpha>\nu$ in \eqref{13.15} leaving quite a gap from the expected $\alpha>1$ (see the next two sections).

(3) One can apply the trace class theory to changes of boundary condition.  The pioneer here is Birman \cite{BirmanBC, BirmanBCFull}; see also Deift--Simon \cite[Appendix]{DeiftSimonTC}.

(4) When $A$ and $B$ are bounded and $A-B$ is trace class. one can define an $L^1(\bbR,dx)$ function, $\xi(x)$, called the Krein spectral shift so that for $f$ a $C^2$ function of compact support, one has that $f(A)-f(B)$ is trace class and
\begin{equation}\label{13.16}
  \tr(f(A)-f(B)) = -\int f'(x)\xi(x)\,dx
\end{equation}
(see Simon \cite[Section 11.4]{SimonTI} or Yafaev \cite[Chap. 8]{Yafaev1} for more on the spectral shift function).  Birman--Krein \cite{BKrein} prove the beautiful \emph{Birman--Krein formula}:
\begin{equation}\label{13.17}
  \det(S(\lambda))=e^{-2\pi i\xi(\lambda)}
\end{equation}
when $A-B$ is trace class.  Here $S=\Omega^-(A,B)^*\Omega^+(A,B)$ is a unitary operator on $\calH_{ac}(B)$ which commutes with $B$, so according to the spectral multiplicity theory (\cite[Section 5.4]{OT}), $B$ has a direct integral decomposition $\calH_{ac}(B) = \int^\oplus_{\sigma_{ac}(B)} \calH_\lambda\,d\lambda, \, B=\int^\oplus_{\sigma_{ac}(B)}$ and $S=\int_{\sigma_{ac}(B)}^{\oplus} S(\lambda)\,d\lambda$ where $S(\lambda)$ is a unitary operator on $\calH_\lambda$.  Birman--Krein prove that $S(\lambda)-\bdone$ is a trace class operator on $\calH_\lambda$ and \eqref{13.17} holds where $\det$ is the Fredholm determinant (\cite[Section 3.10]{OT}).

%%%%%%%%%%%%%%%%%%%%%%%%%%%%%%%%%%%%%%%%%%%%%%%%%%%%%%%%%%%%%%
\section{Scattering and Spectral Theory, II: Kato Smoothness} \lb{s14}
%%%%%%%%%%%%%%%%%%%%%%%%%%%%%%%%%%%%%%%%%%%%%%%%%%%%%%%%%%%%%%

This is the second section on spectral and scattering theory.  We begin with a quick primer on spectral theory that will assume familiarity with the spectral theorem and spectral measures (see \cite[Sections 5.1 and 7.2]{OT}).  For a self--adjoint operator, $H$, on a (complex, separable) Hilbert space, $\calH$, the most basic questions are connected to the Lebesgue decomposition theorem (\cite[Theorem 4.7.3]{RA}) that says that any measure, $d\mu$ on $\bbR$ can be uniquely decomposed $d\mu=d\mu_{ac}+d\mu_{sc}+d\mu_{pp}$ where $d\mu_{pp}$ is pure point, $d\mu_{ac}$ is $dx$--absolutely continuous and $d\mu_{sc}$ has no pure points and is singular with respect to $dx$ (so ``singular continuous'').  There is a corresponding decomposition $\calH = \calH_{ac}(H)\oplus\calH_{sc}(H)\oplus\calH_{pp}(H)$ where $\calH_{y}$ is the set of those vectors, $\varphi$, whose $H$--spectral measure is purely of type $y$.

In simple quantum mechanical systems, $\calH_{ac}$ spectrum is often associated with scattering theory as we've seen, and $\calH_{pp}$ is associated with bound states.  As my advisor, Arthur Wightman, told me there is no reasonable interpretation for states in $\calH_{sc}$ so he called the idea that $\calH_{sc}=\{0\}$ the ``no goo hypothesis''.  A major concern of quantum theoretic spectral theorists in the period from 1960 to 1985, and, in particular, of Kato, was the proof that  $\calH_{sc}=\{0\}$ for two--( and N--)body quantum systems whose potentials obey \eqref{13.15} for $\alpha > 1$.

Ironically, after Kato became less active in NRQM, it was discovered that, in some ways, singular continuous spectrum is ubiquitous.  As I've remarked: ``I seem to have spent the first part of my career proving that singular continuous spectrum never occurs and the second proving that it always does''.  A key breakthrough was the discovery by Pearson \cite{PearsonSC} that sparse potentials with slow decay have purely s.c. spectrum.  I explored this in a series of papers \cite{SimonSC1, SimonSC2, SimonSC3, SimonSC4, SimonSC5, SimonSC6, SimonSC7} of which a typical result concerns $h$ on $\ell^2(\bbZ)$ given by $(hu)_n=u_{n+1}+u_{n-1}+b_n u_n$.  Fix $\alpha > 0$ and let $Q_\alpha$ be the Banach space of $b's$ with $\sup_n\left[(1+|n|)^\alpha|b_n|\right]\equiv\norm{b}_\alpha < \infty$ with $|n|^\alpha |b_n| \to 0$ as $|n| \to \infty$.  Then (see \cite{SimonSC1}), if $\alpha < 1/2$, for a dense $G_\delta$ in $Q_\alpha$, the associated $h$ has purely s.c. spectrum (i.e. $\calH_{sc}(h)=\calH$).

A main tool in the quest to prove that $\calH_{sc}=\{0\}$ is the fact that Stone's formula \cite[Eqn (5.7.30)]{OT}
\begin{equation}\label{14.1}
  \lim_{\epsilon \downarrow 0} \int_{a}^{b} \textrm{Im}\jap{\varphi,R(x+i\epsilon)\varphi}\,dx = \jap{\varphi,\tfrac{1}{2}\left[P_{(a,b)}(H)+P_{[a,b]}(H)\right]\varphi}
\end{equation}
(where $R(z) = (H-z)^{-1}$ for $z \in \bbC\setminus\bbR$) immediately implies that for any $p >1$, we have that
\begin{equation}\label{14.2}
  \sup_{0 < \epsilon < 1} \int_{a}^{b} |\textrm{Im}\jap{\varphi,R(x+i\epsilon)\varphi}|^p\,dx < \infty \Rightarrow P_{(a,b)}(H)\varphi \in \calH_{ac}(H)
\end{equation}
Thus, the most common way of proving that $\calH_{sing}=\{0\}$ is showing that for a dense set of $\varphi$, and enough intervals $(a,b)$, we have that
\begin{equation*}
  \sup_{\substack{\epsilon>0 \\ a<x<b}} |\jap{\varphi,R(x+i\epsilon)\varphi}| < \infty
\end{equation*}
(stronger than needed, but what one often gets).

We'll say a lot more about time--independent scattering in the next section, but we note that in some sense, the key notion of that theory is that control of $\jap{\varphi,R(x+i\epsilon)\varphi}$ as $\epsilon \downarrow 0$ also says something about long time behavior of dynamics as seen in
\begin{equation}\label{14.3}
    \int_{0}^{\infty} e^{-\epsilon t}e^{it\lambda}e^{-itH}\varphi \, dt = -iR(\lambda+i\epsilon)\varphi
\end{equation}
for any $\varphi \in \calH$ because $\int e^{-\epsilon t}e^{i(\lambda-x)t}\,dt=-i(x-\lambda-i\epsilon)^{-1}$.

We turn now to the theory of Kato smoothness which is based primarily on two papers of Kato \cite{KatoSm1, KatoSm2}.  The first is the basic one with four important results: the equivalence of many conditions giving the definition, the connection to spectral analysis, the implications for existence and completeness of wave operators and, finally, a perturbation result.  The second paper concerns the Putnam--Kato theorem on positive commutators.

To me, the 1951 self--adjointness paper is Kato's most significant work (with the adiabatic theorem paper a close second), Kato's inequality his deepest and the subject of this section his most beautiful.  One of the things that is so beautiful is that there isn't just a relation between the time--independent and  time--dependent objects -- there is an equivalence!  Here is the set of equivalent definitions:

\begin{theorem} [Kato \cite{KatoSm1}] \lb{T14.1}  Let $H$ be a self--adjoint operator and $A$ a closed operator.  The following are all equal ($R(\mu)=(H-\mu)^{-1}$):
\begin{equation}\label{14.4}
  \sup_{\substack{\norm{\varphi}=1\\ \epsilon >0}} \frac{1}{4\pi^2} \int_{-\infty}^{\infty} \left(\norm{AR(\lambda+i\epsilon)\varphi}^2+\norm{AR(\lambda-i\epsilon)\varphi}^2\right)\,d\lambda
\end{equation}
\begin{equation}\label{14.5}
  \sup_{\norm{\varphi}=1} \frac{1}{2\pi} \int_{-\infty}^{\infty} \norm{Ae^{-itH}\varphi}^2\,dt
\end{equation}
\begin{equation}\label{14.6}
  \sup_{\substack{\norm{\varphi}=1,\,\varphi \in D(A^*)\\ -\infty <a < b < \infty}} \frac{\norm{P_{(a,b)}(H)A^*\varphi}^2}{b-a}
\end{equation}
\begin{equation}\label{14.7}
  \sup_{\substack{\mu \notin \bbR,\,\varphi \in D(A^*) \\ \norm{\varphi}=1}} \frac{1}{2\pi} |\jap{A^*\varphi,[R(\mu)-R(\bar{\mu})]\varphi}|
\end{equation}
\begin{equation}\label{14.8}
  \sup_{\substack{\mu \notin \bbR,\,\varphi \in D(A^*) \\ \norm{\varphi}=1}} \frac{1}{\pi} \norm{R(\mu)A^*\varphi}^2 \, |\textrm{Im}\mu|
\end{equation}
In particular, if one is finite (resp. infinite), then all are.
\end{theorem}

\begin{remarks}  1.  In \eqref{14.4}/\eqref{14.5}, we set $\norm{A\psi}=\infty$ if $\psi \notin D(A)$, so, for example, to say that \eqref{14.5} is finite implies that for each $\varphi$, we have that $e^{-itH}\varphi \in D(A)$ for Lebesgue a.e. $t \in \bbR$.

2. If one and so all of the above quantities are finite we say that $A$ is \emph{$H$--smooth}.  The common value of these quantities is called $\norm{A}_H^2$.

3. The proof is not hard.  If the integral in \eqref{14.5} has a factor of $e^{-2\epsilon t}$ put inside it, the equality of the integrals in \eqref{14.4} and \eqref{14.5} follows from \eqref{14.3} and the Plancherel theorem.  By monotone convergence, the $\sup$ of the time integral with the $e^{-2\epsilon t}$ factor is the integral without that factor.

4.  The equivalence of \eqref{14.7} and \eqref{14.8} is just $R(\mu)-R(\bar{\mu})={(\mu-\bar{\mu})R(\mu)R(\bar{\mu})}$.

5.  If $d\nu_{A^*\varphi}$ is the $H$--spectral measure for $A^*\varphi$ (so $\int f(\lambda)d\nu_{A^*\varphi}(\lambda) = \jap{A^*\varphi,f(H)A^*\varphi}$), then the equivalence of \eqref{14.6} and \eqref{14.7} involves the relation of $\frac{\epsilon}{\pi} \int \frac{d\nu(\lambda)}{(\lambda-x)^2+\epsilon^2}$ and $\frac{\nu((a,b))}{b-a}$.  A bound like \eqref{14.6} implies a.e. in $d\lambda$ a bound on $\frac{d\nu(\lambda)}{d\lambda}$.  Since $ \frac{\epsilon}{\pi}\int \frac{d\lambda}{(\lambda-x)^2+\epsilon^2} =1$, we get \eqref{14.7}.  Conversely \eqref{14.7} implies \eqref{14.6} via Stone's formula.

6. To see that \eqref{14.6}$\le$\eqref{14.4}, it suffices by taking limits to consider the case where $a$ and $b$ are not eigenvalues of $H$.  One writes $P_{(a,b)}(H)$ by Stone's formula to see that
\begin{align*}
  |\jap{A^*\varphi,P_{(a,b)}(H)\psi}| &\le \frac{1}{2\pi}\norm{\varphi}\limsup_{\epsilon\downarrow 0} \int_{a}^{b} \norm{A[R(\lambda+i\epsilon)-R(\lambda-i\epsilon)]\psi}\,d\lambda  \\
                                      &\le \norm{\varphi}\left(\int_{a}^{b} 1\, d\lambda\right)^{1/2}\left(\frac{1}{4\pi^2}\int_{a}^{b} \textrm{Integrand in \eqref{14.4}}\,d\lambda\right)^{1/2}
\end{align*}
proving that $\norm{P_{(a,b)}(H)A^*\varphi}\le\norm{\varphi}[$RHS of \eqref{14.4}$]^{1/2}|b-a|^{1/2}$.

7.  To see that \eqref{14.4}$\le$\eqref{14.7}, thereby completing the proof of all the equivalences, let $\alpha$ be the $\sup$ in \eqref{14.7}.  For $z \in \bbC_+$, let $K(z)$ be the positive square root of $(2\pi i)^{-1}(R(z)-R(\bar{z}))$.  Then $\norm{AK(z)}^2 \le \alpha$, so
\begin{align*}
  \textrm{Quantity whose }\sup\textrm{ is taken in \eqref{14.4}} &= \int_{-\infty}^{\infty} \norm{AK(\lambda+i\epsilon)^2\varphi}^2\,d\lambda \\
                                                                 &\le \alpha\int_{-\infty}^{\infty} \norm{K(\lambda+i\epsilon)\varphi}^2 \, d\lambda \\
                                                                 &= \alpha \norm{\varphi}^2
\end{align*}

8.  By \eqref{14.3}, if $A$ is $H$--smooth, then
\begin{align*}
  \norm{AR(\lambda+i\mu)\varphi} &\le \int_{0}^{\infty} e^{-\mu t}\norm{Ae^{-itH}\varphi}\,dt \\
                                 &\le \left(\int_{0}^{\infty} e^{-2\mu t}\,dt\right)^{1/2} \left(\int_{0}^{\infty}\norm{Ae^{-itH}\varphi}^2\,dt\right)^{1/2} \\
                                 &\le (2\mu)^{-1/2}(2\pi)^{1/2} \norm{A}_H
\end{align*}
so $A$ $H$--smooth $\Rightarrow$ $A$ is $H$--bounded with relative bound zero.

9.  In \cite{KatoSm1}, Kato states this equivalence in stages since, as the title of the paper indicates, his focus is on controlling certain non--self--adjoint operators (we focus on the self--adjoint case of greatest interest in NRQM).  He first considers general $H$ with $\sigma(H) \subset \bbR$ and proves a version of Theorem \ref{T14.6} below and then (following Friedrichs \cite{FriedCont}) constructs similarity operators using a stationary replacement for wave operators.  He next adds to $H$ a condition that it generate a group $\{U(t)\}_{t \in \bbR}$ of bounded operators with $\norm{U(t)}=\textrm{O}(e^{\epsilon t})$ for all $\epsilon>0$.  Then \eqref{14.3} holds with $e^{-itH}$ replaced by $U(t)$ and Kato proves the equality of \eqref{14.4} and \eqref{14.5} in that case.  Finally, he proves the full Theorem \ref{T14.1} when $H$ is self--adjoint.
\end{remarks}

\begin{example} \lb{E14.2} Let $H = -i\frac{d}{dx}$ on $L^2(\bbR)$ and let $A$ be multiplication by $f(x)$.  Since $e^{-itH}\varphi(x) = \varphi(x-t)$, we compute
\begin{align*}
  \int_{-\infty}^{\infty} \norm{Ae^{-itH}\varphi}^2\,dt &= \int_{\bbR^2} f(x)^2\varphi(x-t)^2 \,dx\,dt \\
                                                        &= \norm{f}_2^2\norm{\varphi}_2^2
\end{align*}
so, if $f \in L^2(\bbR)$, then $A$ is $H$--smooth.
\end{example}

\begin{example}  \lb{E14.3} If $H_0$ is $-\Delta$ on $L^2(\bbR^3)$, it is known (\cite[(6.9.48)]{RA}) that $(H_0+\kappa^2)^{-1}$ with $\textrm{Re}\,\kappa>0$ has integral kernel $\frac{1}{4\pi|x-y|}e^{-\kappa|x-y|}$.  Suppose that
\begin{equation*}
  \frac{1}{4\pi}\int \frac{|V(x)|\,|V(y)}{|x-y|^2}\,d^3x\,d^3y \equiv \norm{V}_R^2 < \infty
\end{equation*}
called the \emph{Rollnik class} in Simon \cite{SimonThesis} after Rollnik \cite{Rollnik}.  Then the Hilbert--Schmidt norm $\norm{|V|^{1/2}(H_0+\kappa^2)^{-1}|V|^{1/2}}_{HS} \le \norm{V}_R$, so, by \eqref{14.7} $|V|^{1/2}$ is $H_0$--smooth with $\norm{|V|^{1/2}}_{H_0} \le \pi^{-1}\norm{V}_R^{1/2}$.  If $V \in L^{3/2}(\bbR^3)$, the HLS inequality (\cite[Thm 6.2.1]{HA}, \cite{LiebHLS, FrankLHLS}) implies that $V$ is Rollnik.
\end{example}

Smoothness has an immediate consequence for the spectral type of $H$:

\begin{theorem} [Kato \cite{KatoSm1}] \lb{T14.4} Let $H$ be a self--adjoint operator and let $A$ be $H$--smooth.  Then $\ran(A^*) \subset \calH_{ac}(H)$.  In particular, if $\ker(A)=\{0\}$, then $H$has purely a.c. spectrum.
\end{theorem}

The proof is very easy.  If $d\nu$ is the $H$--spectral measure for $A^*\varphi$, then \eqref{14.6} says that
\begin{equation}\label{14.9}
  \nu(I) \le \norm{A}_H \norm{\varphi}^2 |I|
\end{equation}
(where $|\cdot|$ is Lebesgue measure) for open intervals, $I$.  By taking unions and using outer regularity, \eqref{14.9} holds for all sets, so $\nu$ is absolutely continuous.

Smoothness also implies existence and completeness of wave operators.

\begin{theorem} [Kato \cite{KatoSm1}] \lb{T14.5} Let $H, H_0$ be two self--adjoint operators.  Let $A, B$ be closed operators so that $A$ is $H$--smooth and $B$ is $H_0$--smooth and so that
\begin{equation}\label{14.10}
  H-H_0 = A^*B
\end{equation}
in the sense that for $\psi \in D(H)$ and $\varphi \in D(H_0)$, we have that
\begin{equation}\label{14.11}
  \jap{H\psi,\varphi} - \jap{\psi,H_0\varphi} = \jap{A\psi,B\varphi}
\end{equation}
Then $\Omega^\pm(H,H_0)$ exist and are complete.
\end{theorem}

\begin{remarks}  1.  Since smoothness implies relative boundedness, if $\psi \in D(H)$ and $\varphi \in D(H_0)$, then the right side of \eqref{14.11} makes sense.

2. In some applications, one assumes that $H-H_0=\sum_{j=1}^{n}A_j^*B_j$ with each $A_j$ $H$--smooth and each $B_j$ is $H_0$--smooth.  The proof in remark 3 extends to this case or, alternatively, one can define smoothness for closed operators, A, from $\calH$, the space on which $H$ is defined to $\calK$, a perhaps distinct Hilbert space, and then pick $\calK=\oplus_{j=1}^n \calH, \, B = \oplus_{j=1}^n B_j,\, A=\oplus_{j=1}^n A_j$ so $A^*B = \sum_{j=1}^{n} A_j^* B_j$.

3.  The proof is again easy (indeed, one of the beauties of Kato smoothness theory is how much one gets with simple proofs).  If $\psi \in D(H)$ and $\varphi \in D(H_0)$, $W(t) = e^{+itH}e^{-itH_0}$, then for $s < t$,
\begin{align*}
  |\jap{\psi,(W(t)-W(s))\varphi}| &= \left| \int_{s}^{t} \jap{Ae^{-iuH}\psi,Be^{-iuH_0}\varphi} \, du \right| \\
                                 &\le \left( \int_{-\infty}^{\infty} \norm{Ae^{-iuH}\psi}^2\,du\right)^{1/2} \left( \int_{-s}^{t} \norm{Be^{-iuH_0}\varphi}^2\,du\right)^{1/2}  \\
                                 &\le \sqrt{2\pi} \norm{A}_H \norm{\psi} \left( \int_{-s}^{t} \norm{Be^{-iuH_0}\varphi}^2\,du\right)^{1/2}
\end{align*}
so
\begin{equation}\label{14.12}
 \norm{(W(t)-W(s))\varphi} \le \sqrt{2\pi} \norm{A}_H \left( \int_{-s}^{t} \norm{Be^{-iuH_0}\varphi}^2\,du\right)^{1/2}
\end{equation}
is Cauchy.  Therefore, $\Omega^\pm(H,H_0)$ exists.  Since $H_0-H = -B^*A$, we conclude that they are also complete by Theorem \ref{T13.1}
\end{remarks}

We say that a closed operator, $A$ is $H$--\emph{supersmooth} if and only if
\begin{equation}\label{14.13}
  \norm{A}_{H,SS}^2 \equiv \sup_{z \in \bbC\setminus\bbR} \norm{A(H-z)^{-1}A^*} < \infty
\end{equation}
The notion is in Kato \cite{KatoSm1} and the name is from Kato--Yajima \cite{KY} in 1989.  The name hasn't stuck but I like it, so I'll use it.  The fourth important result in Kato \cite{KatoSm1} is

\begin{theorem} [Kato \cite{KatoSm1}] \lb{T14.6} Let $H_0$ be a self--adjoint operator.  Let $A$ be $H_0$--supersmooth and $C$ a bounded self-adjoint operator so that
\begin{equation}\label{14.14}
  \alpha \equiv \norm{C}\norm{A}_{H_0,SS}^2 < 1
\end{equation}
Let $B=A^*CA$.  Then $B$ is relatively form bounded with relative form bound at most $\alpha$.  If $H=H_0+B$, then A is also $H$--supersmooth with
\begin{equation}\label{14.14A}
  \norm{A}_{H,SS} \le \norm{A}_{H_0,SS}(1-\alpha)^{-1/2}
\end{equation}
In particular, $\Omega^\pm(H,H_0)$ exist and are complete.
\end{theorem}

\begin{remarks}  1.  Once again, the proofs are simple.  The key is a formal geometric series:
\begin{align}
  A(H&-z)^{-1}A^* =  A(H_0-z)^{-1}A^* \nonumber \\
       &+ \sum_{j=0}^{\infty} (-1)^{j+1}  A(H_0-z)^{-1}A^* \left[C A(H_0-z)^{-1}A^*\right]^j C A(H_0-z)^{-1}A^*  \lb{14.14B}
\end{align}
One proves the form boundedness and uses that to justify a formula like \eqref{14.14B} but with an error term.  Since $\norm{C A(H_0-z)^{-1}A^* } \le \alpha$, the error goes to zero and the series converges.  The final assertion then comes from Theorem \ref{T14.5}.

2.  By the same analysis, the analog of Remark 2 after Theorem \ref{T14.5} holds.  If $H = H_0+\sum_{j=1}^{n} A_j^*B_j$ and $\gamma_{jk} = \sup_{z \in \bbC\setminus\bbR} \norm{B_j(H_0-z)^{-1}A^*_k}$ is finite and $\Gamma = \{\gamma_{jk}\}_{1 \le j,k  \le n}$ is a matrix of norm $\alpha < 1$, and if each $A_j$ and $B_j$ is supersmooth, then $\Omega^\pm(H,H_0)$ exist and are complete.

3.  We repeat that in \cite{KatoSm1}, Kato considers cases where $H_0$ and $C$ need not be self--adjoint.  He assumes that $\sigma(H_0) \subset \bbR$ and ${\norm{C}\sup_z \norm{ A(H_0-z)^{-1}A^* } < 1}$ and then defines an operator $H$ which is formally $H_0+A^*CA$ with a resolvent that obeys \eqref{14.14B}.  He then uses ideas going back to Friedrichs \cite{FriedCont} to define (in terms of resolvents, not time limits) invertible operators $W^\pm$ so that $W^\pm H_0 (W^\pm)^{-1} = H$.
\end{remarks}

That completes our discussion of \cite{KatoSm1}.  The main result of \cite{KatoSm2} is

\begin{theorem} [Putnam--Kato Theorem \cite{PutnPK, KatoSm2}] \lb{T14.7}  Let $A$ and $B$ be bounded self--adjoint operators so that $D \equiv i[A,B]$ is strictly positive in the sense that for all $\varphi \ne 0$, we have that
\begin{equation}\label{14.15}
  \jap{\varphi,D\varphi} > 0
\end{equation}
Then $A$ and $B$ have purely a.c. spectrum.
\end{theorem}

\begin{remarks} 1. The result is due to Putnam.  Kato found the really simple proof in the next remark.

2.  The proof is easy.  For let $C$ be the square root of $i[A,B]$. Then $\frac{d}{dt}\jap{e^{-itA}\varphi,Be^{-itA}\varphi}^2 = \norm{Ce^{-itA}\varphi}^2$ so the integral of $\norm{Ce^{-itA}\varphi}^2$ from $s$ to $t$ is bounded by $2\norm{B}\norm{\varphi}$, Thus $C$ is $A$--smooth and $A$ has only a.c. spectrum on the closure of $\ran(C)$ which is all of $\calH$.
\end{remarks}

\begin{example} [Weak coupling $2$--body] \lb{E14.8} In \cite{KatoSm1}, Kato applied smoothness ideas to Schr\"{o}dinger operators.  If $\nu=3$, as we've seen in Example \ref{E14.3}, if $V \in L^{3/2}$ (indeed, if $V$ is Rollnik), then $|V|^{1/2}$ is $-\Delta$--supersmooth, so for small real $\lambda$, the wave operators, $\Omega^\pm(-\Delta+\lambda V, -\Delta)$ exist and are unitary.  On $(0,\infty)$, if $h_0=-\tfrac{d^2}{dx^2}$ with $u(0)=0$, then $(h_0-z)^{-1}$ has an integral kernel dominated by $\min(x,y)$ (see \cite[(7.9.53]{OT}) for all $z \in \bbC \setminus \bbR$, so if $\int_{0}^{\infty} x|V(x)| dx < \infty$, then $|V|^{1/2}$ is {$h_0$--supersmooth} and one knows that for $\lambda$ small, that $\Omega^\pm(h_0+\lambda V,h_0)$ exists and are unitary.

One knows that if $\nu=1$ or $2$ and $V \in C_0^\infty(\bbR^\nu); V \not\equiv 0$, then for all $\lambda \ne 0$, either $-\Delta+\lambda V$ or $-\Delta-\lambda V$ (or both) have a negative energy bound state (\cite{SimonWeakBound}) so there cannot be $-\Delta$--supersmoothness.

By interpolating between $\norm{e^{it\Delta}\varphi}_\infty \le (4\pi t)^{-\nu/2}\norm{\varphi}_1$ and $\norm{e^{it\Delta}\varphi}_2 = \norm{\varphi}_2$, Kato \cite{KatoSm1} showed that if $\nu \ge 4$ and $V \in L^{\nu/2+\epsilon} \cap L^{\nu/2-\epsilon}$, then $|V|^{1/2}$ is $-\Delta$--supersmooth and he conjectured that this held for $\epsilon=0$.  Indeed, the next theorem is true.
\end{example}

\begin{theorem} \lb{T14.9}  Let $\nu \ge 3$ and $V \in L^{\nu/2}(\bbR^\nu)$.  Then $V$ is supersmooth.  In particular, for $|\lambda|$ small and $H = -\Delta+\lambda V, H_0=-\Delta$, we have that $\Omega^\pm(H,H_0)$ exist and are unitary so that $H$ has purely a.c. spectrum.
\end{theorem}

\begin{remarks}  1.  This result appeared in Kato--Yajima \cite{KY}.  As they added in a ``Note added in proof'', shortly before their paper, Kenig--Ruiz--Sogge \cite{KRS} proved estimates that imply Theorem \ref{T14.9}.

2.  In \cite{IorioOC}, Iorio--O'Carroll used supersmoothness to show $N$--body systems with weak coupling (and $\nu \ge 3$) have unitary wave operators (so no bound states, no non--trivial scattering channels and purely a.c. spectrum).  They required that the two body potentials lie in $L^{\nu/2+\epsilon} \cap L^{\nu/2-\epsilon}$, but given Theorem \ref{T14.9}, their method works for two body potentials in $L^{\nu/2}$.
\end{remarks}

Kato--Yajima \cite{KY} also proved that $(1+|x|^2)^{-1/2}(1-\Delta)^{1/4}$ is $-\Delta$--supersmooth (which says something about $V(x) = |x|^{-2}$ on $L^2(\bbR^\nu); \nu \ge 3$).  Further developments are due to Ben--Artzi--Klainerman \cite{BAKSmooth} and Simon \cite{SimonSmooth}.  In particular, Simon obtained optimal constants in the associated smoothness estimates; for $\nu \ge 3$
\begin{equation}\label{14.15A}
  \int_{-\infty}^{\infty} \norm{(x^2+1)^{-1/2}(-\Delta)^{1/4}e^{it\Delta}\varphi}^2\,dt \le \frac{\pi}{2} \norm{\varphi}^2
\end{equation}
\begin{equation}\label{14.15B}
  \int_{-\infty}^{\infty} \norm{|x|^{-1} e^{it\Delta}\varphi}^2\,dt \le \frac{\pi}{\nu-2} \norm{\varphi}^2
\end{equation}

Next, having completed our discussion of Kato's contributions to smoothness, we turn some applications beginning with repulsive potentials.  In this (and other) regards, it is useful to have the notion of local smoothness due to Lavine \cite{Lav3}.  Let $\Omega \subset \bbR$ be a bounded Borel set.  We say that \emph{$A$ is locally $H$--smooth on $\Omega$} if $AP_\Omega(H)$ is $H$--smooth (where $P_X(H)$ is a spectral projection for $H$ and set $X$ \cite[Section 5.1]{OT}).  It is easy to see \cite[Theorem XIII.30]{RS4} that if $A$ is an operator with $D(H) \subset D(A)$ and either $\sup_{0 < \pm\epsilon < 1; \lambda \in \Omega} \epsilon\, \norm{AR(\lambda+i\epsilon)}< \infty$ or $\sup_{0 < \epsilon < 1; \lambda \in \Omega} \norm{AR(\lambda+i\epsilon)A^*} < \infty$, then $A$ is locally $H$--smooth on $\Omega$.  It is also obvious that if $\ran(A^*)$ is dense, then, if $A$ is locally $H$--smooth, $H \restriction \ran\, P_\Omega(H)$ is purely absolutely continuous.  The following is what makes local $H$--smoothness so useful:

\begin{theorem} [Lavine \cite{Lav3}] \lb{T14.10} Let $H$ and $H_0$ be self--adjoint and $\Omega \subset \bbR$ a bounded open set.  Suppose that $H=H_0+A^*B$ where $B$ is $H_0$--bounded and locally $H_0$--smooth on $\Omega$ and $A$ is $H$--bounded and locally $H$--smooth on $\Omega$.  Then
\begin{equation}\label{14.16}
  \Omega^\pm(H,H_0;P_\Omega(H_0)) = {\textrm{s--}\lim}_{t \to \mp \infty} e^{itH}e^{-itH_0} P_\Omega(H_0)
\end{equation}
exist and have range $P_\Omega(H)$.
\end{theorem}

\begin{remarks} 1.  For complete proofs, see \cite{Lav3} or \cite[Theorem XIII.31]{RS4}.

2. The same proof as Theorem \ref{T14.5} shows that ${\textrm{s--}\lim}_{t \to \mp \infty} P_\Omega(H) e^{itH}e^{-itH_0} P_\Omega(H_0)$ exists.

3. Since $B e^{-itH_0} P_\Omega(H_0)(H_0-z)^{-1}\varphi$ is in $L^2$ with an $L^2$ derivative, we conclude that for any $z \in \bbC\setminus\bbR$
\begin{equation*}
  {\textrm{s--}\lim}_{t \to \mp \infty} Be^{-itH_0} P_\Omega(H_0)(H_0-z)^{-1} = 0
\end{equation*}

4.  Writing $(H-z)^{-1}-(H_0-z)^{-1} = \left[A(H-\bar{z})^{-1}\right]^*B(H_0-z)^{-1}$ and using the assumed boundedness of $A(H-\bar{z})^{-1}$, we conclude by remark 3 that ${\textrm{s--}\lim}_{t \to \mp \infty} \left[(H-z)^{-1}-(H_0-z)^{-1}\right] e^{-itH_0} P_\Omega(H_0)=0$ and then by the Stone--Weierstrass gavotte \cite[Appendix to Chapter 3]{CFKS} that ${\textrm{s--}\lim}_{t \to \mp \infty} \left[f(H)-f(H_0)\right] e^{-itH_0} P_\Omega(H_0)=0$ for any continuous function, f, so that $1-f$ has compact support.  Using this, one sees if $I \subset \Omega$ is a compact set with $\dist(I,\bbR\setminus\Omega) > 0$, then $\textrm{s--}\lim_{t \to \mp \infty} P_{\bbR\setminus\Omega} e^{itH}e^{-itH_0} P_I(H_0) = 0$.  This implies that the limits in \eqref{14.16} exist and that $\ran\,\Omega^\pm(H,H_0;P_\Omega(H_0)) \subset \ran\, P_\Omega(H)$.  This plus symmetry between $H$ and $H_0$ plus the idea behind Theorem \ref{13.1} imply that $\ran\,\Omega^\pm(H,H_0;P_\Omega(H_0)) = \ran\, P_\Omega(H)$.
\end{remarks}

A potential, $V$, on $\bbR^\nu$ is called \emph{repulsive} if and only if $\mathbf{x}\cdot\mathbf{\nabla}V \le 0$ (e.g. $V(x) = (1+|x|)^{-\alpha}$, any $\alpha >0$).  If $V(x) \to 0$ at infinity, then $V(x) \ge 0$.  If ${A=\tfrac{i}{2}(\mathbf{x}\cdot\mathbf{\nabla}+\mathbf{\nabla}\cdot\mathbf{x})}$ is the generator of dilations and $V$ is repulsive, then $i[A,H_0+V]=2H_0- \mathbf{x}\cdot\mathbf{\nabla}V \ge 0$.  One cannot use the Putnam--Kato theorem since neither $A$ nor $H$ is bounded.  If you look at the above proof of the Putnam--Kato theorem, that $H$ is unbounded isn't a problem if our goal is to find a $C$ which is $H$--smooth.  But the unbounded $A$ is.  Lavine's idea was to cutoff $\mathbf{x}$ in the definition of $A$ and get an $\tilde{A}$ which is $H$--bounded and so that $i[\tilde{A},H] \ge c(1+|x|^2)^{-\beta}$ for suitable $\beta$ and as in the Putnam--Kato argument, get that ${(1+|x|^2)^{-\beta/2}(H+1)^{-1}}$ is $H$--smooth.  In this way (he used local smoothness to get wave operators), Lavine proved

\begin{theorem} [Lavine \cite{Lav1, Lav2, Lav3, Lav4}] \lb{T14.11}  Let $H$ be an $N$--body Hamiltonian with center of mass removed on $L^2(\bbR^{(N-1)\nu})$ whose two body potentials $V_{ij}$ lie in $L^p(\bbR^\nu)+L^\infty(\bbR^\nu)$ (with $p$ $\nu$--canonical) and are repulsive.  Then $H$ has purely absolutely continuous spectrum.  If moreover, for some $\beta > 5/2$, we have that $|V_{ij}(x)| \le C(1+|x|)^{-\beta}$, then $\Omega^\pm(H,H_0)$ exist and are complete.
\end{theorem}

\begin{remark}  $5/2$ is an artifact of the proof and when the $V_{ij}$ are spherically symmetric, it has been improved to $\beta > 1$ in Lavine \cite{Lav4}.
\end{remark}

Our final major topic concerns ideas of Vakulenko \cite{Vaku}; the reader should first look at the discussion around equation \eqref{12.6} for definitions of Vakulenko bounding function and Vakulenko potential.

\begin{lemma} [Vakulenko \cite{Vaku}] \lb{L14.12} Let $H$ be self--adjoint and $A$ a closed $H$--bounded operator.  Let $[a,b]$ be a bounded closed interval in $\bbR$ and $B$ a closed operator with $D(H) \subset D(B)$ so that for all $\varphi \in D(H)$ and $\lambda \in [a,b]$, we have that
\begin{equation}\label{14.17}
  \textrm{Re} \jap{(H-\lambda)\varphi,A\varphi} \ge \norm{B\varphi}^2
\end{equation}
Then $B$ is $H$--smooth on $[a,b]$.
\end{lemma}

\begin{remarks} 1. As a preliminary, we note that since $\frac{|x-\lambda|}{|x-(\lambda+i\epsilon)|} \le 1$, we have that
\begin{equation}\label{14.18}
  \norm{(H-\lambda)R(\lambda+i\epsilon)} \le 1
\end{equation}

2. As a second preliminary, if
\begin{equation}\label{14.19}
  \norm{A\varphi} \le \alpha \norm{H\varphi} + \beta\norm{\varphi}
\end{equation}
then
\begin{align}
      \norm{AR(\lambda+i\epsilon)\psi}  &\le \alpha\norm{[(H-\lambda)+\lambda]R(\lambda+i\epsilon)\psi} + \beta\norm{R(\lambda+i\epsilon)\psi} \nonumber \\
                                        &\le (\alpha+\alpha|\lambda|\epsilon^{-1}+\beta\epsilon^{-1})\norm{\psi} \lb{14.20}
\end{align}

3.  Letting $\varphi=R(\lambda+i\epsilon)\psi$ in \eqref{14.17}, we see that
\begin{align}
  \norm{BR(\lambda+i\epsilon)\psi}^2 &\le \norm{(H-\lambda)R(\lambda+i\epsilon)}\norm{AR(\lambda+i\epsilon)}\norm{\psi}^2 \nonumber \\
                                     &\le C\epsilon^{-1} \norm{\psi}^2 \lb{14.21}
\end{align}
(by \eqref{14.18}/\eqref{14.20}) for $0 < \epsilon < 1$ and all $\lambda \in [a,b]$ where $C$ is a constant depending on $\alpha, \beta, a$ and $b$. This implies local smoothness by the discussion prior to Theorem \ref{T14.10}.

4.  Vakulenko's $A$ is close to $i$ times a cutoff dilation generator, so the left side of \eqref{14.17} is like an expectation of a commutator and thus this is a variant of a Mourre estimate but unlike the Mourre estimate, there is no (compact) error term.
\end{remarks}

In Theorem \ref{T12.2}, we stated a bound of the form \eqref{14.17} which immediately implies (given the lemma)

\begin{theorem} [Vakulenko \cite{Vaku}] Let $V(x)$ be a Vakulenko potential with \eqref{12.7} for some Vakulenko bounding function $\eta$.  Then $\sqrt{\eta}$ is $-\Delta+V$ locally smooth on $(0,\infty)$.  In particular, the spectrum of $-\Delta+V$ is purely absolutely continuous on $(0,\infty)$ and the wave operators exist and are complete.
\end{theorem}

\begin{remarks}  1.  Since $\eta$ is everywhere non--vanishing, $\ran\, \sqrt{\eta}$ is dense and this implies the absolute continuity on $(0,\infty)$.

2.  $\sqrt{\eta}$ is locally smooth for both $-\Delta+V$ and $-\Delta$ (since the zero potential is a Vakulenko potential with bounding function $\eta$).  Since $|V|^{1/2} \le \sqrt{\eta}$, we see that $|V|^{1/2}$ is locally smooth which implies that wave operators exist and are complete.

3. The proof of Theorem \ref{T12.2} is particularly easy when $\nu=1$.  Fix $\lambda_0 > 0$ and let
\begin{equation}\label{14.22}
  \omega(x) = \exp \left[\frac{2}{\sqrt{\lambda_0}}\int_{-\infty}^{x} \eta(y) \, dy\right]
\end{equation}
and
\begin{equation}\label{14.23}
  A = 2\omega\frac{d}{dx}
\end{equation}
Since $\eta \in L^1(\bbR)$, $\omega$ is bounded so since $\frac{d}{dx}(-\Delta+V+i)^{-1}$ is bounded, we see that $A$ is $H$--bounded.  It is easy to see (since $\eta$ and $V$ are real) that it suffices to prove \eqref{14.17} when $\varphi$ is real in which case:
\begin{equation}\label{14.24}
  \jap{(H-\lambda)\varphi,A\varphi} = \int_{-\infty}^{\infty} \left[\omega'\left[(\varphi')^2+\lambda(\varphi)^2\right]+2\omega V\varphi\varphi'\right] \, dx
\end{equation}
which we get by integration by parts in
\begin{equation*}
   2\int_{-\infty}^{\infty} (-\varphi''-\lambda\varphi)\omega\varphi'\,dx =  -\int_{-\infty}^{\infty} \omega\left[(\varphi')^2+\lambda(\varphi)^2\right]'\,dx
 \end{equation*}
Since $(|\varphi'|-\sqrt{\lambda}\varphi)^2 = (\varphi')^2+\lambda(\varphi)^2-2\sqrt{\lambda}|\varphi'||\varphi|$ we see that
\begin{equation}\label{14.25}
  \textrm{RHS of \eqref{14.24}} \ge \int_{-\infty}^{\infty} \left(\omega'-\frac{|V(x)|}{\sqrt{\lambda}} \omega\right) \left[(\varphi')^2+\lambda(\varphi)^2\right]'\,dx
\end{equation}
By construction of $\omega$, $|V| \le \eta$, $\omega \ge 1$ and $\lambda > \lambda_0$, we have that
\begin{equation}\label{14.26}
  \omega'-\frac{|V(x)|}{\sqrt{\lambda}} \omega \ge \frac{1}{\sqrt{\lambda_0}}\omega\eta \ge \frac{\eta}{\sqrt{\lambda_0}}
\end{equation}
Thus
\begin{align}
   \textrm{RHS of \eqref{14.25}} &\ge  \int_{-\infty}^{\infty} \frac{\lambda}{\sqrt{\lambda_0}} \eta(x)(\varphi)^2\,dx \nonumber \\
                                 &\ge \sqrt{\lambda_0} \norm{\sqrt{\eta}\varphi}^2
 \end{align}
which is \eqref{14.17}.  The higher dimensional case needs a carefully constructed $\omega$ but is along similar lines.

4.  Since $\eta(x) = (1+|x|)^{-\alpha}, \alpha > 1$ is a Vakulenko bounding function, we get the Corollary below.
\end{remarks}

\begin{corollary} \lb{C14.14} If
\begin{equation}\label{14.28}
  |V(x)| \le C(1+|x|)^{-\alpha}
\end{equation}
for some $\alpha > 1$, then $H=-\Delta+ V$ has purely a.c. spectrum on $(0,\infty)$ and with $H_0=-\Delta$, $\Omega^\pm(H,H_0)$ exist and are complete.
\end{corollary}

Thus Vakulenko obtained a new and beautiful proof of an Agmon--Kato--Kuroda type theorem of the kind we discuss in the next section (albeit 15 years after their work).  Unlike their method, this one seems to require pointwise bounds and doesn't allow for local singularities.

Yafaev \cite{YafaevLR} has an approach to long range $2$--body scattering that exploits some ideas from the theory of smooth perturbations.

We note that the earliest proofs of $N$--body asymptotic completeness for $\textrm{0}(|x|^{-1-\epsilon})$ potentials (at least when $N \ge 4$) were by Sigal--Soffer \cite{SigSof1, SigSof2} and then by Graf \cite{Graf} and Derezi\'{n}ski \cite{Dere}.  \cite{Dere}  and \cite{SigSof2} have results on long range results.  In \cite{YafaevAC1}, Yafaev found a proof that exploits smoothness ideas (as well as some of the tools -- Mourre estimates \cite{Mourre, PSS, FHM}, Deift--Simon wave operators \cite{DeiftSimonWO}, Enss type phase space analysis \cite{Enss, Enss3}--of the earlier approaches).  Kato never considered $N$--body scattering, which is quite involved, so we refer the reader to Yafaev's original paper \cite{YafaevAC1} or lecture notes \cite{YafaevLNM} for details.

%%%%%%%%%%%%%%%%%%%%%%%%%%%%%%%%%%%%%%%%%%%%%%%%%%%%%%%%%%%%%%
\section{Scattering and Spectral Theory, III: Kato--Kuroda Theory} \lb{s15}
%%%%%%%%%%%%%%%%%%%%%%%%%%%%%%%%%%%%%%%%%%%%%%%%%%%%%%%%%%%%%%

This is the third section on spectral and scattering theory; it focuses on stationary, aka time--independent, methods.  As with the prior two sections, we'll include an overview portion but we want to begin by describing the problem we'll discuss and the contributions of Agmon, Kato and Kuroda.  While it is significant that local singularities can be accommodated, we'll mainly discuss the case \eqref{13.15}, i.e.
\begin{equation}\label{15.1}
  |V(x)| \le C(1+|x|)^{-\alpha}
\end{equation}
We consider $H_0=-\Delta, H=-\Delta+V(x)$ on $L^2(\bbR^\nu,d^\nu x)$.  These are the questions that will concern us:

\noindent (A) Existence and Completeness of $\Omega^\pm(H,H_0)$

\noindent (B) Absence of singular continuous spectrum

As a sidelight of the methods, one also gets continuum eigenfunction expansions of a type I will discuss below.  There is also the issue of positive eigenvalues which except for the work of Vakulenko (as discussed in Sections \ref{s12} and \ref{s14}) was studied using very different methods from those used in this section; see Section \ref{s12}.

As we explained in Section \ref{s13}, it follows from Cook's method that $\Omega^\pm(H,H_0)$ exist if $\alpha >1$ while they may not if $\alpha \le 1$.  It is known (see Section 20) that (B) can fail if $\alpha < 1$ (although this was not known in the 1970s), so in the 15 years after 1957, a lot of effort was made on studying problems (A) and (B) when $\alpha >1$.  We'll say a lot more about the detailed history later but start with the best results of Kato--Kuroda on the subject and on the optimal result.

In 1969, Kato \cite{Kato1969KK} using, in part, ideas of Kato--Kuroda (of which we'll say a lot more below) proved

\begin{theorem} [Kato \cite{Kato1969KK}]  \lb{T15.1} Let $V$ obey \eqref{15.1} and $H, H_0$ as above.  Then

(a) If $\alpha>1$, the wave operators exist and are complete.

(b) If $\alpha > 5/4$, $H$ has no singular continuous spectrum.
\end{theorem}

In 1970, Agmon \cite{AgmonAnon} announced.

\begin{theorem} [Agmon \cite{AgmonAnon, AgmonSpec}] \lb{T15.2}  Let $V$ obey \eqref{15.1} and $H$ as above.  If $\alpha > 1$, $H$ has no singular continuous spectrum.
\end{theorem}

While Agmon did not discuss scattering in his announcement, Lavine \cite{Lav5} noted that his estimates and Lavine's theory of local smoothness implied existence and completeness of wave operators (and later, both Agmon and H\"{o}rmander presented other approaches to get completeness).  We also note that as discussed, for example, in \cite[Section XIII.8]{RS4}, one can accommodate local singularities; in place of assuming $(1+|x|)^\alpha V(x)$ is bounded, one need only assume that it is a relatively compact perturbation of $-\Delta$.

Agmon was able to go from $5/4$ to $1$ by an astute observation (Step 8 in the scheme at the end of the chapter).  By using the same idea, Kuroda could extend that Kato--Kuroda argument up to $\alpha >1$.  Later we'll say more about work of others on these problems.

Our goal in this section is to explain the machinery behind certain proofs of Theorems \ref{T15.1} and \ref{T15.2}. We begin with some general overview of the stationary approach to scattering.  The earliest mathematical approach to stationary scattering is in Friedrichs \cite{FriedCont} but we will focus on a slightly later one of Povzner \cite{Povz1} in 1953 and Kato's student, Ikebe \cite{Ikebe1}, in 1960 that discusses eigenfunction expansion.  Their expansions are to be distinguished from what \cite{SimonSmgp} calls BGK expansions after Berezanskii, Browder, G\r{a}rding, Gel'fand and Kac (see references in \cite{SimonSmgp}).  The BGK expansion is essentially a variant of the spectral theorem when an operator $A$ on $L^2(\bbR^\nu,d^\nu x)$ has local trace class properties (i.e. $f(x)P_{[a,b]}(A)f(x)$ is trace class for $f \in C^\infty_0(\bbR^\nu)$).  This expansion is stated in terms of the spectral measures and so has no implications for the spectral properties.  The advantage of BGK expansions is that they are always applicable for Schr\"{o}dinger operators (see \cite{SimonSmgp}) while the Povzner--Ikebe expansion only works in special situations, but when it does, it provides a lot of additional information.

The IP expansion of Povzner--Ikebe involves not spectral measures but $d^\nu k$ which is why it has important spectral consequences.  The model is the Fourier transform for $H_0=-\Delta$ which in this introductory discussion we'll denote as $\hat{f}_0$ (since we'll use $\hat{f}$ for something else) defined on $\bbR^\nu$ by
\begin{equation}\label{15.2}
  \hat{f}_0(\mathbf{k}) = (2\pi)^{-\nu/2} \int \overline{\varphi_0(\mathbf{x},\mathbf{k})} f(\mathbf{x}) \, d^\nu x
\end{equation}
\begin{equation}\label{15.3}
  \varphi_0(\mathbf{x},\mathbf{k}) = e^{i\mathbf{k}\cdot\mathbf{x}}
\end{equation}
(see \cite[Section 6.5]{RA} for the meaning of \eqref{15.2} when $f$ is only in $L^2$ and not in $L^1$).  This provides an eigenfunction expansion of $H_0$ in that (except for places where we want to emphasize  the vector nature of $\mathbf{x}$ and $\mathbf{k}$, we will start using non--boldface)
\begin{equation}\label{15.4}
  f(x) = (2\pi)^{-\nu/2} \int \varphi_0(x,k) \hat{f}_0(k) \, d^\nu k
\end{equation}
\begin{equation}\label{15.5}
  \widehat{H_0f}_0(k) = |k|^2 \hat{f}_0(k)
\end{equation}
so that formally (and much more), $H_0\varphi(\cdot,k)=|k|^2\varphi(\cdot,k)$.

For suitable $V$ and $H=H_0+V$, what Povzner and Ikebe found are functions, $\varphi(\mathbf{x},\mathbf{k})$, so that if $\hat{f}$ is defined by
\begin{equation}\label{15.6}
  \hat{f}(k) = (2\pi)^{-\nu/2} \int \overline{\varphi(x,k)} f(x) \, d^\nu x
\end{equation}
and if $\{\varphi_n(x)\}_{n=1}^N$ is an orthonormal basis of $L^2$ eigenfunctions for $\calH_{pp}(H)$ with $H\varphi_n=E_n\varphi_n$, then
\begin{equation}\label{15.7}
  f(x) = \sum_{n=1}^{N} \jap{\varphi_n,f}\varphi_n(x) + (2\pi)^{-\nu/2} \int \varphi(x,k)\hat{f}(k) \, d^\nu k
\end{equation}
and
\begin{equation}\label{15.8}
  \widehat{Hf}(k) = |k|^2 \hat{f}(k)
\end{equation}
This implies that $H$ has point spectrum plus a.c. spectrum solving problem (B).

They also proved a connection to scattering
\begin{equation}\label{15.9}
  \widehat{\Omega^+f}=\hat{f}_0
\end{equation}
so that formally
\begin{equation}\label{15.9A}
  \Omega^+\varphi = \varphi_0
\end{equation}
(we'll say more about this shortly).  This implies that $\ran\,\Omega^+=\calH_{ac}(H)$ and then, since $\Omega^+\bar{f}=\overline{\Omega^-f}$ (where $\overline{\rule{0pt}{2.75mm}{}\null\quad\null}$ is complex conjugate), we have that $\ran\,\Omega^+=\calH_{ac}(H)$ solving problem (A).

In the physics literature, Gell'Mann  and Goldberger, \cite{GG} appealing to stationary phase arguments (\cite[Section 15.3]{CAB}), considered the meaning of \eqref{15.9A} and formally proved, that pointwise it held if the limit in the definition of wave operator is an abelian limit (i.e. an $e^{-\epsilon t}$ is added to the quantity in the limit and then one takes $\epsilon \downarrow 0$).  Indeed, Ikebe proved \eqref{15.9} in terms of abelian limits and then used the existence of the ordinary limit proven by other means.

Of course, one has to find suitable continuum eigenfunctions, $\varphi(\mathbf{x},\mathbf{k})$, so that \eqref{15.9} holds.  Some thought about Born's ideas suggests one wants $\varphi$ to have the asymptotics \eqref{13.1} near $\mathbf{x}=\infty$.  We'll explain that $\varphi$ obeys an integral equation called the Lippmann--Schwinger equation introduced by two physicists \cite{LippS} in 1950.  Following Lippmann--Schwinger, Povzner and Ikebe, we only consider $\nu=3$ where the integral kernel for $(H_0-k^2)^{-1}$ is especially simple.

Since formally $(H_0+V-k^2)\varphi=0$, we might expect that $\varphi$ obeys $\varphi=-(H_0-k^2)^{-1}V\varphi$.  There are two problems with this.  First, since $k^2$ is in the spectrum of $H_0$, we can't use $(H_0-k^2)^{-1}$ as a bounded operator on $L^2$.  If $\textrm{Im}(k) > 0$ (so $k^2 \notin \bbR$), then $(H_0-k^2)^{-1}$ has an integral kernel, $G_0(\mathbf{x},\mathbf{y};k^2)$, given by
\begin{equation}\label{15.10}
  G_0(x,y;k^2)=\frac{e^{ik|x-y|}}{4\pi|x-y|}
\end{equation}
This has a pointwise limit as $k^2 \to \bbR$, indeed two different limits if one takes $\epsilon\downarrow 0$ for $k^2\pm i\epsilon$.  We thus define for $k>0$
\begin{equation}\label{15.10A}
  G_0(x,y,k^2\pm i0) = \frac{e^{\pm ik|x-y|}}{4\pi|x-y|}
\end{equation}
As we'll see, to get \eqref{15.9}, we want to pick $+i0$, not $-i0$.  It is the use of plus here that led physicists to use $\Omega^+$ for the limit as $t \to -\infty$.  This gives meaning to $-(H_0-k^2)^{-1}V\varphi$.

The second problem with $\varphi=-(H_0-k^2)^{-1}V\varphi$ is that if $V$ has rapid decay (e.g. $V$ has compact support), it is not hard to see that $-(H_0-k^2)^{-1}V\varphi$ looks like the second term in \eqref{13.1}, so it is tempting to try $\varphi=e^{ik.x}-(H_0-k^2)^{-1}V\varphi$.  Notice that since $(H_0-k^2)$  has a kernel (among ``reasonable'' functions), we are allowed to add elements of the kernel when inverting; put differently $(H_0-k^2)[e^{ik.x}-(H_0-k^2)^{-1}V\varphi= -V\varphi$ and thus our formal eigenfunctions will be solutions of the \emph{Lippmann--Schwinger equation}
\begin{equation}\label{15.11}
  \varphi(\mathbf{x},\mathbf{k}) = e^{i\mathbf{k}\cdot\mathbf{x}}-\frac{1}{4\pi} \int \frac{e^{i|\mathbf{k}||\mathbf{x}-\mathbf{y}|}}{|\mathbf{x}-\mathbf{y}|}V(\mathbf{y})\varphi(\mathbf{y}) d^3y
\end{equation}

The pioneer in using the Lippmann--Schwinger equation to prove mathematical results about eigenfunction expansions was Povzner \cite{Povz1, Povz2}.  In \cite{Povz1}, published in 1953, he considered $C^\infty$ potentials, $V$, obeying \eqref{15.1} for $\nu=3, \alpha > 7/2$ and solved problem (B) affirmatively for such $\alpha$.  In 1955, in \cite{Povz2}, for $V$'s of compact support, he solved problem (A) (when $\nu=3$).  Bear in mind that the results of Cook, Hack and Kuroda on existence (via Cook's method) didn't exist when Povzner wrote \cite{Povz2}.  As we'll see, Ikebe's approach to solving problem (A) uses these a priori existence results.

In 1960, Ikebe \cite{Ikebe1} used eigenfunction expansions to solve problems (A) and (B) when $\nu=3$ and $V$ obeys \eqref{15.1} near infinity for $\alpha>2$ and moreover, $V$ is H\"{o}lder continuous away from a finite number of points where it is locally $L^2$.  Let us sketch the ideas that he used:

(i) Let $B$ be the Banach space, $C_\infty(\bbR^3)$, of bounded functions vanishing at $\infty$ with $\norm{\cdot}_\infty$.  For $\textrm{Im}(\kappa) \ge 0$, define
\begin{equation}\label{15.12}
  (T_\kappa g)(x) = -\frac{1}{4\pi} \int \frac{e^{i\kappa|x-y|}}{|x-y|} V(y) g(y)\, d^3y
\end{equation}
Then if $V$ obeys \eqref{15.1} with $\alpha>2$, $T_\kappa$ is a bounded, indeed a compact, operator of $B$ to $B$ which is analytic in $\kappa$ on $\bbC_+$ and H\"{o}lder continuous on $\overline{\bbC_+}\setminus\{0\}$.

(ii) One shows that $T_\kappa\psi = \psi$ has no non--zero solution for $\textrm{Im}(\kappa) > 0$ (since $\psi$ is then exponentially decreasing and so in $L^2$ violating self--adjointness) and then also for $\textrm{Im}(\kappa) = 0, \kappa \ne 0$ since one can use Kato's result mentioned in Remark 4 after Theorem \ref{T12.1}.  In this analysis, Ikebe shows that if $\kappa \in \bbR\setminus\{0\}$ and $\psi$ solves $T_\kappa\psi=\psi$, then $\varphi\equiv\psi \in L^2(\bbR^3)$ obeys
\begin{equation}\label{15.13}
  \int_{|k|=\kappa} |\hat{\varphi}(k)|^2 \, d\omega = 0
\end{equation}
suitably interpreted.  This result, also found by Povzner, is important as we'll see later.

(iii) By Fredholm theory, since $T_\kappa\psi=\psi$ has no solutions, $\bdone-T_\kappa$ is invertible.  One defines $\varphi(\cdot,\mathbf{k})$ to be $(\bdone-T_{|k|})^{-1}\varphi_0(\cdot,\mathbf{k})$ with $\varphi_0$ given by \eqref{15.3} ($\varphi_0 \notin B$ since it doesn't vanish at infinity but if $\eta = \varphi-\varphi_0$, then $\varphi=\varphi_0+T_{|k|}\varphi \iff \eta=T_{|k|}\varphi+T_{|k|}\eta$.  Note that $\eta$ and $T_{|k|}\varphi_0$ are in $B$).  In this way, one gets solutions of the Lippmann--Schwinger equation.

(iv) One also solves $G=G_0+T_\kappa G$ (where $G_0$ is the free Green's function \eqref{15.10A}) and uses this plus Stone's theorem to verify the expansion \eqref{15.7}

(v) By following arguments of Gell'Mann--Goldberger \cite{GG}, one proves \eqref{15.9} where $\Omega^+$ is an abelian limit.  By the results of Cook--Hack--Kuroda, this abelian limit is equal to the ordinary limit.

(vi) \eqref{15.7} solves problem (B) and \eqref{15.9} solves problem (A) as noted above.

(vii) There is a gap in \cite{Ikebe1} found and filled in Simon \cite{SimonThesis} and also filled by Ikebe \cite{Ikebe2}.

We should briefly mention two variants of Ikebe's work.  First, Thoe \cite{ThoeIkebe} extended the result to $\bbR^\nu$ for general $\nu$.  Secondly, for Rollnik potentials (any $V$ obeying \eqref{15.1} for $\alpha>2$ is in $L^{3/2}$ and so Rollnik but Rollnik allows $L^{3/2}$ local singularities), following Rollnik \cite{Rollnik} and Grossman--Wu \cite{GW}, one can rewrite the Lippmann--Schwinger equation in an equivalent form:
\begin{align}\label{15.14}
  \xi(x) &= \xi_0(x) - \frac{1}{4\pi} \int |V(x)|^{1/2}\frac{e^{ik|x-y|}}{|x-y|} V^{1/2}(y) \xi(y)  \nonumber \\
         &\equiv \xi_0(x)+(W_{|k|}\xi)(x)
\end{align}
where $V^{1/2}(y) = |V(y)|^{1/2} \textrm{sgn}(V(y))$ and $\xi(x) = |V(x)|^{1/2}\varphi(x)$.  The point is that the integral kernel in \eqref{15.14} is Hilbert--Schmidt for $\textrm{Im}(k) \ge 0$ if $V$ is Rollnik.  This was used by Simon \cite{SimonThesis} to carry over Ikebe's arguments.  One big difference is that there is no Kato argument to eliminate solutions of the homogeneous equations.  But by Fredholm theory, in any event, the set of points where $1-W_{|k|}$ is not invertible is the set of zeros of a function analytic on $\bbC_+$ and continuous on its closure, so a subset of $\bbR$ with real Lebesgue measure zero.  This allows a proof of completeness but not a solution of problem (B).  We'll say more about this issue below.  We note that this factorization idea is used in several of the approaches to the Agmon--Kato--Kuroda theory and, in particular, an option in the work of Kato and Kuroda.  We'll not discuss this further.

Subsequent to Ikebe solving problem (B) if $\alpha > 2$, the search for the general $\alpha > 1$ result was solved in stages: J\"{a}ger \cite{Jager} did it for $\alpha > 3/2$, Rejto \cite{RejtoNSC1} for $\alpha > 4/3$, Kato \cite{Kato1969KK} using Kato--Kuroda theory did $\alpha > 5/4$ as we've seen, Rejto \cite{RejtoNSC2} did $\alpha > 6/5$ and finally Agmon \cite{AgmonSpec} (and shortly afterwards, independently Saito \cite{Saito}) handled $\alpha > 1$.  As we'll explain using one simple idea from Agmon, Kuroda and Rejto could extend their methods to handle $\alpha > 1$.  Howland \cite{HowlandNSC} had earlier work on this problem and Schechter \cite{SchechterKK} used Kato--Kuroda theory to study higher order elliptic operators (as we'll see Agmon, H\"{o}rmander and Kuroda also did).

In two papers \cite{KKRocky, KKSimple}, Kato and Kuroda developed what they called an abstract theory of scattering.  As Kuroda told me ``it was too abstract to become popular'' (blaming himself for this). In recognition of the history, Reed--Simon dubbed the basic result for $\alpha > 1$ the Agmon--Kato--Kuroda Theorem but it is Agmon's approach that has stuck around.  And this is due not only to the abstraction but also to the elegance and simplicity of Agmon's approach and its flexibility.  Moreover, two early, widely--used monograph presentations (Reed--Simon \cite[Section XIII.8]{RS4} and H\"{o}rmander \cite{Horm2, Horm4}) exposed the Agmon approach.  All this said, while Agmon's technicalities are distinct from Kato--Kuroda, the underlying conceptual framework is similar.  We will describe this scheme using Agmon's approach to explicitly implement the steps.

Agmon uses the spaces $L^2_\beta(\bbR^\nu)$ defined by
\begin{equation}\label{15.14A}
  \norm{\varphi}_\beta^2 = \int (1+|x|^2)^\beta |\varphi(x)|^2 \, d^\nu x < \infty
\end{equation}
These are Hilbert spaces.  One suppresses the natural duality of Hilbert spaces and associates the dual of $L^2_\beta$ with $L^2_{-\beta}$ so that $\psi \in L^2_{-\beta}$ is associated with the linear functional $\varphi \mapsto \int \overline{\psi(x)} \varphi(x) \, d^\nu x$.  Here are the basic facts about Fourier transform on $L^2_\beta$ that we'll need.  For proofs, see \cite[Section IX.9]{RS2}; basically, one proves things for $\nu = 1$ and uses spherical coordinates for the other variables.  We return to using $f \mapsto \hat{f}$ for the Fourier transform.

(1) Let $\beta > 1/2$.  There is for each $\lambda \in (0,\infty)$, a bounded map, $T_\lambda:L^2_\beta(\bbR^\nu) \to L^2(S^{\nu-1},d\omega)$ (where $d\omega$ is unnormalized measure on the unit sphere in $\bbR^\nu$), so that if $f \in \calS(\bbR^\nu)$, then
\begin{equation}\label{15.14B}
  (T_\lambda f)(\omega) = \widehat{f}(\lambda\omega)
\end{equation}

(2) $T_\lambda$ is norm H\"{o}lder continuous in $\lambda$ of order $\beta-1/2$ if $1/2<\beta<3/2$.

(3) Fix $\beta > 1/2$.  As maps of $L^2_\beta$ to $L^2_{-\beta}$, $(-\Delta-\kappa^2)^{-1}$ defined initially for $\textrm{Im}\kappa > 0$ has a continuous extension to $\kappa \in \bbR\setminus\{0\}$.

(4) If $\varphi \in L^2_\beta,\, \beta>1/2$ and $\kappa>0$, then
\begin{equation}\label{15.15}
  \lim_{\epsilon\downarrow 0} \textrm{Im} \jap{\varphi,(-\Delta-(\kappa^2+i\epsilon))^{-1}\varphi} = \frac{\pi \kappa^{\nu-2}}{2} \norm{T_\kappa\varphi}_2^2
\end{equation}
(This is just a version of $\lim_{\epsilon\downarrow 0} \frac{1}{x-i\epsilon}=\calP\left(\frac{1}{x}\right)+i\pi \delta(x)$).

(5) Let $\beta>1/2$.  Fix $\kappa>0$ and suppose that $\varphi \in L^2_\beta$ with $T_\kappa\varphi=0$.  Define $Q_\kappa\varphi$ by
\begin{equation}\label{15.16}
  \widehat{Q_\kappa\varphi}(\mathbf{k}) = (k^2-\kappa^2)^{-1}\hat{\varphi}(\mathbf{k})
\end{equation}
Then for each $\epsilon > 0$, $Q_\kappa\varphi \in L^2_{\beta-1-\epsilon}$ and
\begin{equation}\label{15.17}
  \norm{Q_\kappa\varphi}_{\beta-1-\epsilon} \le C_{\epsilon,\kappa,\nu,\beta}\norm{\varphi}_\beta
\end{equation}
where $C$ depends continuously on its parameters in the region $\epsilon, \kappa > 0, \, \beta>1/2$.  The point here is that without $T_\kappa\varphi=0$, we can define the limit as $\epsilon\downarrow 0$ of $(k^2-\kappa^2-i\epsilon)^{-1}\hat{\varphi}(\mathbf{k})$ which for $\varphi\in L^2_\beta$ with $\beta >1/2$ lies in $L^2_{-\beta}$ but we can never get better than $L^2_{-1/2}$.  When $T_\kappa\varphi=0$, by having $\beta$ large we can get $\varphi$ into a suitable $L^2_{\gamma}$ and, in particular, into $L^2$.

We can now describe the basic strategy of solving problems (A) and (B) for any $\alpha > 1$.

\textbf{Step 1.} Find a triple of spaces $X \subset L^2(\bbR^\nu,d^\nu x) \subset X^*$ where $X$ is a dense subspace of $L^2$ and which is a Banach space in a norm, $\norm{\cdot}_X$, so that for $\psi \in X$, we have that $\norm{\psi}_2 \le \norm{\psi}_X$.  Any $\varphi \in L^2$ acts as a bounded linear functional on $X$ via $\ell_\varphi(\psi) = \jap{\bar{\varphi},\psi}$ so $L^2 \subset X^*$ which can be shown to be dense.  Note that when $\varphi \in L^2$, we have that $\norm{\varphi}_{X^*} \le \norm{\varphi}_2$.  In the Agmon approach, $X = L^2_\beta(\bbR^\nu)$ for some $\beta > 1/2$ and $X^*=L^2_{-\beta}(\bbR^\nu)$.  Let $H_0$ be a self-adjoint operator which in the Agmon setup is a constant coefficient elliptic partial differential operator although we'll mainly be interested in the case $H_0=-\Delta$. By the norm inequalities, for any $z \in \bbC\setminus [E_0,\infty)$ (where $E_0$ is the bottom of the spectrum of $H_0$), $(H_0-z)^{-1}$ is bounded from $X$ to $X^*$.  One must pick $X$ so that $(H_0-z)^{-1}$, as bounded maps from $X$ to $X^*$ has a continuous extension to $[E_0,\infty)$ with a finite set of points removed.  The extension is from above or below the real axis and the two limits need not be equal.  In our case where $E_0=0$, the finite set is only $E_0$.  In the general elliptic case, it is the set of critical points of the defining symbol.  As explained above, in the Agmon setup, where $X=L^2_\beta, \beta > 1/2$, we have the required continuity of the boundary values.  In the Kato--Kuroda theory, $X$ is an abstract space which can be chosen in various ways.

\textbf{Step 2.} Restrict acceptable potentials, $V$, to functions $V: X^* \to X$ or, more generally so that $V(H_0-E_0+1)^{-1}$ is bounded from $X$ to itself.  In fact, we require this to be a compact operator from $X$ to itself.  In the Agmon $L^2_\beta$ case, for $H_0=-\Delta$, one needs that $(1+|x|^2)^\beta V(-\Delta+1)^{-1}$ is compact as an operator on $L^2$.  In particular, if \eqref{15.1} holds, we need that $\alpha > 2\beta$, so if $\alpha > 1$ we can pick $\beta$ with $1/2 < \beta <\alpha/2$.  Thus, the results below will solve problems (A) and (B) when $\alpha > 1$.

\textbf{Step 3.} For simplicity, we henceforth suppose $E_0=0$ and that $E_0$ is the only critical point as happens for the Schr\"{o}dinger case.  Under these assumptions, $B(z) = -(H_0-\kappa^2)^{-1}V$ for $z=\kappa^2; \, \textrm{Im}\kappa \ge 0, \kappa \ne 0$ is compact operator on $X^*$, continuous in $\kappa$ and analytic for $\kappa \in \bbC_+$.  By a version of the analytic Fredholm theorem (see \cite[Theorem VI.14]{RS1}), there is a set $\calE \subset (0,\infty)$, so that $\calE$ is a closed set (i.e. its only limit points are in $\calE$ or are $0$ or $\infty$) of (real) Lebesgue measure $0$ and so that if $z \notin \calE$, then $(\bdone-B(z))^{-1}$ exists and is continuous in $z$ there.  One proves that $(H-z)^{-1} = (1-B(z))^{-1}(H_0-z)^{-1}$ originally for $\textrm{Im} z \ne 0$ and then as maps from $X$ to $X^*$ for $z \notin \calE$.

\textbf{Step 4.}  This suffices to get existence and completeness of wave operators.  Kato--Kuroda \cite{KKRocky, KKSimple} have arguments to get this.  In his original announcement, Agmon \cite{AgmonAnon} didn't mention scattering.  If one can decompose $V=A^*B$ so that $A, B: X^* \to L^2$ (perhaps after multiplication by $(H_0+1)^{-1/2}$), then one can show that $A, B$ are locally smooth for both $H$ and $H_0$ on $(0,\infty)\setminus\calE$ and so by Theorem \ref{T14.10}, one gets existence and completeness (ideas due to Lavine \cite{Lav3, Lav5}).  In later publications, Agmon and H\"{o}rmander have other ways of proving existence and completeness by exploiting a radiation condition.

\textbf{Step 5.} In general, from this, one gets purely a.c. on $(0,\infty)\setminus\calE$ so any singular spectrum on $(0,\infty)$ lies in $\calE$.

\textbf{Step 6.} Suppose we show that any $\lambda_0 \in \calE$ is an $L^2$ eigenvalue of $H$.  Then $\calE \cup \{0,\infty\}$ is a countable closed subset of $\bbR$ which cannot support a singular continuous measure.  In this way, one solves problem (B).

\textbf{Step 7.} If $\varphi \in L^2_{-\beta}$ and $B(\lambda_0+i0)\varphi=\varphi, \, \lambda_0=\kappa^2$, then
\begin{align*}
  0=\textrm{Im}\jap{V\varphi,\varphi} &= \textrm{Im}\jap{V\varphi,(H_0-\lambda_0-i0)^{-1}V\varphi}   \\
                                      &= \frac{\pi \kappa^{\nu-2}}{2} \norm{T_\kappa V\varphi}^2
\end{align*}
so $T_\kappa V\varphi=0$.  Therefore by \eqref{15.17}, $Q_\kappa V\varphi=B(\kappa)\varphi \in L^2_{\alpha-\beta-1-\epsilon}$ for all $\epsilon>0$.  For example, if $\alpha>3/2$, we can pick $\beta > 1/2$ but close to it and $\epsilon$ small so that $\alpha-\beta-1-\epsilon \ge 0$.  Thus $\varphi \in L^2$ and is an eigenfunction.  By invoking Step 6, we see that when $\alpha > 3/2$, we can solve problem (B).  The restriction $\alpha > 5/4$ in Theorem \ref{T15.1} comes from a consideration like this -- what is needed to deduce that $\varphi \in L^2$.

\textbf{Step 8.}  Agmon had the idea of iterating the argument in Step 7!  If we know that $\varphi \in L^2_\gamma$, since $T_\kappa V\varphi = 0$, we have that $\varphi \in L^2_{\alpha+\gamma-1-\epsilon}$, so if $\alpha >1$, we can increase $\gamma$ by an arbitrary amount less than $\alpha-1$.  If $\alpha - 1 > 1/2m$ starting in $L^2_{-\beta}$ with $\beta$ very close to $1/2$, we see by iterating $m$ times that $\varphi$ is an $L^2$ eigenfunction.  In this way, one solves problem (B) for all $\alpha > 1$.

\textbf{Step 9.}  Once one controls the resolvent, one can obtain eigenfunctions via the Lippmann--Schwinger equation.  Knowing that $\calE$ is countable shows the expansion only has a.c. spectrum and point spectrum.

This completes our sketch of the scheme behind the work of Kato--Kuroda and Agmon; see Reed--Simon \cite[Section XIII.8]{RS4} for more details.  After Agmon's argument appeared, various authors realized that the iteration idea in Step 8 could improve their results.  In particular, Kuroda \cite{KurodaAgm1, KurodaAgm2} was able to extend the proof of Theorem \ref{T15.1} to $\alpha > 1$.  He extended this work to fairly general elliptic operators.

The ideas in the Agmon--Kato--Kuroda work have been extended to long range potentials (where \eqref{15.1} holds for suitable $\alpha \in (0,1]$ but we also have $(1+|x|)^{-1-\alpha}$ decay of $\mathbf{\nabla}V$).  One needs to use modified wave operators following Dollard \cite{Dollard}.  There is a vast literature and we will not try to summarize it -- see the books of Derezi\'{n}ski--G\'{e}rard \cite{DerGer} and Yafaev \cite{Yafaev1, Yafaev2}.

The above approach uses the fact that for $L^2_\beta, \, \beta>1/2$, there is a map restricting $\hat{\varphi}$ to the sphere.  One proves this by essentially flattening the sphere.  If we replaced $L^2_\beta$ by $L^p$, we cannot restrict to hyperplanes but remarkably, one can sometimes restrict to curved hypersurfaces like the spheres we needed above.  The associated bounds are known as the Tomas--Stein Theorem (see \cite[Section 6.8]{HA}).  Ionescu--Schlag \cite{IonSch} have developed a theory of scattering and spectral theory under suitable $L^p$ conditions on $V$ using the Tomas--Stein bounds.

%%%%%%%%%%%%%%%%%%%%%%%%%%%%%%%%%%%%%%%%%%%%%%%%%%%%%%%%%%%%%%
\section{Scattering and Spectral Theory, IV: Jensen--Kato Theory} \lb{s16}
%%%%%%%%%%%%%%%%%%%%%%%%%%%%%%%%%%%%%%%%%%%%%%%%%%%%%%%%%%%%%%

This is the last section on ``scattering and spectral'' theory although it involves something closer to diffusion than scattering and the connection to spectral theory is weak.  Still, since it involves large time behavior of $e^{-itH}$, it belongs in this set of ideas.  In any event, we'll discuss a lovely paper of Jensen and Kato \cite{JenKato} involving Schr\"{o}dinger operators, $H=-\Delta+V$, on $\bbR^3$.

One issue that they discuss is the large time behavior of $e^{-itH}$ and its rate of decay.  At first sight, speaking of decay seems puzzling since for $\varphi\in L^2$, we have that $\norm{e^{-itH}\varphi}_2=\norm{\varphi}_2$ has no decay.  But consider the integral kernel when $V=0$ on $\bbR^\nu$
\begin{equation}\label{16.1}
  e^{it\Delta}(x,y) = (4\pi it)^{-\nu/2}e^{i|x-y|^2/4t}
\end{equation}
which shows that
\begin{equation}\label{16.2}
  \sup_{x,y} |e^{it\Delta}(x,y)| = (4\pi |t|)^{-\nu/2}
\end{equation}
so
\begin{equation}\label{16.3}
  \norm{e^{it\Delta}\varphi}_\infty \le (4\pi |t|)^{-\nu/2} \norm{\varphi}_1
\end{equation}
\eqref{16.3} is, in fact, equivalent to \eqref{16.2}.  Since Jensen and Kato use Hilbert space methods, instead of maps from $L^1$ to $L^\infty$, they consider maps between weighted $L^2$ spaces, specifically from $L^2_s$ to $L^2_{-s}$ where $L^2_s$ is given by \eqref{15.14A}.  For example, $\eqref{16.3}$ immediately implies that $\norm{e^{it\Delta}\varphi}_{2,-s} \le C_{\nu,s} t^{-\nu/2}\norm{\varphi}_{2,s}$ so long as $s \ge \nu/2$.

If $H_0=-\Delta$ is replaced by $H=H_0+V$, there is a new issue that arises.  If $H\varphi=E\varphi$ for $\varphi \in L^2$, then $e^{-itH}\varphi=e^{-itE}\varphi$ has no decay in any norm.  Thus one must only try to prove decay of $e^{-itH}P_c(H)$ where $P_c(H)$ (``c'' is for continuous spectrum; if there is no singular continuous spectrum, it is the same as $P_{ac}$) is the projection onto the orthogonal complement of the eigenvectors.  Jensen--Kato don't use $e^{-itH}P_c(H)$ but the equivalent
\begin{equation}\label{16.3A}
  e^{-itH} - \sum_{j=1}^{N}e^{-iE_jt}P_j
\end{equation}
where $\{E_j\}_{j=1}^N$ are the eigenvalues and $P_j$ the projections onto the associated eigenspace $\ker(H-E_j)$.

In the free case, we note that it is easy to see (\cite[Corollary to Theorem XI.14]{RS3}) that if $0 \notin \supp(\widehat{\varphi})$ for $\varphi \in \calS(\bbR^\nu)$, then $\sup_{|x|\le R} |e^{it\Delta}\varphi(x)|$ is $\textrm{O}(t^{-N})$ for all $N$.  That is the diffusive term $t^{-\nu/2}$  is connected to low energies.  A critical realization of Jensen--Kato is that large $t$ asymptotics as maps of $L^2_{s}$ to $L^2_{-s}$ is connected to the behavior of the resolvent $(H-z)^{-1}$ near $z=0$.

For a while now we return to $\nu=3$, the only case considered by Jensen--Kato.  As we'll see, $\nu=3$ is perhaps the simplest case with a rich structure.  Roughly speaking, Jensen--Kato consider $V's$ obeying
\begin{equation}\label{16.3B}
  |V(x)| \le C(1+|x|)^{-\beta}
\end{equation}
They always require $\beta > 2$ and often need $\beta > 3$ or even larger.  In fact, for some of their results, they only need $(1+|x|)^\beta V \in L^{3/2}_{unif}$, but for simplicity we'll only quote results below where the pointwise bound \eqref{16.3B} holds.  Prior to their paper, there was work of Rauch \cite{RauchJK} which motivated them.  He supposed $|V(x)| \le C_1 e^{-C_2|x|}$ and instead of $L^2$--operator norms of $(1+|x|)^{-s}e^{-itH}P_c(H)(1+|x|)^{-s}$, he considered norms $e^{-\epsilon |x|}e^{-itH}P_c(H)e^{-\epsilon |x|}$.   He found for all but a discrete set of $\xi \in \bbR$, with $H(\xi) = -\Delta+\xi V$, one has $t^{-3/2}$ decay for the relevant norms of $e^{-itH(\xi)}$ and, for a discrete set of $\xi$'s, $t^{-1/2}$ decay.  Jensen--Kato extended this result for $L^2_s$ to $L^2_{-s}$ with $s > 5/2$ and $\beta > 3$.  Several years earlier, Yafaev \cite{YafaevVirtual}, in connection with his work on the Efimov effect \cite{YafaevEfimov}, had studied low energy behavior of the resolvent (but not high energy asymptotics of the unitary group) in the case of a zero energy resonance (case (1) in the language of Jensen--Kato).

It is natural to restrict at least to $\beta > 2$ for small energy behavior.  The Birman--Schwinger kernel \cite[Section 7.9]{OT}, $|V(x)|^{1/2}V(y)^{1/2}/4\pi |x-y|$, is Hilbert--Schmidt if \eqref{16.3B} holds for $\beta>2$ and, in general may not even be a bounded operator if $\beta < 2$ (and if $\beta=2$, can be bounded but not compact).  Thus, $\beta > 2$ implies that $-\Delta+V$ has only finitely many negative eigenvalues, each of finite multiplicity.

As we've mentioned, the key input for the Jensen--Kato large time results is an analysis of the resolvent, $R(z)=(H-z)^{-1}$ for $z$ near zero.  The free resolvent $R_0(z)=(H_0-z)^{-1}$ has integral kernel
\begin{equation}\label{16.4}
  G_0(x,y;z) = \frac{e^{i\kappa|x-y|}}{4\pi|x-y|}
\end{equation}
where $\kappa$ obeys $\kappa^2=z$ with $\textrm{Im}(\kappa)>0$ for $z \in \bbC\setminus [0,\infty)$ (with obvious limits if $z$ approaches $\bbR$ from either $\bbC_+$ or $\bbC_-$).  It is only in dimension 3 (and 1) that $G_0$ is so simple; in other dimensions, it is a more complicated Bessel function.  For $z \in \bbC\setminus [0,\infty)$, one has that
\begin{equation}\label{16.5}
  R(z)=(1+R_0(z)V)^{-1}R_0(z)
\end{equation}

Following Agmon and Kuroda (see Section \ref{s15}), Jensen--Kato use the weighted Sobolev spaces, $H^{m,s}(\bbR^3)$ of those $\varphi$ which obey
\begin{equation}\label{16.6}
  \norm{\varphi}_{m,s}=\norm{(1+|x|^2)^{s/2}(1-\Delta)^{m/2}\varphi}_2 < \infty
\end{equation}
For example, we can take the completion of $\calS(\bbR^3)$ in this norm or, since $(1+|x|^2)^{s/2}(1-\Delta)^{m/2}$ is a map of tempered distribution to themselves, we can take those tempered distributions for which the quantity in the norm on the right of \eqref{16.6} is in $L^2$.

Let $K_0$ be the operator with integral kernel $(4\pi|x-y|)^{-1}$, i.e. $G_0(x,y;0)$.  Jensen--Kato prove that if $V$ obeys \eqref{16.3B} with $\beta > 2$, then $K_0V$ is a compact operator on $L^2_{-s}$ if $1/2<s<\beta-1/2$, indeed it is compact on $H^{1,s}$.  It is also true that extended from $\kappa \in \bbC_+$ to it $\kappa \in \bbC_+\cup\bbR$, $VR_0(\kappa^2)$ is H\"{o}lder continuous (and compact).   While Jensen--Kato don't prove it that way, we note that this follows from the generalized Stein--Weiss inequalities \cite[Theorem 6.2.5]{HA}.

Thus, to understand the small $z$ behavior of $R(z)$, we need to know about $(1+K_0V)^{-1}$.  By compactness, this inverse exists if and only if
\begin{equation}\label{16.7}
  (1+K_0V)\varphi=0
\end{equation}
has no non-zero solutions, $\varphi \in H^{1,-s}$.  If $\varphi$ obeys \eqref{16.7}, it is a distributional solution of $(-\Delta+V)\varphi=0$.  Let $\calM$ be the set of all solutions of \eqref{16.7} in $H^{1,-s}$; Jensen--Kato prove that it is independent of which $s$ is chosen in $(1/2,\beta-1/2)$. By compactness $\dim\calM < \infty$.  It is important to know if $\varphi \in L^2$.  \eqref{16.7} says that
\begin{equation}\label{16.8}
  \varphi(x)=-\frac{1}{4\pi}\int \frac{1}{|x-y|} V(y)\varphi(y)\,d^3y
\end{equation}
so that
\begin{equation}\label{16.9}
  \varphi(x) = -\frac{1}{4\pi|x|}\int V(y)\varphi(y) \, d^3y+ \textrm{o}\left(\frac{1}{|x|}\right)
\end{equation}
Thus, if $\int V(y)\varphi(y)\,d^3y \ne 0$, then $\varphi \notin L^2$.  One can show that if $\int V(y)\varphi(y)\,d^3y = 0$, then $\varphi \in L^2$.  Thus, in $\calM$, the set of $L^2$ solutions is either all of $\calM$ or a space of codimension $1$.  If $\calM$ has non--$L^2$--solutions, we say that there is a \emph{zero energy resonance}.  Jensen--Kato thus consider four cases:

(0) (\textbf{regular case}) $\calM=\{0\}$ so $(1+K_0V)^{-1}$ exists.  Since $K_0V$ is compact, the set of $\xi \in \bbR$ for which $\xi V$ is not regular is a discrete set.

(1) (\textbf{pure resonant case}) $\calM \ne \{0\}$ but there are no $L^2$ functions in $\calM$.  This implies that $\dim\calM = 1$.

(2) (\textbf{pure eigenvalue case}) $\calM \ne \{0\}$ and $\calM \subset L^2$.  Thus $0$ is an eigenvalue but there is no resonance.

(3) (\textbf{mixed case}) $\calM \ne \{0\}$ and $\calM$ contains both $L^2$ and non--$L^2$ functions.  Then $\dim\calM\ge 2$ and the set of $L^2$ solutions has codimension 1.

Later, we'll see that in a sense, case (1) is generic among the singular cases.  We'll see similar qualitative behavior in the three singular cases but the detailed expressions for coefficients depends on the case.

Jensen--Kato start by noting the expansion in $\kappa = \sqrt{z}$ when $V=0$.  Given \eqref{16.4}, we see that
\begin{equation}\label{16.9A}
  R_0(\kappa^2) = \sum_{j=0}^{\infty} (i\kappa)^j K_j
\end{equation}
where $K_j$ has the integral kernel
\begin{equation}\label{16.10}
  K_j(x,y) = |x-y|^{j-1}/4\pi j!
\end{equation}
Then, for $j \ge 1$, $K_j$ is bounded from $H^{-1,s}$ to $H^{1,-s}$ if and only if $s > j+1/2$.  That means if we fix $s$, we have an asymptotic series only to any order $J < s-1/2$.  Since $V$ obeying \eqref{16.3B} maps $L^2_{-s}$ to $L^2_{s}$ if and only if $s < \beta/2$, we see that for fixed $\beta$, we can only expect to get an expansion including $\kappa^j$ terms if $j <\tfrac{1}{2}(\beta-1)$.  This explains the conditions on $\beta$ in the theorems below.  Jensen--Kato prove, with explicit formulae for $B^{(0)}_j,\,j=0,1$,

\begin{theorem} [Jensen--Kato \cite{JenKato}] \lb{T16.1} Assume that $V$ is regular at $\kappa=0$, $\beta > 3$ and $s > 3/2$.  Then for explicit operators $B^{(0)}_0 \ne 0$ and $B^{(0)}_1$ from $L^2_s$ to $L^2_{-s}$ as operators between those spaces and $\textrm{Im}\kappa \ge 0$
\begin{equation}\label{16.11}
  R(\kappa^2)= B^{(0)}_0 + i\kappa B^{(0)}_1+\textrm{o}(\kappa)
\end{equation}
If $\beta > 5$ and $s > 5/2$, then $\textrm{o}(\kappa)$ can be replaced by $\textrm{O}(\kappa^2)$.
\end{theorem}

They also prove (with explicit formula for $B^{(k)}_j$) that

\begin{theorem} [Jensen--Kato] \lb{T16.2} Assume that $V$ is not regular at $\kappa=0$, $\beta > 5$ and $s > 5/2$.  Then for explicit operators $B^{(k)}_{-2}$ and $B^{(k)}_{-1}, \, k=1,2,3$ from $L^2_s$ to $L^2_{-s}$ as operators between those spaces and $\textrm{Im}\kappa \ge 0$, one has that
\begin{equation}\label{16.12}
  R(\kappa^2)= -\kappa^{-2} B^{(k)}_{-2} - i\kappa^{-1} B^{(k)}_{-1}+\textrm{O}(1)
\end{equation}
if the singular point is of type $k$.  Moreover, if $k=1$, $B^{(1)}_{-2} = 0,\,B^{(1)}_{-1}\ne 0$ and if $k=2, 3$, then $B_{-2}^{(k)} \ne 0$.
\end{theorem}

\begin{remarks} 1. The explicit formulae have $B^{(k)}_j$ of finite rank for $k=-2,-1$.  If $\beta$ and $s$ are large enough, there should be asymptotic series of any prescribed order and the coefficients are all finite rank \cite{Murata1}, \cite[Prop. 7.1]{JenNen}.

2.  Rauch \cite{RauchJK} says that $B^{(k)}_{-1} \ne 0$ for all $k$ but Jensen--Kato have an explicit example where $B^{(2)}_{-1} = 0$.

3.  Using ideas from Klaus--Simon \cite{KlausSiThr} (discussed further below), one can prove not only that regular $V$'s are generic but among the irregular $V$'s, type (1) is generic and among those not of type (1), type (3) is generic.  For example, one can prove that for any $\beta > 5$, if $X_\beta=\{V \,|\, \norm{V}_\beta = \sup_x |(1+|x|)^\beta |V(x)| < \infty\}$, then the regular $V$'s are a dense open set and, in the set, $\wti{X}_\beta$ of not regular $V$'s (which is closed and so a complete metric space), the set of type (1) $V$'s is a dense open set.  Klaus--Simon only discuss $V \in C^\infty_0(\bbR^3)$ but that is for simplicity and their ideas work in this broader context.  These genericity results are not true for spherically symmetric $V$'s.  In that case. the space $\calM$, if non-zero, generically has a single angular momentum, $\ell$, and always has a finite number of them.  For each $\ell$, the set of $V$'s with only that $\ell$ is a relatively open subset of the closed subset of spherically symmetric elements of $\wti{X}_\beta$, so none is generic in the singular $V$'s.  $\ell=0$ is type (1), $\ell \ne 0$ is of type (2).  Cases of more one $\ell$ are of type (3) or (1) depending only on whether one of the $\ell$ values is $0$.
\end{remarks}

Jensen--Kato also studied low energy asymptotics of the $S$--matrix, and, importantly for the study of asymptotics of $e^{-itH}$, the low energy behavior of the derivative of the spectral measure
\begin{equation}\label{6.13}
  \frac{d}{d\lambda} P_{(-\infty,\lambda)}(H) \equiv P'_H(\lambda)
\end{equation}
A little thought about Stone's formula shows that if $R(z)$ has a limit $R(\lambda+i0)$ uniformly for $\lambda \in (a,b) \subset \bbR$, then
\begin{equation}\label{16.14}
  P'_H(\lambda) = \pi^{-1} \textrm{Im}\,R(\lambda+i0)
\end{equation}
where, for an operator, $A$, one writes $\textrm{Im}\,A = (A-A^*)/2i$.

Since $z=\kappa(z)^2$ with $\textrm{Im}\kappa > 0$ has that $\kappa(\bar{z})=-\overline{\kappa(z)}$, we see that by \eqref{16.14} that if
\begin{equation}\label{16.15}
  R(\kappa^2) = \sum_{j=-2}^{J} (i\kappa)^jQ_j+ \textrm{o}(|\kappa|^J)
\end{equation}
then $Q_j^*=Q_j$ and so, with $L=\left[\frac{J-1}{2}\right]$,
\begin{equation}\label{16.16}
  P'(\lambda)=\pi^{-1}\sum_{\ell=-1}^{L} (-1)^\ell \sqrt{\lambda}^{2\ell+1}Q_{2\ell+1}+\textrm{o}(\sqrt{\lambda}^J)
\end{equation}

In particular, if \eqref{16.11} holds (with an $\textrm{O}(\kappa^2)$ term), then
\begin{equation}\label{16.17}
  P'(\lambda) = \pi^{-1} B^{(0)}_1 \lambda^{1/2} + \textrm{O}(\lambda)
\end{equation}
and if \eqref{16.12} holds, then
\begin{equation}\label{16.18}
  P'(\lambda) = \pi^{-1} B^{(k)}_{-1} \lambda^{-1/2} + \textrm{O}(1)
\end{equation}
In this way Jensen--Kato control $P'(\lambda)$ for small $\lambda$.

They also find a large $\lambda$ result.  They prove that for $k=1,2,\dots$ and $s > k+1/2, \, \beta > 2k+1$, then as maps from $L^2_s$ to $L^2_{-s}$, one has that
\begin{equation} \lb{16.18A}
  \left(\frac{d}{d\lambda}\right)^kP'(\lambda) = \textrm{O}(\lambda^{-(k+1)/2})
\end{equation}
as $\lambda \to \infty$.

With these in hand they can estimate
\begin{equation}\label{16.19}
  e^{-itH}P_c(H) = \int_{0}^{\infty} e^{-it\lambda}P'(\lambda)\,d\lambda
\end{equation}
The large $\lambda$ contribution as $t \to \infty$ can be controlled using repeated integration by parts and the decay estimates in \eqref{16.18A} on derivatives of $P'(\lambda)$.  One sees that the integral on the right side of \eqref{16.19} is dominated by the small $\lambda$ contributions.  Using the fact that the Fourier transform of $\lambda^{(j-1)/2} \chi_{(0,\infty)}(\lambda)$ is the distribution $(-it)^{-(j+1)/2}$ regularized at $t=0$, one sees that

\begin{theorem} [Jensen--Kato \cite{JenKato}] \lb{T16.3} Let $V$ obey \eqref{16.3B} with $\beta>3$, $s>5/2$.  Suppose that $V$ is regular at zero energy.  As a map from $L^2_s$ to $L^2_{-s}$, we have that as $t \to \infty$, \eqref{16.3A} is asymptotic in norm to
\begin{equation}\label{16.20}
  -(4\pi i)^{-1}B^{(0)}_1 t^{-3/2} + \textrm{o}(t^{-3/2})
\end{equation}
\end{theorem}

\begin{theorem} [Jensen--Kato \cite{JenKato}] \lb{T16.4} Let $V$ obey \eqref{16.3B} with $\beta>3$, $s>5/2$.  Suppose that $V$ has an exceptional point of type (1) at zero energy.  Then, for a suitably normalized solution $\psi \in \calM$, we have that as a map from $L^2_s$ to $L^2_{-s}$, as $t \to \infty$, \eqref{16.3A} is asymptotic in norm to
\begin{equation}\label{16.20}
  (\pi i)^{1/2} t^{-1/2} \jap{\psi,\cdot}\psi + \textrm{o}(t^{-1/2})
\end{equation}
\end{theorem}

\begin{remark} $\psi$ is normalized by $\int V(x)\psi(x)\,d^3x=\sqrt{4\pi}$
\end{remark}

That completes our discussion of the Jensen--Kato paper.  One obvious question left open by this work is what happens when $\nu \ne 3$.  This was answered for $\nu \ge 5$ by Jensen \cite{JensenGE5} and for $\nu=4$ by Jensen \cite{Jensen4} and Murata \cite{Murata1} (who also had results for more general elliptic operators); see also Albeverio et al \cite{Alb1, Alb2} .  The case $\nu=2$ with $\int V(x) \, d^2x \ne 0$ was treated by Boll\'{e} et al \cite{Bolle2D} and the general case by Jensen--Nenciu \cite{JenNen}.  For $\nu=1$ with exponentially decaying potentials, the behavior was analyzed by Boll\'{e} et al \cite{Bolle1DA, Bolle1DB} and, in general, by Jensen--Nenciu \cite{JenNen}.  Ito--Jensen \cite{JenIt} discuss Jacobi matrices (discrete $\nu=1$).

For $\nu \ge 5$, an important observation is that there are no resonances at zero energy.  This is because functions $\varphi \in \calM$ obey
\begin{equation}\label{16.21}
  \varphi(x) = - c_\nu \int |x-y|^{-(\nu-2)}V(y)\varphi(y)\,d^\nu y
\end{equation}
and so are $\textrm{O}(|x|^{-(\nu-2)})$ at infinity and thus are in $L^2$ if $\nu \ge 5$.

There is a difference between odd $\nu$ and even $\nu$, so we begin with $\nu \ge 5$, odd.  In that case, for there to be $t^{-\nu/2}$ decay for $e^{-itH_0}$ from $L^2_s$ to $L^2_{-s}$, we need that $P_0'(\lambda) \sim \lambda^{-(\nu-2)/2}$ for small $\lambda$.  At first sight, this seems surprising since $R_0(\kappa^2)$ has $\textrm{O}(1)$ terms, so we might guess also $\textrm{O}(\kappa)$ terms.  In fact, only even powers of $\kappa$ occur until $\kappa^{\nu-2}$.  This can be seen by analyzing the integral kernel for $G_0(x,y;\kappa^2)$ which is a modified Bessel function of the second kind (see \cite[discussion following (6.9.35)]{RA}) which is how Jensen \cite{JensenGE5} does it or by looking at \eqref{15.15}.  (It is an interesting exercise to write $T_\kappa\varphi$ in terms of Taylor coefficients of  $\widehat{\varphi}$ at $k=0$ and so recover the kernels $K_j$ of \eqref{16.10} for $j$ odd.)

If $0$ is not an eigenvalue of $H$, it is easy to prove that as maps from $L^2_s$ to $L^2_{-s}$, for suitable $s$ and $\beta$, one has an asymptotic series for $R(\kappa^2)$ whose first odd term is $(i\kappa)^{\nu-2}$ and then that $e^{-itH}P_c(H)$ as a map between suitable $L^2_r$ spaces is $\textrm{O}(t^{-\nu/2})$.  If zero is an eigenvalue and $\beta$ and $s$ are large enough, one can have any of $\textrm{O}(t^{-\tfrac{1}{2}\nu + 2})$, $\textrm{O}(t^{-\tfrac{1}{2}\nu + 1})$ or $\textrm{O}(t^{-\tfrac{1}{2}\nu})$ and all three possibilities can occur.

Looking at the odd $\nu$ situation, it seems surprising that one can have $\textrm{O}(t^{-m})$ for $m \in \bbZ$ but it happens when $\nu$ is even for the free case.  In fact, if there were an asymptotic series in powers of $\kappa$, the imaginary part cannot have even powers of $\kappa$ as we've seen.  The point is that in even dimensions the Bessel functions have log terms and for $m \in \bbZ$, we have that $\textrm{Im}\left[\lambda^m\log(-\lambda+i0)\right]=\pi \lambda^m$.  Because of this all the above odd $\nu \ge 5$ results extend to even $\nu \ge 5$.

For $\nu=4$, there can be a resonance and/or bound state as when $\nu =3$ so there are three types of singular points.  In the regular case, the leading term is $\textrm{O}(t^{-2})$, but when $\beta$ and $s$ are sufficiently large, the next term is $\textrm{O}(t^{-3}\log(t))$ (unlike $\nu=3$ where the term after $\textrm{O}(t^{-3/2})$ is $\textrm{O}(t^{-5/2})$).  If there is a singular point with only bound states, the leading term is $\textrm{O}(t^{-1})$ but when there are resonances there is only a bound by $\textrm{O}(1/\log t)$.

Jensen--Nenciu \cite{JenNen} analyze $\nu =1,2$ with a new method that also works in general dimension.  These dimensions are special in that there is a zero energy resonance for $H_0=-\Delta$ -- this is especially clear in the coupling constant threshold point of view discussed soon.  For $\nu=1$, if $\int_{-\infty}^{\infty} |x|\,|V(x)|\,dx < \infty$, it is known that every non--zero solution of $-\varphi''+V\varphi=0$ is either asymptotic to $a_\pm x+\textrm{o}(x)$ as $x \to \pm \infty$ with $a_\pm \ne 0$ or is asymptotic to $b_\pm + \mbox{o}(1)$ with $b_\pm \ne 0$ (in which case we say that $a_\pm=0$).  Thus, $0$ is never an eigenvalue and is a resonance if and only if there is $\varphi$ with $a_+=a_-=0$. For suitable $s$ and $\beta$ in the right norm $e^{-itH}$ is $\textrm{O}(t^{-3/2})$ in the regular case, while in the resonance case, one can have $\textrm{O}(t^{-1/2})$ behavior.  $\nu=2$ is very involved.  The resonant subspace can be of dimension up to $3$ and the small $\kappa$ expansion is jointly in $\kappa$ and $\log(\kappa)$.

Next, we want to mention the connection between resonances and coupling constant behavior.  Simon \cite{SimonThresh} considered $A+\xi B$ for general self--adjoint operators, $A$ and $B$, where $A \ge 0$, $|B|^{1/2}(A+1)^{-1/2}$ is compact and $0 \in \sigma_{ess}(A)$ so that $N(\xi) \equiv \dim\ran\,P_{(-\infty,0)}(A+\xi B) < \infty$ for all $\xi \in (0,\Xi)$.  Then $N$ is increasing and there is a discrete set $0 \le \xi_1 \le \xi_2 \le \dots$ so that $N(\xi) \ge j \iff \xi > \xi_j$.  That is, the $\xi_j$ are \emph{coupling constant thresholds}, where, depending on whether you think of $\xi$ as increasing or decreasing, new eigenvalues are born out of $0$ or old ones are absorbed.  Simon proves that
\begin{equation*}
  \lim_{\xi \downarrow \xi_j} -\frac{E_j(\xi)}{\xi-\xi_j}
\end{equation*}
always exists and is non--zero if and only if $0$ is an eigenvalue of $A+\xi_j B$ (with a more complicated statement if $\xi_k = \xi_j$ for some $k \ne j$).

This links up to the Kato--Jensen work in that the $\xi$'s where $\calM(H_0+\xi V) \ne \{0\}$ are exactly the coupling constant thresholds.  If there are eigenvalues $E_j(\xi)$ for $\xi>\xi_j$ with $E_j(\xi_j)=0$ and $E_j(\xi) \le -c(\xi-\xi_j);\, c>0$, then $H_0+\xi_j V$ has a zero eigenvalue.  If instead $E_j = \textrm{o}(\xi-\xi_j)$, then there is a resonance.  For Schr\"{o}dinger operators, this was explored by Rauch \cite{RauchThreshold} and by Klaus--Simon \cite{KlausSiThr}.  In particular, Klaus--Simon show for sufficiently large $\beta$, $-E_j(\xi) = \textrm{O}((\xi-\xi_j)^2)$ and, in that case, if $V$ has compact support, $E_j(\xi)$ is analytic at $\xi=\xi_j$.  In the bound state case, they prove that $E_j(\xi)$ is not analytic at $\xi_j$ (as we'll discuss below, typically, $E_j$ has a non-zero imaginary part for $\xi < \xi_j$ and real).  These ideas also explain why if $\nu=1$ or $\nu=2$, $H_0$ has a resonance at zero energy since it is known (Simon \cite{SimonWeakBound}) that if $V$ obeys \eqref{16.3B} for $\nu =1,2$ and $\beta > 3$ and $\int V(x) \, d^\nu x \le 0$, then for all $\xi >0$, $H_0+\xi V$ has a bound state.

Simon \cite{SimonBrownian, SimonLargeTime} discusses large time behavior of the $L^\infty$ to $L^\infty$ norm of $e^{-tH}$ (note $-t$, not $-it$) when there is and when there is not a zero energy resonance.

If there is a zero energy eigenvalue at a threshold $\xi_j$, then it turns into a negative eigenvalue for $\xi > \xi_j$.  If $\xi < \xi_j$, on the basis of the discussion in Section \ref{s4}, one expects that this half--embedded eigenvalue turns into a resonance (in the sense discussed in that Section, not the notion earlier in this section).  It's imaginary part is not $\textrm{O}((\xi-\xi_j)^{2})$ as it is in the normal Fermi golden rule situation discussed in Section \ref{s4}; rather, as shown in Jensen--Nenciu \cite{JNFermi}, one typically has that it is $\textrm{O}(|\xi-\xi_j|^{3/2})$.  For related results, see Dinu--Jensen--Nenciu \cite{DJN1, DJN2}.

Jensen--Kato discussed dispersive decay in terms of $L^2_s$ spaces but there has been considerable interest in $L^p$ estimates, where for $-\Delta+V$ on $L^2(\bbR^\nu)$, one hopes, based on the case $V=0$, that for $1 \le p \le 2$
\begin{equation}\label{16.22}
  \norm{e^{-itH}P_c(H)\varphi}_{L^{p'}(\bbR^\nu)} \le C |t|^{-\nu\left(\frac{1}{p}-\frac{1}{2}\right)} \norm{\varphi}_{L^p(\bbR^\nu)}
\end{equation}
where $p'=p/p-1$ is the dual index to $p$.  $L^p$ norms are translation invariant making \eqref{16.22} much more suitable for use in the theory of non--linear evolution equations so there is a large literature on such estimates.

The first estimates of the type \eqref{16.22} were found by Schonbek \cite{Schonb} in 1979 who considered $\nu=3,p=1$ and $V$ small.  The first general result for $\nu \ge 3$ and $V$ so that $H$ has neither an eigenvalue nor resonance at zero energy were in a classic paper of Journ\'{e}, Sogge and Soffer \cite{JSS} (see also Schonbek--Zhou \cite{SZ}).

An interesting approach to \eqref{16.22} is due to Yajima \cite{YajDis1, YajDis2, YajDis3, YajDis4} who asked about when the wave operators are bounded from $L^p$ to $L^p$.  You might think that this has nothing to do with \eqref{16.22} but since $(\Omega^\pm)^*$ are then bounded from $L^{p'}$ to $L^{p'}$ and $e^{-itH}P_{ac}(H) = (\Omega^\pm)^*e^{-itH_0}(\Omega^\pm)$, $L^p$ estimates on $\Omega^\pm$ and \eqref{16.22} for $H_0$ imply it for $H$.

There is a considerable literature on $L^p$ dispersive estimates when $0$ is an eigenvalue of resonance.  We refer the reader to Yajima \cite{YajDis5} which includes many references.

Finally, we note that Fournais--Skibsted \cite{FournSkib} and Skibsted--Wang \cite{SkibWang} have results on low energy behavior of the resolvent of $-\Delta+V$ when asymptotically $V(x) \sim -c|x|^{-\beta}$ with $c > 0$ and $\beta\le 2$.  Both discuss low energy resolvent behavior and \cite{FournSkib} also discussed long time asymptotics of $e^{-itH}$.

%%%%%%%%%%%%%%%%%%%%%%%%%%%%%%%%%%%%%%%%%%%%%%%%%%%%%%%%%%%%%%
\section{The Adiabatic Theorem} \lb{s17}
%%%%%%%%%%%%%%%%%%%%%%%%%%%%%%%%%%%%%%%%%%%%%%%%%%%%%%%%%%%%%%

In 1950, Kato published a paper in a physics journal (denoted as based on a presentation in 1948) on the quantum adiabatic theorem.  It is his only paper on the subject but has strongly impacted virtually all the huge literature on the subject and related subjects ever since (there are more Google Scholar citations of this paper than of \cite{KatoHisThm}).  We will begin by describing his theorem and its proof which introduced what he called \emph{adiabatic dynamics} and I'll call the \emph{Kato dynamics}.  We'll see that the Kato dynamics defines a notion of parallel transport on the natural vector bundle over the manifold of all $k$--dimensional subspaces of a Hilbert space, $\calH$, and so a connection.  This connection is called the Berry connection and its holonomy is the Berry phase (when $k=1$).  All this Berry stuff was certainly not even hinted at in Kato's work but it is implicit in the framework.  Then I'll say something about the history before Kato and finally a few brief words about some of the other later developments.

To start, we need a basic result about linear ODEs on Banach spaces:

\begin{proposition} \lb{P17.1}  Let $X$ be a Banach space and $\{\calM_t\}_{0 \le t \le T}$ a family of norm continuous (in $t$) linear maps on X.

(a) For each $x_0 \in X$, there is a function $t \mapsto x(t;x_0);\, 0 \le t \le T$ which is $C^1$ in $t$ which is the unique solution of
\begin{equation}\label{17.1}
  \frac{d}{dt}x(t) = \calM_t(x(t)); \qquad x(0) = x_0
\end{equation}
Moreover, for each $t$, the map $W(t):x_0 \mapsto x(t;x_0)$ is a bounded linear map on $X$ and $t \mapsto W(t)$ is $C^1$ and is the unique solution of \eqref{17.1} when the map $\calM$ acts on the bounded operators on $\calL(X)$ by left operator multiplication by $\calM_t$ with initial condition that $W(0) = \bdone$.

(b) Let $\calH$ be a (separable, complex) Hilbert space and take either $X=\calH$ or $X=\calL(\calH)$ and suppose that
\begin{equation}\label{17.2}
  \calM_t(x) = iA(t)x
\end{equation}
where $A(t$) is a norm continuous map to the bounded self--adjoint operators on $\calH$.  Then there is a $C^1$ family of unitary maps, $U(t)$, with $U(0)=\bdone$ so the solution of \eqref{17.1} is
\begin{equation}\label{17.3}
  t \mapsto U(t)x_0
\end{equation}
\end{proposition}

\begin{remarks} 1.  In \eqref{17.2}, $A(t)x$ is either interpreted as applying $A(t)$ to a vector $x \in \calH$ or as left multiplication if $x \in \calL(\calH)$.

2.  The $U(t)$ in \eqref{17.3} depend only on $\{A(s)\}_{0 \le s \le T}$ (indeed only on $s \le t$) and not on $x_0$.

3.  The proof is elementary. For (a), one shows that the differential equation with initial condition \eqref{17.1} is equivalent to the integral equation
\begin{equation}\label{17.4}
  x(t) = x_0 + \int_{0}^{t} \calM_s(x(s)) \,ds
\end{equation}
on $C([0,T];X)$, the $X$--valued norm continuous functions on $[0,T]$.  One then either uses a contraction mapping theorem (if necessary shrinking $T$ to get a contraction and piecing together unique solutions on several intervals) or else one iterates the integral equation proving an estimate that the $n$th new term in the iteration is bounded by $T^n \left[\sup_{0 \le t \le T} \norm{M_t}\right]^n/n!$ to prove that the iteration converges to a convergent sum.

4.  For (b), one sees that if $U(t)$ solves the equation on $\calL(\calH)$ for $x_0 = \bdone$, then $U(t)x_0$ solves the equation in general. Moreover, by a simple calculation
\begin{equation}\label{17.5}
  \frac{d}{dt} U^*(t)U(t) = 0; \qquad \frac{d}{dt} U(t)U^*(t) = i [A, U(t)U^*(t)]
\end{equation}
The first equation and $U(0)=\bdone$ implies immediately that $U^*(t)U(t) = \bdone$.  The second equation with initial condition $U(0)U^*(0)=\bdone$ is clearly solved by $U(t)U^*(t)=\bdone$ so by uniqueness of solutions we see that $U(t)U^*(t)=\bdone$.  Thus $U(t)$ is unitary.
\end{remarks}

The adiabatic theorem considers a family of time dependent Hamiltonians, $H(s),\, 0 \le s \le 1$ and imagines changing them slowly, i.e. looking at $H(s/T),\, 0 \le s \le T$ for $T$ very large.  Thus, we look for $\tilde{U}_T(s)$ solving
\begin{equation}\label{17.6}
  \frac{d}{ds} \tilde{U}_T(s) = -iH(s/T)\tilde{U}_T(s),\,\, 0 \le s \le T; \qquad \tilde{U}_T(0) = \bdone
\end{equation}
Letting $U_T(s) = \tilde{U}_T(sT), \, 0 \le s \le 1$, we see that $U_T(s),\,0\le s \le 1$ solves
\begin{equation}\label{17.7}
  \frac{d}{ds} U_T(s) = -iT H(s) U_T(s),\,\, 0 \le s \le 1; \qquad U_T(0) = \bdone
\end{equation}
Here is Kato's adiabatic theorem

\begin{theorem} [Kato \cite{KatoAdi}] \lb{T17.2} Let $H(s)$ be a $C^2$ family of bounded self--adjoint operators on a (complex, separable) Hilbert space, $\calH$.  Suppose there is a $C^2$ function, $\lambda(s)$, so that for all $s$, $\lambda(s)$ is an isolated point in the spectrum of $H(s)$ and so that
\begin{equation}\label{17.8}
  \alpha \equiv \inf_{0 \le s \le 1} \dist(\lambda(s), \sigma(H(s))\setminus\{\lambda(s)\}) > 0
\end{equation}
Let $P(s)$ be the projection onto the eigenspace for $\lambda(s)$ as an eigenvalue of $H(s)$.  Then
\begin{equation}\label{17.9}
  \lim_{T \to \infty} (1-P(s))U_T(s)P(0) = 0
\end{equation}
uniformly in $s$ in $[0,1]$.
\end{theorem}

\begin{remarks}  1. Thus if $\varphi_0 \in \ran\, P(0)$, this says that when $T$ is large, $U_T(s)\varphi_0$ is close to lying in $\ran\, P(s)$. That is as $T \to \infty$, the solution gets very close to the ``curve'' $\{\ran\, P(s)\}_{0 \le s \le 1}$.

2.  If there is an eigenvalue of constant multiplicity near $\lambda(0)$ for $s$ small, it follows from \eqref{2.1} that $P(s)$ and $\lambda(s)$ are $C^2$.

3.  It is easy to see that $\dim \ran\, P(s)$ is constant.  It can even be infinite dimensional.

4.  This result is even interesting if $\dim \ran\, P(s)$ is $1$ and/or $\dim \calH < \infty$.

5.  Kato made no explicit assumptions on regularity in $s$ saying ``Our proof given below is rather formal and not faultless from the mathematical point of view. Of course it is possible to retain mathematical rigour by detailed argument based on clearly defined assumptions, but it would take us too far into unnecessary complication and obscure the essentials of the problem.''  It is hard to imagine the Kato of 1960 using such language! In any event, the proof requires that $P(s)$ be $C^2$.

6.  We'll discuss history more later but Kato notes that his work has two advantages over the earlier work of Born--Fock \cite{BF}: (1) They assume complete sets of eigenvectors and do not allow continuous spectrum.  (2) They assume that $\lambda(s)$ is simple, i.e. $\dim \ran\, P(s) = 1$ while Kato can handle degenerate eigenvalues.

7.  As we'll see, the size estimate for \eqref{17.9} is $\textrm{O}(1/T)$.
\end{remarks}

Kato's wonderful realization is that there is an explicit dynamics, $W(s)$ for which \eqref{17.9} is exact, i.e.
\begin{equation}\label{17.10}
 (1-P(s))W(s)P(0) = 0
\end{equation}
He not only constructs it but proves the theorem by showing that (this formula only holds in case $\lambda(s) \equiv 0$))
\begin{equation}\label{17.11}
  \lim_{T \to \infty} [U_T(s)-W(s)]P(0) = 0
\end{equation}
The $W(s)$ that Kato constructs, he called the {\emph{adiabatic dynamics}.  It is sometimes called Kato's adiabatic dynamics.  We call it the \emph{Kato dynamics}.  Here is the basic result:

\begin{theorem} [Kato dynamics \cite{KatoAdi}] \lb{17.3} Let $W(s)$ solve
\begin{equation}\label{17.12}
  \frac{d}{ds}W(s) = iA(s)W(s), 0 \le s \le 1; \qquad W(0)=\bdone
\end{equation}
\begin{equation}\label{17.13}
  iA(s) \equiv [P'(s),P(s)]
\end{equation}
Then $W(s)$ is unitary and obeys
\begin{equation}\label{17.14}
  W(s)P(0)W(s)^{-1} = P(s)
\end{equation}
\end{theorem}

\begin{proof}  That $W(s)$ is unitary follows from Proposition \ref{P17.1}.  Note that since $P(s)^2=P(s)$ we have that
\begin{equation}\label{17.15}
  P'(s) = P'(s)P(s)+P(s)P'(s) \Rightarrow P(s)P'(s)P(s) = 0
\end{equation}
since the first equation and $P^2=P$ imply that $PP'P=2PP'P$.  Expanding the commutator defining $A(s)$ and using $PP'P=0$ yields
\begin{equation}\label{17.16}
  iP(s)A(s) = -P(s)P'(s)
\end{equation}
\begin{equation}\label{17.17}
  iA(s)P(s) = P'(s)P(s)
\end{equation}
so by the first equation in \eqref{17.15}, we have that
\begin{equation}\label{17.18}
  P'(s) = i[A(s),P(s)]
\end{equation}

By \eqref{17.12}
\begin{align}
  (P(s)W(s))' &=  (P'(s)+iP(s)A(s))W(s) \lb{17.19} \\
              &=  iA(s)P(s)W(s) \lb{17.20}
\end{align}
by \eqref{17.18}.  Taking adjoints,
\begin{equation}\label{17.21}
  (W(s)^{-1}P(s))' = -iW(s)^{-1}P(s)A(s)
\end{equation}

Since $W(s)^{-1}P(s)W(s) = (W(s)^{-1}P(s))(P(s)W(s))$, we see that
\begin{align}
  (W(s)^{-1}P(s)W(s))' &= iW(s)^{-1}P(s)A(s)P(s)W(s) \nonumber \\
                       &\null\qquad - iW(s)^{-1}P(s)A(s)P(s)W(s) = 0 \lb{17.22}
\end{align}
At $s=0$, this is $P(0)$ so
\begin{equation}\label{17.23}
  W(s)^{-1}P(s)W(s) = P(0)
\end{equation}
which is equivalent to \eqref{17.14}.
\end{proof}

\begin{proof}  [Proof of Theorem \ref{T17.2}] By replacing $H(s)$ by $H(s) - \lambda(s)\bdone$, we can suppose that $\lambda(s) \equiv 0$ (doing this changes some formulae, particularly the critical \eqref{17.25} -- we'll address this after the proof).  We will prove that
\begin{equation}\label{17.24}
  \norm{U_T(s)^*W(s)P(0)-P(0)} = \textrm{O}(1/T)
\end{equation}
Since $U_T$ is unitary, this implies that
\begin{equation}\label{17.25}
  \norm{W(s)P(0)-U_T(s)P(0)} = \textrm{O}(1/T)
\end{equation}
Since $(1-P(s))W(s)P(0) = (1-P(s))P(s)W(s) = 0$, this implies \eqref{17.9} with an explicit $\textrm{O}(1/T)$ error estimate.

Thus we define
\begin{equation}\label{17.26}
  G(s)=U_T^*(s)W(s)P(0)
\end{equation}
and compute
\begin{equation}\label{17.27}
  G'(s) = (U_T^*(s))'W(s)P(0) + U_T^*(s)W'(s)P(0)
\end{equation}
Applying $\null^*$ to \eqref{17.7} implies that
\begin{equation}\label{17.28}
  (U_T^*(s))' = iTU_T^*(s)H(s)
\end{equation}
so, using \eqref{17.14}, the first term in \eqref{17.27} is
\begin{equation}\label{17.29}
  iTU_T^*(s)H(s)W(s)P(0) = iTU_T^*(s)H(s)P(s)W(s) = 0
\end{equation}
since $\lambda(s) \equiv 0 \Rightarrow H(s)P(s) =0$.  This is useful because it says that a potential $\textrm{O}(T)$ term is zero!

Next note that since $PP'P=0$ we have that $PAP=0$ and thus
\begin{align}
  P(s)W'(s)P(0) &= iP(s)A(s)W(s)P(0) \nonumber \\
                &= iP(s)A(s)P(s)W(s) \nonumber \\
                &= 0 \lb{17.30}
\end{align}

If now $S(s)$ is the reduced resolvent of $H(s)$ (see \eqref{2.4D}) ${S(s) \equiv(1-P(s))H(s)^{-1}}$, then on account of \eqref{17.30}, we have that
\begin{equation}\label{17.39}
  W'(s)P(0) = (1-P(s))W'(s)P(0)=H(s)S(s)W'(s)P(0)
\end{equation}
so, by \eqref{17.21}
\begin{align}
  G'(s) &= U_T^*(s)H(s)S(s)W'(s)P(0) \lb{17.40} \\
        &= (iT)^{-1}[U_T^*(s)]'S(s)W'(s)P(0)
\end{align}
by \eqref{17.28}.  Thus
\begin{equation}\label{17.41A}
  G(s)-P(0) = (iT)^{-1} \int_{0}^{s}[U_T^*(w)]'S(w)W'(w)P(0) \,dw
\end{equation}

As we've seen $U_T'$ is $\textrm{O}(T)$ but we can integrate by parts.  Since $U_T(w)$ has norm one and $S(w)$ and $W'(s)$ are bounded, the boundary terms in the integration by parts are $\textrm{O}(1/T)$.  Since we assumed that $P(s)$ is $C^2$, one has that $S'(s)$ and $W''(s)$ are bounded so the integrand after integration by parts is bounded and we have proven that ${\norm{G(s)-P(0)} = \textrm{O}(1/T)}$, i.e. \eqref{17.24} holds.
\end{proof}

This completes our discussion of what was in this influential paper of Kato.  Kato left at least two important items ``on the table''.  One is the possibility of better estimates than $\textrm{O}(1/T)$.  We discuss this further below.

The other item concerns the fact that \eqref{17.25} says a lot more than \eqref{17.9}.  \eqref{17.9} says that as $T \to \infty$, $U_T(s)$ maps $\ran\, P(0)$ to $\ran\, P(s)$.  \eqref{17.25} actually tells you what the precise limiting map is!  One should note that if $\lambda(s)$ is not identically zero, the proper form of \eqref{17.25} is
\begin{equation}\label{17.42}
  \norm{U_T(s)P(0)-e^{-iT\int_{0}^{s} \lambda(s)\,ds} W(s)P(0)} = \textrm{O}(1/T)
\end{equation}

One fancy pants way of describing this is as follows.  Fix $k \ge 1$ in $\bbZ$.  Let $\calM$ be the manifold of all $k$--dimensional subspaces of some Hilbert space, $\calH$.  We want $\dim(\calH) \ge k$, but it could be finite.  Or $\calM$ might be a smooth submanifold of the set of all such subspaces.  For each $\omega \in \calM$, we have the projection $P(\omega)$.  There is a natural vector bundle of $k$--dimensional spaces over $\calM$, namely, we associate to $\omega \in \calM$, the space $\ran\, P(\omega)$.  If $k=1$, we get a complex line bundle.

The Kato dynamics, $W(s)$, tells you how to ``parallel transport'' a vector $v \in \ran\, P(\gamma(0))$ along a curve $\gamma(s);\,0 \le s \le 1$ in $\calM$.  In the language of differential geometry, it defines a connection and such a connection has a holonomy and a curvature.  In less fancy terms, consider the case $k=1$.  Suppose $\gamma$ is a closed curve.  Then $W(1)$ is a unitary map of $\ran\, P(0)$ to itself, so multiplication by $e^{i\Gamma_B(\gamma)}$.  Returning to $U_T$, it says that the phase change over a closed curve isn't what one might naively expect, namely $\exp(-i\int_{0}^{T} \lambda(s/T)\,ds) = \exp(-iT\int_{0}^{1} \lambda(s)\,ds)$.  There is an additional term, $\exp(i\Gamma_B)$.  This is the \emph{Berry phase} discovered by Berry \cite{Berry} in 1983 (it was discovered in 1956 by Pancharatnam \cite{Pan} but then forgotten).  Simon \cite{SimonBerry} realized that this was just the holonomy of a natural bundle connection and that, moreover, this bundle and connection is precisely the one whose Chern integers are the TKN$^2$ integers of Thouless et al \cite{TKNN} (as discussed by Avron et al \cite{ASSTKNN}).  Thouless got a recent physics Nobel prize in part for the discovery of the TKN$^2$ integers.  The holonomy, i.e. Berry's phase, is an integral of the Kato connection $[P,dP]$.  As usual, this line integral over a closed curve is the integral of its differential $[dP,dP]$ over a bounding surface.  This quantity is the curvature of the bundle and has come to be called the \emph{Berry curvature} (even though Berry did not use the differential geometric language).  Naively $[dP,dP]$ would seem to be zero but it is shorthand for the two--form
\begin{equation}\label{17.43}
  \sum_{i \ne j}\left[\frac{\partial P}{\partial s_i},\frac{\partial P}{\partial s_j}\right] ds_i\wedge ds_j
\end{equation}
This formula of Avron et al \cite{ASSTKNN} for the Berry curvature is a direct descendant of formulae in Kato's paper, although, of course, he did not consider the questions that lead to Berry's phase.

Now, a short excursion into the history of adiabatic theorems.  ``Adiabatic'' first entered into physics as a term in thermodynamics meaning a process with no heat exchange.  In 1916, Ehrenfest \cite{Ehren} discussed the ``adiabatic principle'' in classical mechanics.  The basic example is the realization (earlier than Ehrenfest) that while the energy of a harmonic oscillator is not conserved under time dependent change of the underlying parameters, the action (energy divided by frequency) is fixed in the limit that the parameters are slowly changed (the reader should figure out what Kato's adiabatic theorem says about a harmonic oscillator with slowly varying frequency). See Henrard \cite{Henrard} for discussion of applications of the classical adiabatic invariant.  Interestingly enough, many adiabatic processes in the thermodynamic sense are quite rapid, so the Ehrenfest use has, at best, a very weak connection to the initial meaning of the term!

Ehrenfest used these ideas by asserting that in old quantum theory, the natural quantum numbers were precisely these adiabatic invariants.  Once new quantum mechanics was discovered, Born and Fock \cite{BF} in 1928 discussed what they called the quantum adiabatic theorem, essentially Theorem \ref{T17.2} for simple eigenvalues with a complete set of (normalizable) eigenfunctions.  It was 20 years before Kato found his wonderful extension (and then more than 30 years before Berry made the next breakthrough).

Next, we turn to error estimates.  The error on the right side of \eqref{17.41A} is a sum of two terms after an integration by parts: the boundary term and an integral.  For the integral, one can reuse \eqref{17.29} as we did to get \eqref{17.41A} and see that the integral is $\textrm{O}(1/T^2)$.  The boundary term is $\textrm{O}(1/T)$ but the coefficients will vanish if $P(s)-P(0)$ and $P(t)-P(s)$ vanish sufficiently fast as $s \downarrow 0$ and $s \uparrow t$.  The natural setup is to take $s \in (-\infty,\infty)$ rather than $[0,1]$ and to require that $H(\pm \infty) = \lim_{s \to \pm\infty} H(s)$ exist with approach $\textrm{O}(1/|s|^k)$ for all $k$.  If one does this, one gets an adiabatic theorem with $\textrm{O}(1/T^k)$ errors for all $k$.  Under suitable analyticity conditions on $H(s)$, one can even prove exponential approach, see \cite{NenciuAdi} for an early paper on this subject and \cite{EH, JS, JS2, JP, NenciuAdi2} for additional discussion.  In particular, Joye--Pfister \cite{JP} uses arguments very close to Kato's.

The occurrence of the reduced resolvent, $S$, in Kato's approach suggests that an eigenvalue gap is an important ingredient.  Nevertheless, there are results on adiabatic theorems without gaps, see Avron--Howland--Simon \cite{AHowS} and Hagedorn \cite{HagAdi} for some special situations and Avron--Elgart \cite{AE} for a very general result.  Teufel \cite{Teu} has an alternate proof for this Avron--Elgart result and he has a book \cite{TeuBk} on the subject.  Avron et al \cite{AFGG} and Joye \cite{Joye} have Banach space versions.

For other approaches to adiabatic evolution, see Jansen--Ruskai--Seiler \cite{JRS} and Hastings--Wen \cite{HW}.  For some applications, see Avron--Seiler--Yaffe \cite{ASY}, Klein--Seiler \cite{KSeiler} and Bachmann--de Roeck--Fraas \cite{BdRF}.

%%%%%%%%%%%%%%%%%%%%%%%%%%%%%%%%%%%%%%%%%%%%%%%%%%%%%%%%%%%%%%
\section{Kato's Ultimate Trotter Product Formula} \lb{s18}
%%%%%%%%%%%%%%%%%%%%%%%%%%%%%%%%%%%%%%%%%%%%%%%%%%%%%%%%%%%%%%

We begin this section by describing what is called the Lie product formula.  Let $A, B$ be two finite matrices over $\bbC^n$.  Fix $T > 0$ and for $0 \le s \le T$, define
\begin{equation}\label{18.1}
  g(s) = e^{s(A+B)}-e^{sA}e^{sB}
\end{equation}
Then $g(0)=g'(0)=0$ so, by Taylor's theorem with remainder
\begin{equation}\label{18.2}
   \norm{g(s)} \le Cs^2; \qquad 0 \le s \le T
\end{equation}

Writing
\begin{align*}
  e^{s(A+B)}-\left[e^{sA/n}e^{sB/n}\right]^n  &=  [e^{s(A+B)/n}]^n-\left[e^{sA/n}e^{sB/n}\right]^n \\
                                              &=  \sum_{j=1}^{n} [e^{s(A+B)/n}]^{j-1} g(\tfrac{s}{n}) \left[e^{sA/n}e^{sB/n}\right]^{n-j}
\end{align*}
has norm bounded by $n \exp(s(\norm{A}+\norm{B})) \norm{g(\tfrac{s}{n})} \to 0$ by \eqref{18.2}  Thus, for finite matrices, we have that
\begin{equation}\label{18.3}
  e^{s(A+B)}  = \lim_{n \to \infty} \left[e^{sA/n}e^{sB/n}\right]^n
\end{equation}

This is called the Lie product formula.  Although it seems he never wrote it down explicitly, Lie did consider differential equation results on groups close to \eqref{18.3}.  In 1959, Trotter \cite{TrotterProd} proved a version of the Lie product formula for certain semigroups on Banach spaces:

\begin{theorem} [Trotter Product Formula] \lb{T18.1} Let $X$ be a Banach space and $S(t) = e^{-tA}, \quad t >0$ and $T(t) = e^{-tB}, \quad t>0$ two strongly continuous semigroups on $X$ that obey
\begin{equation}\label{18.4}
  \textrm{s}-\lim_{t \downarrow 0} S(t) = \textrm{s}-\lim_{t \downarrow 0} T(t) = \bdone; \qquad \norm{S(t)}+\norm{T(t)} \le Ce^{Dt}
\end{equation}
Suppose that the operator closure of A+B on $D(A) \cap D(B)$ generates a strongly continuous semigroup, $W(t) \textrm{\textit{``=''}} e^{-t(A+B)}$ obeying \eqref{18.4},  Then
\begin{equation}\label{18.5}
  \textrm{s}-\lim_{n \to \infty} \left[S(\tfrac{t}{n})T(\tfrac{t}{n})\right]^n = W(t)
\end{equation}
\end{theorem}

\begin{remarks} 1. If $S(t)$ is a semigroup obeying \eqref{18.4}, then one defines
\begin{equation*}
  D(A) = \{\varphi \,|\, \lim_{t \downarrow 0} \left(\frac{\bdone-S(t)}{t}\right)\varphi \textrm{ exists}\}
\end{equation*}
and sets $A\varphi$ to be the limit.  One then writes $S(t) = e^{-tA}$.

2.  If $X$ is a Hilbert space, $S(t)$ is self--adjoint and a contraction, then $S(t) = e^{-tA}$ for a positive (possibly unbounded) self--adjoint operator, $A$.  This sets up a $1-1$ correspondence between such semigroups and positive self--adjoint operators.

3.  It is a famous theorem of Stone \cite[Section 7.3]{OT} that when $X$ is a Hilbert space, then $S(t)$ is unitary for all $t$ and strongly continuous at $0$ (with $S(0)=\bdone$) if and only if $S(t) = e^{-itA}$ for a self--adjoint operator $A$.
\end{remarks}

For a very simple proof when $X$ is a Hilbert space, $A$ and $B$ are self--adjoint and $A+B$ is self--adjoint (rather than only esa) on $D(A) \cap D(B)$, see \cite[Theorem VIII.30]{RS1}.  The proof is due to Nelson \cite{NelsonFK} and looks like the finite matrix proof plus one use of the uniform boundedness principle.

The limitation that $A+B$ have a closure that is a semigroup generator is quite strong.  For example, there are cases where $D(A) \cap D(B) = \{0\}$ but formally $A+B$ makes sense.  Remarkably, Kato proved a result that, at least for self--adjoint contraction semigroups, always holds.  Let $A$ and $B$ be self--adjoint operators and $q_A$, $q_B$ their closed quadratic forms as discussed in Example \ref{E10.3}.  Their form sum $q_C=q_A+q_B$ is always a closed form but $V_{q_C}$ may not be dense.  We'll write $C=A \dot{+} B$.  We need to define $e^{-tC}$ for $C$'s which are associated to closed quadratic forms where $V_q$ might not be dense.  We follow the philosophy discussed in Section 10 in the discussion of monotone convergence.  If $q$ is a closed quadratic form and $C$ is the self--adjoint operator on $\overline{V_q}$ with $V_q=D(C^{1/2})$ and $q(\varphi) = \jap{C^{1/2}\varphi,C^{1/2}\varphi}$ for $\varphi \in V_q$, then we define $e^{-t\dot{C}}$ to be the operator
\begin{equation}\label{18.6}
  e^{-t\dot{C}} = e^{-tC}P
\end{equation}
where $P$ is the orthogonal projection onto $\overline{V_q}$.  Here is Kato's result

\begin{theorem} [Kato's Ultimate Trotter Product Formula \cite{KatoTP2}] \lb{T18.2} Let $q_1, q_2$ be two closed quadratic forms on a Hilbert space, $\calH$, with associated semigroups $e^{-t\dot{A}},\,e^{-t\dot{B}}$.  Let $e^{-t\dot{C}}$ be the semigroup associated to the closed form sum $q_1+q_2$.  Then
\begin{equation}\label{18.7}
  \textrm{s}-\lim_{n \to \infty}\left[e^{-t\dot{A}/n}e^{-t\dot{B}/n}\right]^n = e^{-t\dot{C}}
\end{equation}
\end{theorem}

\begin{remarks} 1.  The proof is somewhat technical; we refer the reader to the original paper \cite{KatoTP2} or to Reed--Simon \cite[Theorem S.21]{RS1}.  The proof relies on a general result of Chernoff \cite{ChernoffTP1} (see also \cite[Theorem S.19]{RS1}).

2.  Earlier results on Trotter product formula for form sums include Chernoff \cite{ChernoffTP1, ChernoffTP2, ChernoffTP3}, Faris \cite{FarisTP} and Kato himself \cite{KatoTP1}.

3.  It would be nice to have some kind of result for $e^{-it\dot{C}}$ but it is unlikely there is one when the approximation is applied to a vector not in $\overline{V_q}$.  That said, \eqref{18.7} holds for all $t \in \bbC$ with $|\arg(t)| < \pi/2$ and, as explained by Kato in a Note to his paper \cite{KatoTP2}, by an argument that he got from me, one can extend the result from positive  self--adjoint $A, B$, to generators of holomorphic contraction semigroups.

4.  It could be argued with some justice that this paper doesn't so much belong in Kato's work on NRQM but to his work on linear semigroups.  But, as found by Nelson \cite{NelsonFK} (see also Simon \cite{SimonFI}) the Trotter product formula is central to the proof of the Feynman--Kac formula and also to interpreting Feynman integrals for $e^{-itH}$.  Moreover, we saw its appearance in Section 9 -- see Theorem \ref{T9.3}.

5.  Kato--Masuda \cite{KM} found an extension to nonlinear semigroups.  Their paper also has a new result in the linear case, namely instead of $A\dot{+}B$, one can consider $k$ positive, self--adjoint operators, $A_1,\dots,A_k$ and their form sum $A_1\dot{+}\dots\dot{+}A_k$.

\end{remarks}

\begin{example} \lb{E18.3}  Let $P$ and $Q$ be two orthogonal projections on a Hilbert space.  Define
\begin{equation}\label{18.8}
  q_1(\varphi) = \left\{
                   \begin{array}{ll}
                     0, & \hbox{ if } \varphi \in \ran\, P \\
                     \infty, & \hbox{ if } \varphi \notin \ran\, P
                   \end{array}
                 \right.
\end{equation}
and similarly for $q_2$ and $Q$.  Then $e^{-t\dot{A}}=P,\,e^{-t\dot{B}}=Q$ for all $t$.  It is easy to see that the form sum $q_1+q_2$ has the same structure as \eqref{18.8} but with $\ran\, P$ replaced by $\ran\, P\cap\ran Q$.  If $R$ is the projection onto this intersection, then Kato's result says that
\begin{equation}\label{18.9}
  \textrm{s}-\lim_{n \to \infty} (PQ)^n = R
\end{equation}
It is interesting that this geometrically well known fact is a special case of Kato's result \eqref{18.7}.
\end{example}

%%%%%%%%%%%%%%%%%%%%%%%%%%%%%%%%%%%%%%%%%%%%%%%%%%%%%%%%%%%%%%
\section{Regularity of Eigenfunctions and the Kato Cusp Condition} \lb{s19}
%%%%%%%%%%%%%%%%%%%%%%%%%%%%%%%%%%%%%%%%%%%%%%%%%%%%%%%%%%%%%%

If one wants to understand the wider impact of Kato's work, a good place to get insight is to look at citations at Google Scholar (\href{https://scholar.google.co.il/scholar?hl=en&q=tosio+kato}{https://scholar.google.co.il/scholar?hl=en\&q=tosio+kato}).  Of course, the publication with the most references by far is Kato's book \cite{KatoBk} with over 20,000 citations.  In second place (with over 1700 citations) is the 1957 paper \cite{KatoEig} discussed in this section.  This may be surprising to some, but it reflects its importance to quantum chemists and atomic physicists.

In this paper, Kato begins by saying that he regards this paper as a continuation of \cite{KatoHisThm}.   In that earlier paper, he stated  ``If $V$ is the Coulomb potential as in the case of real atoms, it follows that the eigenfunctions satisfy the wave equation everywhere except at singular points of the potential (they are even analytic since the Coulomb potential is an analytic function). Regarding their behavior at these singular points, we can derive no conclusion from the above theorem. A detailed study shows, however, that they are bounded even at such points''.  He is interested in the properties of $L^2$--eigenfunctions and what he calls generalized eigenfunctions or \emph{wave packets} by which he means $\psi \in \calH$ with $\psi \in \ran E_\Omega(H)$ where $H$ is a quantum Hamiltonian, $\Omega = [a,b]$, a bounded interval, and $E_\Omega(H)$ is a spectral projection \cite[Section 5.1]{OT}.  In fact, we'll soon see that $\psi \in \ran(e^{-sH})$ for some $s>0$ suffices for some of the results that Kato proved.  Kato focused on local regularity of $\psi$ with some global estimates (like on $\norm{\boldsymbol{\nabla}\psi}_\infty$).  In particular, he delivered on the boundedness result he claimed in 1951.

There is a huge literature on other aspects of eigenfunctions which we'll not discuss except for a few words now.  First, there is the issue of exponential decay which we mentioned briefly at the end of Section \ref{s12}; below all we'll discuss, in the context of proving pointwise bounds, is how to go from $L^2$ exponential decay to pointwise exponential decay.  Secondly, there is literature on the structure of nodes (i.e. the zero set); see, for example, Zelditch \cite{ZelNodes}.  Finally there are the issues of continuum eigenfunction expansions and the related theorem that $\sigma(H)$ is the closure of the set of $E$ for which $H\psi=E\psi$ has a polynomially bounded solution; see \cite[Corollary C.5.5]{SimonSmgp}.

Kato considers two classes of Hamiltonians.  The first, which we'll call general $H$, acts on $L^2(\bbR^{\nu N})$ with $\boldsymbol{x} = (x_1,\dots,x_N);\, x_j \in \bbR^\nu$ (Kato only considers the case $\nu=3$, but we'll discuss the more general case below).  $H$ then has the form
\begin{equation}\label{19.1}
  H=-\Delta+\sum_{j=1}^{N} V_j(x_j)+\sum_{1 \le j < k \le N} V_{jk}(x_j-x_k)
\end{equation}
with each $V_j,V_{jk} \in L^p(\bbR^\nu)+L^\infty(\bbR^\nu)$, where $p$ is $\nu$--canonical (see just prior to Theorem \ref{T7.9}) so that H is esa--$\nu$ (see Section \ref{s7}).  $-\Delta$ assumes equal masses of the light particle and an infinite mass heavy particle but one easily accommodates general masses using the formalism in Section \ref{s11}.

Kato also considered what we will call \emph{atomic Hamiltonians}
\begin{equation}\label{19.2}
  H = -\Delta-\sum_{j=1}^{N}\frac{Z}{|x_j|} + \sum_{1 \le j < k \le N} \frac{1}{|x_j-x_k|}
\end{equation}
on $L^2(\bbR^{3N})$.  Kato allows Hughes--Eckart terms, allows $Z$ to be $j$ dependent and allows $\tfrac{z_{jk}}{|x_j-x_k|}$ rather than $\tfrac{1}{|x_j-x_k|}$.  All these are easy to accommodate as is the molecular case where $\tfrac{Z}{|x_j|}$ is replaced by
\begin{equation}\label{19.3}
  \sum_{\ell=1}^{L} \frac{Z_\ell}{|x_j-R_\ell|}
\end{equation}
Most of the time, for simplicity of exposition, we'll discuss the atomic case.

In the atomic case, we'll be especially interested in the set of singularities where some $|x_j|$ or $|x_j-x_k|$ vanish, i.e.
\begin{equation}\label{19.3A}
  \Sigma = \{\boldsymbol{x}=(x_1,\dots,x_N) \,|\, \prod_{j=1}^{N} |x_j| \,\prod_{1 \le j < k \le N} |x_j-x_k| =0\}
\end{equation}

In \cite{KatoEig}, Kato proved three main theorems.  For the first two, we need a definition.  Let $0 < \alpha \le 1$ and $j =0$ or $1$.  Then
\begin{equation*}
  C^{j,\alpha} =\{\psi \,| \, \psi \textrm{ is } C^j \textrm{ and obeys \eqref{19.4}} \}
\end{equation*}
\begin{equation}\label{19.4}
  \exists C\, \forall_{x,y \, |\, |x-y| \le 1} \, |D^{(j)}\psi(x) - D^{(j)}\psi(y)| \le C|x-y|^\alpha
\end{equation}
($\alpha=1$ is called Lipschitz; otherwise, we are saying the derivative is H\"{o}lder continuous).  If the constant $C$ in \eqref{19.4} is allowed to depend on a compact $K$ requiring $x,y\in K \subset \bbR^\nu$, we say that $\psi \in C^{j,\alpha}_{loc}$.

\begin{theorem} [(Kato\cite{KatoEig})] \lb{T19.1} Let $\nu=3$ and let $V_j, V_{jk} \in L^\sigma(\bbR^3)+L^\infty(\bbR^3)$ for some $\sigma \ge 2$. Let $\psi$ be an eigenfunction or wave packet.  Then:

(a) For all $\alpha$ with $\alpha \le 1$ and $\alpha < 2-\tfrac{3}{\sigma}$, we have that
\begin{equation}\label{19.5}
  \psi \in C^{0,\alpha}
\end{equation}

(b) If $\sigma > 3$, we have that for all $\alpha < 1-\tfrac{3}{\sigma}$ that
\begin{equation}\label{19.6}
  \psi \in C^{1,\alpha}
\end{equation}
\end{theorem}

The Coulomb case allows any $\sigma$ with $\sigma < 3$ but not $\sigma=3$ so it is borderline for $\psi$ being Lipschitz.  Nevertheless, Kato proved that

\begin{theorem} [(Kato \cite{KatoEig})] \lb{T19.2} Let $\nu=3$ and let $H$ be an atomic Hamiltonian.  Let $\psi$ be an eigenfunction or wave packet.  Then $\psi \in C^{0,1}$ (i.e. is Lipschitz).  Indeed $\psi$ is $C^1$ on $\bbR^{3n}\setminus\Sigma$ with $\boldsymbol{\nabla}\psi \in L^\infty$.
\end{theorem}

\begin{remarks} 1.  It is easy to see by the fact that $\Sigma$ is closed of measure zero, that the $C^1$ result with bounded derivative implies the $C^{0,1}$ result.

2.  As Kato remarks, in the atomic case, there were no previous positive results on regularity of eigenfunctions if $N \ge 2$ although it was known that certain series expansions did not work.

3.  Since the potentials are real analytic on $\bbR^{3N}\setminus\Sigma$, it is known by elliptic regularity \cite{Pic, FriedAnal, GT}  that genuine eigenfunctions are real analytic on $\bbR^{3N}\setminus\Sigma$.  So the point of the theorem is control on $\Sigma$ and the uniformity of the bounds.
\end{remarks}

Kato's third result concerns the exact behavior at the two particle coincidences.  To understand why he states the theorem as he does, consider Hydrogen--like Hamiltonians where the eigenfunctions are exactly known.

\begin{example} \lb{E19.3} Let $h = -\Delta-\tfrac{2}{|x|}$ on $L^2(\bbR^3)$.  It is known \cite{Griff} that the unnormalized ground (1s) state is given by
\begin{equation}\label{19.7}
  \varphi_0(\boldsymbol{r}) = e^{-r}; \qquad r = |\boldsymbol{r}|
\end{equation}
obeying $h\varphi_0 = -\varphi_0$.  Notice that $\varphi_0$ is not $C^1$ at $\boldsymbol{r}=0$ but has a cusp there, i.e.
\begin{equation}\label{19.8}
  \boldsymbol{\nabla}\varphi_0(\boldsymbol{r}) = -\frac{\boldsymbol{r}}{r} e^{-r}
\end{equation}
so that the limit of the derivative ar $\boldsymbol{r}=0$ is directionally dependent.

The 2p state (with $m=0$) is given by
\begin{equation}\label{19.9}
  \varphi_1(\boldsymbol{r}) = z e^{-r/2}; \qquad \boldsymbol{r} = (x,y,z) \in \bbR^3
\end{equation}
obeying $h\varphi_1 = -\tfrac{1}{4} \varphi_1$.  Thus
\begin{equation}\label{19.10}
  \boldsymbol{\nabla}\varphi_1(\boldsymbol{r}) = -\frac{1}{2} z \frac{\boldsymbol{r}}{r} e^{-r/2} + (0,0,1) e^{-r/2}
\end{equation}
This derivative is continuous at $\boldsymbol{r}=0$ and non--zero at $\boldsymbol{r}=0$.  Kato had the realization that by taking a spherical average of $\psi$, one captures (at least in the one electron case) exactly the s states which have cusps.  That explains why he took the average in the next Theorem.
\end{example}

\begin{theorem} [Kato Cusp Condition \cite{KatoEig}] \lb{T19.4} Let $H$ be an atomic Hamiltonian and let $\psi$ be an $L^2$ eigenfunction for $H$.  Let $\boldsymbol{x}=(x_1,\dots,x_N)$.  Define on $(0,\infty) \times \bbR^{3(N-1)}$
\begin{equation}\label{19.11}
  \wti{\psi}(r,x_2,\dots,x_N) = \frac{1}{4\pi} \int_{S^2} \psi(r\omega,x_2,\dots,x_N) \, d\omega
\end{equation}
where $d\omega$ is the surface measure on the two dimensional sphere, so $\wti{\psi}$ is a spherical average.  Then except for $(x_2,\dots,x_N)$ in a set of lower dimension (i.e. less than $3N-3$), one has that
\begin{equation}\label{19.12}
  \left.\frac{\partial\wti{\psi}}{\partial r}\right|_{r=0} = -\frac{Z}{2} \psi(0,x_2,\dots,x_N)
\end{equation}
\end{theorem}

\begin{remarks} 1.  \eqref{19.12} is the celebrated \emph{Kato cusp condition}.

2.  There is a similar result at $x_j-x_k=0$; $-\tfrac{Z}{2}$ is replaced by $+\tfrac{1}{2}$.

3.  In \eqref{19.12}, the left side means to compute the derivative for $r > 0$ (using that $\psi$ is $C^1$ there $\Rightarrow \wti{\psi}$ is $C^1$) and then take $r \downarrow 0$.  \eqref{19.12} says that $\wti{\psi}(r,x_2,\dots,x_N) = -\tfrac{Z}{2}r\psi(0,x_2,\dots,x_N)+\textrm{o}(r)$ so that, if $\psi(0,x_2,\dots,x_N) \ne 0$, $\wti{\psi}$ has a cusp as it does for Hydrogen.

4.  Most modern variational calculations for atoms and molecules use basis elements that have the cusp condition, so this theorem is very influential.
\end{remarks}

Kato's proofs depend on rewriting the time--independent Schr\"{o}dinger equation as an integral equation and analyzing that equation.  This completes what we want to say about Kato's paper itself.  We turn to later work, first concerning general Hamiltonians and Theorem \ref{T19.1}.  The most powerful results use path integral methods (pioneered by Herbst--Sloan \cite{HerbS}, Carmona \cite{Carm} and Aizenman--Simon \cite{AizSimon}; two comprehensive references are \cite{SimonFI, SimonSmgp}) and are expressed in terms of a class of spaces $K_\nu^{(\alpha)}; \nu = 1,2,\dots;\, \alpha \in [0,2)$ defined by (we suppose $\nu \ge 2$ and when $\alpha = 0$ that $\nu \ge 3$; we refer the reader to \cite{SimonSmgp} for the other cases):

\begin{definition} $K_\nu^{(\alpha)}$ is defined by

(a) for $\alpha \in (0,1)\cup (1,2)$ and $\nu \ge 2$ as those $V$ with
\begin{equation}\label{19.13}
  \sup_{x} \int_{|x-y| \le 1} |x-y|^{-(\nu-2+\alpha)} |V(y)| \, dy < \infty
\end{equation}

(b) if $\alpha = 0$ or $\alpha =1$ and $\nu \ge 3$ by
\begin{equation}\label{19.14}
  \lim_{r \downarrow 0} \sup_{x} \int_{|x-y| \le r} |x-y|^{-(\nu-2+\alpha)} |V(y)| \, dy  = 0
\end{equation}
\end{definition}

\begin{remarks}  1.  If $\alpha = 0$, $K_\nu^{(0)} = K_\nu$ as defined in \eqref{9.33}.

2.  If $\alpha_1 > \alpha$, then $K_\nu^{(\alpha_1)} \subset   K_\nu^{(\alpha)}$

3.  If $p > \nu/(2-\alpha)$, then $L^p_{unif} \subset K^{(\alpha)}_\nu$ by H\"{o}lder's inequality.  In particular $v \in L^\sigma(\bbR^3)+L^\infty(\bbR^3) \Rightarrow v \in K_3^{(\alpha)}$ so long as $\alpha < 2-3/\sigma$.

4.  As with $K_\nu$, $v(x) \in K_\nu^{(\alpha)}$ for $x \in \bbR^\nu$ implies that $V(x,y) \equiv v(x),\, x\in \bbR^\nu, y\in\bbR^{\mu-\nu} \Rightarrow V \in K_\mu^{(\alpha)}$.  Thus in the context of Theorem \ref{T19.1} $V_j(x_j)$ and $V_{jk}(x_j-x_k)$ on $\bbR^{3N}$ will lie in $K_{3N}^{(\alpha)}$ if the $V$'s, $\alpha$ and $\sigma$ are as in Remark 2.  This means that Theorem \ref{T19.1} follows from Theorem \ref{T19.6} below.

5.  As with $K_\nu$, these spaces are special cases of a class of spaces of Schechter \cite{Schechter}.  In this context, they were introduced by Simon \cite{SimonSmgp}.

6. $K_{\nu,loc}^{(\alpha)}$ is those $V$ whose restriction to each ball in $\bbR^\nu$ lies in $K_\nu^{(\alpha)}$.
\end{remarks}

One of Kato's realizations is that eigenfunctions are bounded and continuous.  In this regard, the following is useful.

\begin{theorem} [Subsolution estimate] \lb{T19.5} Let $V$ be a function on $\bbR^\nu$ with $V \in K_\nu$.  Let $\psi \in L^2_{loc}$ solve $(-\Delta+V)\psi = 0$ in distributional sense.  Then $\psi$ is a continuous function and for any $r > 0$, there is $C$ depending only on the $K_\nu$--norm of $V_- \equiv \max(V(x),0)$ (and, in particular, not on $\psi$) so that
\begin{equation}\label{19.15}
  |\psi(x)| \le C \int_{|x-y| \le r} |\psi(y)| \, dy
\end{equation}
\end{theorem}

\begin{remarks} 1.  Such estimates go back to Stampacchia \cite{Stamp} and Trudinger \cite{Trud} who had stronger hypotheses on $V$.  For $V \in K_\nu$, Agmon \cite[Chapter 5]{AgmonBk} has an analytic proof and Aizenman--Simon \cite{AizSimon} a path integral proof; see also \cite{SimonSmgp}.

2.  It is enough to have $V_- \in K_\nu$ and $V_+ \equiv V+V_- \in K_{\nu, loc}$.

3.  The name comes from the fact that it is a result proven for positive functions, $u$ with $(-\Delta+V)u \le 0$ (so subsolutions rather than solutions as in subharmonic rather than harmonic).  Kato's inequality shows that if $(-\Delta+V)\psi = 0$, then $u = |\psi|$ is a subsolution.  In this form, the inequality is intimately connected to Harnack's inequality \cite{AizSimon, SimonSmgp}.
\end{remarks}

Subsolution estimates are important because they say that $\psi \in L^2 \Rightarrow \psi \in L^\infty$ (with, in fact, the function going pointwise to zero at $\infty$) and so they give the bounded continuous part of Kato's Theorem \ref{T19.1} (for eigenfunctions; for wave packets, see below).  They also show that $e^{ar}\psi \in L^2 \Rightarrow e^{ar}\psi \in L^\infty$ and so the $L^2$ exponential decay estimates discussed in Theorem \ref{T12.7} imply pointwise exponential decay.

The following has Theorem \ref{T19.1} as a special case:

\begin{theorem} \lb{T19.6}  Let $0 < \alpha < 2$.  Let $V_-\in K_{\nu}^{(\alpha)}, \,V_+\in K_{\nu, loc}^{(\alpha)}$.  Let $f \in L^2(\bbR^\nu)$.  Then, for each $t > 0$, $e^{-tH}f$ lies in

(a) $C^{0,\alpha}$ if $\alpha\in (0,1)$

(b) Is $C^1$ and in $C^{0,1}$ if $\alpha=1$

(c) $C^{1,\alpha-1}$ if $\alpha \in (1,2)$

\noindent and the norms only depend on $t$, the $L^2$--norm of $f$ and the $K_\nu$ norm of $V_-$.
\end{theorem}

\begin{remarks} 1.  The proof using functional integration can be found in Simon \cite[Theorem B.3.5]{SimonSmgp}.

2.  For eigenfunctions, there are subsolution type estimates for the constants in H\"{o}lder estimates; see \cite[Theorem C.2.5]{SimonSmgp}.

3.  To control $\boldsymbol{\nabla}\psi$, one needs $\alpha=1$.  The Coulomb potentials in atomic and molecular Hamiltonians are in $K_{3N}^{(\alpha)}$ for $\alpha \in [0,1)$ but not for $\alpha=1$.  Nevertheless, Hoffmann--Ostenhof et al. \cite{HOHOOSWF} have proven for such potentials and $L^2$ eigenfunctions, one has that
\begin{equation}\label{19.16A}
  \sup_{|y-x| \le R} |\boldsymbol{\nabla}\psi(y)| \le C \sup_{|y-x| \le 2R} |\psi(y)|
\end{equation}
for any $x$, where $C$ is a universal constant depending only on $R$ and $H$. This includes and improves Kato's theorem \ref{T19.2}; one improvement is that exponential decay of $\psi$ implies exponential decay of its first derivatives.
\end{remarks}

There has been considerable literature dealing with the questions discussed in Kato's Theorems \ref{T19.2} and \ref{T19.4}; a substantial fraction of this literature is by Maria and Thomas Hoffmann--Ostenhof and their collaborators.  We want to discuss some of the highlights.

The first result sheds additional light on the behavior near pair singularities.  We define
\begin{equation}\label{19.16}
  \Sigma_j = \{x \,|\, |x_j|=0\}; \qquad \Sigma_{jk} = \{x\,|\, |x_j-x_k|=0\},\, j< k
\end{equation}
so $\Sigma=\left(\bigcup_{j=1}^N \Sigma_j\right)\cup\left(\bigcup_{j < k} \Sigma_{jk}\right)$.

\begin{theorem} [Fournais et al \cite{FHOHOOSAnal}] \lb{T19.7} Let $x^{(0)} \in \Sigma_1, \, x^{(0)} \notin \left(\bigcup_{j=2}^N \Sigma_j\right)\cup\left(\bigcup_{j < k} \Sigma_{jk}\right)$.  Let $\psi$ be an $L^2$ eigenfunction of $H$.  Then there are two functions, $\varphi_1$ and $\varphi_2$, defined and analytic in a neighborhood, $Q$, of $x^{(0)} \in \bbR^{3N}$, so that in $Q$
\begin{equation}\label{19.17}
  \psi(x) = \varphi_1(x)+ |x_1| \varphi_2(x)
\end{equation}
\end{theorem}

\begin{remarks}  1.  Near $x^{(0)}$,
\begin{equation*}
\psi(x) = \varphi_1(x^{(0)})+|x_1| \varphi_2(x^{(0)}) + \boldsymbol{\nabla}\varphi_1(x^{(0)})\cdot (x-x^{(0)}) + \textrm{O}((x-x^{(0)})^2)
\end{equation*}
clearly showing the cusp.

2.  Similar results hold for each $\Sigma_j$ and each $\Sigma_{jk}$.

3.  For a proof, see \cite{FHOHOOSAnal}.  They were motivated by earlier work of Hill \cite{Hill}.

4.  This shows a cusp, but supplements rather than proves the Kato cusp equality \eqref{19.12}.  Indeed, that equality implies that $\varphi_2(x^{(0)}) = -\tfrac{Z}{2}\varphi_1(x^{(0)})$.
\end{remarks}

The cusp condition only holds at simple singular points where only a single pair among $\{0,x_1,\dots,x_N\}$ coincides (in the atomic case).  In 1954, Fock \cite{Fock} (the same Fock of Born--Fock 26 years earlier and the Hartree--Fock approximation 24 years earlier and of Fock space 22 years earlier!) gave arguments that there are $\jap{x_j,x_k} \log(|x_j|^2+|x_k|^2)$ terms at points where both $|x_j|$ and $|x_k|$ go to zero.  These are called \emph{Fock terms}.

The following includes and improves the Kato cusp condition, Theorem \ref{T19.4},

\begin{theorem} [Fournais et al \cite{FHOHOOSCusp}] On $\bbR^{3N}$, let
\begin{align}
  F_2(x) &= - \frac{Z}{2}\sum_{j=1}^{N}|x_j|+\frac{1}{4}\sum_{1\le j<k\le N}|x_j-x_k| \lb{19.18} \\
  F_3(x) &= \frac{2-\pi}{12\pi} \sum_{1\le j<k\le N} \jap{x_j,x_k} \log(|x_j|^2+|x_k|^2) \lb{19.19}
\end{align}
For any $\varphi$, write
\begin{equation}\label{19.20}
  \psi = e^{F_2+F_3}\varphi
\end{equation}
Then, if $\psi$ solves $H\psi=E\psi$ on a bounded set, $\Omega$, we have that
\begin{equation}\label{19.21}
  \varphi \in C^{1,1}
\end{equation}
\end{theorem}

\begin{remarks} 1.  Writing $\psi$ in the form $e^F\varphi$ is often called a \emph{Jastrow trial function} after Jastrow \cite{Jastrow} who had the idea of modifying Slater determinants, $\varphi$ by multiplying by $e^F$ with $F$ a simple rational function of the $|x_j|$ and $|x_j-x_k|$.

2.  The weaker result where $F_3$ isn't included and $\varphi \in C^{1,\alpha}$ was proven earlier by Hoffmann-Ostenhof et al. \cite{HOHOOSWF}.  The above theorem is from Fournais et al \cite{FHOHOOSCusp} where the reader can find a proof that depends on looking at the PDE that $\varphi$ obeys and standard elliptic estimates.  All depend on noting that
\begin{equation}\label{19.22}
  \Delta F_2 = V
\end{equation}

3.  The reader may be puzzled by $-\tfrac{Z}{2}$ but $\tfrac{1}{4}$ rather than $\tfrac{1}{2}$ (since the effective $Z$ for the $jk$ pair is $+1$).  But $\boldsymbol{\nabla}F_2$ has only one $\boldsymbol{\nabla}_j$ acting non--trivially on $|x_j|$ but both $\boldsymbol{\nabla}_j$ and $\boldsymbol{\nabla}_k$ act non-trivially on $|x_j-x_k|$ turning the $\tfrac{1}{4}$ into a $\tfrac{1}{2}$ which is also why we get \eqref{19.22}.

4.  \cite{HOHOOSWF} noted that their result implies that
\begin{equation}\label{19.23}
  \boldsymbol{\nabla}\psi - \psi\boldsymbol{\nabla}F_2 \in C^{1,\alpha}, \, \alpha \in (0,1)
\end{equation}
while \cite{FHOHOOSCusp} note that their results imply that
\begin{equation}\label{19.24}
  \boldsymbol{\nabla}\psi - \psi\boldsymbol{\nabla}(F_2+F_3) \in C^{1,1}
\end{equation}
The first is a strong form of the Kato cusp condition (which follows from the continuity of $\boldsymbol{\nabla}\psi - \psi\boldsymbol{\nabla}F_2$) and unlike Kato, they prove results at multiple coincidences.  The second result implies that second derivatives of $\psi$ are bounded at simple coincidences and have a logarithmic blow up at points where $|x_j|$ and $|x_k|$ go to zero.

5.  The obvious extensions hold for molecular Hamiltonians.

6.  A interesting alternate approach to understanding the Kato cusp conditions in terms of singularities at corners is found in Ammann--Carvalho--Nistor \cite{ACN}.
\end{remarks}

That completes what we want to say about regularity of eigenfunctions; we end this section with a few remarks on the closely related subject of regularity of the one electron density defined by
\begin{equation}\label{19.25}
  \rho_\psi(x) = N \int |\psi(x,\x_2,\dots,x_N)|^2 \, d^3x_2\dots d^3x_n
\end{equation}
(this is the formula if $\psi$ is symmetric or antisymmetric; otherwise the ``$N$'' in front is replaced by summing against putting $x$ in each of the $N$ slots).  It measures the electron density.

\begin{theorem} [Fournais et al \cite{FHOHOOSAnalDensity}] \lb{T19.9} For any atomic or molecular eigenfunction, the density, $\rho_\psi$ is real analytic away from the nuclei ($x=0$ in the atomic case and $x=R_j,\, j=1,\dots,K$ in the molecular case).
\end{theorem}

This was proven in \cite{FHOHOOSAnalDensity}.  Earlier the same authors had proven that $\rho_\psi$ is $C^\infty$ \cite{FHOHOOSSmoothDensity}.  Jecko \cite{Jecko} has an alternate proof of Theorem \ref{T19.9}.
%%%%%%%%%%%%%%%%%%%%%%%%%%%%%%%%%%%%%%%%%%%%%%%%%%%%%%%%%%%%%%
\section{Two Conjectures} \lb{s20}
%%%%%%%%%%%%%%%%%%%%%%%%%%%%%%%%%%%%%%%%%%%%%%%%%%%%%%%%%%%%%%

I thought it would be appropriate to end this paper with two open questions in the areas that interested Kato.  One dates from 1971 when Kato was still active and the other from 2000, the year after he died.

\begin{conjecture} [J\"{o}rgens' Conjecture] \lb{C20.1}  Let $\Omega \subset \bbR^\nu$ be open.  Let $V \in L^2_{loc}(\Omega)$ so that $-\Delta+V$ is bounded from below and esa on $C_0^\infty(\Omega)$.  Suppose that $V_1 \ge V$ is also in $L^2_{loc}(\Omega)$.  Then $-\Delta+V_1$ is also esa on $C_0^\infty(\Omega)$.
\end{conjecture}

This result would be interesting even for $\Omega=\bbR^\nu$, where, of course when $V \equiv 0$, this is just the famous result of Kato in Section \ref{s9}. The case where $\Omega = \bbR^\nu; \, \nu \ge 5$ and $V(x) = -\nu(\nu-4)|x|^{-2}$ (results of Kalf--Walter and Simon) is mentioned in Section \ref{s9}.

In the early 1970s, there were a set of almost annual meetings at Oberwolfach on spectral and scattering theory and frequent PDE meetings.  They were quite important.  For example, Agmon announced his result Theorem \ref{T15.2} in 1970 \cite{AgmonAnon} but only published the full paper \cite{AgmonSpec} in 1975.  In between, the standard source for his work were personal notes some people took of a series of lectures that he gave at one of these Oberwolfach meetings.  At the 1971 PDE conference, Konrad J\"{o}rgens (1926--1974), who died tragically of a brain tumor less than three years later, made the above conjecture during the discussions but never published it.  He made the conjecture during the discussion of the Kalf--Walter \cite{KalfWalter1} and clearly had in mind the case where $\Omega$ is $\bbR^\nu$ with a finite set of points removed. Note that Simon's and Kato's papers discussed in the historical part of Section \ref{s9} were both preprinted in early 1972 before this conjecture, so the conjecture was made for a local Stummel space rather than $L^2_{loc}$ but eventually, it was updated to $L^2_{loc}$.

In one dimension, this is related to a result of Kurss \cite{Kurss} who proved the result for continuous $V$ although his argument doesn't need continuity (essentially it follows from a simple comparison argument for positive solutions and limit point--limit circle methods).  In 1966, in \cite{StHa}, Stetkaer--Hansen extended Theorem \ref{T8.6} to the case where V is locally Stummel.  Since, if $\nu \le 3$, $L^2_{loc}$ is the same as locally Stummel, this implies J\"{o}rgen's conjecture for $\Omega = \emptyset$ and these $\nu$ (indeed without the need of a comparison potential!)  H. Kalf has informed me that neither he nor J\"{o}rgens knew of these papers when the conjecture was made.

Many people, especially in the various German groups studying Schr\"{o}dinger operators worked hard on this problem.  In 1980, Cycon \cite{CyconJC} proved a result when there was an additional technical condition on $-\Delta+V$.  He remarked that given the failure to find a proof, some researchers began to suspect that it might be false.

\begin{conjecture} [Simon's Conjecture] \lb{C20.2}  Let $V$ be a measurable function on $\bbR^\nu,\,\nu \ge 2$ obeying
\begin{equation}\label{20.1}
  \int |x|^{-\nu+1} |V(x)|^2\, d^\nu x < \infty
\end{equation}
Then $-\Delta+V$ has a.c. spectrum of infinite multiplicity on $[0,\infty)$.
\end{conjecture}

This was made by Simon \cite{Si21Cent}.  While not explicit, there is a presumption that $-\Delta+V$ is esa-$\nu$.  If $V$ obeys \eqref{16.3B}, one needs $\beta > 1/2$.  It would be interesting to prove the conjecture for all $V$'s obeying \eqref{16.3B} for any fixed $\beta \in (1,\tfrac{1}{2})$.

Here is some background on the conjecture.  Kato--Kuroda--Agmon and others studied $V$'s obeying \eqref{16.3B} for any $\beta>1$ and found (much more than) $\sigma_{ac}(-\Delta+V)=[0,\infty)$. As noted in Section \ref{s14}, when $\nu=1$, if is known that for any $\beta < 1/2$, there are $V$'s with no a.c. spectrum; in fact, in a sense, this is generic.  In the mid 1990s, I realized that determining what happened when $1 > \beta > 1/2$ was a natural problem and alerted my graduate student advisees to this fact.  Kiselev \cite{Kis1} proved that when $\nu=1$ and $\beta > 3/4$, one could prove that $\sigma_{ac}(-\Delta+V)=[0,\infty)$.  (It was eventually realized that this regime differed from $\beta > 1$ in that one could also have singular continuous spectrum  mixed in).  This was then pushed, again when $\nu=1$ to $\beta > 1/2$ by Christ--Kiselev \cite{CK} and Remling \cite{Rem}.  It seemed natural that the precise borderline was $V \in L^2$ and in 1999, Deift and Killip (then my PhD. student) \cite{DeiftKillip} proved

\begin{theorem} [Deift--Killip \cite{DeiftKillip}] \lb{T20.3}  Let $V \in L^2(\bbR,dx)$.  Then $H = -\frac{d^2}{dx^2}+V$ on $L^2$ has a.c. spectrum $[0,\infty)$ with multiplicity 2.
\end{theorem}

If $V(\mathbf{x})=V(|x|)$ is spherically symmetric on $\bbR^\nu$, then \eqref{20.1}$\Rightarrow \int_{0}^{\infty} |V(r)|^2\,dr < \infty$, so the Deift--Killip result implies and is essentially equivalent to Conjecture \ref{C20.2} for spherically symmetric $V$.

Several people have worked quite hard on this conjecture without success (although sometimes they found weaker results that they published).  The reader trying to understand the Deift--Killip result should also consult Killip--Simon \cite{KS1, KS2}.

%%%%%%%%%%%%%%%%%%%%%%%%%%%%%%%%%%%%%%%%%%%%%%%%%%%%%%%%%%%

\appendix
\renewcommand{\thesection}{A}
\section{Kato's Proof of His $x^{-1}$ inequality}

Kato \cite[Remark V.5.12]{KatoBk} states, without a full proof, that for each $\varphi\in C_0^\infty(\bbR^3)$, one has that
\begin{equation}\label{A.1}
  \int |x|^{-1} |\varphi(x)|^2 \, d^3x \le \frac{\pi}{2} \int |k| |\hatt{\varphi}(k)|^2 \, d^3k
\end{equation}
He does say that this is equivalent to bounding the integral operator \eqref{A.5} below, but he doesn't explain how to go further.  When Hubert Kalf visited Kato in Berkeley in 1975, he asked Kato for the proof.  Hubert shared what Kato showed him and gave me permission to include it here. Recall that, as we explained after \eqref{10.21}, this is a special case of a result of Herbst getting the optimal constant for all these scale invariant inequalities.

\begin{lemma} \lb{LA.1} Let $C$ be an integral operator on $L^2(X,d\mu)$ with integral kernel
\begin{equation}\label{A.2}
  C(x,y) = A(x,y)B(x,y)
\end{equation}
Suppose that
\begin{equation}\label{A.3}
  \sup_y \int |A(x,y)|^2 d\mu(x) = M_1^2; \qquad \sup_x \int |B(x,y)|^2 d\mu(y) = M_2^2
\end{equation}
Then
\begin{equation}\label{A.4}
  \norm{C} \le M_1 M_2
\end{equation}
\end{lemma}

\begin{proof}  Let $\varphi,\psi \in L^2(M,d\mu)$.  Then
\begin{align*}
   |\jap{&\varphi,C\psi}| = \left|\int A(x,y) B(x,y) \overline{\varphi(x)} \psi(y) \, d\mu(x)\,d\mu(y)\right| \\
             & \le \left(\int |A(x,y)|^2|\psi(y)|^2 \, d\mu(x)\,d\mu(y)\right)^{1/2} \left(\int |B(x,y)|^2|\varphi(y)|^2 \, d\mu(x)\,d\mu(y)\right)^{1/2}
\end{align*}
by the Schwarz inequality.  By \eqref{A.3}, the first integral (integrating first over $x$) is bounded by $M_1^2\norm{\psi}^2$, so $|\jap{\varphi,C\psi|}| \le M_1 M_2 \norm{\varphi}\norm{\psi}$.
\end{proof}

Next, we include the part that was in Kato's book.  \eqref{A.1} is equivalent to $\norm{|p|^{-1/2}|x|^{-1}|p|^{-1/2}} \le \pi/2$ where $|p|$ is the operator $\hatt{|p|\varphi}(p) = |p| \hat{\varphi}(p)$.  In ``$p$--space'', $x^{-1}$ is convolution with the function $(2\pi)^{-3/2} \hatt{x^{-1}}$ (see \cite[(6.2.45)]{RA}) and, by \cite[Theorem 6.8.1]{RA}, $\hatt{x^{-1}}(k) = \sqrt{\tfrac{2}{\pi}} |k|^{-2}$.  Thus \eqref{A.1} is equivalent to a bound on an integral operator
\begin{equation}\label{A.5}
  \norm{C}\le \frac{\pi}{2}; \qquad C(k,p) = \frac{1}{2\pi^2}\frac{1}{k^{1/2}}\frac{1}{|\boldsymbol{k}-\boldsymbol{p}|^2}\frac{1}{p^{1/2}}
\end{equation}

We can write $2\pi^2C$ in the form of \eqref{A.2} where
\begin{equation}\label{A.6}
  A(k,p) = \frac{p^{1/2}}{k |\boldsymbol{k}-\boldsymbol{p}|}; \qquad B(k,p) = \frac{k^{1/2}}{p |\boldsymbol{k}-\boldsymbol{p}|}
\end{equation}
This factorization isn't what one might guess but has the naive expectation multiplied/divided by $(k/p)^{1/2}$ (without doing this the integral in \eqref{A.7} would diverge).  R. L. Frank pointed out that a similar idea was used in Hardy--Littlewood--Polya \cite[Section 9.3]{HLP} for not unrelated (but one dimensional) integral operators.  It might have motivated Kato.  By the Lemma, we need to compute
\begin{align}
  \int \frac{p}{k^2 |\boldsymbol{k}-\boldsymbol{p}|^2}\,d^3k &= \int_{0}^{\infty} 2\pi p \left[\int_{-1}^{1} \frac{d\omega}{k^2+p^2-2kp\omega}\right]\,dk \lb{A.7} \\
       &= 2\pi \int_{0}^{\infty} \frac{1}{k} \log\left[\frac{k+p}{|k-p|}\right]\,dk \lb{A.8}
 \end{align}
where \eqref{A.7} comes from shifting to polar coordinates with $d^3k=(k^2dk)\,d\varphi \, d(\cos\theta)$ and $\omega=\cos\theta$.  The inner integral gives $\frac{1}{2pk} \log\frac{(k+p)^2}{(k-p)^2}$ yielding \eqref{A.8}.

Using
\begin{equation}\label{A.9}
  \int_{0}^{a}\frac{1}{x}\log\left(\frac{a+x}{a-x}\right)\,dx = \int_{a}^{\infty}\frac{1}{x}\log\left(\frac{x+a}{x-a}\right)\,dx = \frac{\pi^2}{4}
\end{equation}
(we defer this calculation) and Lemma \ref{LA.1}, we see that
\begin{equation}\label{A.10}
  \norm{C} \le \frac{1}{2\pi^2}(2\pi)2 \frac{\pi^2}{4} = \frac{\pi}{2}
\end{equation}
proving \eqref{A.1}.

By scaling and changing $x$ to $1/x$, one sees that the integrals in \eqref{A.9} are equal and $a$ independent, so we can take $a=1$ in the first integral.  Using first $u=\frac{1+x}{1-x}$ and then $y = \log u$, one sees that
\begin{align}
  \int_{0}^{1}\frac{1}{x}\log\left(\frac{1+x}{1-x}\right)\,dx  &= 2\int_{1}^{\infty} \frac{1}{u^2-1} \log u \, du\lb{A.11} \\
                             &= 2\sum_{n=1}^{\infty} \int_{1}^{\infty} u^{-2n} \log u \, du \nonumber \\
                             &= 2\sum_{n=1}^{\infty} \int_{0}^{\infty} y e^{-(2n-1)y}\,dy \nonumber \\
                             &= 2\sum_{n=1}^{\infty} \frac{1}{(2n-1)^2} = 2\left[\frac{\pi^2}{8}\right] \lb{A.12}
\end{align}
since if $Q=\sum_{n=1}^{\infty} 1/n^2$, then $Q=\frac{1}{4} Q+ \sum_{n=1}^{\infty} (2n-1)^{-2}$ so that
\begin{equation*}
  Q=\frac{\pi^2}{6} \Rightarrow \sum_{n=1}^{\infty} \frac{1}{(2n-1)^2} = \frac{3Q}{4} = \frac{\pi^2}{8}
\end{equation*}
by the Euler sum.  Alternatively (and this is a remark I got from Martin Klaus) one can evaluate the right side, $I$, of \eqref{A.11} by the method of contour integrals. First, using that the integral is unchanged by $u \to u^{-1}$ coordinate changes, note that $I$ is ${\int_{0}^{\infty} (u^2-1)^{-1} \log u \, du}$.  This integral is unchanged under rotating the contour by $\SI{90}{\degree}$.  Since, for $y>0$ we have that $\log(iy)=\log(y)+i\pi/2$, and, under $u=iy$, we have that $du/(u^2-1) = -i\,dy/(1+y^2)$, we see that
\begin{equation*}
  I = -i\int_{0}^{\infty} \frac{\log(y)}{y^2+1}dy + \frac{\pi}{2}\int_{0}^{\infty}\frac{dy}{y^2+1} = \frac{\pi}{2}\frac{\pi}{2}=\frac{\pi^2}{4}
\end{equation*}
The first integral is $0$ by $y \mapsto y^{-1}$.

It might be surprising that \eqref{A.1} has equality in the norm since the proof has some inequalities.  But the integrals in \eqref{A.3} are independent of the variable one is taking a $\sup$ over so the only inequality is the Schwarz inequality.  It is believable that one can come close to saturating that.

Having completed our exposition of Kato's proof, we note that one standard proof (see, e.g. \cite[pg. 169]{RS2}) of the classical Hardy's inequality in $\bbR^3$ uses
\begin{equation}\label{A.13}
  \norm{\nabla\varphi}_2^2 - \tfrac{1}{4}\norm{r^{-1}\varphi}_2^2 = \norm{r^{-1/2} \nabla(r^{1/2}\varphi)}_2^2
\end{equation}
Frank, Lieb and Seiringer \cite{FLSHardy} have found an analogous formula that proves \eqref{A.1} (they also do this for other fractional powers)
\begin{align}
 \left\langle \psi, \left( \sqrt{-\Delta}\right.\right.& - \left.\left.\frac{2}{\pi |x|} \right)\psi \right\rangle \nonumber \\
    &= \frac{1}{2\pi^2} \iint_{\bbR^3\times\bbR^3} \frac{\left| |x| \psi(x) - |y| \psi(y)\right|^2}{|x-y|^4}\,\frac{dx}{|x|}\,\frac{dy}{|y|} \lb{A.14}
\end{align}
This proves strict positivity for any function $\psi$ and, by taking $\psi(x)$ to be $|x|^{-1}$ cutoff near the origin and near infinity, one sees that the constant in \eqref{A.1} is optimal.

%%%%%%%%%%%%%%%%%%%%%%%%%%%%%%%

\end{document}